\pdfoutput=1
\documentclass[prd,aps,amsfonts,eqsecnum,superscriptaddress,nofootinbib,notitlepage,longbibliography,10pt]{revtex4-2}
\usepackage{booktabs}
\usepackage{xcolor}
\usepackage[a4paper, left=2cm, right=2cm, top=2cm, bottom=2cm]{geometry}
\usepackage{setspace}
\usepackage[english]{babel}
\addto\extrasenglish{%
	\def\appendixautorefname{Appendix}%
}
\usepackage{appendix}
\usepackage[utf8]{inputenc}
\usepackage[T1]{fontenc}
\usepackage{lmodern}
\usepackage{amsmath,amsthm,mathdots,amssymb,bm,bbm,mathrsfs,mathtools}
\usepackage[final]{graphicx}
\usepackage{algorithm,algpseudocode}
\usepackage{subcaption}
\usepackage{tikz,pgfplots}\pgfplotsset{compat=1.16}

\usepackage{comment}

\usepackage[hidelinks, hyperindex=true, linktocpage=true, draft=false]{hyperref}
\usepackage{mleftright}\mleftright

\usepackage{enumitem}
\mathtoolsset{showonlyrefs=true}
\usepackage{needspace}
\usepackage{microtype,color,graphicx}

\usepackage{tikz-cd}
\usepackage{xcolor}
\usepackage{multirow}
\usepackage{tablefootnote}

\usepackage{macros}
\BeforeBeginEnvironment{thm}{\Needspace{6\baselineskip}}
\BeforeBeginEnvironment{lem}{\Needspace{6\baselineskip}}
\BeforeBeginEnvironment{cor}{\Needspace{6\baselineskip}}
\BeforeBeginEnvironment{conj}{\Needspace{6\baselineskip}}
\BeforeBeginEnvironment{prop}{\Needspace{6\baselineskip}}
\BeforeBeginEnvironment{defn}{\Needspace{6\baselineskip}}
\numberwithin{algorithm}{section}
\crefname{algorithm}{Algorithm}{Algorithms}
\makeatletter
\def\ftype@algorithm{8}
\makeatother

\usepackage{ifdraft}

\usetikzlibrary{quantikz}
\makeatletter
\def\l@subsubsection#1#2{}
\makeatother
\begin{document}
\input{inputs/equation-reference-tags}
\def\appendixautorefname{Appendix}
\renewcommand{\appendixautorefname}{Appendix}
\renewcommand{\chapterautorefname}{Chapter}
\renewcommand{\sectionautorefname}{Section}
\renewcommand{\subsectionautorefname}{Section}
\renewcommand{\subsubsectionautorefname}{Section}

\title{Efficient learning of quantum interactions from thermal metastable states}

\author{Bingrun Wang}
\email{wangbr25@mails.tsinghua.edu.cn}
\affiliation{Center for Quantum Information, IIIS, Tsinghua University, Beijing 100084, China}

\author{Qi Ye}
\email[Corresponding author: ]{yeq22@mails.tsinghua.edu.cn}
\affiliation{Center for Quantum Information, IIIS, Tsinghua University, Beijing 100084, China}
\affiliation{Shanghai Qi Zhi Institute, Shanghai 200232, China}
\affiliation{Hefei National Laboratory, Hefei 230088, China}

\author{Chi-Fang Chen}
\email[Corresponding author: ]{achifchen@gmail.com}
\affiliation{University of California, Berkeley, California 94720, USA}

\begin{abstract}

Learning quantum interactions from finite-temperature many-body systems is a central task in emerging quantum platforms. Recently, the problem of learning from lattice quantum Gibbs states has found rigorous, efficient protocols. Nevertheless, exact Gibbs states, as the input premise, are in fact computationally intractable to prepare and may not faithfully represent generic finite-temperature quantum systems. In contrast, a system coupled to a heat bath can be stuck at an approximate stationary state (\textit{metastable state}) long before it truly equilibrates.
Here, we formulate a physically and algorithmically consistent alternative: learning from such metastable states of detailed-balanced master equations (Lindbladians) arising from system-bath interactions. We distill the algorithmic mechanism and structural condition underlying Gibbs-state learning and extend it in full to metastable states, attaining nearly optimal sample and computational complexity (in the system size and the precision). More broadly, we sharpen notions of metastability and develop a unified framework for finite-temperature learning.

\end{abstract}

\maketitle

\tableofcontents
\section{Introduction}
\label{sec:introduction}

Learning from the quantum physical world has emerged as a frontier in many-body physics, quantum information theory, quantum sensing, and learning theory~\cite{gebhartLearningQuantumSystems2023,anshuSurveyComplexityLearning2023,huangQuantumAdvantageLearning2022}. However, the full description of a generic $n$-qubit quantum state is simply intractable to learn~\cite{haahSampleoptimalTomographyQuantum2016,odonnellEfficientQuantumTomography2016}. Therefore, any hope of efficiently obtaining a succinct description of many-body systems must sharply recognize the essential structure of the family of interest~\cite{cramerEfficientQuantumState2010,flammiaQuantumTomographyCompressed2012,huangPredictingManyProperties2020}. A natural and highly effective ansatz class comes from quantum thermodynamics: equilibrium phenomena are often specified by the interaction (many-body Hamiltonian) and the temperature~\cite{jaynesInformationTheoryStatistical1957,alhambraQuantumManyBodySystems2023}.

Today, the task of learning the Hamiltonian from quantum Gibbs states has surfaced as a clean proxy problem for abstracting finite-temperature learning in a range of applications, perhaps reminiscent of the story of classical Boltzmann machines~\cite{ackleyLearningAlgorithmBoltzmann1985,breslerEfficientlyLearningIsing2014,vuffrayInteractionScreeningEfficient2016,klivansLearningGraphicalModels2017,aminQuantumBoltzmannMachine2018}. In the analog sensing setting, we often have access to experimental apparatus, such as a fridge, that cools a piece of complex material which we wish to understand~\cite{samarakoonMachinelearningassistedInsightSpin2020}; in controllable quantum simulators, we might already have some desired Hamiltonian in mind, but wish to benchmark whether the prepared state precisely matches the target~\cite{wiebeHamiltonianLearningCertification2014,eisertQuantumCertificationBenchmarking2020,carrascoTheoreticalExperimentalPerspectives2021,kokailEntanglementHamiltonianTomography2021}. In any case, provably efficient (sample-wise and computation-wise) learning protocols have been ultimately found~\cite{anshuSampleefficientLearningQuantum2021,haah2022optimal, bakshiLearningQuantumHamiltonians2023, chen2025learning}, in principle, for learning the coefficients of any lattice Hamiltonian, given multiple copies of finite-temperature Gibbs states.

However, a major conceptual inconsistency of the above premise is that Gibbs states might well be intractable to prepare in the first place, therefore invalidating the input assumptions, let alone the algorithm guarantees. Indeed, even classical Gibbs states are NP-hard to sample from~\cite{barahonaComputationalComplexityIsing1982,slyComputationalTransitionUniqueness2010,slyCountingTwospinModels2014}; therefore, we do not expect a quantum simulator to always efficiently prepare the Gibbs states~\cite{kitaevClassicalQuantumComputation2002,kempeComplexityLocalHamiltonian2006}. In experimental cooling apparatus, if we approximate the interaction with the bath as a master equation, it, as a quantum computational process, is likewise constrained by complexity-theoretic hardness~\cite{chenLocalMinimaQuantum2023,rouzeEfficientThermalizationUniversal2026} and can mix exponentially slowly~\cite{gamarnikSlowMixingQuantum2024,rakovszkyBottlenecksQuantumChannels2024}. A sharper class of states is then needed to capture finite-temperature quantum systems~\cite{berthierTheoreticalPerspectiveGlass2011,brownQuantumMemoriesFinite2016,bergamaschi2025metastability}. Unfortunately, the arguments behind the Gibbs learning protocol often rely on exact analytic properties of the Gibbs formulation. Even though we morally expect --- just as the role of classical Boltzmann machines --- that Gibbs states are intended to be a toy model, it is still an open-ended question what might be a suitable extension beyond the Gibbs family, and whether such a candidate admits an effective mathematical description and conceptually transparent learning algorithms.

In this paper, we seek to articulate a conceptually consistent setting of efficient Hamiltonian learning in quantum systems at finite temperatures, and provide a protocol with provable guarantees. Departing from ideal Gibbs states, we start with a microscopic dynamical model of how the system interacts with the bath. In particular, we model the system-bath interaction as a (quasi)-local, detailed-balance master equation~\cite{daviesMarkovianMasterEquations1974,chen2025thermal,chen2023exact,scandiThermalizationOpenManyBody2026}, and study its approximate stationary states (or metastable states)~\cite{bergamaschi2025metastability}, a self-consistent model of finite-temperature quantum states.

What we found is a conceptually unified picture: from approximate stationary states we can also learn the unique Hamiltonian, through the same mechanism as if they were actually Gibbs states. Consequently, we also give a range of generalizations and robustness of the scope of applicability of the results: the experimental samples need not be \textit{identical} states, but any sequence of metastable states suffices. In particular, since metastable states can always be efficiently prepared from Lindbladian evolution~\cite[Lemma~II.1]{bergamaschi2025metastability}, our learning protocol also works when we are only given access to detailed-balanced Lindblad evolution as a black box, which can be natural in digital quantum simulation contexts~\cite{temmeQuantumMetropolisSampling2011,cleveEfficientQuantumAlgorithms2017,chen2025thermal,chen2023thermal,chen2023exact,dingEfficientQuantumGibbs2024,gilyen2024glauber}. Of course, the scope of this work implicitly assumes a model of thermalization where the system interacts weakly with a Markovian bath~\cite{daviesMarkovianMasterEquations1974,breuerTheoryOpenQuantum2007,mozgunovCompletelyPositiveMaster2020}, even though we believe the learning protocol can still qualitatively work. At least, our metastable learning setup is actually robust to perturbative errors in the microscopic master equation.

The main technical contribution of our work is partly a conceptual and technical extension of the initial framework of~\cite{bergamaschi2025metastability} and the synergies with learnability of the Hamiltonian. The `approximate detailed balance' condition, among many other notions introduced in~\cite{bergamaschi2025metastability}, appears to be the essential workhorse that ultimately recovers the learning results analogous to the ideal Gibbs case.

\subsection{Setup}
\label{sec:setup}
We consider $n$ qubits on a finite lattice $\Lambda\subset\mathbb{Z}^D$ with a geometrically local, $k$-local Hamiltonian
\begin{align}
    \vH=\sum_{\gamma\in\Gamma}h_\gamma\vP_\gamma,\quad \text{with unknown coefficients}
    \qquad h_\gamma\in[-1,1].
    \label{eq:intro-lattice-Hamiltonian}
\end{align}
The distinct, nonidentity Pauli strings $\vP_\gamma$ are known and have weight at most $k$ and respect the $D$-dimensional lattice geometry. We treat the geometric parameters as fixed constants throughout the paper and collect them in full detail in~\Cref{defn:local-Hamiltonian}.

Our main physical model of finite-temperature quantum states regards those arising from interaction with a finite-temperature bath. Under the time-Markovian, weak-coupling assumption, the fundamental dynamical description is modeled by a Lindbladian or master equation
\begin{align}
        \frac{\dd\vsigma}{\dd t}
    &=\CL[\vsigma].
\end{align}

While there are many extra parameters and microscopic nuances that may play a role in realistic system-bath interaction (e.g., bath memory effects or strong coupling), we will focus on a minimal quantitative model~\cite{chen2025thermal} for interacting with a heat bath.  Given inverse-temperature $\beta$ and Hamiltonian $\vH$, there exists an explicit master equation (\Cref{subsec:thermal-lindbladian-definition})
\begin{align}
        \CL[\cdot]&=-i[\vH,\cdot]+\sum_{a\in\mathcal A}\CL_a[\cdot],
    \label{eq:intro-thermal-dynamics}
\end{align}
where the generators\footnote{Strictly speaking, we can include a coupling strength parameter for $\sum_{a\in\mathcal A}\CL_a[\cdot]$ which will propagate through the bounds.} $\CL_a$ are quasi-local in the lattice geometry and satisfy Kubo--Martin--Schwinger (KMS) detailed balance at inverse temperature $\beta$. The set ${a\in\mathcal A}$ prescribes the particular system-bath interaction terms, which explicitly parameterize each $\CL_a$. For simplicity, we take them to be the set of single-site Pauli jumps $\{\vA_a\}_{a\in\mathcal A}=\{\vX_i,\vY_i,\vZ_i\}_{i\in\Lambda}$, i.e., the entire lattice interacts with the bath, not just the boundary.
We also include the system Hamiltonian term $-i[\vH,\cdot]$, which is merely the standard Schrödinger equation that cannot be `turned off'. We do note that it is sometimes mathematically cleaner to drop the Hamiltonian term, which is completely viable in digital quantum thermal simulation settings.

A defining feature of exact detailed balance is that the Gibbs state (for Hamiltonian $\vH$ at inverse temperature $\beta$) is exactly a stationary state such that
\begin{align}
    \vrho:=\frac{e^{-\beta\vH}}{\Tr(e^{-\beta\vH})},
    \qquad \CL_a[\vrho]=0\quad\text{for all }a\in\mathcal A, \quad \text{thus}
    \quad \CL[\vrho]=0.
    \label{eq:intro-thermal-fixed-point}
\end{align}
Although the Gibbs state is a stationary state, evolving the Lindbladian may not always efficiently converge to the Gibbs state (i.e., the \textit{mixing time} can be prohibitively large for complexity-theoretic and physical reasons).
Therefore, our main model of finite-temperature state --- under the Lindbladian~\eqref{eq:intro-thermal-dynamics} --- is then any \textit{approximate stationary state} of the dynamics.
Following Ref.~\cite{bergamaschi2025metastability}, we say a state $\vsigma$ is $\epsilon$-metastable if
\begin{align}
    \|\CL[\vsigma]\|_1\leq\epsilon.
    \label{eq:intro-global-metastability}
\end{align}
When $\epsilon >0$, such $\vsigma$ can be very far from the Gibbs state $\vrho$, but still appear effectively in equilibrium with the heat bath. Indeed, metastable states can always be efficiently prepared; for example, any initial state becomes $\frac{2}{T}$-metastable after evolving the Lindbladian over uniformly random times $[0,T].$ Of course, a metastable state may arise in many other ways, and we would like an analytic framework that is agnostic of its past and depends only on the parameter $\epsilon$.

Our main topic, then, is whether such relaxed classes of finite-temperature states indeed make a more physical case of efficient and local learning at finite temperatures, like in the Gibbs case.

\subsection{Main results}
\label{sec:main-results}

\begin{figure}[!t]
\centering
\includegraphics[width=\textwidth]{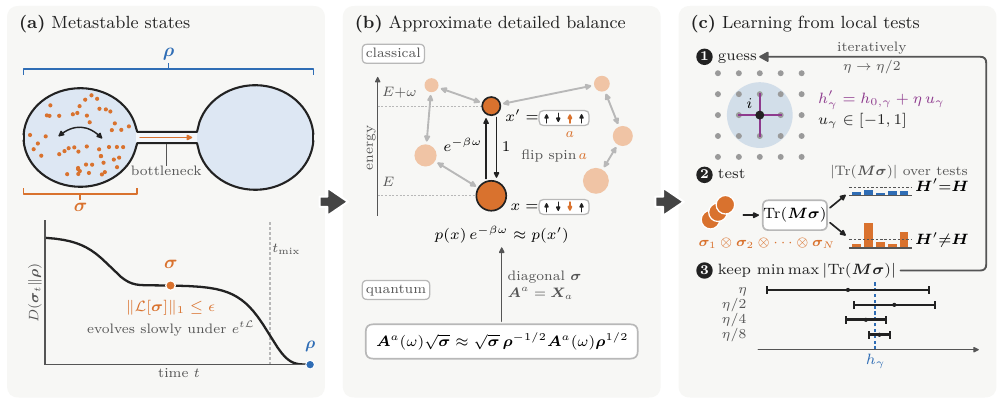}
\caption{Overview of the results.
(a)~In this example, a metastable state $\vsigma$ is concentrated in one region of configuration space, whereas the Gibbs state $\vrho$ spreads over all of it.
Probability leaks out of that region only slowly, so $\vsigma$ is approximately stationary under the Lindbladian evolution, $\|\CL[\vsigma]\|_1\leq\epsilon$, yet still far from equilibrium (the plateau of $D(\vsigma_t\|\vrho)$).
(b)~Metastability implies approximate detailed balance (ADB).
Classically, if configurations $x$ and $x'$ differ by one spin flip and their energies differ by $\omega$, the probability fluxes between them balance, $p(x)e^{-\beta\omega}\approx p(x')$.
ADB is the quantum version for single-site jumps $\vA^a$, and it reduces to this classical relation when $\vsigma$ is diagonal with entries $p(x)$ and $\vA^a=\vX_a$.
(c)~Algorithm sketch. Guesses $h_{0,\gamma}+\eta u_\gamma$ for the terms near each site are checked by local tests $\Tr(\vM\vsigma)$ measured on independent preparations.
By ADB, every $|\Tr(\vM\vsigma)|$ is small for the true Hamiltonian, whereas a wrong guess makes some of them large.
The algorithm keeps the guess whose largest test value is smallest, which reduces the coefficient error from $\eta$ to $\eta/2$, and repeats this step with the updated estimates.}
\label{fig:schematic}
\end{figure}

Our representative learning result considers access to a streamed sequence of unknown metastable states under the same master equation (thus the same Hamiltonian).

\begin{thm}[Hamiltonian learning from metastable states]
\label{thm:intro-global-metastable-learning}
Consider a $D$-dimensional, $n$-qubit lattice Hamiltonian $\vH$ (\Cref{defn:local-Hamiltonian}) and a detailed-balanced Lindbladian $\CL$ (as in~\eqref{eq:intro-thermal-dynamics}), and assume access to a stream of
independent, possibly different $\epsilon$-metastable states $\{\vsigma_i\}$.
Then, there exists an algorithm that estimates every coefficient $h_\gamma$
to additive error $\eta$ using
\begin{align}
    N_{\mathrm{samp}}
    &=O\!\left(\frac{e^{\poly(\beta^{\pm 1})}}{\eta^2}\log\!\frac{n}{\delta}\,\polylog\left(\frac{1}{\eta}\right)\right) \quad \text{sample complexity and} \\
    T_{\mathrm{tot}}
    &=O\!\left(n \cdot \frac{e^{\poly(\beta^{\pm 1})}}{\eta^2}\log\!\frac{n}{\delta}\,\polylog\left(\frac{1}{\eta}\right)\right) \quad \text{quantum and classical
time complexity},
    \label{eq:intro-global-learning-complexities}
\end{align}
with probability at least $1-\delta$, provided that the precision is above the floor
$\eta\geq C_\beta\epsilon^{\alpha_1}n^{\alpha_2}$. In the above, big $O(\cdot)$ and constants
($C_\beta=e^{\poly(\beta^{\pm1})}$,
$\alpha_1,\alpha_2=1/\poly(\beta^{\pm1})>0$) hide only the geometric constants.
\end{thm}

The sample complexity reflects a union bound, so that each of the $O(n)$ coefficients $\{h_{\gamma}\}$ is learned within the failure-probability budget. The computational complexity comes from covering each sample with non-overlapping local measurements, with greedy classical post-processing. That is, we attain the nearly optimal $n$, $\eta$, and $\delta$ dependence of Gibbs-state learning~\cite[Theorem I.2]{chen2025learning}\footnote{As stated, the temperature prefactor $e^{\poly(\beta^{\pm 1})}$ is worse than $e^{\poly(\beta)}/\beta^2$ but likely improvable.}; recall the lower bound of $\Omega\!\left(e^\beta\log(n/\delta)/(\beta^2\eta^2)\right)$ for Gibbs learning~\cite[Theorem 1.2]{haah2022optimal}, which also applies here.

Meanwhile, the above covers a much larger class of states beyond Gibbs states, as long as each $\vsigma_i$ is sufficiently metastable, regardless of how each of them is obtained. More precisely, since the required metastability $\epsilon$ depends polynomially on the desired precision $\eta$ and system size $n$, these metastable states can efficiently arise under a Lindbladian evolution time $T \sim 1/\epsilon$ from arbitrary initial states. Therefore, our results also give a natural way to learn the Hamiltonian given black-box access to the master equation.
\begin{cor}[Hamiltonian learning from detailed-balanced Lindbladians]
\label{cor:intro-lindbladian-learning}

\Cref{thm:intro-global-metastable-learning} continues to hold, when access to unknown metastable states $\{\vsigma_i\}$ is replaced with black-box access to the Lindbladian evolution itself
$e^{\CL t}$ (associated with the unknown target Hamiltonian). For error $\eta$, it suffices to use Lindbladian evolution time up to $T=C_\beta n^{\alpha_2/\alpha_1}\eta^{-1/\alpha_1}$ per sample.

\end{cor}
In digital simulation applications, this gives a way to quantitatively benchmark a claimed implementation of the Lindbladian, which might come from heuristic compilations, acting on choices of initial states.

Of course, the algorithm has a fundamental precision floor for the Hamiltonian that depends on the quality of the metastable state, where having more samples simply cannot further improve the Hamiltonian learning precision (see Proposition~\ref{prop:adb-information-floor} and the remark following it for a simple example where the precision floor scales as $\epsilon/\beta$, independently of $n$).  Whether the system-size-dependent factor between the learning precision $\eta$ and metastability $\epsilon$ in \Cref{thm:intro-global-metastable-learning} is necessary remains open.

At least, we can remove the system-size dependence on the precision floor under a slightly stronger metastability assumption.

\begin{thm}[Hamiltonian learning from locally metastable states]
\label{thm:intro-local-metastable-learning}
In the setting of~\Cref{thm:intro-global-metastable-learning}, assume the streamed metastable states are instead metastable with respect to every Lindbladian term: $$\max_{a\in\mathcal A}\|\CL_a[\vsigma_i]\|_1\leq\epsilon.$$ Then, the sample and computational complexities continue to hold, provided that $\eta\geq C_\beta\epsilon^{\alpha_1}$.
\end{thm}

That is, we can learn a local term in the Hamiltonian, as long as the input state is metastable for all\footnote{We believe that a local Hamiltonian term can be learned only assuming metastability for local terms of the Lindbladian nearby.}  local terms of the Lindbladian. Now, the system-size independence of this precision threshold is optimal.
For every $n$ and sufficiently small $\epsilon$, the example in
\Cref{sec:adb-learning-lower-bound} gives two Hamiltonians whose coefficients
differ by $\Omega_\beta(\epsilon)$ and share the same
$\epsilon$-locally metastable input state.
No number of copies can distinguish these instances, so the precision
floor cannot vanish as $n$ grows.

To what extent does our model of metastable states capture real, analog quantum states at finite temperatures? To be very careful, we are already restricting the Hamiltonian family to be geometrically local, respecting some $D$-dimensional lattice (instead of the full infinite-dimensional Schrödinger equation). Furthermore, the master equation is inherently an approximation of the system-bath dynamics, which need not capture long-time trajectories of the true system-bath joint evolution. Fortunately, our metastable state learning setting is robust: it suffices if the \textit{generator} of the true dynamics $\CL'$ is well approximated by the detailed-balanced family $\CL$.

\begin{cor}[Robustness of Hamiltonian learning under imperfect dynamics]
\label{cor:lindblad-learning}
In the setting of
\Cref{thm:intro-global-metastable-learning},
suppose the inputs
$\vsigma_i$ each satisfy $\|\CL'[\vsigma_i]\|_1\leq\epsilon$ for some physical generator $\CL'$ satisfying
$\|\CL'-\CL\|_\diamond\leq\mu$.
Then, the same
sample and time complexities hold to the precision floor
$\eta\geq C_\beta(\epsilon+\mu)^{\alpha_1}n^{\alpha_2}$.
\end{cor}

Indeed, observe that if $\|\CL'-\CL\|_\diamond\leq\mu$, then every state satisfies $\|\CL[\vsigma]\|_1\leq\|\CL'[\vsigma]\|_1+\mu$. That is, if we evolve the true dynamics $\CL'$ for a very long time and obtain a very metastable state $\vsigma$ for $\CL'$, then it is still a metastable state also for the ideal $\CL$. An analogous trick $\epsilon \rightarrow \epsilon+ \mu$ holds when assuming $\norm{\CL_a'[\vsigma]}_1\le \epsilon$ and $\norm{\CL_a'-\CL_a}_{\diamond}\le \mu.$
Broadly, we hope our results relax the Gibbs input assumption to more physical contexts, which gives quantitative evidence that the learning protocol should work in a range of natural models of finite-temperature quantum states.

\subsection{Related works}
\label{sec:related-works}
Our Hamiltonian learning setting sits among other learning theory threads.

\paragraph*{Hamiltonian learning from Gibbs states.}

The most algorithmically and conceptually relevant framework is learning a local Hamiltonian from copies of its Gibbs state; implicitly, a large body of work has taken the Gibbs states as the fundamental model of finite-temperature quantum states.
Anshu et al.~\cite{anshuSampleefficientLearningQuantum2021} gave the sample-efficient result valid at any finite temperature, although their classical post-processing requires evaluating the partition function, which can be computationally hard. At high temperatures, Haah, Kothari, and Tang~\cite{haah2022optimal} gave an algorithm with optimal sample and time complexity. Beyond this regime, Bakshi et al.~\cite{bakshiLearningQuantumHamiltonians2023} obtained a polynomial-time algorithm at any constant temperature, and Narayanan~\cite{narayananImprovedAlgorithmsLearning2024} improved its dependence on the inverse temperature. More recently, Chen, Anshu, and Nguyen~\cite{chen2025learning} gave a local algorithm that learns each term from measurements near it, with near-optimal sample complexity on lattices at any temperature; our work adapts from their protocol. Its key ingredient, the \textit{identifiability equation}, relies on exact properties of the Gibbs state, including the KMS condition and commutation with the Hamiltonian; in fact, all of the works above take copies of the Gibbs state as input.

\paragraph*{Detailed-balanced Lindbladians and their approximate stationary states.}
The main analytical framework and the way we model finite-temperature systems trace back to the microscopic models of system-bath interaction. What has recently accelerated quantitative study of this equilibration process is the construction of detailed-balanced, quasi-local Lindbladians~\cite{chen2023exact,gilyen2024glauber,dingEfficientQuantumGibbs2024}. While serving as a quantum algorithmic counterpart of MCMC methods initially, it also turns out to quantitatively match physical derivations from suitable weak-coupling in some regime~\cite{scandiThermalizationOpenManyBody2026}. Motivated by classical Markov chains, the question of mixing times and bottlenecks in quantum many-body systems naturally becomes a next target. Gamarnik, Kiani, and Zlokapa~\cite{gamarnikSlowMixingQuantum2024} and Rakovszky et al.~\cite{rakovszkyBottlenecksQuantumChannels2024} proposed quantum bottleneck theorems and proved exponential lower bounds on the mixing time for certain Hamiltonians at low temperatures; at certain low-temperature limits, quantum local minima of the energy ~\cite{chenLocalMinimaQuantum2023} have been defined and proven to be highly expressive. Most relevant to us is the study of approximately stationary states of the detailed-balance Lindbladian by Bergamaschi, Chen, and Vazirani~\cite{bergamaschi2025metastability}. In particular, we will draw a direct connection from mathematical notions of metastability (the `approximate detailed-balance condition') to Hamiltonian learning, and conversely enlarge their analytic framework.

\paragraph*{Classical graphical models.}
The classical counterpart of our problem, learning Ising models and other graphical models, has a long history. For i.i.d.\ samples from the Gibbs distribution, Bresler~\cite{breslerEfficientlyLearningIsing2014} gave an efficient algorithm that recovers the interaction graph of bounded-degree Ising models. Interaction screening~\cite{vuffrayInteractionScreeningEfficient2016,lokhovOptimalStructureParameter2018} recovers both the graph and the couplings with near-optimal sample complexity, and Klivans and Meka~\cite{klivansLearningGraphicalModels2017} gave a multiplicative-weights algorithm that applies to general graphical models. Other works learn instead from a trajectory of Glauber dynamics~\cite{breslerLearningGraphicalModels2018}, and Gaitonde, Moitra, and Mossel~\cite{gaitondeBypassingNoisyParity2024} showed that such trajectories can even bypass computational barriers of the i.i.d.\ setting. Liu et al.~\cite{liuLocallyStationaryDistributions2024} studied locally stationary measures of slow-mixing chains, a classical counterpart of the approximately stationary states considered here. Closest to our work, Jayakumar et al.~\cite{jayakumarDiscreteDistributionsAre2026} showed that classical models can be learned from metastable samples. Their notion of strong metastability is an approximate detailed balance condition for a reversible Markov chain, and it ensures that single-variable conditionals remain close to those of the Gibbs distribution. Our ADB condition plays the same role in the quantum setting. Since noncommuting local terms admit no convenient analogue of these conditionals, we work instead with approximate versions of the KMS and commutation identities.

While not directly relevant, we also mention some broader connections.

\paragraph*{Hamiltonian learning from other input assumptions.}
Beyond Gibbs states, a large body of work learns the Hamiltonian from its real-time evolution $e^{-i\vH t}$~\cite{haah2022optimal,stilckfrancaEfficientRobustEstimation2024,yuRobustEfficientHamiltonian2023a,huangLearningManyBodyHamiltonians2023,dutkiewiczAdvantageQuantumControl2024,liHeisenberglimitedHamiltonianLearning2023,mobusHeisenberglimitedHamiltonianLearning2025,niQuantumHamiltonianLearning2024,miraniLearningInteractingFermionic2024,bakshiStructureLearningHamiltonians2024,zhaoLearningStructureAny2025,maLearningKbodyHamiltonians2026,huAnsatzfreeHamiltonianLearning2025,zhouOptimalAnsatzfreeHamiltonian2026,chenLowerBoundsLearning2026,pradenneLearningHamiltoniansLong2026}, where the relevant resource is the total evolution time. Huang et al.~\cite{huangLearningManyBodyHamiltonians2023} showed that interleaving the evolution with single-qubit control reaches the Heisenberg limit $O(1/\epsilon)$, and Dutkiewicz, O'Brien, and Schuster~\cite{dutkiewiczAdvantageQuantumControl2024} proved that without control the total evolution time is $\Omega(\epsilon^{-2})$ for large classes of Hamiltonians. Later work also removed the need to know the interaction structure in advance~\cite{bakshiStructureLearningHamiltonians2024,huAnsatzfreeHamiltonianLearning2025}, and related methods learn local Lindbladians~\cite{lewisLearningStructureOpen2026,mobusRobustStructureLearning2026,aradNearOptimalLearningLocal2026}. While the goal seems similar to ours, these techniques do not directly transfer to learning from static states.

Static states other than Gibbs states have been used as well, although the choice of state family is less unified. For eigenstates and other states that commute with $\vH$, local correlations give linear constraints that determine $\vH$ up to scale~\cite{qiDeterminingLocalHamiltonian2019a,baireyLearningLocalHamiltonian2019}, and similar ideas apply to steady states of open-system dynamics~\cite{baireyLearningDynamicsOpen2020}. Our learning protocol contains approximate versions of these commutator constraints, while requiring neither control over the dynamics nor exact stationarity.

\section{Proof idea}
\label{sec:proof-idea}

Why should a metastable state, which may be far from the Gibbs state, still determine every local term of the Hamiltonian?
For classical spin systems, the answer is a condition called `strong metastability' in Ref.~\cite{jayakumarDiscreteDistributionsAre2026}.
A distribution $p$ satisfies strong metastability if the probability flows between configurations $x$ and $x'$ that differ by a single spin flip remain nearly balanced in total variation, so that $p(x)/p(x')\approx\pi(x)/\pi(x')=e^{\beta\Delta E}$ for most such pairs, where $\pi$ is the Gibbs distribution and $\Delta E=E(x')-E(x)$ depends only on the Hamiltonian terms near the flipped spin (\Cref{fig:schematic}(b)).
Estimating these ratios from samples therefore reveals the energy change $\Delta E$ of each single spin flip, which in turn determines the local terms of the Hamiltonian, even when $p$ is globally far from $\pi$.

In the quantum study of metastability, Ref.~\cite{bergamaschi2025metastability} started from the stationary condition $\norm{\CL[\vsigma]}_1 \approx 0$ and proved a local Markov property and an area law of mutual information. The approximate detailed balance condition, as one of the intermediate conditions, was introduced as part of the theory building. Of course, the approximate detailed balance condition (ADB) itself does not obviously address learning. Unfortunately, the classical strategy of learning from conditionals requires at least a change in language due to the noncommuting nature of quantum Hamiltonians and quantum states.

We therefore carefully open up the Gibbs-state algorithm of Ref.~\cite{chen2025learning}.
To attain near-optimal performance, the algorithm relied extensively on exact properties of Gibbs states, such as the exact Kubo--Martin--Schwinger (KMS) condition, commutation with the Hamiltonian, and faithfulness of the KMS norm.
Surprisingly, all of these properties are approximately implied by liberal uses of ADB.
Overcoming these obstacles --- and recovering the near-optimal scaling in system size $n$ and precision $\epsilon$ --- is the main technical challenge of our work.
\Cref{thm:intro-global-metastable-learning,thm:intro-local-metastable-learning} therefore follow from the chain of reductions illustrated in \Cref{fig:schematic},
\begin{align*}
    \text{metastability}
    \quad\Longrightarrow\quad\text{ADB}
    \quad\Longrightarrow\quad\text{measurable tests}
    \quad\Longrightarrow\quad\text{learning}.
\end{align*}
Below, we describe these reductions in reverse order.
The reduction from measurable tests to learning follows Ref.~\cite{chen2025learning}, and the reduction from metastability to ADB builds on Ref.~\cite{bergamaschi2025metastability} with several improvements.
The reduction from ADB to measurable tests is the main technical contribution of this work.
Along the way, we establish further relations among different notions of approximate stationarity, which together give a complete framework, summarized in \Cref{fig:metastability-relations}.

\paragraph*{From measurable tests to learning.}
Following the local min--max strategy of Ref.~\cite[Algorithm~IV.3]{chen2025learning}, the learner guesses $\vH'$ and tests its closeness to $\vH$ through measurements on $\vsigma$.
For a single-site Pauli $\vA$, Pauli orthogonality gives $\|[\vA,\vH-\vH']\|_{F}^2=4\sum_{\gamma:\{\vA,\vP_\gamma\}=0}|h_\gamma-h'_\gamma|^2$, where $\|\vX\|_{F}^2:=2^{-n}\Tr(\vX^\dagger\vX)$.
Since every term touching a site anticommutes with a Pauli there, these local commutators control all coefficient errors.
Since all available data come from measurements on the input state $\vsigma$, we need to relate these commutators to that state. We therefore consider the $\vsigma$-weighted Kubo--Martin--Schwinger (KMS) inner product and norm,
\begin{align*}
    \langle\vX,\vY\rangle_{\vsigma}:=\Tr\bigl(\vX^\dagger\sqrt{\vsigma}\,\vY\sqrt{\vsigma}\bigr),
    \qquad
    \|\vX\|_{\vsigma}:=\sqrt{\langle\vX,\vX\rangle_{\vsigma}}.
\end{align*}
For the maximally mixed state $\vtau:=\vI/2^n$, this norm reduces to the normalized Frobenius norm, $\|\vX\|_{\vtau}=\|\vX\|_F$.
For a general metastable state, we need it to still detect local operators, $\|\vB\|_{F}\leq C_{\beta}\|\vB\|_{\vsigma}$ for $\vB$ of bounded support.
Ref.~\cite{chen2025learning} proved local faithfulness for Gibbs states. Our extension to states with small ADB error uses a substantially different proof and recovers the Gibbs-state result as a corollary (\Cref{thm:meta-local-faithfulness}).

To compare a guess $\vH'$ with the unknown $\vH$, the algorithm uses a family of tests $Q_{\vsigma}(\vO,\vA;\vG,\vH')$.
They are indexed by local test operators $\vO$, single-site Paulis $\vA$, and local auxiliary Hamiltonians $\vG$ chosen by the learner.
The tests have three properties,
\begin{align*}
    &\text{(P1) Completeness:}
    &&Q_{\vsigma}(\vO,\vA;\vG,\vH)\approx0
    \quad\text{for every auxiliary }\vG,\\
    &\text{(P2) Soundness:}
    &&Q_{\vsigma}(\vO,\vA;\vH,\vH')\approx\frac{\beta}{2}\langle\vO,[\vA,\vH-\vH']\rangle_{\vsigma},\\
    &\text{(P3) Measurability:}
    &&Q_{\vsigma}(\vO,\vA;\vG,\vH')=\Tr(\vM\vsigma)
    \quad\text{for a known bounded observable }\vM.
\end{align*}
Together, (P1) and (P2) let the tests distinguish the true Hamiltonian from wrong candidates, as illustrated in \Cref{fig:schematic}(c).
If $\vH'=\vH$, every test vanishes for every $\vO$, $\vA$, and $\vG$.
If $\vH'$ differs from $\vH$ near the site of $\vA$, then the choice $\vG=\vH$ and $\vO=[\vA,\vH-\vH']$ gives the test value $\frac{\beta}{2}\|[\vA,\vH-\vH']\|_{\vsigma}^2$, which is far from zero.
The operator $[\vA,\vH-\vH']$ is unknown, but it is a linear combination of the known operators $[\vA,\vP_\gamma]$, so a large test value already appears for some $\vO=[\vA,\vP_\gamma]$.

The learner does not know $\vH$ and therefore cannot choose $\vG=\vH$ directly.
Instead, it selects the candidate whose worst test is smallest,
\begin{align*}
    \vH'_{*}=\arg\min_{\vH'}\;\max_{\vG,\vO,\vA}\bigl|Q_{\vsigma}(\vO,\vA;\vG,\vH')\bigr|.
\end{align*}
By (P1), the true Hamiltonian makes every test small, so the minimizer also has small test values for every auxiliary $\vG$.
The auxiliary Hamiltonians range over a local net that contains a point close to $\vH$, so this includes $\vG\approx\vH$, and the previous paragraph shows that $\vH'_*$ is close to $\vH$.

The algorithm measures each test near $\vO$ and $\vA$ and iteratively refines local guesses, halving the coefficient error in each round until it reaches the accuracy set by the tests.
This yields the sample and time complexities of \Cref{thm:intro-global-metastable-learning} (\Cref{sec:learning-protocol}).
By adapting the tests and sampling procedure, the same guarantees extend to independent, nonidentical metastable states (\Cref{sec:streaming-learning}).

\paragraph*{From approximate detailed balance to measurable tests.}
The learning problem therefore reduces to constructing tests with properties (P1)--(P3).
Our starting point is the identifiability equation of Ref.~\cite{chen2025learning}, which extends to an arbitrary state (\Cref{lem:state-identifiability}),
\begin{align}
\frac{\beta}{2}\langle\vO,[\vA,\vH-\vH']\rangle_{\vsigma}
={}&\frac{1}{\sqrt{2\pi}}\int_{\mathbb R}g_{\mathrm{id}}(t)
\Tr\!\left[\vO^\dagger\sqrt{\vsigma}
\left(
\vrho^{-1/2}\vrho'^{1/2}\vA_{\beta\vH'/2}(t)\vrho'^{-1/2}\vrho^{1/2}
\right.\right.\nonumber\\[-2pt]
&\hspace{3.2cm}\left.\left.
-\vrho^{1/2}\vrho'^{-1/2}\vA_{\beta\vH'/2}(t)\vrho'^{1/2}\vrho^{-1/2}
\right)_{\beta\vH/2}(-t)\sqrt{\vsigma}\right]\dd t.
\label{eq:intro-identifiability}
\end{align}
Here $\vrho$ and $\vrho'$ are the Gibbs states of $\vH$ and $\vH'$, $\vX_{\vH}(t)=e^{it\vH}\vX e^{-it\vH}$ denotes Heisenberg evolution, and $g_{\mathrm{id}}$ is a fixed, exponentially decaying kernel.
The left-hand side is exactly the pairing in (P2).

For a Gibbs input $\vsigma=\vrho$, the right-hand side of \eqref{eq:intro-identifiability} can be turned into tests with (P1)--(P3) using two exact properties of $\vrho$.
First, the outer factors $\sqrt{\vrho}$ combine with the adjacent Gibbs conjugations, which leaves an expression linear in the input state.
Second, $[\vH,\vrho]=0$ allows the evolution under the unknown $\vH$ to be moved from the state onto $\vO$.
The learner then replaces $\vH$ acting on $\vO$ by an auxiliary $\vG$, which yields a measurable test as in (P3).
The choice $\vG=\vH$ recovers \eqref{eq:intro-identifiability}, which gives (P2).
For $\vH'=\vH$, the two terms cancel by the KMS condition for every $\vG$, which gives (P1).
The imaginary-time conjugations must also be regularized in frequency, which we handle as in Ref.~\cite{chen2025learning} (\Cref{thm:double-frequency-truncation}).

However, for a metastable input, both properties of $\vrho$ fail.
\begin{enumerate}[label=(\roman*),leftmargin=*,topsep=3pt,itemsep=2pt,parsep=0pt]
    \item \emph{Nonlinearity in the state.} The separated factors of $\sqrt{\vsigma}$ make the right-hand side of \eqref{eq:intro-identifiability} generally nonlinear in $\vsigma$, so it is not directly the expectation value of a state-independent observable.
    \item \emph{Noncommutation with the Hamiltonian.} A metastable state need not commute with $\vH$, so moving the unknown dynamics onto $\vO$ also introduces time evolution of the state.
\end{enumerate}
Overcoming these two difficulties, which do not arise for Gibbs states, is the core technical contribution of our work.

Both difficulties are resolved by ADB.
Quantitatively, ADB measures the deviation
\begin{align}
    \vA^a(\omega)\sqrt{\vsigma}
    -\sqrt{\vsigma}\,\vrho^{-1/2}\vA^a(\omega)\vrho^{1/2},
    \label{eq:intro-adb-deviation}
\end{align}
where $\vA^a(\omega)$ is a quasi-local operator Fourier transform of the jump $\vA^a$, the component of $\vA^a$ that changes the energy of $\vH$ by about $\omega$.
The deviation vanishes for $\vsigma=\vrho$.
When $\vH$ consists only of $\vI$ and $\vZ$ Pauli terms, $\vsigma$ is diagonal, and $\vA^a$ is a Pauli $\vX$, it reduces to the classical flow balance described above.
We write $\epsilon_{\mathrm{ADB}}$ for the largest weighted average of the squared Hilbert--Schmidt norm of \eqref{eq:intro-adb-deviation} over all single-site jumps (\Cref{defn:adb}).

For difficulty~(i), ADB allows local operators to pass through $\sqrt{\vsigma}$ at the cost of a Gibbs conjugation.
This merges the two factors of $\sqrt{\vsigma}$ into a single $\vsigma$, schematically $\sqrt{\vsigma}\,\vrho^{-1/2}\vX\vrho^{1/2}\sqrt{\vsigma}\approx\vX\vsigma$, and makes the expression linear in the state.
The subtlety is that ADB applies to operators evolved under the true $\vH$, whereas \eqref{eq:intro-identifiability} also involves the candidate $\vH'$, so we expand the candidate evolution around $\vH$ in a Dyson series.
Each term of this series is a product of local operators evolved under $\vH$, which can be moved through $\sqrt{\vsigma}$ as above (\Cref{thm:kms-to-gns-reduction}).

We address difficulty~(ii) in two ways. The first uses ADB to bound $\|[\vH,\vsigma]\|_1$ and hence the motion of the entire state, at the cost of a factor of $n$ (part~(a) below). To our knowledge, this global stationarity result is new and of independent interest. The second bounds only the tested commutators $\Tr(\vsigma[\vH,\vB])$ for products $\vB$ of a few local operators evolved under $\vH$. Quasi-locality confines the relevant contributions to $[\vH,\vB]$ to neighborhoods of these operators, up to controlled tails. The Gibbs-conjugation identity in \Cref{lem:weak-commutator-identity} then lets us apply local ADB to obtain a bound independent of $n$, as in part~(b) below.

\begin{thm}[ADB implies approximate stationarity under $\vH$, informal version of \Cref{thm:approximate-commute-deduce,thm:weak}]
\label{thm:intro-tested-commutator}
There are constants $C_\beta=e^{\poly(\beta^{\pm1})}$ and $\alpha_\beta=1/\poly(\beta^{\pm1})$, independent of $n$, such that the following holds.
\begin{enumerate}[label=(\alph*),leftmargin=*,topsep=3pt,itemsep=2pt,parsep=0pt]
    \item The state approximately commutes with the Hamiltonian, $\|[\vH,\vsigma]\|_1\leq C_\beta\,n\,\epsilon_{\mathrm{ADB}}^{\alpha_\beta}$.
    \item For every product $\vB=\prod_{j=1}^{q}\vA_j$ of $q$ single-site Paulis, $|\Tr(\vsigma[\vH,\vB])|\leq C_\beta^{\,q}\,\epsilon_{\mathrm{ADB}}^{\alpha_\beta}$.
\end{enumerate}
\end{thm}
Part~(a) alone would already suffice, since $|\Tr(\vsigma[\vH,\vB])|=|\Tr(\vB[\vH,\vsigma])|\leq\|\vB\|\,\|[\vH,\vsigma]\|_1$, but it would introduce a factor of $n$ into the learning precision.
Part~(b) avoids this factor by bounding only the tested commutators that the transfer actually requires.

After both steps, the learner replaces $\vH$ acting on $\vO$ by an auxiliary $\vG$, as in the Gibbs case.
The resulting tests are measurable, and the error in (P2) is bounded by a sufficiently small multiple of the current coefficient error $\eta$, plus $C_\beta\epsilon_{\mathrm{ADB}}^{1/\poly(\beta^{\pm1})}$ (\Cref{lem:tested-identification}). The term proportional to $\eta$ is absorbed into each refinement step, while the ADB term sets the final precision guaranteed by the algorithm (\Cref{subsec:learning-precision}).

\paragraph*{From metastability to approximate detailed balance.}
It remains to derive ADB from metastability.
Ref.~\cite{bergamaschi2025metastability} bounds ADB through entropy production (EP), the rate at which relative entropy to the Gibbs state decreases under the thermal dynamics.
For a full-rank state, EP is given by
\begin{align*}
    \mathsf{EP}_{\CL}[\vsigma]
    &:=-\left.\frac{\dd}{\dd t}
       D\bigl(e^{t\CL}[\vsigma]\|\vrho\bigr)\right|_{t=0}
    =-\Tr\bigl(\CL[\vsigma](\log\vsigma-\log\vrho)\bigr).
\end{align*}
Consequently,
\[
    \mathsf{EP}_{\CL}[\vsigma]
    \leq \|\CL[\vsigma]\|_1\,
          \|\log\vsigma-\log\vrho\|.
\]
This direct estimate becomes singular when $\vsigma$ has zero eigenvalues and therefore does not yield a finite ADB bound for states that are not full rank.
Regularization with the global Gibbs state controls this divergence, but the resulting bound on the global logarithmic term introduces a factor of $n$.

We improve this approach in two respects.
First, combining regularization with continuity of ADB gives a universal bound for the original state under global metastability, including states with zero eigenvalues.
Second, exploiting locality and regularizing only near each jump gives a bound under local metastability that is independent of the system size.
The following theorem summarizes these two improvements.

\begin{thm}[Metastability implies ADB, informal version of \Cref{thm:global-adb,thm:local-adb}]
\label{thm:intro-ms-to-adb}
For a state $\vsigma$, define the global and local metastability errors
\begin{align}
    \epsilon_{\mathrm{ms}}^{\mathrm{glo}}(\vsigma):=\|\CL[\vsigma]\|_1,
    \qquad
    \epsilon_{\mathrm{ms}}^{\mathrm{loc}}(\vsigma):=\max_{a\in\mathcal A}\|\CL_a[\vsigma]\|_1.
    \label{eq:intro-local-metastability}
\end{align}
There is a constant $C_\beta=e^{\poly(\beta^{\pm1})}$, independent of $n$, such that the following holds.
\begin{enumerate}[label=(\alph*),leftmargin=*,topsep=3pt,itemsep=2pt,parsep=0pt]
    \item If $\epsilon_{\mathrm{ms}}^{\mathrm{glo}}\leq1/2$, then $\epsilon_{\mathrm{ADB}}\leq C_\beta\,\epsilon_{\mathrm{ms}}^{\mathrm{glo}}\bigl[n+\log(1/\epsilon_{\mathrm{ms}}^{\mathrm{glo}})\bigr]$.
    \item If $\epsilon_{\mathrm{ms}}^{\mathrm{loc}}\leq1/2$, then $\epsilon_{\mathrm{ADB}}\leq C_\beta\,\epsilon_{\mathrm{ms}}^{\mathrm{loc}}\bigl[\log(e/\epsilon_{\mathrm{ms}}^{\mathrm{loc}})\bigr]^{D+1}$.
\end{enumerate}
\end{thm}
Part~(a) uses the thermal-state regularization technique of Ref.~\cite[Appendix~F]{bergamaschi2025metastability}, applying the EP argument to the mixture $\vsigma_r=(1-r)\vsigma+r\vrho$ with $r=\epsilon_{\mathrm{ms}}^{\mathrm{glo}}$.
Our new continuity lemma for ADB (\Cref{lem:state-cont}) then transfers the bound back to the original state, without requiring full rank.
Together with the previous reduction, part~(a) proves \Cref{thm:intro-global-metastable-learning}, and it is the only source of the system-size-dependent factor in its precision.
Part~(b) instead regularizes the state only near each jump, so that the relevant relative logarithm stays bounded independently of the system size.
It shows that under local metastability, the precision of our algorithm does not depend on $n$. Together with the previous reduction, part~(b) proves \Cref{thm:intro-local-metastable-learning}.

\paragraph*{Byproduct: a complete framework of approximate stationarity.}
Ref.~\cite{bergamaschi2025metastability} related several notions of approximate stationarity through the chain of implications
\begin{align*}
    \text{global metastability}
    \quad\Longrightarrow\quad\text{small EP}
    \quad\Longrightarrow\quad\text{ADB}
    \quad\Longrightarrow\quad\text{local metastability},
\end{align*}
in which the first two steps require full-rank states.
Our results complete this picture, as summarized in \Cref{fig:metastability-relations}.
In particular, we supply the missing direction from local to global metastability.
By \Cref{thm:intro-ms-to-adb}(b), a locally metastable state satisfies ADB, and by \Cref{thm:intro-tested-commutator}(a), ADB forces the state to approximately commute with $\vH$.
The thermal Lindbladian $\CL=-i[\vH,\cdot]+\sum_a\CL_a$ consists of this Hamiltonian part and $3n$ thermal components, so the triangle inequality $\|\CL[\vsigma]\|_1\leq\|[\vH,\vsigma]\|_1+3n\,\epsilon_{\mathrm{ms}}^{\mathrm{loc}}(\vsigma)$ yields the following corollary.

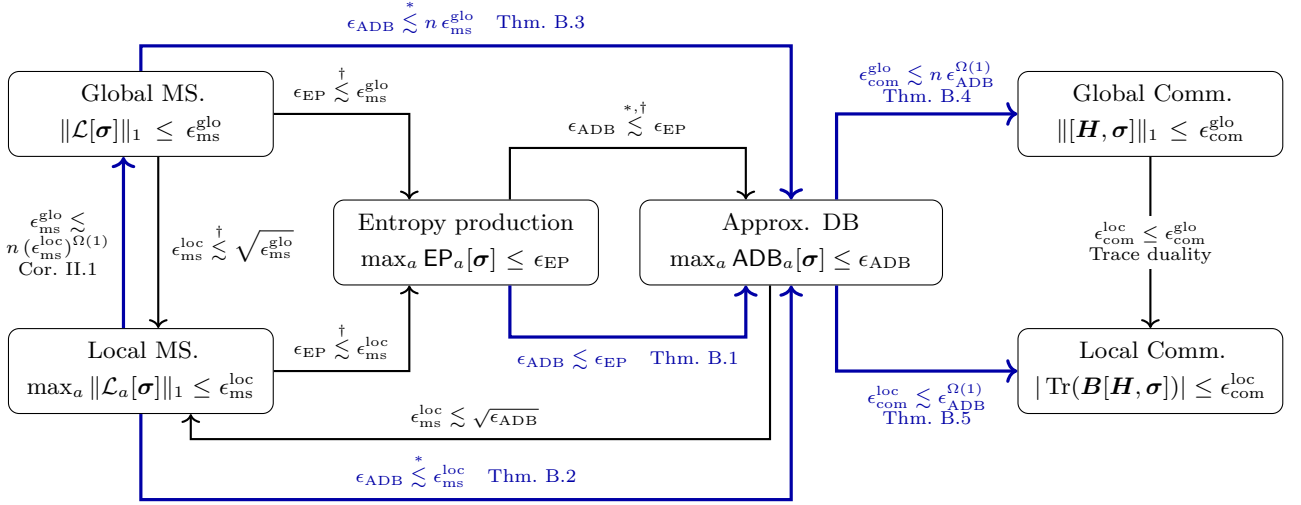
\begin{figure}[!t]
\centering
\begingroup
\def\epscg{\epsilon_{\mathrm{com}}^{\mathrm{glo}}}
\def\epscl{\epsilon_{\mathrm{com}}^{\mathrm{loc}}}
\def\epsmg{\epsilon_{\mathrm{ms}}^{\mathrm{glo}}}
\def\epsml{\epsilon_{\mathrm{ms}}^{\mathrm{loc}}}
\def\epsep{\epsilon_{\mathrm{EP}}}
\def\epsadb{\epsilon_{\mathrm{ADB}}}
\begin{tikzpicture}[
 box/.style={draw,rounded corners,align=center,text width=3.15cm,
   minimum height=1.05cm,inner sep=5pt,font=\small},
 prior/.style={->,thick},
 added/.style={->,very thick,blue!65!black},
 note/.style={font=\scriptsize\normalcolor,align=center,fill=white,inner sep=2pt}]

\node[box] (gm) at (0,2.10) {Global MS.\\[3pt]
 $\|\CL[\vsigma]\|_1\leq\epsmg$};
\node[box] (lm) at (0,-1.30) {Local MS.\\[3pt]
 $\max_a\|\CL_a[\vsigma]\|_1\leq\epsml$};
\node[box] (ep) at (4.30,.40) {Entropy production\\[3pt]
 $\max_a\mathsf{EP}_a[\vsigma]\leq\epsep$};
\node[box,text width=3.65cm] (adb) at (8.60,.40) {Approx. DB\\[3pt]
 \mbox{$\max_a\mathsf{ADB}_a[\vsigma]\leq\epsadb$}};
\node[box] (gc) at (13.35,2.10) {Global Comm.\\[3pt]
 $\|[\vH,\vsigma]\|_1\leq\epscg$};
\node[box] (lc) at (13.35,-1.30) {Local Comm.\\[3pt]
 $|\Tr(\vB[\vH,\vsigma])|\leq\epscl$};

\draw[added] ([xshift=-.23cm]lm.north)--([xshift=-.23cm]gm.south);
\draw[prior] ([xshift=.23cm]gm.south)--([xshift=.23cm]lm.north);
\node[note,anchor=east] at (-0.33,.40)
 {$\epsmg\lesssim$\\$n\,(\epsml)^{\Omega(1)}$\\Cor.~\ref{cor:intro-local-to-global-metastability}};
\node[note,anchor=west] at (0.33,.40)
 {$\epsml\overset{\dagger}{\lesssim}\sqrt{\epsmg}$};

\draw[prior] (gm.east)--node[note,above=2pt,pos=.50]
 {$\epsep\overset{\dagger}{\lesssim}\epsmg$}
 (3.55,2.10)--([xshift=-.75cm]ep.north);
\draw[prior] (lm.east)--node[note,above=2pt,pos=.50]
 {$\epsep\overset{\dagger}{\lesssim}\epsml$}
 (3.55,-1.30)--([xshift=-.75cm]ep.south);

\draw[added] (gm.north)--(0,3.00)--
 node[note,above=2pt,pos=.50]
 {$\epsadb\overset{*}{\lesssim} n\,\epsmg$\quad Thm.~\ref{thm:global-adb}}
 (8.60,3.00)--(adb.north);

\draw[prior] ([xshift=.58cm]ep.north)--(4.88,1.65)--
 node[note,above=2pt,pos=.5]
 {$\epsadb\overset{*,\dagger}{\lesssim}\epsep$}
 (8.00,1.65)--([xshift=-.60cm]adb.north);
\draw[added] ([xshift=.58cm]ep.south)--(4.88,-.85)--
 node[note,below=2pt,pos=.5]
 {$\epsadb\lesssim\epsep$\quad Thm.~\ref{thm:constant-fi-adb}}
 (8.00,-.85)--([xshift=-.60cm]adb.south);

\draw[added] ([xshift=.60cm]adb.north)--(9.20,2.10)--
 node[note,above=2pt,pos=.50]
 {$\epscg\lesssim n\,\epsadb^{\Omega(1)}$\\Thm.~\ref{thm:approximate-commute-deduce}}
 (gc.west);
\draw[added] ([xshift=.60cm]adb.south)--(9.20,-1.30)--
 node[note,below=2pt,pos=.50]
 {$\epscl\lesssim\epsadb^{\Omega(1)}$\\Thm.~\ref{thm:weak}}
 (lc.west);
\draw[prior] (gc.south)--node[note,pos=.50]
 {$\epscl\leq\epscg$\\Trace duality} (lc.north);

\draw[prior] ([xshift=-.28cm]adb.south)--(8.32,-2.20)--
 node[note,above=2pt,pos=.51]
 {$\epsml\lesssim\sqrt{\epsadb}$}
 (0.66,-2.20)--([xshift=.66cm]lm.south);
\draw[added] (lm.south)--(0,-3.00)--
 node[note,above=2pt,pos=.50]
 {$\epsadb\overset{*}{\lesssim}\epsml$\quad Thm.~\ref{thm:local-adb}}
 (8.60,-3.00)--(adb.south);
\end{tikzpicture}

\caption{Relations between metastability (MS), entropy production (EP), approximate detailed balance (ADB), and commutation with the Hamiltonian (Comm.).
Black arrows are results from Ref.~\cite{bergamaschi2025metastability} or direct consequences, and blue arrows are our results.
At fixed temperature and local parameters, $\lesssim$ suppresses constants, and $\Omega(1)$ denotes a positive exponent.
Asterisks ($*$) mark bounds that hold up to additional factors polylogarithmic in the inverse error.
Daggers ($\dagger$) mark additional $\|\log\vsigma-\log\vrho\|$ dependence, which requires full-rank states.
For Local Comm., $\vB$ is a norm-one local operator. Our ADB-to-local-commutation bound extends this class to products of a fixed number of local operators evolved under $\vH$ for bounded times.
Definitions are given in \Cref{sec:adb-to-commutation}.}
\label{fig:metastability-relations}
\endgroup
\end{figure}

\begin{cor}[Local metastability implies global metastability]
\label{cor:intro-local-to-global-metastability}
There exist constants $K_\beta=e^{\poly(\beta^{\pm 1})}$ and $\nu_\beta=1/\poly(\beta^{\pm 1})>0$, independent of $n$, such that every state $\vsigma$ satisfies $\epsilon_{\mathrm{ms}}^{\mathrm{glo}}(\vsigma)\leq K_\beta\, n\,\bigl(\epsilon_{\mathrm{ms}}^{\mathrm{loc}}(\vsigma)\bigr)^{\nu_\beta}$.
\end{cor}

We also sharpen one of the known implications.
Ref.~\cite{bergamaschi2025metastability} bounded ADB by EP only up to a logarithmic factor involving $\|\log\vsigma-\log\vrho\|$, which diverges as eigenvalues of $\vsigma$ approach zero.
We show that ADB is at most a constant multiple of EP, with a constant independent of the state (\Cref{thm:constant-fi-adb}).
Together, global metastability, local metastability, and ADB are equivalent at fixed temperature, up to polynomial losses in the error and factors of $n$.

We have not attempted to optimize the exponents of the errors in these relations, and we focus on their dependence on $n$.
The factor $n$ in the conversion from local to global metastability (Corollary~\ref{cor:intro-local-to-global-metastability}) is necessary, and so is the factor $n$ in the conversion from ADB to global commutation (\Cref{thm:intro-tested-commutator}(a)).
Both are witnessed by a product example in which the local errors stay bounded while the global errors grow linearly in $n$ (Proposition~\ref{prop:commutation-error-lower-bound}).
Whether the factor $n$ in the conversion from global metastability to ADB (\Cref{thm:intro-ms-to-adb}(a)) is necessary remains open.
Removing it would also remove the system-size-dependent factor from the precision of \Cref{thm:intro-global-metastable-learning}.

\section{Discussion}
\label{sec:discussions}

We provided efficient Hamiltonian learning protocols for quantum metastable states of detailed-balanced Lindbladians. To extend the Gibbs learning framework, we heavily adopted the notion of the approximate detailed balance condition, which is simultaneously a way to quantify metastability and a key structure that enables identifiability of the Hamiltonian.

\paragraph*{Theory of metastability.}
Much of the technical work amounts to digesting and sharpening the initial framework of~\cite{bergamaschi2025metastability}. In fact, when handling the Hamiltonian term, we had to dive deep into the framework and ultimately discover that, up to polynomial losses in the error and extra factors of the system size $n$, the approximate detailed balance condition is equivalent to global and local metastability and implies approximate commutation with the Hamiltonian (\Cref{fig:metastability-relations}). Some of these factors of $n$ are necessary, but we do not know whether the factor in the conversion from global metastability to ADB is, and it is this factor that makes the learning precision floor $\eta$ depend on the system size.

\paragraph*{Connection to Markov properties.}
Ref.~\cite{chen2025learning} asked whether the local Markov property of quantum Gibbs states~\cite{chen2025GibbsMarkov} can be used directly for learning.
Metastable states are also locally Markov~\cite{bergamaschi2025metastability}, and classical metastable samples can be learned through their single-variable conditionals~\cite{jayakumarDiscreteDistributionsAre2026}.
Can a quantum local Markov property similarly be turned into a learning algorithm for metastable states?

\paragraph*{Beyond detailed-balanced Lindbladians.}
As the fundamental assumption, the Lindbladians of Ref.~\cite{chen2023exact} with single-site Pauli jumps appeared throughout this paper. We believe, without a proof, that similar results will hold under general conditions on the jump operators and transition rates, and furthermore the learning protocol applies to other detailed-balanced families~\cite{gilyen2024glauber,dingEfficientQuantumGibbs2024,scandiThermalizationOpenManyBody2026}. On the other hand, the present argument is very much restricted to the weak-coupling, Markovian bath settings and there might be other classes of states that can be counted as `finite-temperature'. While we believe the approximate detailed balance condition can continue to play a role --- and that the presented protocol can very well work --- it is not clear how to derive this condition beyond Lindbladian stationary states.

\vspace{\baselineskip}
\noindent\textbf{Acknowledgements.}
Concurrently and independently, Thiago Bergamaschi also obtained an argument for learning Hamiltonians from quantum metastable states. We thank him for discussing this work and for coordinating the arXiv posting. We also thank Sitan Chen and Yunchao Liu for early discussions that motivated this project, and Anurag Anshu, Quynh Nguyen, Thiago Bergamaschi, and Umesh Vazirani for collaboration on related topics. We thank Dong-ling Deng, Weiyuan Gong, and Laura Lewis for helpful discussions.
QY is supported by the National Natural Science Foundation of China (Grant No.~T24B2002) and the Quantum Science and Technology-National Science and Technology Major Project (Grant No.~2021ZD0302203). CFC is supported by a Simons-CIQC postdoctoral fellowship through NSF QLCI Grant No.~2016245.

\vspace{\baselineskip}
\noindent\textbf{AI Disclosure.}
The authors conceived and wrote the original statement and proof of \Cref{thm:intro-global-metastable-learning}. To further improve the learning precision floor (\Cref{thm:intro-local-metastable-learning}), LLMs helped discover the system-size-independent ADB bound under local metastability (\Cref{thm:local-adb,lem:state-cont,lem:adb4-relative-log,lem:adb4-regularized-ADB}) and develop the idea underlying \Cref{thm:weak}. LLMs also assisted with polishing the writing and figure preparation. The authors reviewed and checked all incorporated LLM-generated material and take full responsibility for the manuscript. The models used were GPT-6 Pro and Opus 5.5.

\appendix
\clearpage
\hypersetup{hidelinks}
\section{Preliminaries}
\label{sec:basic-notation}

Bold uppercase Latin letters denote operators and Hamiltonians, while bold Greek letters denote density operators. We write $\|\cdot\|$ for the operator norm, $\|\cdot\|_p$ for the Schatten $p$-norm, and $\supp(\vO)$ for the support of an operator.
\begin{defn}[Local Hamiltonians]\label{defn:local-Hamiltonian}
On $n$ qubits, a local Hamiltonian has the form
\begin{align}
 \vH=\sum_{\gamma\in\Gamma}h_\gamma\vP_\gamma,
 \qquad h_\gamma\in[-1,1],
 \label{eq:local-Hamiltonian-definition}
\end{align}
where the known, distinct, nonidentity Pauli strings $\vP_\gamma$ have weight at most $k$. We choose an interaction-degree bound $d\geq1$ so that each site belongs to at most $d$ interaction terms and each interaction term overlaps at most $d$ terms, including itself. For the geometric learning protocol, the qubits lie on a finite subset of $\mathbb Z^D$ and the interaction terms have bounded diameter. The parameters $k,d,D$ and the interaction range are fixed independently of $n$.
\end{defn}
Unless stated otherwise, local Hamiltonians satisfy \Cref{defn:local-Hamiltonian}. We write $X\lesssim Y$ when $X\leq {C_{\beta}}Y$ for a positive constant ${C_{\beta}}$ that depends only on $\beta$ and the fixed locality parameters. Such constants are independent of $n$; an absolute constant is independent of all parameters. The coefficients and degrees of $\poly(\beta^{\pm 1})$ may also depend on the fixed locality parameters. We keep dependence on variable support sizes and product lengths explicit.

At inverse temperature $\beta>0$, the Gibbs state and the maximally mixed state are
\begin{align}
 \vrho=\frac{e^{-\beta\vH}}{\Tr(e^{-\beta\vH})},
 \qquad \vtau=\frac{\vI}{2^n}.
 \label{eq:basic-gibbs-state}
\end{align}
For any density operator $\vsigma$, define the KMS pairing and norm by
\begin{align}
 \langle\vO,\vP\rangle_{\vsigma}
 :=\Tr[\vO^\dagger\vsigma^{1/2}\vP\vsigma^{1/2}],
 \qquad \|\vO\|_{\vsigma}:=\sqrt{\langle\vO,\vO\rangle_{\vsigma}}.
 \label{eq:kms-inner-product-basic}
\end{align}
Strictly speaking, when $\vsigma$ is not full rank, $\|\cdot\|_{\vsigma}$ is only a seminorm. We nevertheless refer to it as the KMS norm throughout. For $\vsigma=\vtau$, this is the normalized Hilbert--Schmidt norm.

The ideal property of Schatten norms gives, for compatible operators,
\begin{align}
    \|\vX\vY\vZ\|_2
    \leq \|\vX\|\,\|\vY\|_2\,\|\vZ\|,
    \qquad
    \|\vX\|_{\vsigma}\leq\|\vX\|,
    \qquad
    \|\vsigma^{1/2}\|_2=1.
    \label{eq:basic-schatten-estimates}
\end{align}
We use these inequalities below without further comment when converting trace expressions into bounds in the operator or Hilbert--Schmidt norm.

\subsection{Thermal dynamics and detailed balance}
\label{subsec:thermal-lindbladian-definition}
A system coupled to a thermal reservoir exchanges energy with its environment. Under suitable weak-coupling and Markovian approximations, its reduced dynamics is described by a Lindblad master equation. Classical Glauber dynamics provides a useful model in which local updates relax the system toward the Gibbs distribution. Detailed balance equates the equilibrium probability flow of each transition with that of its reverse. Here, quantum detailed balance means that the Heisenberg generator is self-adjoint in the Gibbs KMS inner product. This property implies that the Gibbs state is stationary~\cite{chen2023exact,gilyen2024glauber}.

We use the construction of Chen, Kastoryano, and Gily\'en~\cite{chen2023exact}, in the form studied in Ref.~\cite[Section IV C]{bergamaschi2025metastability}. It combines exact KMS detailed balance with quasi-local jump operators. These jumps resolve energy differences through Gaussian-filtered Hamiltonian evolution. For an operator $\vA$, write $\vA_{\vH}(t):=e^{i\vH t}\vA e^{-i\vH t}$ and define
\begin{align}
 \widehat{\vA}_{\vH}(\omega)
 &:=\frac1{\sqrt{2\pi}}\int_{-\infty}^{\infty}\vA_{\vH}(t)e^{-i\omega t}f_\sigma(t)\dd t,
 &f_\sigma(t)&:=e^{-\sigma^2t^2}\sqrt{\sigma\sqrt{2/\pi}}.
 \label{eq:operator-fourier-transform-basic}
\end{align}
The scalar width $\sigma$ is distinct from the state $\vsigma$.

For the single-site Pauli jumps $\{\vA^a\}_{a\in\mathcal A}=\{\vX_i,\vY_i,\vZ_i\}_{i\in\Lambda}$, define $\widehat{\vA}^{a}(\omega)=\widehat{\vA^a}_{\vH}(\omega)$ by \eqref{eq:operator-fourier-transform-basic}, with $0<\sigma\leq1/\beta$, and set
\begin{align}
 \gamma(\omega)&=\exp\!\left[-\beta\max\left(\omega+\frac{\beta\sigma^2}{2},0\right)\right],
 &\widehat{\vA}^{a}(\omega,t)&=\bigl(\widehat{\vA}^{a}(\omega)\bigr)_{\vH}(t).
 \label{eq:thermal-jumps-and-weight}
\end{align}
Each local generator is
\begin{align}
 \CL_a[\cdot]={}&-i[\vC^a,\cdot]
 +\int_{-\infty}^{\infty}\gamma(\omega)
 \left(\widehat{\vA}^{a}(\omega)(\cdot)\widehat{\vA}^{a}(\omega)^\dagger
 -\frac12\left\{\widehat{\vA}^{a}(\omega)^\dagger\widehat{\vA}^{a}(\omega),\cdot\right\}\right)\dd\omega,
 \label{eq:local-thermal-lindbladian}\\
 \vC^a:={}&\operatorname{p.v.}\!\int_{-\infty}^{\infty}\int_{-\infty}^{\infty}
 \gamma(\omega)c(t)\widehat{\vA}^{a}(\omega,t)^\dagger\widehat{\vA}^{a}(\omega,t)
 \dd\omega\dd t,
 &c(t)&:=\frac{1}{\beta\sinh(2\pi t/\beta)}.
 \label{eq:thermal-coherent-correction}
\end{align}
The principal value is taken symmetrically at $t=0$. The Hermitian correction $\vC^a$ is part of the thermal generator. As shown in Refs.~\cite{chen2023exact,bergamaschi2025metastability}, the Heisenberg adjoint of $\CL_a$ satisfies the KMS detailed-balance identity
\begin{align}
 \langle\vO,\CL_a^\dagger[\vP]\rangle_{\vrho}
 =\langle\CL_a^\dagger[\vO],\vP\rangle_{\vrho},
 \qquad\CL_a[\vrho]=0.
 \label{eq:thermal-KMS-balance}
\end{align}
Including the physical Hamiltonian evolution gives the full dynamics
\begin{align}
    \CL[\cdot]=-i[\vH,\cdot]+\sum_a\CL_a[\cdot],
    \qquad \CL_a[\vrho]=0.
    \label{eq:thermal-lindbladian-basic}
\end{align}
Thus $\CL[\vrho]=0$. KMS self-adjointness above refers to each thermal component $\CL_a$. The stationarity conditions for this dynamics are introduced in \Cref{sec:adb-to-commutation}.

The Gaussian Fourier representation makes this generator particularly convenient for our estimates. Quantum Glauber and Metropolis constructions are effective models of thermal relaxation, with connections to master equations obtained from system--bath approximations~\cite[Section 8]{gilyen2024glauber}. Establishing the same learning guarantees for other thermal generators remains open. We use the explicit detailed-balance structure of this construction throughout.

\subsection{Operator Fourier transform}
\label{subsec:operator-fourier}
We collect properties of the transform \eqref{eq:operator-fourier-transform-basic}, following Refs.~\cite{chen2023exact,chen2025learning,chen2025GibbsMarkov}.
The Fourier transform of the scalar Gaussian is
\begin{align}
    \widehat f_\sigma(\omega)
    =\frac{e^{-\omega^2/(4\sigma^2)}}{\sqrt{\sigma\sqrt{2\pi}}}.
    \label{eq:gaussian-fourier-pair-basic}
\end{align}
For $\vH=\sum_E E\vP_E$, define the Bohr components $\vA_\nu:=\sum_{E_2-E_1=\nu}\vP_{E_2}\vA\vP_{E_1}$. Then
\begin{align}
    \vA_{\vH}(t)=\sum_\nu e^{i\nu t}\vA_\nu,
    \qquad
    \widehat{\vA}_{\vH}(\omega)
    =\sum_\nu \vA_\nu\widehat f_\sigma(\omega-\nu),
    \label{eq:bohr-decomposition-basic}
\end{align}
and the reconstruction identity is
\begin{align}
    \vA=\frac{1}{\sqrt{2\sigma\sqrt{2\pi}}}
    \int_{-\infty}^{\infty}\widehat{\vA}_{\vH}(\omega)\,\mathrm d\omega.
    \label{eq:operator-ft-reconstruction-basic}
\end{align}
These identities and normalization conventions agree with~\cite{chen2023exact,chen2025learning,chen2025GibbsMarkov}.

\paragraph*{Imaginary-time conjugation.}

For the Gibbs state $\vrho\propto e^{-\beta\vH}$, Gaussian filtering converts imaginary-time conjugation into a frequency shift according to
\begin{align}
    \vrho^{1/2}\widehat{\vA}_{\vH}(\omega)\vrho^{-1/2}
    &=e^{-\beta\omega/2+\sigma^2\beta^2/4}
      \widehat{\vA}_{\vH}(\omega-\sigma^2\beta),
    \label{eq:imaginary-time-positive-basic}\\
    \vrho^{-1/2}\widehat{\vA}_{\vH}(\omega)\vrho^{1/2}
    &=e^{\beta\omega/2+\sigma^2\beta^2/4}
      \widehat{\vA}_{\vH}(\omega+\sigma^2\beta).
    \label{eq:imaginary-time-negative-basic}
\end{align}
Since \(\|\widehat{\vA}_{\vH}(\omega)\|\leq
\|\vA\|/[\sqrt{\sigma}(2\pi)^{1/4}]\), \Cref{eq:imaginary-time-positive-basic,eq:imaginary-time-negative-basic} imply, for \(|\omega|\leq\Omega\),
\begin{align}
    \left\|\vrho^{\pm1/2}\widehat{\vA}_{\vH}(\omega)\vrho^{\mp1/2}\right\|
    \leq
    \frac{e^{\beta\Omega/2+\sigma^2\beta^2/4}}
    {\sqrt{\sigma}(2\pi)^{1/4}}\,\|\vA\|.
    \label{eq:imaginary-time-norm-basic}
\end{align}
The frequency cutoffs in \Cref{sec:identifiability-bounds} bound these imaginary-time factors independently of the system size~\cite{chen2025learning,chen2025GibbsMarkov}.

The formulas above hold for general $0<\sigma\leq\beta^{-1}$. In the subsequent appendices, we set the scalar Fourier width to $\sigma=\beta^{-1}$. The notation $f_\beta$ denotes $f_{\beta^{-1}}$ from \eqref{eq:operator-fourier-transform-basic}; $\sigma_f$, when used, denotes this same scalar width.

\subsection{Lieb--Robinson bounds}

We use the truncation and perturbation estimates of Ref.~\cite[Lemmas B.1 and B.2]{chen2025learning}, which are based on Lieb--Robinson bounds~\cite{lieb1972finite}. For this subsection, we have similar definition requirements as in \Cref{defn:local-Hamiltonian}.

Let $\vA$ be an observable supported on a region $\mathsf{A}$. The interaction graph has one vertex for each set $X_\gamma$ and an edge between distinct vertices whose sets overlap. We write $\operatorname{dist}(\gamma,\mathsf{A})$ for the minimum graph distance from $\gamma$ to a vertex whose set intersects $\mathsf{A}$; it is infinite if no such path exists. Thus terms meeting $\mathsf{A}$ have distance zero. For an integer $\ell\geq1$, let $\vH_\ell$ retain precisely the terms at distance less than $\ell-1$ from $\mathsf{A}$. The corresponding site patch is $\mathsf{A}$ together with the full sets $X_\gamma$ of all retained terms. For a single-site jump $\vA^a$, we denote this patch by $S_a(\ell)$.

\begin{lem}[Lieb--Robinson truncation bound~{\cite[Lemma B.1]{chen2025learning}}]
\label{lem:basic-lr}
For a Hamiltonian with local terms as above, an observable $\vA$ supported on $\mathsf{A}$, every integer $\ell\geq1$, and every $t\in\mathbb R$,
\begin{align}
    \|\vA_{\vH_\ell}(t)-\vA_{\vH}(t)\|
    \lesssim \|\vA\|\min\left\{2,
    |\supp(\vA)|\frac{(2d|t|)^\ell}{\ell!}\right\}.
    \label{eq:basic-lr-bound}
\end{align}
\end{lem}

\makeatletter
\MT@newlabel{eq:termwise-perturbative-lr-basic}
\MT@newlabel{{eq:termwise-perturbative-lr-basic}}
\MT@newlabel{eq:perturbative-lr-basic}
\MT@newlabel{{eq:perturbative-lr-basic}}
\makeatother
\begin{lem}[Stability under local perturbations]
\label{lem:basic-lr-perturbation}
Let $\vF=\sum_{\gamma\in\Gamma}\vf_\gamma$ and $\vF'=\sum_{\gamma\in\Gamma}\vf'_\gamma$, where $\vf_\gamma$ and $\vf'_\gamma$ are Hermitian, supported on the same sets $X_\gamma$, and satisfy $\|\vf_\gamma\|,\|\vf'_\gamma\|\leq1$. Distances and the degree bound $d$ refer to the common interaction graph defined by these sets. Let $\vA$ be supported on $\mathsf{A}$ and set $r_\delta:=\operatorname{dist}(\delta,\mathsf{A})$. Then, for every real $t$,
\begin{align}
    &\|\vA_{\vF'}(t)-\vA_{\vF}(t)\|
    \nonumber\\
    &\quad\lesssim
    \frac{\|\vA\|\,|\mathsf{A}|}{d}
    \sum_{\substack{\delta\in\Gamma\\r_\delta<\infty}}
    \|\vf'_\delta-\vf_\delta\|
    \min\left\{2|t|,
    \frac{(2d|t|)^{r_\delta+1}}{(r_\delta+1)!}\right\}.
    \label{eq:termwise-perturbative-lr-basic}
\end{align}
Terms at infinite distance do not affect the evolution of $\vA$.
\end{lem}

\begin{proof}[Proof]
Follow the termwise interpolation proof of Ref.~\cite[Lemma B.2]{chen2025learning}, retaining the factor $\|\vA\|\,|\mathsf{A}|$ from Lemma B.1 there to allow general support and normalization. Summing over the changed terms gives \eqref{eq:termwise-perturbative-lr-basic}.
\end{proof}

\clearpage
\begingroup
\makeatletter
\@for\BAppendixEquationLabel:=%
eq:local-metastability-basic,eq:intro-uniform-ADB-error,eq:relations-triangle%
\do{%
  \expandafter\MT@newlabel\expandafter{\BAppendixEquationLabel}%
  \edef\BAppendixEquationTag{{\BAppendixEquationLabel}}%
  \expandafter\MT@newlabel\expandafter{\BAppendixEquationTag}%
}
\makeatother
\endgroup
\section{Metastability and approximate stationarity}

\label{sec:adb-to-commutation}
This appendix introduces the notions of metastability and establishes the basic properties used in our learning analysis.

\Cref{subsec:metastability-notions} fixes the notation used throughout the appendices. \Cref{subsec:entropy-production-to-adb} bounds the approximate detailed balance (ADB) error, introduced in Ref.~\cite{bergamaschi2025metastability}, by a constant multiple of the entropy production rate. \Cref{subsec:thermal-lindbladian-to-adb} then bounds this error in terms of local and global metastability for use in the learning analysis. \Cref{subsec:adb-to-hamiltonian-stationarity} shows that metastable states are approximately stationary under Hamiltonian evolution. Finally, \Cref{sec:adb-appendix,subsec:learning-test-stationarity} establish the ADB and tested-commutator estimates for local operator products that we repeatedly use in the learning argument.

\subsection{Review of notions of metastability and their relations}
\label{subsec:metastability-notions}

The Gibbs state $\vrho$ describes thermal equilibrium and is an exact fixed point of the dynamics introduced in \Cref{subsec:thermal-lindbladian-definition}. Reaching it can nevertheless require a mixing time exponential in the system size. As discussed in \Cref{sec:main-results}, this motivates considering metastable states, whose evolution is slow even when they remain far from $\vrho$. Approximate stationarity can be attained without waiting for global equilibration. In particular, averaging the evolution over a uniformly chosen time in $[0,T]$ produces a state with global metastability error at most $2/T$~\cite[Lemma II.1]{bergamaschi2025metastability}. We introduce the metastability errors and local equilibrium conditions that quantify this behavior, following Ref.~\cite{bergamaschi2025metastability}.

\paragraph*{Global and local metastability.}
The global metastability error is defined by
\begin{align}
 {\epsilon_{\mathrm{ms}}^{\mathrm{glo}}}(\vsigma):=\|\CL[\vsigma]\|_1.
 \label{eq:metastability-definition-basic}
\end{align}
For the single-site Pauli jump family, the local metastability error is defined by the maximum over the thermal components,
\begin{align}
 {\epsilon_{\mathrm{ms}}^{\mathrm{loc}}}(\vsigma):=\max_{a\in\mathcal A}\|\CL_a[\vsigma]\|_1.
 \label{eq:local-metastability-basic}
\end{align}
Here locality refers to the quasi-local generators $\CL_a$. Each metastability error is measured in the full trace norm.

\paragraph*{Entropy production.}
The relative entropy $D(\vsigma\|\vrho)=\Tr[\vsigma(\log\vsigma-\log\vrho)]$ differs from $\beta$ times the free energy by a constant. Its decay therefore measures relaxation of the free energy. For a full-rank state, define
\begin{align}
 \mathsf{EP}_{\CL}[\vsigma]
 &:=-\left.\frac{d}{dt}D(e^{t\CL}[\vsigma]\|\vrho)\right|_{t=0}
 =-\Tr[\CL[\vsigma](\log\vsigma-\log\vrho)],
 \label{eq:entropy-production-definition}\\
 \mathsf{EP}_a[\vsigma]
 &:=-\Tr[\CL_a[\vsigma](\log\vsigma-\log\vrho)].
 \label{eq:local-entropy-production}
\end{align}
Each thermal component preserves $\vrho$, so its entropy production is nonnegative. Hamiltonian evolution leaves the relative entropy to $\vrho$ unchanged. Hence
\begin{align}
 \mathsf{EP}_{\CL}[\vsigma]=\sum_a\mathsf{EP}_a[\vsigma],
 \qquad \mathsf{EP}_a[\vsigma]\geq0.
 \label{eq:entropy-production-sum}
\end{align}
Ref.~\cite[Theorem B.1]{bergamaschi2025metastability} expresses each local entropy production rate $\mathsf{EP}_a$ as a nonnegative quadratic form in commutators of the relative logarithm $\log\vsigma-\log\vrho$. This representation connects dynamical metastability to the local equilibrium relation below. States with zero eigenvalues are treated by regularization in \Cref{subsec:thermal-lindbladian-to-adb}.

\paragraph*{Approximate detailed balance.}
In thermal equilibrium, a local transition and its reverse are related by a Boltzmann weight. Approximate detailed balance (ADB) expresses an approximate version of this relation for $\sqrt{\vsigma}$. We use the ADB error throughout the learning analysis to control the errors. With the Fourier convention of \eqref{eq:operator-fourier-transform-basic}, set $\vA^a(\omega,t)=(\widehat{\vA^a}_{\vH}(\omega))_{\vH}(t)$.
\begin{defn}
    [Approximate detailed balance~{\cite[Definition C.1]{bergamaschi2025metastability}}]\label{defn:adb}
    For a state $\vsigma$ and jump operator $\vA^a$, define
    \begin{equation}
        \mathsf{ADB}_a[\vsigma]
        :=\iint_{-\infty}^{\infty}\gamma(\omega)g(t)
        \times\Bigl\|\vA^a(\omega,t)\sqrt{\vsigma}
        -\sqrt{\vsigma}\vrho^{-\frac12}
          \vA^a(\omega,t)\vrho^{\frac12}\Bigr\|_2^2
          \,\rd\omega\,\rd t,
    \end{equation}
where $\gamma(\omega)$ is the shifted-Metropolis weight
    \begin{align}
         \gamma(\omega) = \exp\L(-\beta\max\left(\omega +\frac{\beta \sigma^2}{2},0\right)\R),\qquad
         g(t) = \frac{1}{\beta\cosh(2\pi t/ \beta)}.\label{eq:Metropolis}
    \end{align}
By Ref.~\cite[Lemma C.1]{bergamaschi2025metastability}, the corresponding time-independent defect satisfies
     \begin{align}
         \int_{-\infty}^{\infty} \lnormp{\vA^a(\omega)\sqrt{\vsigma}-\sqrt{\vsigma}\vrho^{-\frac{1}{2}}\vA^a(\omega)\vrho^{\frac{1}{2}}}{2}^2  \gamma(\omega) \rd \omega \lesssim \mathsf{ADB}_a[\vsigma].
     \end{align}
\end{defn}
We use the uniform single-site error
\begin{align}
 \epsilon_{\mathrm{ADB}}:=\max_a\mathsf{ADB}_a[\vsigma].
 \label{eq:intro-uniform-ADB-error}
\end{align}
The estimates extracted from this defect are collected in \Cref{sec:adb-appendix}.

\paragraph*{Hamiltonian stationarity.}
The stationarity error under Hamiltonian evolution is
\begin{align}
 \epsilon_H(\vsigma):=\|[\vH,\vsigma]\|_1.
 \label{eq:H-stationarity-basic}
\end{align}
It controls the change of the state under $\vH$ alone, as quantified in \Cref{lem:Hamiltonian-stationarity}.

For full-rank states, global metastability bounds every local entropy-production rate through the relative-logarithm estimate \eqref{eq:metastability-to-EP}; the same estimate with the corresponding single-jump metastability error bounds each rate separately. Small local entropy production implies ADB~\cite[Theorem C.1]{bergamaschi2025metastability}. Conversely, ADB implies local metastability~\cite[Theorem C.2]{bergamaschi2025metastability}. There is a constant $C_\beta\leq e^{\poly(\beta^{\pm1})}$ such that, for every state $\vsigma$ and every jump $\vA^a$,
\begin{align}
 \|\CL_a[\vsigma]\|_1\leq C_\beta\|\vA^a\|\sqrt{\mathsf{ADB}_a[\vsigma]}.
 \label{eq:adb-to-local-metastability}
\end{align}
The new implication is from local metastability to Hamiltonian stationarity, using the ADB estimates proved below. Together with
\begin{align}
 {\epsilon_{\mathrm{ms}}^{\mathrm{glo}}}(\vsigma)
 \leq\epsilon_H(\vsigma)+3n{\epsilon_{\mathrm{ms}}^{\mathrm{loc}}}(\vsigma),
 \label{eq:relations-triangle}
\end{align}
 this gives a quantitative local-to-global metastability bound once $\epsilon_H(\vsigma)$ is controlled by $\epsilon_{\mathrm{ms}}^{\mathrm{loc}}$. The reverse direction is established in Ref.~\cite[Remark C.1]{bergamaschi2025metastability}. \Cref{fig:metastability-relations} also shows how the ADB estimates enter the learning argument.

\begingroup
\makeatletter
\@for\BAppendixEquationLabel:=%
eq:EP-to-adb,eq:constant-fi-adb,%
eq:twirl-base-convolution,%
eq:twirled-adb-definition,%
eq:twirl-adb-weight,%
eq:twirl-original-fi,%
eq:twirl-gradient-identity,%
eq:twirl-base-comparison,%
eq:twirl-uniform-comparison,%
eq:twirl-fi-lower-envelope,%
eq:twirl-singular-part,%
eq:twirl-tail-part,%
eq:twirl-integrated-equivalence%
\do{%
  \expandafter\MT@newlabel\expandafter{\BAppendixEquationLabel}%
  \edef\BAppendixEquationTag{{\BAppendixEquationLabel}}%
  \expandafter\MT@newlabel\expandafter{\BAppendixEquationTag}%
}
\makeatother
\endgroup
\subsection{From entropy production rate to approximate detailed balance}
\label{subsec:entropy-production-to-adb}
For a full-rank state and a norm-one Hermitian jump, Ref.~\cite[Theorems B.1 and C.1]{bergamaschi2025metastability} proves
\begin{align}
 \mathsf{ADB}_a[\vsigma]\lesssim\mathsf{EP}_a[\vsigma]
 \left(1+\log\frac{\|\log\vsigma-\log\vrho\|^2}{\mathsf{EP}_a[\vsigma]}\right).
 \label{eq:EP-to-adb}
\end{align}
The logarithmic factor can diverge as eigenvalues of $\vsigma$ approach zero, so this estimate does not give a uniform constant-factor comparison. In \Cref{thm:constant-fi-adb}, we remove this factor and bound ADB by a constant multiple of entropy production, independently of the system size and the input state. The auxiliary functional below averages the original ADB time weight against the normalized Gaussian.

Throughout this subsection, the scalar Fourier width is $\sigma=\beta^{-1}$, and $f_\beta$ denotes $f_{\beta^{-1}}$ from \eqref{eq:operator-fourier-transform-basic}.

\begin{thm}[Constant comparison between entropy production and ADB]
\label{thm:constant-fi-adb}
There is a numerical constant $C_{\beta}\leq e^{\poly(\beta^{\pm 1})}$, independent of {$\vH$, the norm-one jump,} the system size and the input state, such that every full-rank density operator
$\vsigma$ satisfies
\begin{equation}
    \mathsf{ADB}_{a}(\vsigma)
    \leq C_{\beta}\mathsf{EP}_{a}
              [\vsigma].
    \label{eq:constant-fi-adb}
\end{equation}
{The entropy production is computed for the original single-jump Gaussian, shifted-Metropolis generator.}
\end{thm}

To obtain \Cref{eq:EP-to-adb}, Ref.~\cite[Definition~B.1 and Theorem~B.1]{bergamaschi2025metastability} introduces the relative Fisher information, which equals the entropy production for full-rank states:
\begin{align}
    \mathsf{FI}_{a}^{\vH}[\vsigma\Vert\vrho_{\vH}]
    &=\int_{-\infty}^{\infty}\int_{-\infty}^{\infty}
      \int_{-1/2}^{1/2}h_s(\omega)g_s(t)
    \times\left\|
       \vsigma^{\frac{1+2s}{4}}
       [\widehat{\vA}_{\vH}^{a}(\omega,t),
          \log\vsigma-\log\vrho_{\vH}]
       \vsigma^{\frac{1-2s}{4}}
     \right\|_2^2
     \dd s\dd\omega\dd t.
    \label{eq:twirl-original-fi}
\end{align}
The scalar kernels used for FI and ADB are, for $|s|<1/2$,
\begin{align}
    h_s(\omega)
    &:=\exp\!\left(
          s\beta\omega-\frac{s^2}{2}-\frac18
          -\frac\beta2\left|\omega-\frac{s}{\beta}\right|
       \right),
    \label{eq:twirl-frequency-weight}\\
    g_s(t)
    &:=\frac{2\cos(\pi s)\cosh(2\pi t/\beta)}
       {\beta[\cosh(4\pi t/\beta)+\cos(2\pi s)]},
    \label{eq:twirl-fi-weight}\\
    g_s^{\mathsf{ADB}}(t)
    &:=\frac1{2\pi\beta}
       \log\frac{\cosh(2\pi t/\beta)+\cos(\pi s)}
                    {\cosh(2\pi t/\beta)-\cos(\pi s)}
      =\frac12\int_{-(1/2-|s|)}^{1/2-|s|}g_u(t)\dd u,
    \label{eq:twirl-adb-weight}
\end{align}
with $g_0=g$. By Ref.~\cite[Lemma~C.2]{bergamaschi2025metastability}, replacing $g_s(t)$ by $g_s^{\mathsf{ADB}}(t)$ in \eqref{eq:twirl-original-fi} gives ADB. Thus the comparison reduces to their kernel ratio, which diverges near $s=t=0$. Controlling this singularity by a cutoff produces the logarithmic factor in \Cref{eq:EP-to-adb}~\cite[Theorem~C.1]{bergamaschi2025metastability}.

\Cref{eq:constant-fi-adb} strengthens \Cref{eq:EP-to-adb} by replacing the logarithmic factor with a constant independent of the state and system size. To prove it, we introduce twirled ADB by Gaussian averaging of the ADB time weight. \Cref{lem:twirl-gradient} expresses twirled ADB using the same entropy gradients and frequency weights as FI, so only the time kernels need to be compared. \Cref{lem:twirl-base-comparison} shows that Gaussian averaging changes the original ADB weight by at most constant factors; hence a bound on twirled ADB also controls ADB. Finally, \Cref{lem:twirl-fi-weight-comparison} shows that the averaging controls the singularity: the resulting time kernel is bounded by a constant multiple of the FI kernel, uniformly in $s$ and $t$. Together with the identity between FI and entropy production, these comparisons give the constant bound in \Cref{eq:constant-fi-adb} and prove \Cref{thm:constant-fi-adb}.

\begin{defn}[Twirled approximate detailed balance]
\label{defn:twirled-adb}
Define
\begin{equation}
    \widetilde g(t)
    :=\int_{-\infty}^{\infty}|f_\beta(u)|^2g(t-u)\dd u.
    \label{eq:twirl-base-convolution}
\end{equation}
For any density operator $\vsigma$, its twirled ADB error is
\begin{align}
    \operatorname{tADB}_{\vH}^{a}(\vsigma)
    &:=\int_{-\infty}^{\infty}\int_{-\infty}^{\infty}
        \gamma(\omega)\widetilde g(t)
    \times\left\|
       \widehat{\vA}_{\vH}^{a}(\omega,t)\sqrt{\vsigma}
       -\sqrt{\vsigma}e^{\beta\vH/2}
          \widehat{\vA}_{\vH}^{a}(\omega,t)e^{-\beta\vH/2}
     \right\|_2^2
     \dd\omega\dd t.
    \label{eq:twirled-adb-definition}
\end{align}
The original $\operatorname{ADB}_{\vH}^{a}$ has the same integrand with
$g(t)$ in place of $\widetilde g(t)$.
\end{defn}

The twirled ADB error has a direct dynamical interpretation.
{Gaussian averaging of the original single-jump generator gives}
\begin{align}
    \widetilde{\mathcal L}_a^{\vH}[\vsigma]
    :=
    \int_{-\infty}^{\infty}|f_\beta(u)|^2
    e^{iu\vH}\mathcal L_a^{\vH}
    [e^{-iu\vH}\vsigma e^{iu\vH}]
    e^{-iu\vH}\dd u.
    \label{eq:twirled-thermalizing-generator}
\end{align}

The twirled ADB error is precisely the ADB error associated
with this generator, evaluated using its Gaussian mixture
of conjugated jump operators. Its ADB error is
\begin{align}
    \operatorname{tADB}_{\vH}^{a}(\vsigma)
    &=
    \int_{-\infty}^{\infty}\!
    \int_{-\infty}^{\infty}\!
    \int_{-\infty}^{\infty}
    |f_\beta(u)|^2\gamma(\omega)g(t)
    \times
    \left\|
    \widehat{\vA}_{\vH}^{a}(\omega,t+u)\sqrt{\vsigma}
    -
    \sqrt{\vsigma}e^{\beta\vH/2}
    \widehat{\vA}_{\vH}^{a}(\omega,t+u)e^{-\beta\vH/2}
    \right\|_2^2
    \dd t\dd\omega\dd u.
\end{align}
This recovers the definition of twirled ADB above by integrating over $u$.
Equivalently, unitary invariance of the Hilbert--Schmidt
norm gives
\begin{equation}
    \operatorname{tADB}_{\vH}^{a}(\vsigma)
    =
    \int_{-\infty}^{\infty}|f_\beta(u)|^2
    \operatorname{ADB}_{\vH}^{a}
    (e^{-iu\vH}\vsigma e^{iu\vH})\dd u.
\end{equation}
Thus, twirled ADB measures the original detailed-balance
defect averaged over Gaussian-distributed Hamiltonian
rotations of the input state.

The following lemma gives an entropy-gradient representation of twirled ADB, which we use to bound it by the Fisher information.

\begin{lem}[Entropy-gradient representation of twirled ADB]
\label{lem:twirl-gradient}
For a full-rank state $\vsigma$, define
\begin{equation}
    \widetilde g_s^{\mathsf{ADB}}(t)
    :=\int_{-\infty}^{\infty}
       |f_\beta(u)|^2g_s^{\mathsf{ADB}}(t-u)\dd u.
    \label{eq:twirl-gradient-convolution}
\end{equation}
Then
\begin{align}
    \operatorname{tADB}_{\vH}^{a}(\vsigma)
    &=\int_{-\infty}^{\infty}\int_{-\infty}^{\infty}
      \int_{-1/2}^{1/2}h_s(\omega)\widetilde g_s^{\mathsf{ADB}}(t)
    \times\left\|
       \vsigma^{\frac{1+2s}{4}}
       [\widehat{\vA}_{\vH}^{a}(\omega,t),
          \log\vsigma-\log\vrho_{\vH}]
       \vsigma^{\frac{1-2s}{4}}
     \right\|_2^2
     \dd s\dd\omega\dd t.
    \label{eq:twirl-gradient-identity}
\end{align}
\end{lem}
\begin{proof}
For each $u$, the rotated jump $e^{iu\vH}\vA^ae^{-iu\vH}$ is Hermitian
and has filtered evolution
$\widehat{\vA}_{\vH}^{a}(\omega,t+u)$.
Expanding \eqref{eq:twirl-base-convolution} in
\eqref{eq:twirled-adb-definition} and
applying the entropy-gradient identity of
\cite[Lemma~C.2]{bergamaschi2025metastability} separately to each
rotated jump gives
\begin{align}
    \int_{-\infty}^{\infty}|f_\beta(u)|^2
    \int_{-\infty}^{\infty}\int_{-\infty}^{\infty}
    \int_{-1/2}^{1/2}
     h_s(\omega)g_s^{\mathsf{ADB}}(t)
    \times\left\|
       \vsigma^{\frac{1+2s}{4}}
       [\widehat{\vA}_{\vH}^{a}(\omega,t+u),
          \log\vsigma-\log\vrho_{\vH}]
       \vsigma^{\frac{1-2s}{4}}
     \right\|_2^2
     \dd s\dd\omega\dd t\dd u.
    \label{eq:twirl-average-gradients}
\end{align}
Change variables from $t+u$ to $t$ and integrate over $u$.
The resulting weight is exactly
$\widetilde g_s^{\mathsf{ADB}}(t)$, proving
\eqref{eq:twirl-gradient-identity}.
\end{proof}

The next lemma compares the time weights of the original and twirled ADB and shows that they agree up to constant factors. This is the first step in relating the original ADB to Fisher information.

\begin{lem}[Comparison of the original and convolved time weights]
\label{lem:twirl-base-comparison}
There are {absolute} constants $0<c_-\leq c_+<\infty$ such that
\begin{equation}
    c_-g(t)\leq\widetilde g(t)\leq c_+g(t)
    \qquad\text{for every }t\in\mathbb R.
    \label{eq:twirl-base-comparison}
\end{equation}
\end{lem}
\begin{proof}
The elementary inequality
$\cosh x\leq e^{|x-y|}\cosh y$, used in both directions, gives
\begin{equation}
    e^{-2\pi|u|/\beta}g(t)
    \leq g(t-u)
    \leq e^{2\pi|u|/\beta}g(t).
    \label{eq:twirl-shift-comparison}
\end{equation}
Multiply by $|f_\beta(u)|^2$ and integrate. The exponential factor given by $g$ could be suppressed by the Gaussian, so the integral is finite and could be bounded by a constant. This gives the desired comparison.
\end{proof}
The next lemma bounds the time weight in the entropy-gradient representation of twirled ADB by a constant multiple of the original FI time weight. This completes the comparison with Fisher information.

\begin{lem}[Uniform comparison with the FI time weight]
\label{lem:twirl-fi-weight-comparison}
There is a numerical constant ${C_{\beta}\leq e^{\poly(\beta^{\pm 1})}}$ such that
\begin{equation}
    \widetilde g_s^{\mathsf{ADB}}(t)\leq {C_{\beta}}g_s(t)
    \qquad
    \text{for every }|s|<1/2\text{ and }t\in\mathbb R.
    \label{eq:twirl-uniform-comparison}
\end{equation}
\end{lem}
\begin{proof}
We first reduce the comparison to the single weight
$g_0^{\mathsf{ADB}}$. Put $c=\cos(\pi s)\in(0,1]$.
We have
\begin{equation}
    g_s(t)
    =\frac{c\cosh(2\pi t/\beta)}
      {\beta[\cosh^2(2\pi t/\beta)-1+c^2]}
    \geq\frac{c}{\beta\cosh(2\pi t/\beta)}
    =c\,g(t).
    \label{eq:twirl-fi-lower-envelope}
\end{equation}
The function $x\mapsto\log((1+x)/(1-x))$ is convex on $[0,1)$
and vanishes at zero. Hence
 taking $x=1/\cosh(2\pi t/\beta)$ proves $g_s^{\mathsf{ADB}}(t)\leq c\,g_0^{\mathsf{ADB}}(t)$
for $t\neq0$, and hence almost everywhere. Convolution preserves
this inequality. Together with
\eqref{eq:twirl-fi-lower-envelope}, it follows that
\begin{equation}
    \frac{\widetilde g_s^{\mathsf{ADB}}(t)}{g_s(t)}
    \leq
    \frac{\displaystyle
       \int_{-\infty}^{\infty}
       |f_\beta(t-v)|^2g_0^{\mathsf{ADB}}(v)\dd v}{g(t)}.
    \label{eq:twirl-reduction-to-zero}
\end{equation}

We bound the numerator by separating $|v|\leq\beta$ and
$|v|>\beta$. The possible singularity occurs only in the first part.

For $|v|\leq\beta$, we have
$(t-v)^2\geq t^2/2-v^2\geq t^2/2-\beta^2$, and consequently
\begin{align}
    \int_{-\beta}^{\beta}
      |f_\beta(t-v)|^2g_0^{\mathsf{ADB}}(v)\dd v
    &\leq
      \frac{e^2}{\beta}\sqrt{\frac2\pi}
      e^{-t^2/\beta^2}
      \int_{-\beta}^{\beta}g_0^{\mathsf{ADB}}(v)\dd v
    \leq
      \frac14\sqrt{\frac2\pi}e^{2+\pi^2}g(t).
    \label{eq:twirl-singular-part}
\end{align}

For $|v|>\beta$, the number $x=1/\cosh(2\pi v/\beta)$ is
at most $1/2$. Since
$\log((1+x)/(1-x))\leq4x$ on this interval, we have $
    g_0^{\mathsf{ADB}}(v)\leq\frac2\pi g(v)$ for $|v|>\beta.$
Lemma~\ref{lem:twirl-base-comparison} therefore gives
\begin{align}
    &\int_{-\infty}^{-\beta}
        |f_\beta(t-v)|^2g_0^{\mathsf{ADB}}(v)\dd v
      +\int_{\beta}^{\infty}
        |f_\beta(t-v)|^2g_0^{\mathsf{ADB}}(v)\dd v
      \nonumber\\
    &\leq
      \frac2\pi\int_{-\infty}^{\infty}|f_\beta(t-v)|^2g(v)\dd v
      =\frac2\pi\widetilde g(t)
      \leq\frac{2c_+}{\pi}g(t).
    \label{eq:twirl-tail-part}
\end{align}
Combining \eqref{eq:twirl-singular-part} and
\eqref{eq:twirl-tail-part} proves
\eqref{eq:twirl-uniform-comparison}.
\end{proof}

\begin{proof}[Proof of Theorem~\ref{thm:constant-fi-adb}]
The original and twirled ADB definitions contain the same nonnegative
square-root-defect integrand. Integrating
\eqref{eq:twirl-base-comparison} against that integrand gives
\begin{equation}
    c_-\operatorname{ADB}_{\vH}^{a}(\vsigma)
    \leq\operatorname{tADB}_{\vH}^{a}(\vsigma)
    \leq c_+\operatorname{ADB}_{\vH}^{a}(\vsigma).
    \label{eq:twirl-integrated-equivalence}
\end{equation}
This part applies to every density operator, including singular states.
For full-rank $\vsigma$, use Lemma~\ref{lem:twirl-gradient} and
integrate the pointwise comparison
\eqref{eq:twirl-uniform-comparison} against the nonnegative
entropy-gradient integrand. By \eqref{eq:twirl-original-fi},
\begin{equation}
    \operatorname{tADB}_{\vH}^{a}(\vsigma)
    \leq {C_{\beta}}\mathsf{FI}_{a}^{\vH}[\vsigma\Vert\vrho_{\vH}].
\end{equation}
Combining the two inequalities yields a constant bound between the Fisher information and the ADB norm. At last, the entropy-dissipation identity
\cite[Theorem~B.1]{bergamaschi2025metastability} identifies this
Fisher information with $\mathsf{EP}_{a}^{\vH}[\vsigma]$.
{Enlarging $C_\beta\geq1$ to absorb $c_+$, $c_-^{-1}$, and the preceding comparison factors proves all claims.}
\end{proof}

\begingroup
\makeatletter
\@for\BAppendixEquationLabel:=%
eq:metastability-to-EP,%
eq:adb-to-local-metastability,%
eq:adb3-hamiltonian-result,%
eq:lindbladian-locality,%
eq:lindbladian-locality-split,%
eq:adb4-state-distance,%
eq:adb4-Delta,%
eq:adb4-bipartite-order,%
eq:adb4-gibbs-order,%
eq:adb4-u-definition,%
eq:u-ell,%
eq:adb4-extended-reference%
\do{%
  \expandafter\MT@newlabel\expandafter{\BAppendixEquationLabel}%
  \edef\BAppendixEquationTag{{\BAppendixEquationLabel}}%
  \expandafter\MT@newlabel\expandafter{\BAppendixEquationTag}%
}
\makeatother
\endgroup
\subsection{From dynamical metastability to approximate detailed balance}
\label{subsec:thermal-lindbladian-to-adb}
For a full-rank state, H\"older's inequality gives
\begin{align}
 \mathsf{EP}_{\CL}[\vsigma]
 \leq{\epsilon_{\mathrm{ms}}^{\mathrm{glo}}}(\vsigma)\|\log\vsigma-\log\vrho\|.
 \label{eq:metastability-to-EP}
\end{align}
The same estimate applies to each local component.

Without regularization, \Cref{eq:metastability-to-EP,eq:EP-to-adb} do not give a finite ADB bound when $\vsigma$ has zero eigenvalues, since $\|\log\vsigma\|$ is unbounded. Ref.~\cite[Appendix F, Lemma F.3]{bergamaschi2025metastability} handles this divergence by bounding ADB of a state regularized with the Gibbs state. In contrast, the two theorems below give bounds for ADB of the original state under the stated metastability assumptions, without requiring full rank. In particular, \Cref{thm:local-adb} uses locality and local regularization to obtain a bound independent of the system size, while \Cref{thm:global-adb} gives the corresponding global bound with linear dependence on $n$.
\begin{thm}[Uniform ADB from local metastability]
\label{thm:local-adb}
Let $\vH$ satisfy \Cref{defn:local-Hamiltonian}. For any density operator $\vsigma$, let ${\epsilon_{\mathrm{ms}}^{\mathrm{loc}}}=\max_a\|\mathcal L_a^{\vH}[\vsigma]\|_1$ be the local metastability error from \eqref{eq:local-metastability-basic}. If $0<{\epsilon_{\mathrm{ms}}^{\mathrm{loc}}}\leq1/2$, then
\begin{equation}
    \max_a\operatorname{ADB}_{\vH}^{a}(\vsigma)
    \leq {C_{\beta}}{\epsilon_{\mathrm{ms}}^{\mathrm{loc}}}
          \left[\log\frac e{{\epsilon_{\mathrm{ms}}^{\mathrm{loc}}}}\right]^{D+1}.
    \label{eq:main-adb}
\end{equation}
The constant depends only on $\beta,d,k,D$. {At ${\epsilon_{\mathrm{ms}}^{\mathrm{loc}}}=0$, the right-hand side is defined by its zero limit, and the bound remains valid.}
\end{thm}
For the ADB error under global metastability, we have the following bound.
\begin{thm}[Uniform ADB from global metastability]
\label{thm:global-adb}
Let $\vH$ satisfy \Cref{defn:local-Hamiltonian}. For any state $\vsigma$ with $0<\|\mathcal{L}^{\vH}[\vsigma]\|_1=\epsilon_{\mathrm{ms}}^{\mathrm{glo}}\leq1/2$, the uniform ADB error obeys
\begin{equation}
    \epsilon_{\mathrm{ADB}}
    \leq C_{\beta}\epsilon_{\mathrm{ms}}^{\mathrm{glo}}
    \left[n+\log\frac1{\epsilon_{\mathrm{ms}}^{\mathrm{glo}}}\right].
    \label{eq:global-ms-to-adb}
\end{equation}
The constant $C_{\beta}$ depends only on $\beta$ and the fixed locality parameters.
\end{thm}

The main argument is the proof of \Cref{thm:local-adb}, which we give first. \Cref{lem:state-cont} controls the change in the square root of ADB in terms of the trace-norm distance between density operators. By using this lemma we could bound the ADB error for states that are not even full rank. \Cref{lem:H-cont,lem:lindbladian-locality} then use locality to bound the changes in the ADB functional and the single-jump generator under Hamiltonian truncation. This allows us to relate local metastability to ADB using a Hamiltonian restricted to a neighborhood of the jump. \Cref{lem:adb4-relative-log,lem:adb4-regularized-ADB} regularize the state using the local thermal state tensored with the exterior marginal of the original state, rather than the global Gibbs state used in Ref.~\cite[Appendix F]{bergamaschi2025metastability}. This controls the relative logarithm without dependence on the full system size and yields an ADB bound for the regularized state. Finally, combining these estimates and choosing the truncation radius proves \Cref{thm:local-adb}.

\begin{lem}[Continuity in the input state]
\label{lem:state-cont}
For any density operators $\vsigma,\vxi$ and the fixed jump,
\begin{equation}
    \left|\sqrt{\operatorname{ADB}_{\vH}^{a}(\vsigma)}
          -\sqrt{\operatorname{ADB}_{\vH}^{a}(\vxi)}\right|
    \leq2\|\sqrt{\vsigma}-\sqrt{\vxi}\|_2
    \leq2\sqrt{\|\vsigma-\vxi\|_1}.
    \label{eq:state-cont}
\end{equation}
\end{lem}

\begin{proof}
The operator equation is linear in its square-root input, so the reverse triangle inequality in Hilbert--Schmidt norm gives
\begin{align*}
    &\Bigl|\|\widehat{\vA}^{a}_{\vH}(\omega,t)
                    \sqrt{\vsigma}-\sqrt{\vsigma}e^{\beta\vH/2}
          \widehat{\vA}^{a}_{\vH}(\omega,t)e^{-\beta\vH/2}\|_2
    -\|\widehat{\vA}^{a}_{\vH}(\omega,t)
                    \sqrt{\vxi}-\sqrt{\vxi}e^{\beta\vH/2}
          \widehat{\vA}^{a}_{\vH}(\omega,t)e^{-\beta\vH/2}\|_2\Bigr|
          \\
    &\leq\left\|
          \widehat{\vA}^{a}_{\vH}(\omega,t)
                    (\sqrt{\vsigma}-\sqrt{\vxi})
          -(\sqrt{\vsigma}-\sqrt{\vxi})e^{\beta\vH/2}
          \widehat{\vA}^{a}_{\vH}(\omega,t)e^{-\beta\vH/2}
          \right\|_2.
\end{align*}
{Set $\vB=\sqrt{\vsigma}-\sqrt{\vxi}$. The reverse triangle inequality in the weighted Hilbert space yields}
\begin{align*}
    \left|\sqrt{\operatorname{ADB}_{\vH}^{a}(\vsigma)}
          -\sqrt{\operatorname{ADB}_{\vH}^{a}(\vxi)}\right|
    &\leq{\left[
      \int_{-\infty}^{\infty}\int_{-\infty}^{\infty}
          \gamma(\omega)g(t)\,
          \times\left\|
          \widehat{\vA}_{\vH}^{a}(\omega,t)\vB
          -\vB e^{\beta\vH/2}
          \widehat{\vA}_{\vH}^{a}(\omega,t)e^{-\beta\vH/2}
          \right\|_2^2\dd\omega\dd t
      \right]^{1/2}}.
\end{align*}
This is the weighted norm of the linear defect in $\vB$; it does not require $\vB$ to be positive. For the remaining estimates, abbreviate $\vrho=\vrho_{\vH}$, $\vA=\vA^a$, and $\vA^a(\omega,t)=\widehat{\vA}_{\vH}^{a}(\omega,t)$. At fixed $t$, expand the frequency integral:
\begin{align*}
    \int_{-\infty}^{\infty}
       \lnormp{\vA^a(\omega,t)\vB-\vB\vrho^{-\frac12}
                   \vA^a(\omega,t)\vrho^{\frac12}}{2}^2
       \gamma(\omega)\rd\omega
    \leq2\Bigg(\int_{-\infty}^{\infty}
       \lnormp{\vA^a(\omega,t)\vB}{2}^2
       \gamma(\omega)+
       \lnormp{\vB\vrho^{-\frac12}
       \vA^a(\omega,t)
       \vrho^{\frac12}}{2}^2\gamma(\omega)\rd\omega\Bigg).
\end{align*}
Since $0\leq\gamma(\omega)\leq1$, we can ignore the weight in the first integral. For the first term, we can expand $\|X\|_2^2=\Tr(XX^\dagger)$ and apply Parseval to obtain
\begin{align*}
    &\int_{-\infty}^{\infty}\Tr\Bigg[
       \Big(\frac1{\sqrt{2\pi}}\int_{\mathbb R}
       \vA_{\vH}(u+t)e^{-i\omega u}f_\sigma(u)\dd u\Big)
       \vB\vB^\dagger
       \times\Big(\frac1{\sqrt{2\pi}}\int_{\mathbb R}
       \vA_{\vH}(u'+t)e^{i\omega u^{{\prime}}}
       f_\sigma(u')\dd u'\Big)
       \Bigg]\rd\omega
    \\
    &=\int_{-\infty}^{\infty}|f_\sigma(u)|^2
       \Tr\Big(\vA_{\vH}(u+t)\vB\vB^\dagger
       \vA_{\vH}(u+t)\Big)\dd u
    \\
    &=\int_{-\infty}^{\infty}|f_\sigma(u)|^2
       \|\vA_{\vH}(u+t)\vB\|_2^2\dd u
    \leq\int_{-\infty}^{\infty}|f_\sigma(u)|^2
       \|\vB\|_2^2\dd u
       =\|\vB\|_2^2.
\end{align*}
The last inequality uses the fact that the Gaussian is normalized in $L^2$. For the conjugated contribution we have
\begin{align*}
    \int_{-\infty}^{\infty}
       \lnormp{\vB\vrho^{-\frac12}\vA^a(\omega,t)\vrho^{\frac12}}{2}^2
       \gamma(\omega)\rd\omega
    &=\int_{-\infty}^{\infty}
       \lnormp{\vB\vA^a(\omega+\beta^{-1},t)}{2}^2
       e^{\beta\omega+1/2}\gamma(\omega)\rd\omega
    \\
    &\leq\int_{-\infty}^{\infty}
       \lnormp{\vB\vA^a(\omega+\beta^{-1},t)}{2}^2\rd\omega
       \leq\|\vB\|_2^2.
\end{align*}
{Combining these estimates and using $\int_{-\infty}^{\infty}g(t)\dd t=1/2\leq1$ bounds the weighted squared defect by $4\|\vB\|_2^2$. Taking its square root proves the first inequality in \Cref{lem:state-cont}.}

{For the second inequality, apply the Powers--St{\o}rmer inequality}~\cite{powers1970free},
which holds for arbitrary positive semidefinite operators:
\begin{equation}
    \|\sqrt{\vsigma}-\sqrt{\vxi}\|_2^2
    =\tr\!\left(\vsigma+\vxi
       -2\sqrt{\vsigma}\sqrt{\vxi}\right)
    \leq\tr|\vsigma-\vxi|
    =\|\vsigma-\vxi\|_1.
\end{equation}
Neither commutativity nor a full-rank assumption is required.
\end{proof}

Both the single-jump thermalization generator and the corresponding ADB functional are quasi-local, suggesting a bound independent of the system size. We make this intuition precise using spatial truncation. To prove \Cref{thm:local-adb}, we use the following lemmas. For each single-site Pauli jump $\vA^a$ and integer radius $\ell\geq1$, let $\vH_\ell$ retain the Hamiltonian terms at interaction-graph distance less than $\ell-1$ from the jump site, using the convention in Appendix~\ref{sec:basic-notation}. Let $S_a(\ell)$ contain that site and the full support of every retained term. We first compare the ADB functionals constructed from $\vH$ and $\vH_\ell$.

\begin{lem}[Centered Hamiltonian comparison]
\label{lem:H-cont}
There exists $C_{\beta}\leq e^{\poly(\beta^{\pm 1})}$ such that for any density operator $\vsigma$,
\begin{align}
    \left|\sqrt{\operatorname{ADB}_{\vH}^{a}(\vsigma)}
          -\sqrt{\operatorname{ADB}_{\vH_{{\ell}}}^{a}(\vsigma)}\right|^2
    &\leq {C_{\beta}}
          \left(4^{-\ell}+e^{-\pi\ell/(4ed\beta)}\right).
    \label{eq:adb3-hamiltonian-result}
\end{align}
\end{lem}
\begin{proof}
The difference between the two operator defects is
\begin{align}
    &(\widehat{\vA}^{a}_{\vH}(\omega,t)\sqrt{\vsigma}
       -\sqrt{\vsigma}e^{\beta\vH/2}
           \widehat{\vA}^{a}_{\vH}(\omega,t)e^{-\beta\vH/2})
       -(\widehat{\vA}^{a}_{\vH_{{\ell}}}(\omega,t)\sqrt{\vsigma}
       -\sqrt{\vsigma}e^{\beta\vH_{{\ell}}/2}
           \widehat{\vA}^{a}_{\vH_{{\ell}}}(\omega,t)e^{-\beta\vH_{{\ell}}/2})
       \\
    &=\left(\widehat{\vA}^{a}_{\vH}(\omega,t)
             -\widehat{\vA}^{a}_{\vH_{{\ell}}}(\omega,t)\right)\sqrt{\vsigma}
    -e^{\beta\omega/2+1/4}\sqrt{\vsigma}
       \left(\widehat{\vA}^{a}_{\vH}(\omega+\beta^{-1},t)
             -\widehat{\vA}^{a}_{\vH_{{\ell}}}(\omega+\beta^{-1},t)\right).
    \label{eq:adb3-explicit-hamiltonian-difference}
\end{align}
For the first term, $\|\sqrt{\vsigma}\|_2=1$ {gives}
\begin{align}
    \int_{\mathbb R}\gamma(\omega)
       \left\|\left(\widehat{\vA}^{a}_{\vH}(\omega,t)
              -\widehat{\vA}^{a}_{\vH_{{\ell}}}(\omega,t)\right)
              \sqrt{\vsigma}\right\|_2^2\dd\omega
    &\leq\int_{\mathbb R}|f_\beta(u)|^2
       \left\|\left(\vA^a_{\vH}(t+u)-\vA^a_{\vH_{{\ell}}}(t+u)\right)
                   \sqrt{\vsigma}\right\|_2^2\dd u\nonumber\\
    &\leq\int_{\mathbb R}|f_\beta(u)|^2
       \|\vA^a_{\vH}(t+u)-\vA^a_{\vH_{{\ell}}}(t+u)\|^2\dd u.
    \label{eq:adb3-difference-parseval}
\end{align}
The second term satisfies the same bound after shifting $\omega\mapsto\omega+\beta^{-1}$, since $\gamma(\omega)e^{\beta\omega+1/2}\leq1$. The reverse triangle inequality and \eqref{eq:adb3-difference-parseval} therefore bound the left-hand side of \eqref{eq:adb3-hamiltonian-result} by
\[
    4\int_{\mathbb R}\widetilde g(t)
    \|\vA^a_{\vH}(t)-\vA^a_{\vH_\ell}(t)\|^2\dd t.
\]
By \Cref{lem:basic-lr}, the squared norm is at most $C_{\beta}4^{-\ell}$ for $|t|\leq\ell/(4ed)$, and at most $4$ otherwise. Since \Cref{lem:twirl-base-comparison} gives $\widetilde g(t)\lesssim g(t)\lesssim\beta^{-1}e^{-\pi|t|/\beta}$, splitting the integral at this cutoff proves \eqref{eq:adb3-hamiltonian-result}.
\end{proof}
\Cref{lem:H-cont} transfers a bound for the localized ADB functional to the one constructed from $\vH$. We localize the Hamiltonian entering the central single-jump generator while retaining the full state and its correlations with the exterior. The next lemma compares the single-jump metastability errors for $\vH$ and $\vH_\ell$, evaluated on the full state.

\begin{lem}[Locality of the single-jump thermal generator]
\label{lem:lindbladian-locality}
Let
$\mathcal L_{{a}}^{{\vH}}$ and $\mathcal L_{{a}}^{{\vH_\ell}}$ be the
corresponding Gaussian, shifted-Metropolis KMS-detailed-balanced
generators, with width $\beta^{-1}$ and the same unit jump rate.
Both generators act on the full system. There exists a constant $C_{\beta}\leq e^{\poly(\beta^{\pm 1})}$ such that
for {every integer} ${\ell}\geq\max\{1,4e^2d\beta\}$ and every density operator
$\vsigma$,
\begin{align}
    \left|
    \|\mathcal L_{{a}}^{{\vH}}(\vsigma)\|_1
    -\|\mathcal L_{{a}}^{{\vH_\ell}}(\vsigma)\|_1
    \right|
    &\leq
    \|(\mathcal L_{{a}}^{{\vH}}
       -\mathcal L_{{a}}^{{\vH_\ell}})(\vsigma)\|_1
    \leq {C_{\beta}}\left(2^{-{\ell}}+e^{-c{\ell}/(d\beta)}\right).
    \label{eq:lindbladian-locality}
\end{align}
The constants are independent of $n,{\ell}$, and $\vsigma$.
\end{lem}

\begin{proof}
Let $\vC_{\vK}^{a}$ denote the coherent correction in
$\mathcal L_{{a}}^{{\vK}}$, for $\vK=\vH,\vH_{{\ell}}$.
Expanding the quadratic jump terms {and} applying Cauchy--Schwarz in frequency gives
\begin{align}
    \|(\mathcal L_{{a}}^{{\vH}}
          -\mathcal L_{{a}}^{{\vH_\ell}})(\vsigma)\|_1
    &\leq
    4\left(
      \int_{-\infty}^{\infty}|f_\beta(t)|^2
      \|\vA_{\vH}^{a}(t)-\vA_{\vH_{{\ell}}}^{a}(t)\|^2
      \dd t
    \right)^{1/2}
    +2\|\vC_{\vH}^{a}-\vC_{\vH_{{\ell}}}^{a}\|.
    \label{eq:lindbladian-locality-split}
\end{align}
Here the dissipative estimate uses only
$0\leq\gamma\leq1$ and $\int_{-\infty}^{\infty}|f_\beta(t)|^2
\dd t=1$; it is the quadratic-jump estimate in
\cite[Lemma~A.2]{chen2025GibbsMarkov}. Parseval is applied to the
Fourier-transformed vectors, not to operator norms.

Split the integral in \eqref{eq:lindbladian-locality-split} at $|t|=\ell/(4ed)$. By \Cref{lem:basic-lr}, the Gaussian tail, and $\ell\geq4e^2d\beta$, its square root is at most $C_{\beta}(2^{-\ell}+e^{-\ell/(4d\beta)})$. The coherent correction is bounded by $C_{\beta}(2^{-\ell}+e^{-c\ell/(d\beta)})$ by Ref.~\cite[Corollary~A.2]{chen2025GibbsMarkov}. Substituting these bounds into \eqref{eq:lindbladian-locality-split} and decreasing $c$ if necessary proves the second inequality in \eqref{eq:lindbladian-locality}; the first is the reverse triangle inequality.
\end{proof}

{It remains to bound the localized ADB functional. We use a local version of the regularization argument in Ref.~\cite{bergamaschi2025metastability}. Unlike the global regularization, we use the local reference $\vtau_a=\vrho_S\otimes\vsigma_E$.}

\begin{lem}[Local regularization and its relative logarithm]
\label{lem:adb4-relative-log}
Fix a jump $a$, a centered patch $S=S_a(\ell)$, and $E=\Lambda\setminus S$.
Write $m=|S|$. Define
\begin{equation}
    \vrho_S:=\frac{e^{-\beta\vH_{{\ell}}}}
                        {\tr_S(e^{-\beta\vH_{{\ell}}})},
    \qquad\vsigma_E:=\tr_S\vsigma,
    \qquad \vtau_a:=\vrho_S\otimes\vsigma_E,
    \qquad\vsigma_{a,r}:=(1-r)\vsigma+r\vtau_a,
    \label{eq:adb4-regularization}
\end{equation}
for $0<r\leq1/2$. On
$\mathcal H_S\otimes\operatorname{supp}(\vsigma_E)$, both
$\vtau_a$ and $\vsigma_{a,r}$ are strictly positive, and
\begin{align}
    \|\vsigma_{a,r}-\vsigma\|_1&\leq2r,
    \label{eq:adb4-state-distance}\\
    \|\log\vsigma_{a,r}-\log\vtau_a\|
    &\leq\Delta_{\ell,r}
       :=1+\log\frac1r+2m\log2+2\beta\|\vH_{{\ell}}\|.
    \label{eq:adb4-Delta}
\end{align}
In particular, $\Delta_{\ell,r}$ does not involve the dimension of $E$
or the smallest nonzero eigenvalue of $\vsigma_E$.
\end{lem}
\begin{proof}
Choose an orthonormal basis $\{|i\rangle\}_{i=1}^{2^m}$ of $S$ and write
$\vsigma=\sum_{i,j}|i\rangle\langle j|\otimes\vsigma_{ij}$.
Positivity and Cauchy--Schwarz imply the inequality
\begin{equation}
    \vsigma
    \leq2^m\sum_i|i\rangle\langle i|\otimes\vsigma_{ii}
    \leq2^m\vI_S\otimes\vsigma_E.
    \label{eq:adb4-bipartite-order}
\end{equation}

The second inequality in \eqref{eq:adb4-bipartite-order} uses
$0\leq\vsigma_{ii}\leq\sum_j\vsigma_{jj}=\vsigma_E$.
In particular, \eqref{eq:adb4-bipartite-order} shows that
$\operatorname{supp}(\vsigma)$ is contained in
$\mathcal H_S\otimes\operatorname{supp}(\vsigma_E)$.
All subsequent logarithms are evaluated on this subspace.

The finite-region Gibbs state obeys
\begin{equation}
    e^{-\beta\vH_{{\ell}}}\geq e^{-\beta\|\vH_{{\ell}}\|}\vI_S,
    \qquad \tr_S(e^{-\beta\vH_{{\ell}}})\leq2^m e^{\beta\|\vH_{{\ell}}\|},
    \qquad \vrho_S\geq2^{-m}e^{-2\beta\|\vH_{{\ell}}\|}\vI_S.
    \label{eq:adb4-gibbs-order}
\end{equation}
Combining \eqref{eq:adb4-bipartite-order} and \eqref{eq:adb4-gibbs-order}
yields
\begin{equation}
    \vsigma\leq2^{2m}e^{2\beta\|\vH_{{\ell}}\|}\vtau_a,
    \qquad
    r\vtau_a\leq\vsigma_{a,r}
         \leq2^{2m}e^{2\beta\|\vH_{{\ell}}\|}\vtau_a.
    \label{eq:adb4-order-sandwich}
\end{equation}
The last upper bound uses
$(1-r)2^{2m}e^{2\beta\|\vH_{{\ell}}\|}+r\leq
2^{2m}e^{2\beta\|\vH_{{\ell}}\|}$.
Operator monotonicity of the logarithm therefore gives
\begin{equation}
    (\log r)\vI
    \leq\log\vsigma_{a,r}-\log\vtau_a
    \leq\bigl(2m\log2+2\beta\|\vH_{{\ell}}\|\bigr)\vI.
    \label{eq:adb4-log-sandwich}
\end{equation}
The operator monotonicity of the logarithm follows from its
integral representation
\[
    \log\vA
    =
    \int_0^\infty
    \left(\frac{I}{1+t}-(\vA+tI)^{-1}\right)\dd t,
    \qquad \vA>0,
\]
which is obtained by applying the scalar identity to each
eigenvalue of $\vA$. {Since} inversion reverses the operator order, for
$0<\vA\leq\vB$ and every $t\geq0$ we have
\[
    \vA+tI\leq\vB+tI
    \quad\Longrightarrow\quad
    (\vA+tI)^{-1}\geq(\vB+tI)^{-1}.
\]
Therefore, the integral representation of the logarithm gives
\[
    \log\vB-\log\vA
    =
    \int_0^\infty
    \bigl[(\vA+tI)^{-1}-(\vB+tI)^{-1}\bigr]\dd t
    \geq0.
\]
Thus $\vA\leq\vB$ implies $\log\vA\leq\log\vB$,
without any commutativity assumption.

This proves \eqref{eq:adb4-Delta} without assuming that the two states
commute. Finally,
$\|\vsigma_{a,r}-\vsigma\|_1=r\|\vtau_a-\vsigma\|_1\leq2r$.
\end{proof}

Combining \Cref{lem:lindbladian-locality,lem:adb4-relative-log} yields the following ADB bound for the locally regularized state $\vsigma_{a,r}$:

\begin{lem}[ADB of the locally regularized state]
\label{lem:adb4-regularized-ADB}
Use the setting of Lemma~\ref{lem:adb4-relative-log}, and assume
{that $\ell$ is an integer with} $\ell\geq{\max\{1,}4e^2\beta d{\}}$. There are constants
${C_{\beta}}\leq e^{\poly(\beta^{\pm 1})},c_{\mathrm L}>0$, independent of $n$, such that,
with
\begin{equation}
    u_\ell:={\epsilon_{\mathrm{ms}}^{\mathrm{loc}}}
       +{C_{\beta}}
          \left(e^{-c_{\mathrm L}\ell/(d\beta)}+2^{-\ell}\right),
    \label{eq:adb4-u-definition}
\end{equation}
one has
\begin{equation}
    \|\mathcal L_a^{\vH_{{\ell}}}(\vsigma)\|_1\leq u_\ell,
    \qquad
    \|\mathcal L_a^{\vH_{{\ell}}}(\vsigma_{a,r})\|_1\leq u_\ell,
    \label{eq:u-ell}
\end{equation}
and
\begin{equation}
    \operatorname{ADB}_{\vH_{{\ell}}}^{a}(\vsigma_{a,r})
    \leq {C_{\beta}}u_\ell\Delta_{\ell,r}.
    \label{eq:adb4-regularized-bound}
\end{equation}
The generator and ADB functional act on the full system. Only the
metastability error of the central jump $a$ is used.
\end{lem}
\begin{proof}

By \Cref{lem:lindbladian-locality}, we have
\begin{align}
    \|\mathcal L_a^{\vH_{{\ell}}}(\vsigma)\|_1
    &\leq\|\mathcal L_a^{\vH}(\vsigma)\|_1
        +\|(\mathcal L_a^{\vH_{{\ell}}}-\mathcal L_a^{\vH})(\vsigma)\|_1
     \leq u_\ell.
    \label{eq:adb4-centered-metastability-error}
\end{align}
Because $\mathcal L_a^{\vH_{{\ell}}}$ acts on $S$ and fixes $\vrho_S$,
\begin{equation}
    \mathcal L_a^{\vH_{{\ell}}}(\vtau_a)=0,
    \qquad
    \mathcal L_a^{\vH_{{\ell}}}(\vsigma_{a,r})
       =(1-r)\mathcal L_a^{\vH_{{\ell}}}(\vsigma).
    \label{eq:adb4-regularized-metastability-error}
\end{equation}
This proves \eqref{eq:u-ell}.

The reference for entropy production must be $\vtau_a$, not the global
Gibbs state. On the support from Lemma~\ref{lem:adb4-relative-log},
\begin{equation}
    \vtau_a=
    \frac{
      e^{-\beta(\vH_{{\ell}}\otimes\vI_E
                  -\beta^{-1}\vI_S\otimes\log\vsigma_E)}}
      {\tr\!\left[e^{-\beta(\vH_{{\ell}}\otimes\vI_E
                  -\beta^{-1}\vI_S\otimes\log\vsigma_E)}\right]}.
    \label{eq:adb4-extended-reference}
\end{equation}
The two terms in this Hamiltonian act on distinct tensor factors, and
$\vA^a$ is supported on $S$. Hence
\begin{align}
    e^{it(\vH_{{\ell}}\otimes\vI_E-\beta^{-1}\vI_S\otimes\log\vsigma_E)}
       (\vA^a\otimes\vI_E)
     e^{-it(\vH_{{\ell}}\otimes\vI_E-\beta^{-1}\vI_S\otimes\log\vsigma_E)}
       &=\bigl(e^{it\vH_{{\ell}}}\vA^ae^{-it\vH_{{\ell}}}\bigr)\otimes\vI_E.
    \label{eq:adb4-exterior-cancellation-time}
\end{align}
Thus rebuilding the Gaussian jumps and their KMS coherent correction
from \eqref{eq:adb4-extended-reference} gives exactly
$\mathcal L_a^{\vH_{{\ell}}}$. Similarly, for every operator $\vX$ on $S$,
\begin{align}
    \vtau_a^{-1/2}(\vX\otimes\vI_E)\vtau_a^{1/2}
    &=(\vrho_S^{-1/2}\vX\vrho_S^{1/2})\otimes\vI_E
     =(e^{\beta\vH_{{\ell}}/2}\vX e^{-\beta\vH_{{\ell}}/2})\otimes\vI_E.
    \label{eq:adb4-exterior-cancellation-imaginary}
\end{align}
It follows that the ADB functional with reference $\vtau_a$ is exactly
$\operatorname{ADB}_{\vH_{{\ell}}}^{a}$, still evaluated on $\vsigma_{a,r}$.

Define the scalar entropy production
\begin{equation}
    F_{a,S}:=-\tr\!\left[
       \mathcal L_a^{\vH_{{\ell}}}(\vsigma_{a,r})
       (\log\vsigma_{a,r}-\log\vtau_a)\right].
    \label{eq:adb4-EP}
\end{equation}
It is nonnegative because $\mathcal L_a^{\vH_{{\ell}}}$ preserves $\vtau_a$
and its semigroup decreases relative entropy to that reference.
By \eqref{eq:u-ell}, Lemma~\ref{lem:adb4-relative-log}, and H\"older's
inequality,
\begin{equation}
    0\leq F_{a,S}\leq u_\ell\Delta_{\ell,r}.
    \label{eq:adb4-EP-bound}
\end{equation}
{\Cref{thm:constant-fi-adb} gives}
\begin{equation}
    \operatorname{ADB}_{\vH_{{\ell}}}^{a}(\vsigma_{a,r})
    \leq {C_{\beta}}F_{a,S}.
    \label{eq:adb4-imported-EP}
\end{equation}
Its right-hand side is zero when $F_{a,S}=0$.
This input uses only the relative logarithm; its constant is independent
of the dimension and does not require locality of the reference
Hamiltonian.
\end{proof}

Combining \Cref{lem:adb4-regularized-ADB,lem:H-cont,lem:state-cont} transfers the ADB bound from the locally regularized state back to the original state and proves \Cref{thm:local-adb}.
\begin{proof}[{Proof of Theorem~\ref{thm:local-adb}}]
Insert the intermediate quantity
$\sqrt{\operatorname{ADB}_{\vH_{{\ell}}}^{a}(\vsigma)}$ and use
Lemmas~\ref{lem:H-cont} and \ref{lem:state-cont}:
\begin{align}
    \sqrt{\operatorname{ADB}_{\vH}^{a}(\vsigma)}
    &\leq\sqrt{\operatorname{ADB}_{\vH_{{\ell}}}^{a}(\vsigma)}
        +\left[{C_{\beta}}
            \left(4^{-\ell}+e^{-\pi\ell/(4ed\beta)}\right)\right]^{1/2}
        \nonumber\\
    &\leq\sqrt{\operatorname{ADB}_{\vH_{{\ell}}}^{a}(\vsigma_{a,r})}
        +{2}\sqrt{\|\vsigma-\vsigma_{a,r}\|_1}
        +\left[{C_{\beta}}
            \left(4^{-\ell}+e^{-\pi\ell/(4ed\beta)}\right)\right]^{1/2}.
    \label{eq:adb4-return-calculation}
\end{align}
{Squaring with $(x+y+z)^2\leq3x^2+3y^2+3z^2$, and applying \eqref{eq:adb4-state-distance} and \Cref{lem:adb4-regularized-ADB}, gives}
\begin{equation}
    {\operatorname{ADB}_{\vH}^{a}(\vsigma)
    \leq {C_{\beta}}\left[u_\ell\Delta_{\ell,r}+r+4^{-\ell}
                     +e^{-\pi\ell/(4ed\beta)}\right].}
\end{equation}
{For $0<{\epsilon_{\mathrm{ms}}^{\mathrm{loc}}}\leq1/2$, choose}
\begin{equation}
    {r={\epsilon_{\mathrm{ms}}^{\mathrm{loc}}},\qquad
    \ell=\left\lceil {C_{\beta}}
                     \log\frac e{{\epsilon_{\mathrm{ms}}^{\mathrm{loc}}}}\right\rceil.}
\end{equation}
Then $u_\ell=O({\epsilon_{\mathrm{ms}}^{\mathrm{loc}}})$.
{Thus \eqref{eq:adb4-Delta} implies}
\begin{equation}
    {\Delta_{\ell,r}
    \leq {C_{\beta}}\left[1+\log\frac1r+(1+\ell)^D\right]
    \leq {C_{\beta}}\left[\log\frac e{{\epsilon_{\mathrm{ms}}^{\mathrm{loc}}}}\right]^{D+1}.}
\end{equation}
Substitution yields \Cref{thm:local-adb}.
\end{proof}

We now prove \Cref{thm:global-adb} using Gibbs-state regularization and the continuity estimate in \Cref{lem:state-cont}.

\begingroup
\begin{proof}[Proof of Theorem~\ref{thm:global-adb}]
Following the Gibbs-state regularization in Ref.~\cite[Appendix F]{bergamaschi2025metastability}, take $0<r\leq1/2$ and set
\begin{equation*}
    \vsigma_r=(1-r)\vsigma+r\vrho,
    \qquad \|\vsigma_r-\vsigma\|_1\leq2r.
\end{equation*}
Since $\CL[\vrho]=0$, we have $\|\CL[\vsigma_r]\|_1=(1-r)\epsilon_{\mathrm{ms}}^{\mathrm{glo}}$. Operator monotonicity of the logarithm, applied to $r\vrho\leq\vsigma_r\leq\vI$, and the bound $\|\vH\|\leq dn$ give
\begin{equation*}
    \|\log\vsigma_r-\log\vrho\|
    \leq\log\frac1r+n\log2+2\beta\|\vH\|
    \leq C_{\beta}\left[n+\log\frac1r\right].
\end{equation*}
By \Cref{thm:constant-fi-adb}, the nonnegativity of each local entropy production rate, and H\"older's inequality,
\begin{align*}
    \mathsf{ADB}_a[\vsigma_r]
    \leq C_{\beta}\mathsf{EP}_a[\vsigma_r]
    \leq C_{\beta}\mathsf{EP}_{\CL}[\vsigma_r]\leq C_{\beta}\epsilon_{\mathrm{ms}}^{\mathrm{glo}}
    \left[n+\log\frac1r\right].
\end{align*}
Finally, \Cref{lem:state-cont} and $(x+y)^2\leq2x^2+2y^2$ yield
\begin{equation*}
    \mathsf{ADB}_a[\vsigma]
    \leq2\mathsf{ADB}_a[\vsigma_r]+8\|\vsigma-\vsigma_r\|_1
    \leq2\mathsf{ADB}_a[\vsigma_r]+16r.
\end{equation*}
Choose $r=\epsilon_{\mathrm{ms}}^{\mathrm{glo}}$ and take the maximum over $a$.
\end{proof}
\endgroup

\begingroup
\makeatletter
\@for\BAppendixEquationLabel:=%
eq:adb-implies-approximate-commutation,%
eq:low-frequency-part-commutation,%
eq:filtered-Hamiltonian-scalar,%
eq:filtered-Hamiltonian-local-sum,%
eq:sqrt-state-commutator-bound%
\do{%
  \expandafter\MT@newlabel\expandafter{\BAppendixEquationLabel}%
  \edef\BAppendixEquationTag{{\BAppendixEquationLabel}}%
  \expandafter\MT@newlabel\expandafter{\BAppendixEquationTag}%
}
\makeatother
\endgroup
\subsection{From ADB to Hamiltonian stationarity}
\label{subsec:adb-to-hamiltonian-stationarity}
In \Cref{subsec:thermal-lindbladian-to-adb}, we proved that both local and global metastability imply ADB. We next show that ADB controls the Hamiltonian commutator. Together, these results show that a state with sufficiently small dynamical metastability error approximately commutes with the Hamiltonian. The proof uses the local-product estimate in \Cref{lem:low_freq_bound}, proved in the following subsection and reused in the learning analysis.
\begin{thm}[From approximate detailed balance to approximate commutation]
\label{thm:approximate-commute-deduce}
Let $\vH$ be a local Hamiltonian as in \Cref{defn:local-Hamiltonian}, with Gibbs state $\vrho$, and suppose that the uniform single-site ADB error is bounded by
$\epsilon_{\mathrm{ADB}}$ as in
\eqref{eq:intro-uniform-ADB-error}. Then
\begin{align}
    \epsilon_H
    :=\norm{[\vH,\vsigma]}_1
    \leq
    n\,e^{\poly(\beta^{\pm 1})}
    \epsilon_{\mathrm{ADB}}^{\frac{1}{32\beta d+2}}.
    \label{eq:adb-implies-approximate-commutation}
\end{align}
\end{thm}

\begin{proof}
The proof uses the fact that the Hamiltonian commutes with its Gibbs state $\vrho$ and is invariant under its own real-time evolution.
We use the same Gaussian Fourier convention as in
\eqref{eq:operator-fourier-transform-basic}, {with scalar width $\sigma=\beta^{-1}$}. For clarity, define the
low-frequency component of an operator $\vB$ by
\begin{align}
    \vB_{\leq\Delta}
    :=
    \int_{-\Delta}^{\Delta}
    \widehat{\vB}_{\vH}(\omega)\,\dd\omega.
    \label{eq:low-frequency-part-commutation}
\end{align}

Since $\vH$ is invariant under its own Heisenberg evolution, \(\vH_{\vH}(t)=\vH\). It follows from the Gaussian Fourier transform that
\begin{align}
    \widehat{\vH}_{\vH}(\omega)
    =
    \widehat f_{{\sigma}}(\omega)\vH =
    \frac{e^{-\omega^2/(4{\sigma}^2)}}
    {\sqrt{{\sigma}\sqrt{2\pi}}}\vH.
\end{align}
{Consequently, with a positive scalar factor,}
\begin{align}
    \vH_{\leq\Delta}
    ={\sqrt{2{\sigma}\sqrt{2\pi}}\,
    \erf\!\left(\frac{\Delta}{2{\sigma}}\right)}\vH.
    \label{eq:filtered-Hamiltonian-scalar}
\end{align}

We fix \(\Delta_0:={\sigma}=\frac{1}{\beta}\). Then $[\sqrt{2{\sigma}\sqrt{2\pi}}\,\erf(1/2)]^{-1}=O(\sqrt{\beta})$.
By linearity of the Fourier filter,
\begin{align}
    \vH_{\leq\Delta_0}
    =
    \sum_{\gamma\in\Gamma}h_\gamma
    (\vP_\gamma)_{\leq\Delta_0}.
    \label{eq:filtered-Hamiltonian-local-sum}
\end{align}

For each nonidentity Pauli string $\vP_\gamma$, apply
\Cref{lem:low_freq_bound} to its single-site Pauli factors. Since
$\vP_\gamma$ has weight at most $k$ and every one of its single-site factors
has ADB defect at most $\epsilon_{\mathrm{ADB}}$, we obtain
\begin{align}
    \Big\|
    \sqrt{\vsigma}(\vP_\gamma)_{\leq\Delta_0}
    -
    \sqrt{\vrho}(\vP_\gamma)_{\leq\Delta_0}
    \sqrt{\vrho^{-1}}\sqrt{\vsigma}
    \Big\|_2
    \leq e^{\poly(\beta^{\pm 1})}
    \epsilon_{\mathrm{ADB}}^{\frac{1}{32\beta d+2}}.
    \label{eq:local-Pauli-filtered-adb-bound}
\end{align}
An identity term, if present, contributes zero and may be omitted.

Because $\vrho$ is a function of $\vH$, the two operators commute and $\sqrt{\vrho}\vH\sqrt{\vrho^{-1}}=\vH$.
Using \eqref{eq:filtered-Hamiltonian-scalar},
\eqref{eq:filtered-Hamiltonian-local-sum}, and the triangle inequality, we
obtain
\begin{align}
    &\sqrt{2{\sigma}\sqrt{2\pi}}\,\erf(1/2)
    \norm{\sqrt{\vsigma}\vH-\vH\sqrt{\vsigma}}_2
    =\Big\|
    \sqrt{\vsigma}\vH_{\leq\Delta_0}
    -
    \sqrt{\vrho}\vH_{\leq\Delta_0}
    \sqrt{\vrho^{-1}}\sqrt{\vsigma}
    \Big\|_2 \\
    &\quad \leq
    \sum_{\gamma\in\Gamma}|h_\gamma|
    \Big\|
    \sqrt{\vsigma}(\vP_\gamma)_{\leq\Delta_0}
    -
    \sqrt{\vrho}(\vP_\gamma)_{\leq\Delta_0}
    \sqrt{\vrho^{-1}}\sqrt{\vsigma}
    \Big\|_2 \nonumber\\
    &\quad\leq \left(\sum_{\gamma\in\Gamma}|h_\gamma|\right)
    e^{\poly(\beta^{\pm 1})}
    \epsilon_{\mathrm{ADB}}^{\frac{1}{32\beta d+2}}.
    \label{eq:sqrt-state-commutator-bound}
\end{align}

It remains to convert the Hilbert--Schmidt commutator with
$\sqrt{\vsigma}$ into the trace-norm commutator with $\vsigma$. We have
\begin{align}
    [\vsigma,\vH]
    &=
    \sqrt{\vsigma}
    \bigl(\sqrt{\vsigma}\vH-\vH\sqrt{\vsigma}\bigr) +\bigl(\sqrt{\vsigma}\vH-\vH\sqrt{\vsigma}\bigr)
    \sqrt{\vsigma}.
\end{align}
Hence H\"older's inequality for Schatten norms gives
\begin{align}
    \norm{[\vH,\vsigma]}_1
    &=\norm{[\vsigma,\vH]}_1 \nonumber\\
    &\leq
    2\norm{\sqrt{\vsigma}}_2
    \norm{\sqrt{\vsigma}\vH-\vH\sqrt{\vsigma}}_2 \nonumber\\
    &=2\norm{\sqrt{\vsigma}\vH-\vH\sqrt{\vsigma}}_2,
    \label{eq:trace-commutator-from-sqrt-commutator}
\end{align}
where \(\norm{\sqrt{\vsigma}}_2^2=\Tr(\vsigma)=1\).
Finally, the assumption $|h_\gamma|{\leq}1$ implies \(\sum_{\gamma\in\Gamma}|h_\gamma| \leq|\Gamma|\). Bounded interaction degree implies \(|\Gamma| \leq c_{\mathrm{int}}(d,k)n\) for a system-size-independent constant $c_{\mathrm{int}}(d,k)$. Substituting
these estimates into \eqref{eq:sqrt-state-commutator-bound} and absorbing the scalar prefactor into $e^{\poly(\beta^{\pm 1})}$ proves
\eqref{eq:adb-implies-approximate-commutation}.
\end{proof}

Combining \Cref{thm:local-adb,thm:approximate-commute-deduce} gives a commutator bound linear in $n$ when the metastable error is given for every local thermal generator. \Cref{sec:lower-bounds} shows that this linear dependence on the system size is optimal for the global stationarity estimate.

\begingroup
\makeatletter
\@for\BAppendixEquationLabel:=%
eq:product-raw-cutoff-time,%
eq:product-outer-cutoff-bound,%
eq:product-high-low-balance%
\do{%
  \expandafter\MT@newlabel\expandafter{\BAppendixEquationLabel}%
  \edef\BAppendixEquationTag{{\BAppendixEquationLabel}}%
  \expandafter\MT@newlabel\expandafter{\BAppendixEquationTag}%
}
\makeatother
\endgroup
\subsection{Approximate detailed balance for time-evolved local products}\label{sec:adb-appendix}
The {ADB condition} in \Cref{defn:adb} supplies the following estimates for time-evolved local observables and low-frequency products. They provide the local-product estimate used in \Cref{thm:approximate-commute-deduce}, and will also be used for ADB-based linearization and local faithfulness.

\begin{lem}[ADB bound for a time-evolved operator]\label{lem:time-evolved-adb}
{Let $\vH$ satisfy \Cref{defn:local-Hamiltonian}. For each single-site Pauli jump $\vA^a$ and every real $t$,}
\begingroup
\begin{align}
    \int_{-\infty}^{\infty}
    \lnormp{\vA^a(\omega,t)\sqrt{\vsigma}
      -\sqrt{\vsigma}\vrho^{-\frac12}\vA^a(\omega,t)\vrho^{\frac12}}{2}^2
      \gamma(\omega){\rd\omega} &\lesssim e^{2\pi|t|/\beta}\mathsf{ADB}_a[\vsigma].
\end{align}
\endgroup
\end{lem}
\begin{proof}
Compare the shifted and unshifted ADB integrals and use $g(t)=1/[\beta\cosh(2\pi t/\beta)]$.
By definition,
\begingroup
\begin{align}
    \mathsf{ADB}_{[{\vA^a_{\vH}(t)}]}[\vsigma]
    &=\iint_{-\infty}^{\infty}\gamma(\omega)g(t') \times\lnormp{\vA^a(\omega,t'+t)\sqrt{\vsigma}
      -\sqrt{\vsigma}\vrho^{-\frac12}\vA^a(\omega,t'+t)\vrho^{\frac12}}{2}^2
      \rd\omega\rd t'.
\end{align}
\endgroup
After the change of variables $t'\mapsto t'-t$,
\begingroup
\begin{align}
    &\iint_{-\infty}^{\infty}\gamma(\omega)g(t'-t)
      \lnormp{\vA^a(\omega,t')\sqrt{\vsigma}
      -\sqrt{\vsigma}\vrho^{-\frac12}\vA^a(\omega,t')\vrho^{\frac12}}{2}^2
      \rd\omega\rd t' \nonumber\\
    &\leq e^{2\pi|t|/\beta}\iint_{-\infty}^{\infty}\gamma(\omega)g(t')
      \lnormp{\vA^a(\omega,t')\sqrt{\vsigma}
      -\sqrt{\vsigma}\vrho^{-\frac12}\vA^a(\omega,t')\vrho^{\frac12}}{2}^2
      \rd\omega\rd t' =e^{2\pi|t|/\beta}\mathsf{ADB}_a[\vsigma].
\end{align}
\endgroup
{Apply the time-independent Gaussian-ADB estimate from Ref.~\cite[Lemma C.1]{bergamaschi2025metastability}, recalled in \Cref{defn:adb}, to the Hermitian jump $\vA^a_{\vH}(t)$. This proves the claim.}
\end{proof}

We next bound the low-frequency component of a local product of Pauli operators.
\begin{lem}\label{lem:low_freq_bound}
{Let $\vH$ satisfy \Cref{defn:local-Hamiltonian}, let $\Delta>0$, and write $\vB=\prod_{i=1}^q\vB_i$ in a fixed order, with $q\geq1$ norm-one single-site Pauli factors. Sites may repeat. Set $\epsilon_B:=\max_{1\leq i\leq q}\mathsf{ADB}_{\vB_i}[\vsigma]$, and use the unnormalized cutoff from \eqref{eq:low-frequency-part-commutation}.}
\begin{itemize}
    \item Its low-frequency component obeys
    \begingroup
\begin{align}
    \|\sqrt{\vsigma}\vB_{\leq\Delta}
        -\sqrt{\vrho}\vB_{\leq\Delta}\sqrt{\vrho^{-1}}\sqrt{\vsigma}\|_2 &\leq{\epsilon_B}^{\frac{1}{32\beta d+2}}
        e^{\frac{8\beta^2d+1/4}{16\beta d+1}\Delta}
        {e^{q\,\poly(\beta^{\pm1})}}.
\end{align}
\endgroup
    \item For $\vB$ evolved for time $t_0$, the bound is
    \begingroup
\begin{align}
    \|\sqrt{\vsigma}\vB(t_0)_{\leq\Delta}
        -\sqrt{\vrho}\vB(t_0)_{\leq\Delta}\sqrt{\vrho^{-1}}\sqrt{\vsigma}\|_2 &\leq{\epsilon_B}^{\frac{1}{32\beta d+2}}
        e^{\frac{8\beta^2d+1/4}{16\beta d+1}\Delta}
        {e^{q\,\poly(\beta^{\pm1})}}
        e^{\frac{\pi|t_0|}{\beta(16\beta d+1)}}.
\end{align}
\endgroup
\end{itemize}
{The polynomial is independent of $q$, the sites, $\Delta$, $t_0$, and the system size. The empty product has zero defect.}
\end{lem}
For a fixed frequency cutoff, a state-independent bound on the low-frequency ADB defect follows directly from the operator-norm bound under Gibbs conjugation. The point of this lemma is to retain dependence on the local Pauli ADB errors, with a prefactor that grows exponentially in the number of Pauli factors for fixed temperature, cutoff, and evolution time.
\begin{proof}
{We prove the time-evolved estimate; setting $t_0=0$ gives the first bound. Throughout, $\sigma=\beta^{-1}$, $\vB(t)=\vB_{\vH}(t)$, and $q$ counts factors, not support sites. Let ${C_{\beta}}\geq1$ denote a constant, allowed to increase between estimates, independent of $q$, the cutoffs, the times, and $n$, with ${C_{\beta}}\leq e^{\poly(\beta^{\pm1})}$.}

{The exact time representation of the unnormalized cutoff is}
\begin{align}
    &\sqrt{\vsigma}\vB(t_0)_{\leq\Delta}
       -\sqrt{\vrho}\vB(t_0)_{\leq\Delta}\sqrt{\vrho^{-1}}\sqrt{\vsigma}
       \nonumber\\
    =&{\frac{2}{\sqrt{2\pi}}}
       \int_{{-\infty}}^{{\infty}}
       \Bigl(\sqrt{\vsigma}\vB(t+t_0)
       -\sqrt{\vrho}\vB(t+t_0)\sqrt{\vrho^{-1}}\sqrt{\vsigma}\Bigr)
       \frac{\sin(t\Delta)}{t}
       {f_\sigma(t)}\,{\dd t}.
    \label{eq:product-raw-cutoff-time}
\end{align}
{Write $\hat{\vB}_i(\omega_i,t)=\widehat{\vB_i}_{\vH}(\omega_i,t)$.} {The reconstruction identity in Appendix~\ref{sec:basic-notation} gives}
\begin{align}
    \vB(t)
    &={{\frac{1}{(2\sigma\sqrt{2\pi})^{q/2}}}}
       \int_{{-\infty}}^{{\infty}}\!\cdots\!
       \int_{{-\infty}}^{{\infty}}
       {\prod_{i=1}^q}\hat{\vB}_i(\omega_i,t)\dd\vomega,
       \qquad \vomega=(\omega_1,\ldots,\omega_q).
    \label{eq:product-exact-reconstruction}
\end{align}
Define the low-total-frequency region
$\vU_{\vOmega}=\{\vomega:\sum_{i=1}^q|\omega_i|\leq\vOmega\}$
{for $\vOmega\geq0$}, and let $\vU_{\vOmega}^c$ denote its complement. Decompose
\begingroup
\begin{align}
    \vB_{\in\vU_{\vOmega}}
    &={\frac{1}{(2\sigma\sqrt{2\pi})^{q/2}}}\int_{\vomega\in\vU_{\vOmega}}
       {\prod_{i=1}^q}\hat{\vB}_i(\omega_i)\dd\vomega,\qquad \vB_{\notin\vU_{\vOmega}}
    ={\frac{1}{(2\sigma\sqrt{2\pi})^{q/2}}}\int_{\vomega\in\vU_{\vOmega}^c}
       {\prod_{i=1}^q}\hat{\vB}_i(\omega_i)\dd\vomega.
\end{align}
\endgroup
We estimate these two contributions separately.

{For the low-total-frequency part, telescope the product. The summand at position $i=0,\ldots,q-1$, before the common reconstruction factor and frequency integral, is}
\begingroup
\begin{align}
    &\sqrt{\vrho}{\prod_{j=1}^i}\hat{\vB}_j(\omega_j,t)
      \sqrt{\vrho^{-1}}\sqrt{\vsigma}
      {\prod_{j=i+1}^q}\hat{\vB}_j(\omega_j,t) -\sqrt{\vrho}{\prod_{j=1}^{i+1}}\hat{\vB}_j(\omega_j,t)
      \sqrt{\vrho^{-1}}\sqrt{\vsigma}
      {\prod_{j=i+2}^q}\hat{\vB}_j(\omega_j,t) \nonumber\\
    =&\sqrt{\vrho}{\prod_{j=1}^i}\hat{\vB}_j(\omega_j,t)\sqrt{\vrho^{-1}}
      \Bigl(\sqrt{\vsigma}\hat{\vB}_{i+1}(\omega_{i+1},t)
      -\sqrt{\vrho}\hat{\vB}_{i+1}(\omega_{i+1},t)
       \sqrt{\vrho^{-1}}\sqrt{\vsigma}\Bigr) {{}\times{}}{\prod_{j=i+2}^q}\hat{\vB}_j(\omega_j,t).
\end{align}
\endgroup
Empty products are identities. By the Gaussian conjugation and norm bounds in Appendix~\ref{sec:basic-notation}, the remaining $q-1$ factors satisfy, on $\vU_{\vOmega}$,
\begin{align}
    &\|\sqrt{\vrho}{\prod_{j=1}^i}\hat{\vB}_j(\omega_j,t)\sqrt{\vrho^{-1}}\|
     \|{\prod_{j=i+2}^q}\hat{\vB}_j(\omega_j,t)\|
     \leq{{C_{\beta}}^{q-1}}e^{\beta\vOmega/2}.
\end{align}
{For the active factor $i$, enlarge its frequency interval to $[-\Omega,\Omega]$ and integrate the other $q-1$ frequencies:}
\begin{align}
    &\int_{\vomega\in\vU_{\vOmega}}
      \bigl\|\sqrt{\vsigma}\hat{\vB}_i(\omega_i,t)
      -\vrho^{1/2}\hat{\vB}_i(\omega_i,t)\vrho^{-1/2}\sqrt{\vsigma}\bigr\|_2\dd\vomega
      \nonumber\\
    \leq&{\int_{\sum_{j\ne i}|\omega_j|\leq\Omega}
       \prod_{j\ne i}\dd\omega_j}\,
       \int_{-\Omega}^{\Omega}
      \bigl\|\sqrt{\vsigma}\hat{\vB}_i(\omega_i,t)
      -\vrho^{1/2}\hat{\vB}_i(\omega_i,t)\vrho^{-1/2}\sqrt{\vsigma}\bigr\|_2\dd\omega_i
      \nonumber\\
    =&\frac{(2\Omega)^{{q-1}}}{({q-1})!}
       \int_{-\Omega}^{\Omega}
      \bigl\|\sqrt{\vsigma}\hat{\vB}_i(\omega_i,t)
      -\vrho^{1/2}\hat{\vB}_i(\omega_i,t)\vrho^{-1/2}\sqrt{\vsigma}\bigr\|_2\dd\omega_i.
\end{align}
{For $q=1$, the zero-dimensional volume is one. Since $\hat{\vB}_i(\omega_i,t)^\dagger=\hat{\vB}_i(-\omega_i,t)$, \Cref{lem:time-evolved-adb}, adjoint/frequency reversal, and Cauchy--Schwarz give}
\begingroup
\begin{align}
    &\int_{-\Omega}^{\Omega}
      \bigl\|\sqrt{\vsigma}\hat{\vB}_i(\omega_i,t)
      -\vrho^{1/2}\hat{\vB}_i(\omega_i,t)\vrho^{-1/2}\sqrt{\vsigma}\bigr\|_2\dd\omega_i \nonumber\\
    \leq&\left[\int_{-\Omega}^{\Omega}\gamma(-\omega_i)
      \bigl\|\sqrt{\vsigma}\hat{\vB}_i(\omega_i,t)
      -\vrho^{1/2}\hat{\vB}_i(\omega_i,t)\vrho^{-1/2}\sqrt{\vsigma}\bigr\|_2^2\dd\omega_i\right]^{1/2} \times{\left[\int_{-\Omega}^{\Omega}
      \frac{\dd\omega_i}{\gamma(-\omega_i)}\right]^{1/2}}\\ \leq&{C_{\beta}\epsilon_B^{1/2}}
      e^{\pi|t|/\beta}\frac{e^{\beta\vOmega/2}}{\beta^{1/2}}.
\end{align}
\endgroup
{Summing the $q$ telescoping terms and restoring the reconstruction factor bounds the low-total-frequency defect at time $t$ by}
\begingroup
\begin{align}
    {q  C_{\beta}^{q}}
      e^{\beta\vOmega/2}\frac{(2\Omega)^{{q-1}}}{({q-1})!}
      {\epsilon_B^{1/2}}e^{\pi|t|/\beta}
      \frac{e^{\beta\vOmega/2}}{\beta^{1/2}} &\leq{C_{\beta}^{q}\sqrt{\epsilon_B}}
      e^{2\beta\vOmega}e^{\pi|t|/\beta}.
\end{align}
\endgroup
Applying \eqref{eq:product-raw-cutoff-time} to this component gives a norm bound of
\begingroup
\begin{align}
    {\frac{2}{\sqrt{2\pi}}}
      \int_{{-\infty}}^{{\infty}}
      {C_{\beta}^{q}\sqrt{\epsilon_B}}e^{2\beta\vOmega}
      e^{\pi(|t|+|t_0|)/\beta}
      {\left|\frac{\sin(t\Delta)}{t}\right|f_\sigma(t)}\dd t &\leq{C_{\beta}^{q}\sqrt{\epsilon_B}}
      e^{2\beta\vOmega}{\Delta}\,e^{\pi|t_0|/\beta} \leq{C_{\beta}^{q}\sqrt{\epsilon_B}}
      e^{2\beta\vOmega}e^{{\Delta/4}}e^{\pi|t_0|/\beta}.
\end{align}
\endgroup
The time growth is absorbed after integration against the Gaussian.

For the complementary contribution $\vB_{\notin\vU_{\vOmega}}$, first note that
\begin{align}
    \|\vB_{\notin\vU_{\vOmega}}\|
    &\leq{{\frac{1}{(2\sigma\sqrt{2\pi})^{q/2}}}}\int_{\vomega\in\vU_{\vOmega}^c}
      {\prod_{i=1}^q}\|\hat{\vB}_i(\omega_i)\|\dd\vomega.
\end{align}
{The unitaries generated by $\vH$ commute with $\vrho$, so the same norm bounds hold after real-time evolution. The bounded-degree commutator expansion and Gaussian shift (also detailed in the proof of \Cref{lem:weak-high-frequency}) give}
\begin{align}
    \|\hat{\vB}_i(\omega_i)\|
    &\leq\frac{2e^{-|\omega_i|/(4d)+\sigma^2/(16d^2)}}{\sqrt{\sigma}(2\pi)^{1/4}}.
\end{align}
{Factorizing the full frequency integral avoids any power of the cutoff:}
\begingroup
\begin{align}
    \|\vB_{\notin\vU_{\vOmega}}\|\leq C_{\beta}^q\int_{\vomega\in\vU_{\vOmega}^c}
       e^{-\sum_i|\omega_i|/(4d)}\dd\vomega
    &\leq C_{\beta}^q e^{-\vOmega/(8d)}
       \prod_{i=1}^q\int_{-\infty}^{\infty}e^{-|\omega_i|/(8d)}\dd\omega_i =C_{\beta}^q e^{-\vOmega/(8d)}.
\end{align}
\endgroup
For the outer cutoff, the Gaussian shift in Appendix~\ref{sec:basic-notation} gives, for any operator $\vX$ and either sign,
\begin{align}
    {\vrho^{\pm1/2}\vX_{\leq\Delta}\vrho^{\mp1/2}}
    &{{}=\int_{-\Delta}^{\Delta}
       e^{\mp\beta\omega/2+1/4}
       \widehat{\vX}_{\vH}(\omega\mp\beta^{-1})\dd\omega.}
    \label{eq:imaginary-time-minus-kernel-adb}
\end{align}
{Using $\|\widehat{\vX}_{\vH}(\omega)\|\leq\sqrt{\beta}(2\pi)^{-1/4}\|\vX\|$ in this integral and in the unconjugated cutoff shows that}
\begin{align}
    {\|\vX_{\leq\Delta}\|
      +\|\vrho^{\pm1/2}\vX_{\leq\Delta}\vrho^{\mp1/2}\|}
    &{{}\leq {C_{\beta}}e^{\beta\Delta/2}\|\vX\|.}
    \label{eq:product-outer-cutoff-bound}
\end{align}
{Apply this with $\vX=\vB_{\notin\vU_{\vOmega}}(t_0)$. Since $\|\sqrt{\vsigma}\|_2=1$, the high-total-frequency defect is bounded by}
\begingroup
\begin{align}
    \|\sqrt{\vsigma}(\vB_{\notin\vU_{\vOmega}}){(t_0)}_{\leq\Delta}
      -\sqrt{\vrho}(\vB_{\notin\vU_{\vOmega}}){(t_0)}_{\leq\Delta}
       \sqrt{\vrho^{-1}}\sqrt{\vsigma}\|_2 &\leq\|(\vB_{\notin\vU_{\vOmega}}){(t_0)}_{\leq\Delta}\|
      +\|\sqrt{\vrho}(\vB_{\notin\vU_{\vOmega}}){(t_0)}_{\leq\Delta}\sqrt{\vrho^{-1}}\| \nonumber\\
    &\leq e^{-\vOmega/(8d)}e^{\beta\Delta/2}{C_{\beta}^q}.
\end{align}
\endgroup
{Combining the two contributions gives}
\begingroup
\begin{align}
    \|\sqrt{\vsigma}\vB(t_0)_{\leq\Delta}
      -\sqrt{\vrho}\vB(t_0)_{\leq\Delta}\sqrt{\vrho^{-1}}\sqrt{\vsigma}\|_2 &\leq{C_{\beta}^q}\Bigl[
      e^{-\vOmega/(8d)}e^{\beta\Delta/2}
      +{\sqrt{\epsilon_B}}e^{2\beta\vOmega}
       e^{{\Delta/4}}e^{\pi|t_0|/\beta}\Bigr].
    \label{eq:product-high-low-balance}
\end{align}
\endgroup
{For $0<\epsilon_B\leq1$, balance the two bracketed terms at the formal cutoff}
\begingroup
\begin{align}
    {\Omega_*}
    &{:=\frac{4d}{16\beta d+1}\log\frac1{\epsilon_B}
      +\frac{2d(2\beta-1)}{16\beta d+1}\Delta
      -\frac{8\pi d}{\beta(16\beta d+1)}|t_0|,}\qquad {\vOmega}{:=\max\{0,\Omega_*\}.}
    \label{eq:product-optimizing-cutoff}
\end{align}
\endgroup
{If $\Omega_*\geq0$, substitution gives}
\begin{align}
    {{C_{\beta}}^q\epsilon_B^{\frac{1}{32\beta d+2}}}
    e^{\frac{8\beta^2d+1/4}{16\beta d+1}\Delta}
    e^{\frac{\pi|t_0|}{\beta(16\beta d+1)}}.
\end{align}
{If $\Omega_*<0$, then $\epsilon_B^{1/2}e^{\Delta/4+\pi|t_0|/\beta}>e^{\beta\Delta/2}$. Their weighted geometric mean, with respective weights $1/(16\beta d+1)$ and $16\beta d/(16\beta d+1)$, is the preceding bound without ${C_{\beta}}^q$ and is at least $e^{\beta\Delta/2}$. Thus \eqref{eq:product-outer-cutoff-bound} applied to the full product, whose norm is one, proves the same conclusion. The same case split also applies when $\epsilon_B>1$.}
\end{proof}

\begingroup
\makeatletter
\@for\BAppendixEquationLabel:=%
eq:weak-product,%
eq:weak-final,%
eq:weak-commutator-identity,%
eq:weak-commutator-defects,%
eq:weak-scalar-kernel,%
eq:weak-cutoff,%
eq:weak-single-frequency-envelopes,%
eq:weak-low-frequency,%
eq:weak-combined-cutoff-bound%
\do{%
  \expandafter\MT@newlabel\expandafter{\BAppendixEquationLabel}%
  \edef\BAppendixEquationTag{{\BAppendixEquationLabel}}%
  \expandafter\MT@newlabel\expandafter{\BAppendixEquationTag}%
}
\makeatother
\endgroup
\subsection{Approximate stationarity for learning tests}\label{subsec:learning-test-stationarity}

The estimates in \Cref{sec:adb-appendix} yield the global commutator bound of \Cref{thm:approximate-commute-deduce}, which contains a factor of the system size $n$. For learning, we only need this commutator tested against products $\vB$ of evolved local observables. By cyclicity of the trace,
\begin{equation*}
    \left|\tr\!\left(\vsigma[\vH,\vB]\right)\right|
    =\left|\tr\!\left(\vB[\vsigma,\vH]\right)\right|
    \leq\|\vB\|\,\|[\vH,\vsigma]\|_1.
\end{equation*}
Using the global bound here would retain the factor $n$. We therefore bound the left-hand side directly using local ADB, obtaining an estimate independent of $n$ for a fixed number of factors and fixed evolution times. This is the form needed to transfer Hamiltonian dynamics in \Cref{subsec:tested-dynamics-transfer}.

\begin{thm}[Tested commutators of evolved Pauli products]\label{thm:weak}
{Let $\vH$ satisfy \Cref{defn:local-Hamiltonian}, and let $\epsilon_{\mathrm{ADB}}$ be the uniform single-site error in \eqref{eq:intro-uniform-ADB-error}. Consider an ordered product with $q\geq1$,}
\begin{equation}
    \vB=\prod_{j=1}^{q}(\vA_j)_{\vH}(t_j)
    \label{eq:weak-product}
\end{equation}
{where each $\vA_j$ is a norm-one single-site Pauli and $T=\max_{1\leq j\leq q}|t_j|$. The sites may repeat. There are fixed positive polynomials, denoted by $\poly(\beta^{\pm1})$ and allowed to differ between occurrences, such that}
\begin{equation}
    \left|\tr\!\left(\vsigma[\vH,\vB]\right)\right|
    \leq e^{\poly(\beta^{\pm 1})\cdot q}\,
    \epsilon_{\mathrm{ADB}}^{\frac{1}{\poly(\beta^{\pm 1})}}
    \exp\!\left(
        \frac{\pi T}
             {\beta(16\beta{d}+1)}
    \right).
    \label{eq:weak-final}
\end{equation}
{The polynomial coefficients depend only on the fixed locality parameters and are independent of $n,q$, the times, and the factor locations.}
\end{thm}
The operator norm $\|[\vH,\vB]\|$ is at most $2dq$, a bound that does not use approximate detailed balance. To control $|\Tr(\vsigma[\vH,\vB])|$ in terms of ADB, we first use \Cref{lem:weak-commutator-identity} to rewrite the commutator using Gibbs conjugations. After taking its expectation in $\vsigma$, this gives Hilbert--Schmidt norms of differences with the same structure as ADB. We then split the product into high- and low-frequency parts and estimate them separately. \Cref{lem:weak-high-frequency} uses locality to bound the high-frequency contribution by a quantity that decays exponentially with the cutoff. For the low-frequency part, \Cref{lem:weak-low-frequency} expands the product difference into single-factor differences and controls them by the local ADB error. Finally, choosing the cutoff to balance these two bounds proves \Cref{thm:weak}.

\begin{lem}[Commutator identity]\label{lem:weak-commutator-identity}
{Let $\beta>0$ and use the signed identifiability kernel, distinct from the positive ADB weight $g$ in \Cref{defn:adb},}
\begin{equation}
    g_{{\mathrm{id}}}(t):=-\frac{\pi^{3/2}}
    {2\sqrt{2}\bigl(1+\cosh(\pi t)\bigr)}.
    \label{eq:weak-g-kernel}
\end{equation}
For any operator $\vB$ and Hermitian $\vH$ on a
finite-dimensional Hilbert space,
\begin{equation}
    [\vH,\vB]
    =-\frac{2}{\beta\sqrt{2\pi}}
    \int_{-\infty}^{\infty}
       \left(
       e^{\beta\vH/2}\vB_{\beta\vH/2}(t)e^{-\beta\vH/2}
       -e^{-\beta\vH/2}\vB_{\beta\vH/2}(t)e^{\beta\vH/2}
       \right)
       g_{{\mathrm{id}}}(t)\dd t,
    \label{eq:weak-commutator-identity}
\end{equation}
where
$\vB_{\beta\vH/2}(t)
=e^{i\beta t\vH/2}\vB e^{-i\beta t\vH/2}$.
Consequently, for every density operator $\vsigma$,
\begin{align}
    \left|\tr\!\left(\vsigma[\vH,\vB]\right)\right|
    \leq{}&
    \frac{2}{\beta\sqrt{2\pi}}
    \left\|\int_{-\infty}^{\infty}
       \left(
       \sqrt{\vsigma}\vB_{\beta\vH/2}(t)
       -e^{-\beta\vH/2}\vB_{\beta\vH/2}(t)
        e^{\beta\vH/2}\sqrt{\vsigma}
       \right)
       g_{{\mathrm{id}}}(t)\dd t
    \right\|_2
    \nonumber\\
    &+
    \frac{2}{\beta\sqrt{2\pi}}
    \left\|\int_{-\infty}^{\infty}
       \left(
       \vB_{\beta\vH/2}(t)\sqrt{\vsigma}
       -\sqrt{\vsigma}\,e^{\beta\vH/2}
        \vB_{\beta\vH/2}(t)e^{-\beta\vH/2}
       \right)
       g_{{\mathrm{id}}}(t)\dd t
    \right\|_2.
    \label{eq:weak-commutator-defects}
\end{align}
\end{lem}
\begin{proof}
{Apply Ref.~\cite[Lemma III.3]{chen2025learning}, the operator identity underlying \Cref{lem:state-identifiability}, with $\vA=\vB$, $\vH_1=\beta\vH/2$, and $\vH_2=0$. Its left-hand side is $(\beta/2)[\vH,\vB]$, which gives \eqref{eq:weak-commutator-identity}.}

{For \eqref{eq:weak-commutator-defects}, trace cyclicity expresses the trace of the integrand in \eqref{eq:weak-commutator-identity} against $\vsigma$ as the difference of the two defect traces, each paired with $\sqrt{\vsigma}$. Apply the triangle inequality and Schatten Cauchy--Schwarz after integration, using $\|\sqrt{\vsigma}\|_2=1$.}
\end{proof}

For the product \eqref{eq:weak-product}, introduce notation analogous to that in \Cref{sec:adb-appendix} for a frequency cutoff $\Omega\geq0$,
\begin{align}
    \vB_{\leq\Omega}^{\mathrm{prod}}
    &:=\left(\frac{\beta}{2\sqrt{2\pi}}\right)^{q/2}
       \int_{U_\Omega}
       \prod_{j=1}^{q}\widehat{\vA}_{j,\vH}(\omega_j,t_j)
       \dd\boldsymbol\omega,\qquad
    \vB_{>\Omega}^{\mathrm{prod}}
    :=\left(\frac{\beta}{2\sqrt{2\pi}}\right)^{q/2}
       \int_{\mathbb R^q\setminus U_\Omega}
       \prod_{j=1}^{q}\widehat{\vA}_{j,\vH}(\omega_j,t_j)
       \dd\boldsymbol\omega.
    \label{eq:weak-cutoff}
\end{align}
We estimate the high- and low-frequency contributions separately. We first use the operator Fourier transform to show that the high-frequency contribution decays exponentially with the cutoff $\Omega$.
\begin{lem}[High-frequency part]\label{lem:weak-high-frequency}
For the decomposition \eqref{eq:weak-cutoff}, uniformly in $t_1,\ldots,t_q$,
\begin{equation}
    \left\|[\vH,\vB_{>\Omega}^{\mathrm{prod}}]\right\|
    \leq e^{\poly(\beta^{\pm 1})\cdot q}e^{-\Omega/(8{d})}.
    \label{eq:weak-high-frequency}
\end{equation}
\end{lem}
\begin{proof}
For a single-site Pauli, the bounded-degree nested-commutator count gives
\begin{equation}
    \Bigl\|
       \underbrace{[\vH,[\vH,\ldots,[\vH,\vA_j]]\ldots]}_
                  {m\ \text{commutators}}
    \Bigr\|\leq(2d)^m m!.
    \label{eq:weak-nested-commutators}
\end{equation}
The first term in a nonzero commutator chain must meet the site of $\vA_j$, giving at most $d$ choices. At the $r$th step, a contributing term must meet the preceding union of supports, giving at most $rd$ choices. Each commutator costs at most a factor of two. Therefore, for real $|z|\leq1/(4d)$,
\begin{align}
    \left\|e^{z\vH}[\vH,\vA_j]e^{-z\vH}\right\|
    &\leq2d\sum_{m\geq0}(m+1)(2{d}|z|)^m
      =\frac{2{d}}{(1-2{d}|z|)^2}\leq8{d}.
    \label{eq:weak-small-imaginary-time}
\end{align}
{The same expansion also gives the companion bound}
\begin{align}
    {\|e^{z\vH}\vA_j e^{-z\vH}\|}
    &{{}\leq\sum_{m=0}^{\infty}(2d|z|)^m
      =\frac1{1-2d|z|}\leq2,\qquad |z|\leq\frac1{4d}.}
\end{align}
{Both estimates are uniform in the real-time argument, since evolution by $\vH$ commutes with these conjugations and preserves operator norms.}

The Gaussian shift identity, valid for any operator $\vX$, reads
\begin{align}
    \widehat{\vX}_{\vH}(\omega,t)
    &=e^{-z\omega+z^2/\beta^2}
       \widehat{\bigl(e^{z\vH}\vX e^{-z\vH}\bigr)}_{\vH}
       (\omega-2z/\beta^2,t).
    \label{eq:weak-gaussian-shift}
\end{align}
It follows by completing the square in the scalar Gaussian transform of each Bohr component; see also \cite[Lemma IX.2]{chen2025GibbsMarkov}. {Take $z=\operatorname{sgn}(\omega)/(4d)$ and use the respective bounds above for $\vX=\vA_j$ and $\vX=[\vH,\vA_j]$.} Since the Fourier transform commutes with $[\vH,\cdot]$,
\begin{align}
    {\|\widehat{\vA}_{j,\vH}(\omega,t)\|}
    &{{}\leq\frac{2\sqrt\beta}{(2\pi)^{1/4}}
      e^{\frac1{16d^2\beta^2}-\frac{|\omega|}{4d}},}
      \nonumber\\
    \left\|[\vH,\widehat{\vA}_{j,\vH}(\omega,t)]\right\|
    &\leq\frac{8d\sqrt\beta}{(2\pi)^{1/4}}
       e^{\frac{1}{16d^2\beta^2}-\frac{|\omega|}{4d}}.
    \label{eq:weak-single-frequency-envelopes}
\end{align}
{Applying Leibniz's rule to the product commutator gives}
\begin{align}
    &\left\|[\vH,\vB_{>\Omega}^{\mathrm{prod}}]\right\|\nonumber\\
    &\qquad\leq
       (q+1)e^{\poly(\beta^{\pm 1})\cdot q}
       \int_{\sum_j|\omega_j|>\Omega}
       e^{-\sum_j|\omega_j|/(4{d})}\dd\boldsymbol\omega\nonumber\\
    &\qquad\leq
       (q+1)e^{\poly(\beta^{\pm 1})\cdot q}e^{-\Omega/(8{d})}
       {\prod_{j=1}^q\int_{-\infty}^{\infty}e^{-|\omega_j|/(8d)}\dd\omega_j}\nonumber\\
    &\qquad=
       (q+1)e^{\poly(\beta^{\pm 1})\cdot q}{(16d)^q}e^{-\Omega/(8{d})}
       \leq e^{\poly(\beta^{\pm 1})\cdot q}e^{-\Omega/(8{d})}.
    \label{eq:weak-high-frequency-calculation}
\end{align}
\end{proof}
Here the last inequality enlarges the fixed polynomial to absorb $(q+1)(16d)^q$ for $q\geq1$. We next use the ADB error to bound the low-frequency contribution.

\begin{lem}[Low-frequency part]\label{lem:weak-low-frequency}
{With the dimensionless time and kernel from \Cref{lem:weak-commutator-identity}, the low-frequency part satisfies}
\begin{align}
    &\left\|\int_{{-\infty}}^{{\infty}}
       \left(
       \sqrt{\vsigma}(\vB_{\leq\Omega}^{\mathrm{prod}})_{{\beta\vH/2}}(t)
       -e^{-\beta\vH/2}(\vB_{\leq\Omega}^{\mathrm{prod}})_{{\beta\vH/2}}(t)
        e^{\beta\vH/2}\sqrt{\vsigma}
       \right)
       g_{{\mathrm{id}}}(t)\dd t\right\|_2
       \nonumber\\
    &\quad+
    \left\|\int_{{-\infty}}^{{\infty}}
       \left(
       (\vB_{\leq\Omega}^{\mathrm{prod}})_{{\beta\vH/2}}(t)\sqrt{\vsigma}
       -\sqrt{\vsigma}\,e^{\beta\vH/2}
        (\vB_{\leq\Omega}^{\mathrm{prod}})_{{\beta\vH/2}}(t)e^{-\beta\vH/2}
       \right)
       g_{{\mathrm{id}}}(t)\dd t\right\|_2
       \nonumber\\
    &\leq {C_{\beta}}^{q+1}
       \epsilon_{\mathrm{ADB}}^{\frac1{32\beta d+2}}
       e^{\frac{16\beta^2d+1/2}{16\beta{d}+1}\Omega}
       e^{\frac{\pi\max_j|t_j|}{\beta(16\beta{d}+1)}}.
    \label{eq:weak-low-frequency}
\end{align}
{The constant satisfies $1\leq {C_{\beta}}\leq e^{\poly(\beta^{\pm1})}$ and is independent of $q$, $\Omega$, the times, and $n$.}
\end{lem}

\begin{proof}
{First bound the first norm. Inserting \eqref{eq:weak-cutoff}, advance each factor time by $\beta t/2$ and telescope the ordered product:}
\begin{align}
    &\sqrt{\vsigma}\prod_{i=1}^q
       \widehat{\vA}_{i,\vH}(\omega_i,t_i+{\beta t/2})
    -e^{-\beta\vH/2}
       \left(\prod_{i=1}^q\widehat{\vA}_{i,\vH}(\omega_i,t_i+{\beta t/2})\right)
       e^{\beta\vH/2}\sqrt{\vsigma}
       \nonumber\\
    &=\sum_{j=1}^q e^{-\beta\vH/2}
       \left(\prod_{i<j}\widehat{\vA}_{i,\vH}(\omega_i,t_i+{\beta t/2})\right)
       e^{\beta\vH/2}
       \nonumber\\
    &\quad\times\left(
       \sqrt{\vsigma}\,\widehat{\vA}_{j,\vH}(\omega_j,t_j+{\beta t/2})
       -e^{-\beta\vH/2}\widehat{\vA}_{j,\vH}(\omega_j,t_j+{\beta t/2})
        e^{\beta\vH/2}\sqrt{\vsigma}
       \right)
       \prod_{i>j}\widehat{\vA}_{i,\vH}(\omega_i,t_i+{\beta t/2}).
    \label{eq:weak-explicit-telescope}
\end{align}
{For completeness, we establish the needed ordered-word estimate directly. This is the fixed-time estimate used in \eqref{eq:dyson-low-frequency-word-bound}, with $q=m+1$ factors and before any Dyson time integration. Set}
\begin{align}
    {r_j:=\sum_{i\ne j}|\omega_i|,\quad
       \Delta_j:=\Omega-r_j,\quad s_j:=t_j+\beta t/2.}
\end{align}
{For each summand, fix the other frequencies with $r_j\leq\Omega$ and integrate $\omega_j$ over $[-\Delta_j,\Delta_j]$ before taking the norm. Apply \Cref{lem:low_freq_bound} to this single active Pauli, with outer cutoff $\Delta_j$ and time $s_j$. The Gaussian shifts and \eqref{eq:weak-single-frequency-envelopes} bound the remaining factors by ${C_{\beta}}^{q-1}e^{(\beta/2-1/(4d))r_j}$. Since $\beta/2-1/(4d)\leq {\frac{8\beta^2d+1/4}{16\beta d+1}}$ and $r_j+\Delta_j=\Omega$,}
\begingroup
\begin{align}
    {e^{(\beta/2-1/(4d))r_j}e^{{\frac{8\beta^2d+1/4}{16\beta d+1}}\Delta_j}}
    &{{}\leq e^{2{\frac{8\beta^2d+1/4}{16\beta d+1}}\Omega}e^{-{\frac{8\beta^2d+1/4}{16\beta d+1}} r_j},}
    \nonumber\\
    {\int_{r_j\leq\Omega}e^{-{\frac{8\beta^2d+1/4}{16\beta d+1}} r_j}\prod_{i\ne j}\dd\omega_i}
    &{{}\leq\prod_{i\ne j}\int_{-\infty}^{\infty}
       e^{-{\frac{8\beta^2d+1/4}{16\beta d+1}}|\omega_i|}\dd\omega_i
       =\left({\frac{2(16\beta d+1)}{8\beta^2d+1/4}}\right)^{q-1}.}
\end{align}
\endgroup
{Restoring the reconstruction factor and summing over the $q$ positions gives}
\begin{align}
    &\Biggl\|\left(\frac{\beta}{2\sqrt{2\pi}}\right)^{q/2}
       \int_{U_\Omega}\Biggl(
       \sqrt{\vsigma}\prod_{i=1}^q
        \widehat{\vA}_{i,\vH}(\omega_i,t_i+{\beta t/2})
       -e^{-\beta\vH/2}
        \left(\prod_{i=1}^q\widehat{\vA}_{i,\vH}(\omega_i,t_i+{\beta t/2})\right)
        e^{\beta\vH/2}{\sqrt{\vsigma}}
       \Biggr){\dd\boldsymbol\omega}\Biggr\|_2
       \nonumber\\
    &\leq(q+1){C_{\beta}}^{q+1}
       \epsilon_{\mathrm{ADB}}^{\frac1{32\beta d+2}}
       e^{\frac{16\beta^2d+1/2}{16\beta{d}+1}\Omega}
       e^{\frac{\pi(\max_j|t_j|+{\beta|t|/2})}
                    {\beta(16\beta{d}+1)}}.
    \label{eq:weak-ordered-word-bound}
\end{align}
{The opposite orientation follows by applying this estimate to the adjoint word $\vB^\dagger=(\vA_q)_{\vH}(t_q)\cdots(\vA_1)_{\vH}(t_1)$. Indeed, $\widehat{\vA}_{j,\vH}(\omega_j,t_j)^\dagger=\widehat{\vA}_{j,\vH}(-\omega_j,t_j)$; reverse the order and change every $\omega_j$ to $-\omega_j$. The frequency region, maximum factor time, and real kernel are unchanged.}

{Finally, integrate against $|g_{\mathrm{id}}(t)|\leq {C_{\beta}} e^{-\pi|t|}$. For the first orientation, the norm is at most}
\begin{align}
    \int_{-\infty}^{\infty}(q+1){C_{\beta}}^{q+1}
       \epsilon_{\mathrm{ADB}}^{\frac1{32\beta d+2}}
       e^{{\frac{16\beta^2d+1/2}{16\beta d+1}}\Omega}
       e^{\frac{\pi({T}+{\beta|t|/2})}{\beta(16\beta{d}+1)}}
       {|g_{\mathrm{id}}(t)|}\dd t
    &\leq {C_{\beta}}^{q+1}
       \epsilon_{\mathrm{ADB}}^{\frac1{32\beta d+2}}
       e^{\frac{16\beta^2d+1/2}{16\beta{d}+1}\Omega}
       e^{\frac{\pi{T}}{\beta(16\beta{d}+1)}}.
\end{align}
{The time integral is finite because $\pi/[2(16\beta d+1)]<\pi$. Enlarging ${C_{\beta}}$ absorbs this integral, $q+1$, and the sum of the two orientations, proving \eqref{eq:weak-low-frequency}.}
\end{proof}

Finally, we choose the cutoff $\Omega$ to balance the two contributions: the high-frequency bound decays exponentially in $\Omega$, while the low-frequency bound grows exponentially. This gives a cutoff that depends logarithmically on the ADB error $\epsilon_{\mathrm{ADB}}$.

\begin{proof}[{Proof of \Cref{thm:weak}}]
{Combine \Cref{lem:weak-high-frequency,lem:weak-low-frequency} with \eqref{eq:weak-commutator-defects}. Write}
\begin{align}
    {a:=\frac1{8d},\qquad
       b:=\frac{16\beta^2d+1/2}{16\beta d+1},\qquad
       \alpha:=\frac1{32\beta d+2},\qquad
       \lambda:=\frac{\pi}{\beta(16\beta d+1)}.}
\end{align}
{Absorbing the fixed prefactor $2/(\beta\sqrt{2\pi})$ into ${C_{\beta}}^q$ gives}
\begin{align}
    {|\Tr(\vsigma[\vH,\vB])|}
    &{{}\leq {C_{\beta}}^q\left[e^{-a\Omega}
       +\epsilon_{\mathrm{ADB}}^\alpha e^{b\Omega}e^{\lambda T}\right].}
    \label{eq:weak-combined-cutoff-bound}
\end{align}
{For $0<\epsilon_{\mathrm{ADB}}\leq1$, choose the logarithmic cutoff}
\begin{align}
    {\Omega:=\frac{\alpha}{a+b}\log\frac1{\epsilon_{\mathrm{ADB}}}.}
    \label{eq:weak-logarithmic-cutoff}
\end{align}
{Then}
\begin{align}
    {|\Tr(\vsigma[\vH,\vB])|}
    &{{}\leq2{C_{\beta}}^q\epsilon_{\mathrm{ADB}}^{\alpha a/(a+b)}e^{\lambda T},}
       \nonumber\\
    {\frac{\alpha a}{a+b}}
    &{{}={\frac{1}{\poly(\beta^{\pm1})}}.}
\end{align}
{This proves \eqref{eq:weak-final}. For $\epsilon_{\mathrm{ADB}}=0$, take $\Omega\to\infty$ in \eqref{eq:weak-combined-cutoff-bound}. For $\epsilon_{\mathrm{ADB}}>1$, the bound $\|[\vH,\vB]\|\leq2dq$ gives the theorem after enlarging the fixed prefactor.}
\end{proof}
The resulting transfer bound is used in \Cref{subsec:tested-dynamics-transfer}.

\clearpage
\makeatletter
\MT@newlabel{eq:dyson-high-frequency-kms-integral}
\MT@newlabel{{eq:dyson-high-frequency-kms-integral}}
\MT@newlabel{eq:dyson-high-frequency-operator-bound}
\MT@newlabel{{eq:dyson-high-frequency-operator-bound}}
\MT@newlabel{eq:dyson-low-frequency-active-factor}
\MT@newlabel{{eq:dyson-low-frequency-active-factor}}
\MT@newlabel{eq:dyson-low-frequency-ordered-bound}
\MT@newlabel{{eq:dyson-low-frequency-ordered-bound}}
\MT@newlabel{eq:dyson-low-frequency-telescoping}
\MT@newlabel{{eq:dyson-low-frequency-telescoping}}
\MT@newlabel{eq:dyson-low-frequency-time-integral}
\MT@newlabel{{eq:dyson-low-frequency-time-integral}}
\MT@newlabel{eq:dyson-low-frequency-trace-to-defect}
\MT@newlabel{{eq:dyson-low-frequency-trace-to-defect}}
\MT@newlabel{eq:dyson-low-frequency-word-bound}
\MT@newlabel{{eq:dyson-low-frequency-word-bound}}
\MT@newlabel{eq:identification-F}
\MT@newlabel{{eq:identification-F}}
\MT@newlabel{eq:identification-F-cut}
\MT@newlabel{{eq:identification-F-cut}}
\MT@newlabel{eq:truncated-identifiability-integral}
\MT@newlabel{{eq:truncated-identifiability-integral}}
\MT@newlabel{eq:linearization-error-term}
\MT@newlabel{{eq:linearization-error-term}}
\MT@newlabel{eq:linearized-I}
\MT@newlabel{{eq:linearized-I}}
\MT@newlabel{eq:measurable-Q}
\MT@newlabel{{eq:measurable-Q}}
\MT@newlabel{eq:measurement-operator}
\MT@newlabel{{eq:measurement-operator}}
\MT@newlabel{eq:metastable-identifiability}
\MT@newlabel{{eq:metastable-identifiability}}
\MT@newlabel{eq:move6-order-bound}
\MT@newlabel{{eq:move6-order-bound}}
\MT@newlabel{eq:move6-smallness}
\MT@newlabel{{eq:move6-smallness}}
\MT@newlabel{eq:move6-tested-duhamel}
\MT@newlabel{{eq:move6-tested-duhamel}}
\MT@newlabel{eq:outer-cutoff-operator-bounds}
\MT@newlabel{{eq:outer-cutoff-operator-bounds}}
\MT@newlabel{eq:outer-cutoff-trace-bound}
\MT@newlabel{{eq:outer-cutoff-trace-bound}}
\MT@newlabel{eq:spatial-truncation-inside-conjugation}
\MT@newlabel{{eq:spatial-truncation-inside-conjugation}}
\MT@newlabel{eq:spatially-truncated-dyson}
\MT@newlabel{{eq:spatially-truncated-dyson}}
\MT@newlabel{eq:state-motion-identity}
\MT@newlabel{{eq:state-motion-identity}}
\MT@newlabel{eq:transfer-state-difference}
\MT@newlabel{{eq:transfer-state-difference}}
\MT@newlabel{eq:transferred-Q}
\MT@newlabel{{eq:transferred-Q}}
\MT@newlabel{eq:assembled-identification-budgets}
\MT@newlabel{{eq:assembled-identification-budgets}}
\makeatother
\section{From the identifiability equation to measurable tests}
\label{sec:identifiability-bounds}

In this section, we use the tools of \Cref{sec:adb-to-commutation} to construct observables whose expectations approximate pairings with the commutator $[\vA,\vH-\vH']$, where $\vH'$ is a candidate Hamiltonian and $\vA$ is a single-site Pauli operator. These observables will be used in \Cref{sec:learning-protocol} to iteratively search for the coefficients of the unknown Hamiltonian. We work with the $\vsigma$-weighted KMS norm of this commutator; in \Cref{subsec:faithfulness}, we will show that this norm controls the usual operator norm for operators of bounded support. We therefore study the pairing
\begin{align}
 \mathcal F_\sigma(\vO,\vA;\vH,\vH')
 :=\frac{\beta}{2}\langle\vO,[\vA,\vH-\vH']\rangle_{\vsigma}.
 \label{eq:identification-F}
\end{align}
Taking $\vO=[\vA,\vH-\vH']$ gives $(\beta/2)\|[\vA,\vH-\vH']\|_{\vsigma}^2$.

Throughout this section, $\vH,\vH'$ and the auxiliary $\vG$ are local Hamiltonians as in \Cref{defn:local-Hamiltonian}. The input state has sufficiently small local ADB error as in \Cref{defn:adb,eq:intro-uniform-ADB-error}. Let $a_\beta=e^{-P(\beta^{\pm1})}$ for a prescribed positive polynomial $P$.
The test class consists of local operators $\vO$ with support size bounded by $k$, and the Hamiltonians have interaction degree bounded by $d$. We retain explicit norms in the intermediate estimates. The parameters $d,k,D$ denote the interaction degree, Pauli weight, and lattice dimension, respectively, as in \Cref{defn:local-Hamiltonian}; we use $\sigma=\beta^{-1}$, distinct from the state $\vsigma$, and write $\eta=\|h-h'\|_\infty$.

The following theorem is the main result of this section. It constructs a family of measurable tests $Q_\sigma(\vO,\vA;\vG,\vH')$ that are linear in $\vsigma$ and approximate the identifiability pairing $\mathcal F_\sigma(\vO,\vA;\vH,\vH')$ with controlled error when the auxiliary Hamiltonian is $\vH$. These tests allow us to search the Hamiltonian parameter space for an accurate candidate.

\begin{thm}[Measurable identification tests]\label{lem:tested-identification}
For arbitrary $a_{\beta}=e^{-P(\beta^{\pm1})}$ for a prescribed positive polynomial $P$, there exists a family of measurable tests $Q_\sigma(\vO,\vA;\vG,\vH')$, with the following properties for normalized $\vO$ in the fixed bounded-support test class and single-site Pauli $\vA$. The allowed local Hamiltonians satisfy \Cref{defn:local-Hamiltonian}, and the input state satisfies the conditions above, with $\epsilon_{\mathrm{ADB}}\leq e^{-\poly(\beta^{\pm 1})}$.
\begin{enumerate}
\item There is an operator $\vM_{\vO,\vA;\vG,\vH'}$, determined by the selected Hamiltonians, test operators and known parameters, such that
\begin{align}
 Q_\sigma=\Tr(\vM_{\vO,\vA;\vG,\vH'}\vsigma),
 \qquad \|\vM_{\vO,\vA;\vG,\vH'}\|\leq K_\beta.
 \label{eq:summary-measurability}
\end{align}
The Hermitian real and imaginary parts of $\vM$ can be measured separately.
\item For every allowed auxiliary $\vG$, the true candidate satisfies the following bound on its test value
\begin{align}
 |Q_\sigma(\vO,\vA;\vG,\vH)|\leq C_\beta\epsilon_{\mathrm{ADB}}^{{\frac{1}{\poly(\beta^{\pm1})}}}.
 \label{eq:summary-completeness}
\end{align}
\item Whenever $\|h-h'\|_\infty\leq2$,
\begin{align}
 \left|Q_\sigma(\vO,\vA;\vH,\vH')-\mathcal F_\sigma(\vO,\vA;\vH,\vH')\right|
 \leq a_\beta\|h-h'\|_\infty+C_\beta\epsilon_{\mathrm{ADB}}^{{\frac{1}{\poly(\beta^{\pm1})}}}.
 \label{eq:summary-identification}
\end{align}
\end{enumerate}
{Here $K_\beta,C_\beta=e^{\poly(\beta^{\pm 1})}$.} The polynomials and cutoffs may depend on the prescribed $P$ and the fixed test-support bound, but not on $n$.
\end{thm}

To construct these tests, we first express the commutator pairing in a concrete integral form. The following identifiability equation, valid for any input state, writes it in terms of time evolution and Gibbs conjugations.

\begin{lem}[Identifiability for an arbitrary input state]\label{lem:state-identifiability}
Let $\vrho$ and $\vrho'$ be the Gibbs states of $\vH$ and $\vH'$ at inverse temperature $\beta$. For every state $\vsigma$ and operators $\vO,\vA$,
\begin{align}
\mathcal F_\sigma
={}&\frac{1}{\sqrt{2\pi}}\int_{\mathbb R}g_{\mathrm{id}}(t)
\Tr\!\left[\vO^\dagger\sqrt{\vsigma}
\left(
\vrho^{-1/2}\vrho'^{1/2}\vA_{\beta\vH'/2}(t)\vrho'^{-1/2}\vrho^{1/2}
\right.\right.\nonumber\\[-2pt]
&\hspace{4cm}\left.\left.
-\vrho^{1/2}\vrho'^{-1/2}\vA_{\beta\vH'/2}(t)\vrho'^{1/2}\vrho^{-1/2}
\right)_{\beta\vH/2}(-t)\sqrt{\vsigma}\right]\dd t,
\label{eq:metastable-identifiability}
\end{align}
where $g_{\mathrm{id}}(t)=-\pi^{3/2}/[2\sqrt{2}(1+\cosh(\pi t))]$.
\end{lem}
\begin{proof}
Ref.~\cite[Lemma III.3]{chen2025learning} gives, for Hermitian $\vH_1,\vH_2$,
\begin{align}
 [\vA,\vH_2]-[\vA,\vH_1]
 =\frac1{\sqrt{2\pi}}\int_{\mathbb R}g_{\mathrm{id}}(t)
 \Bigl(&e^{\vH_2}e^{-\vH_1}\vA_{\vH_1}(t)e^{\vH_1}e^{-\vH_2}\nonumber\\
 &-e^{-\vH_2}e^{\vH_1}\vA_{\vH_1}(t)e^{-\vH_1}e^{\vH_2}\Bigr)_{\vH_2}(-t)\dd t.
 \label{eq:underlying-operator-identifiability}
\end{align}
Set $\vH_1=\beta\vH'/2$ and $\vH_2=\beta\vH/2$. The Gibbs normalization constants cancel within each conjugation, so the right-hand side becomes the operator integral in \eqref{eq:metastable-identifiability}. Taking the KMS pairing with $\vO$, namely $\langle\vO,\vY\rangle_{\vsigma}=\Tr[\vO^\dagger\sqrt{\vsigma}\vY\sqrt{\vsigma}]$, proves the claim.
\end{proof}

\paragraph*{The four obstacles and their resolution.}
\Cref{fig:identification-flow} summarizes the four steps that turn \eqref{eq:metastable-identifiability} into measurable tests. In \Cref{subsec:double-frequency-truncation}, double-frequency truncation controls amplification by the two Gibbs conjugations. In \Cref{subsec:adb-linearization}, ADB-based linearization combines the separated square-root state factors into a single state factor. In \Cref{subsec:tested-dynamics-transfer}, tested stationarity allows the remaining true-Hamiltonian dynamics to be transferred to the local test operator. Finally, \Cref{subsec:auxiliary-tests} replaces that unknown generator by a selected auxiliary Hamiltonian to define a measurable family. The true candidate has a small test value for every auxiliary choice, while the choice equal to the true Hamiltonian retains the identification guarantee.

 The main technical difficulty arises in the second step because the available ADB estimate concerns true-Hamiltonian dynamics, whereas the expression to be linearized contains candidate dynamics. The perturbative argument in \Cref{subsec:adb-linearization} bridges this mismatch.

\begin{figure}[htbp]
\centering
\begin{tikzpicture}[box/.style={draw,rounded corners,inner sep=7pt,font=\small}]
\node[box] (f) at (0,0) {\begin{tabular}{@{}p{4.9cm}p{10.0cm}@{}}Exact identification & $\mathcal F_\sigma=\frac\beta2\langle\vO,[\vA,\vH-\vH']\rangle_{\vsigma}$\quad\eqref{eq:identification-F}\end{tabular}};
\node[box] (fc) at (0,-1.65) {\begin{tabular}{@{}p{4.9cm}p{10.0cm}@{}}Frequency-regularized pairing & $\mathcal F_{\mathrm{cut}}(\vO,\vA;\vH,\vH')$\quad\eqref{eq:identification-F-cut}\end{tabular}};
\node[box] (i) at (0,-3.3) {\begin{tabular}{@{}p{4.9cm}p{10.0cm}@{}}Linear in the input state & $\mathcal I_\sigma(\vO,\vA;\vH,\vH')$\quad\eqref{eq:linearized-I}\end{tabular}};
\node[box] (qh) at (0,-4.95) {\begin{tabular}{@{}p{4.9cm}p{10.0cm}@{}}Target dynamics acts on $\vO$ & $Q_\sigma(\vO,\vA;\vH,\vH')$\quad\eqref{eq:transferred-Q}\end{tabular}};
\node[box] (qg) at (0,-6.6) {\begin{tabular}{@{}p{4.9cm}p{10.0cm}@{}}Measurable auxiliary tests & $Q_\sigma(\vO,\vA;\vG,\vH')=\Tr(\vM\vsigma)$\quad\eqref{eq:measurable-Q}\end{tabular}};
\draw[->] (f) -- node[right,font=\scriptsize]{\ref{subsec:double-frequency-truncation}: double-frequency truncation} (fc);
\draw[->] (fc) -- node[right,font=\scriptsize]{\ref{subsec:adb-linearization}: ADB-based linearization} (i);
\draw[->] (i) -- node[right,font=\scriptsize]{\ref{subsec:tested-dynamics-transfer}: transfer of Hamiltonian dynamics} (qh);
\draw[->] (qh) -- node[right,font=\scriptsize]{\ref{subsec:auxiliary-tests}: replace the test evolution $\vH$ by $\vG$} (qg);
\end{tikzpicture}
\caption{Construction of the measurable tests from the identifiability equation \eqref{eq:metastable-identifiability}. The first three steps produce the expression with true auxiliary dynamics. The fourth defines the family with selectable dynamics. ADB-based linearization is the main technical step.}
\label{fig:identification-flow}
\end{figure}
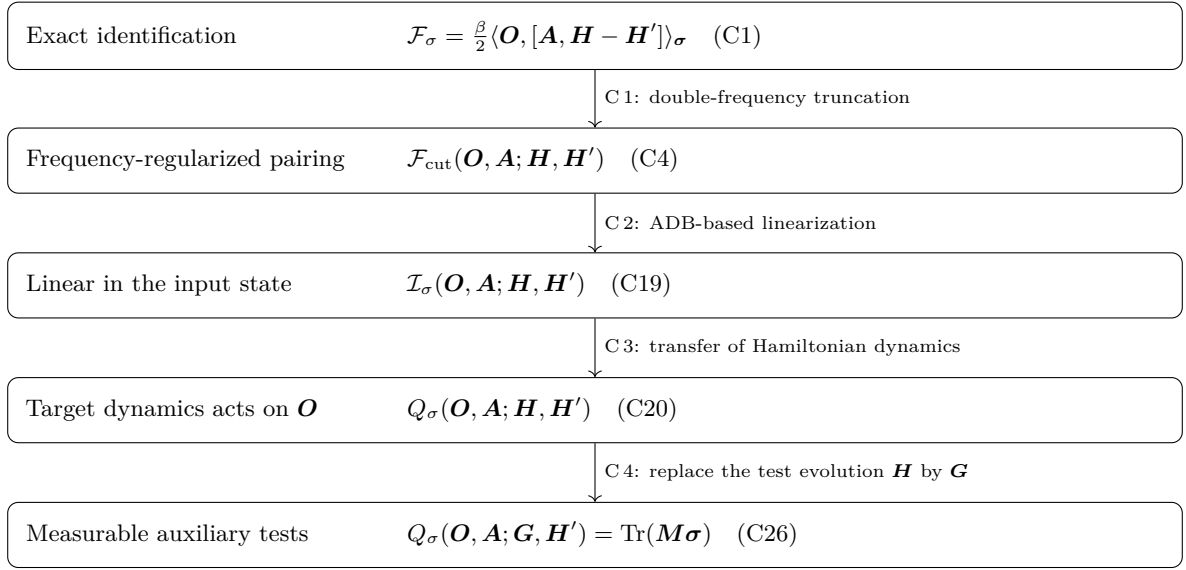

\subsection{Double-frequency truncation}
\label{subsec:double-frequency-truncation}
We next prove \Cref{thm:double-frequency-truncation}, which regularizes $\langle\vO,[\vA,\vH-\vH']\rangle_{\vsigma}$. The truncation error is tracked in terms of the current coefficient error $\eta$ between $\vH$ and $\vH'$.

We introduce cutoffs for both frequency variables.
{The single-transform reconstruction factor from Appendix~A is $(2\sigma\sqrt{2\pi})^{-1/2}$, so a double reconstruction has factor $1/(2\sigma\sqrt{2\pi})$.}
With the normalization in \Cref{thm:double-frequency-truncation}, the approximation to \eqref{eq:identification-F} is
\begin{align}
 \mathcal F_{\mathrm{cut}}(\vO,\vA;\vH,\vH')
 :=\frac{\beta}{{4\sigma\sqrt{2\pi}}}
 \int_{-\Omega_2}^{\Omega_2}\int_{-\Omega_1}^{\Omega_1}
 \left\langle\vO,[\vA(\omega_2),\vH-\vH'](\omega_1)\right\rangle_{\vsigma}
 \dd\omega_1\dd\omega_2.
 \label{eq:identification-F-cut}
\end{align}
Here the inner frequency refers to $\vH'$ and the outer frequency to $\vH$.

\begin{thm}[Double-frequency truncation]
\label{thm:double-frequency-truncation}
Let $\vH,\vH'$ be local Hamiltonians as in \Cref{defn:local-Hamiltonian}, let $\vA$ be a {single-site} Pauli operator, and suppose every Pauli coefficient of $\vH-\vH'$ has magnitude at most $\eta$.
{For every prescribed polynomial $P(\beta^{\pm1})$, there are cutoffs $\Omega_1,\Omega_2=\poly(\beta^{\pm1})$, depending also on the fixed locality parameters and $P$, such that}
\begin{align}
 {\Bigg|}\ipc{\vO}{[\vA,\vH-\vH']}_{\vsigma}
 -{\frac{1}{2\sigma\sqrt{2\pi}}}
 \int_{-\Omega_2}^{\Omega_2}\int_{-\Omega_1}^{\Omega_1}
 \ipc{\vO}{[\vA(\omega_2),\vH-\vH'](\omega_1)}_{\vsigma}
 \,\dd\omega_1\dd\omega_2{\Bigg|}
    \leq\|\vO\|\|\vA\|\eta e^{-{P}(\beta^{\pm1})}.
\end{align}
\end{thm}
\begin{proof}
{Set $\vX=[\vA,\vH-\vH']$. The case $\eta=0$ is immediate. Otherwise write $\vH'-\vH=\eta\vV$, with every coefficient of $\vV$ bounded by one. Only terms meeting $\supp(\vA)$ contribute to $\vX$, so $\|\vX\|\leq2d|\supp(\vA)|\eta\|\vA\|$. Expanding this commutator into bounded-weight Pauli terms and using the local Fourier bound gives, with $C_{\beta}=e^{\poly(\beta^{\pm1})}$ and fixed $c_0>0$,}
\begin{align}
 &\ipc{\vO}{\vX}_{\vsigma}
 -{\frac{1}{\sqrt{2\sigma\sqrt{2\pi}}}}\int_{-\Omega_1}^{\Omega_1}
       \ipc{\vO}{\vX_{\vH}(\omega_1)}_{\vsigma}\,\dd\omega_1
 ={\frac{1}{\sqrt{2\sigma\sqrt{2\pi}}}}\int_{|\omega_1|>\Omega_1}
       \ipc{\vO}{\vX_{\vH}(\omega_1)}_{\vsigma}\,\dd\omega_1,
 \nonumber\\
 &{\left|\frac{1}{\sqrt{2\sigma\sqrt{2\pi}}}\int_{|\omega_1|>\Omega_1}
       \ipc{\vO}{\vX_{\vH}(\omega_1)}_{\vsigma}\,\dd\omega_1\right|
 \leq C_{\beta}\eta\|\vO\|\|\vA\|e^{-\Omega_1/(c_0d)}.}
\end{align}
{For the inner tail, the operator-norm estimate in the proof of Ref.~\cite[Lemma III.5]{chen2025learning}, applied to the coefficient-normalized $\vV$, gives}
\begin{align}
 \|[\vA_{\vH'}(\omega_2),\vH-\vH']\|
 &\leq{C_{\beta}\eta\|\vA\|e^{-|\omega_2|/(4d)}}.
\end{align}
{Here $C_{\beta}$ includes the degree and Gaussian prefactors in that estimate; it is independent of the system size and $\eta$. Gaussian filtering also gives $\|\vB_{\vH}(\omega_1)\|\leq\|\vB\|/[\sqrt{\sigma}(2\pi)^{1/4}]$ for any $\vB$. Therefore the double-filtered quantity minus the single-filtered quantity is}
\begin{align}
 &{\frac{1}{2\sigma\sqrt{2\pi}}}\int_{-\Omega_2}^{\Omega_2}\int_{-\Omega_1}^{\Omega_1}
 \ipc{\vO}{[\vA(\omega_2),\vH-\vH'](\omega_1)}_{\vsigma}
 \,\dd\omega_1\dd\omega_2
 -{\frac{1}{\sqrt{2\sigma\sqrt{2\pi}}}}\int_{-\Omega_1}^{\Omega_1}
 \ipc{\vO}{\vX_{\vH}(\omega_1)}_{\vsigma}\,\dd\omega_1
 \nonumber\\
 &={-\frac{1}{2\sigma\sqrt{2\pi}}}
 \int_{|\omega_2|>\Omega_2}\int_{-\Omega_1}^{\Omega_1}
 \ipc{\vO}{[\vA(\omega_2),\vH-\vH'](\omega_1)}_{\vsigma}
 \,\dd\omega_1\dd\omega_2.
\end{align}
{Its absolute value is at most}
\begin{align}
 {\frac{\Omega_1}{\sigma^{3/2}(2\pi)^{3/4}}
 \|\vO\|\int_{|\omega_2|>\Omega_2}
 \|[\vA_{\vH'}(\omega_2),\vH-\vH']\|\,\dd\omega_2}
    {\leq C_{\beta}\eta\|\vO\|\|\vA\|
 \Omega_1e^{-\Omega_2/(4d)}.}
\end{align}
{Choose $\Omega_1$ first and then $\Omega_2$ so that each tail is at most half the prescribed error. Both cutoffs are polynomial in $\beta^{\pm1}$ for fixed locality parameters. For the pairing $\mathcal F_\sigma$, absorb the additional factor $\beta/2$ into the chosen accuracy polynomial.}
\end{proof}

\subsection{Linearization via approximate detailed balance}
\label{subsec:adb-linearization}

{
From this point onward, time is physical time.
\begin{align}
 {g(t):=\frac2\beta g_{\mathrm{id}}(2t/\beta),\qquad
 |g(t)|\leq C_{\beta}\beta^{-1}e^{-2\pi|t|/\beta}.}
\end{align}
}

After frequency truncation, the expression is still nonlinear in the input state $\vsigma$ because the KMS pairing contains square-root state factors:
\begin{align}
 \mathcal F_{\mathrm{cut}}
 ={}&\frac{1}{4\pi\sigma}
 \int_{-\Omega_2}^{\Omega_2}\int_{-\Omega_1}^{\Omega_1}\int_{-\infty}^{\infty}
 \Tr\Bigl[\vO^\dagger\sqrt{\vsigma}
 \Bigl(\vrho^{-1/2}\vrho'^{1/2}\vA_{\vH'}(t,\omega_2)
 \vrho'^{-1/2}\vrho^{1/2}
 \nonumber\\
 &\hspace{2.5cm}
 -\vrho^{1/2}\vrho'^{-1/2}\vA_{\vH'}(t,\omega_2)
 \vrho'^{1/2}\vrho^{-1/2}\Bigr)_{\vH}(-t,\omega_1)
 \sqrt{\vsigma}\Bigr]g(t)\,\dd t\dd\omega_1\dd\omega_2.
 \label{eq:truncated-identifiability-integral}
\end{align}
We next use the ADB property to obtain an expression linear in $\vsigma$. This is one of the main technical steps in constructing measurable tests.

\begin{thm}[Refined ADB-based linearization]\label{thm:kms-to-gns-reduction}
    Let $\vH,\vH'$ satisfy \Cref{defn:local-Hamiltonian}. Assume that the coefficient of each Pauli term in $\vH-\vH'$ is bounded in magnitude by $\eta$.
    {Fix a prescribed positive polynomial $\poly_1(\beta^{\pm1})$.} There exist polynomials $\poly$ and $\poly'$ in the same variables such that, whenever
        \begin{align}
            \epsilon_{\mathrm{ADB}}<e^{-\poly(\beta^{\pm 1})}
            \quad\text{and}\quad
            \eta<e^{-\poly'(\beta^{\pm 1})},
        \end{align}
        the following term, denoted $E_{\mathrm{lin}}^{+}$,
        \begin{align}
            \frac{1}{{4\pi\sigma}}
            \int_{-\Omega_2}^{\Omega_2}\int_{-\Omega_1}^{\Omega_1}
            \int_{-\infty}^{\infty}\Tr\Bigg[
            \vO^\dagger
            \Bigg(
            &\sqrt{\vsigma}\Big(\sqrt{\vrho}^{-1}\sqrt{\vrho'}\vA_{\vH'}(t,{\omega_2})\sqrt{\vrho'^{-1}}\sqrt{\vrho}\Big)_{\vH}(-t,{\omega_1})\sqrt{\vsigma}\\
            &-
            \Big(\sqrt{\vrho'}\vA_{\vH'}(t,{\omega_2})\sqrt{\vrho'^{-1}}\Big)_{\vH}(-t,{\omega_1})\vsigma
            \Bigg)\Bigg]g(t)\dd t
            \dd\omega_1\dd\omega_2
            \label{eq:linearization-error-term}
\end{align}
        {satisfies}
\begin{align}
 {|E_{\mathrm{lin}}^+|
 \leq\|\vO\|\|\vA\|\left(\eta e^{-\poly_1(\beta^{\pm1})}
 +e^{\poly(\beta^{\pm1})}\epsilon_{\mathrm{ADB}}^{\frac{1}{32\beta d+2}}\right).}
\end{align}
\end{thm}

\Cref{thm:kms-to-gns-reduction} bounds the difference between one branch of the truncated KMS pairing in \eqref{eq:truncated-identifiability-integral} and a term containing only a single factor $\vsigma$. The latter is linear in the state and has the form of an observable expectation. Applying the same replacement to both conjugation branches yields \eqref{eq:linearized-I}; \Cref{subsec:tested-dynamics-transfer,subsec:auxiliary-tests} then remove the dependence on the unknown Hamiltonian dynamics to make the tests measurable.

The main difficulty is the $\vH'$-dynamics inside the expression, whereas the available ADB estimates apply to $\vH$-dynamics. We use the Dyson expansion to write the former as a sum of Pauli products evolved under $\vH$, allowing us to apply the local-product ADB estimate of \Cref{lem:low_freq_bound} in \Cref{sec:adb-appendix}. \Cref{lem:kms-gns-perturbation-truncation} first restricts the Hamiltonian perturbation to a local region, so that the number of contributing Pauli products grows at most exponentially with the Dyson order, independently of the system size. \Cref{lem:zeroth-order-dyson,lem:dyson_term_adb_bound_m} then control the zeroth-order term and all positive-order terms, respectively.

Summing these bounds requires $\eta$ to be sufficiently small, so that the factors $\eta^m$ suppress the growth with the Dyson order. For general $0\leq\eta\leq2$, we instead truncate time and the Dyson series to obtain the following coarse estimate.
\begin{thm}[Coarse ADB-based linearization]\label{thm:coarse-adb-linearization}
Let $\vH,\vH'$ be local Hamiltonians on the fixed-dimensional lattice of \Cref{defn:local-Hamiltonian}, with coefficient distance bounded by two, and let $\vA$ be a single-site Pauli and $\vO$ a fixed-support test with $\|\vO\|\leq1$. For every positive polynomial $P(\beta^{\pm1})$, there is a positive polynomial $P'(\beta^{\pm1})$ such that, if $\epsilon_{\mathrm{ADB}}\leq e^{-P'(\beta^{\pm1})}$, the linearization error in \eqref{eq:linearization-error-term} {satisfies $|E_{\mathrm{lin}}^+|\leq e^{-P(\beta^{\pm1})}$}, {for any fixed polynomial observable cutoffs $\Omega_1,\Omega_2$. The polynomial $P'$ may depend on these cutoffs.}
\end{thm}

{The inner kernels and outer window are}
\begin{align}
 {h_+(u)}&{:=\frac{f_\sigma(u)}{\sqrt{2\pi}}
 \int_{-\Omega_2}^{\Omega_2}e^{-\beta\omega/2+\sigma^2\beta^2/4}
 e^{-i(\omega-\sigma^2\beta)u}\,\dd\omega,}\nonumber\\
 {h_-(u)}&{:=\frac{f_\sigma(u)}{\sqrt{2\pi}}
 \int_{-\Omega_2}^{\Omega_2}e^{\beta\omega/2+\sigma^2\beta^2/4}
 e^{-i(\omega+\sigma^2\beta)u}\,\dd\omega
 =\overline{h_+(u)},}\nonumber\\
 {w_{\Omega_1}(u)}&{:=f_\sigma(u)\frac{2\sin(\Omega_1u)}{u},
 \qquad w_{\Omega_1}(0)=2\Omega_1f_\sigma(0).}
 \label{eq:outer-time-weight}
\end{align}
{In particular, $|h_\pm(u)|\leq C_{\beta}\sqrt{\sigma}\,\beta^{-1}e^{\beta\Omega_2}e^{-\sigma^2u^2}$. The inner Fourier factor is included in $h_\pm$, whereas the outer Fourier factor is not included in $w_{\Omega_1}$. The coefficient of the physical-time, double-frequency identifiability integral is $1/(4\pi\sigma)$; after converting the outer filter to time, it is $1/(4\pi\sigma\sqrt{2\pi})$.}

{The kernels just defined give the exact representation}
\begin{align}
 \int_{-\Omega_2}^{\Omega_2}\sqrt{\vrho'}\vA_{\vH'}(t,{\omega_2})
 \sqrt{\vrho'^{-1}}\,\dd\omega_2
 =\int_{\mathbb R}\vA_{\vH'}(t+t_2)h_+(t_2)\,\dd t_2.
\end{align}
{The opposite Gibbs conjugation has the same representation with $h_-$. We first establish the spatial and Dyson estimates, and then assemble the linearization bound.}

Write $\vH'=\vH+\eta\vV$, with every Pauli coefficient of $\vV$ bounded by one.
If $\eta=0$, set $\vV=0$. Following Appendix~A's patch convention, let $\vV_\ell$ retain terms at interaction-graph distance less than $\ell-1$ from $\supp(\vA)$. Use the common support graph of $\vH$ and $\vH'$. The hybrid Hamiltonian $\vH+\eta\vV_\ell$ has the coefficients of $\vH'$ inside this patch and those of $\vH$ outside, so each local term still has norm at most one.
By the perturbative Lieb--Robinson bound in \Cref{lem:basic-lr-perturbation}, truncating the Hamiltonian difference $\vH-\vH'$ to a neighborhood of radius $\ell$ around $\supp(\vA)$ introduces an error proportional to $\eta$ that decays exponentially in $\ell$. The following lemma states this spatial-truncation estimate for the perturbation.

\begin{lem}[Spatial truncation of the perturbation]
\label{lem:kms-gns-perturbation-truncation}
For every polynomial $\poly(\beta^{\pm1})$, there exists a polynomial $\poly'(\beta^{\pm1})$ such that, whenever $\ell>\poly'(\beta^{\pm1})$, replacing $\vH'$ by $\vH+\eta\vV_\ell$ in the regularized KMS-to-GNS expression changes its value by at most
$\eta{\|\vO\|\|\vA\|}e^{-\poly(\beta^{\pm1})}$.
{The already selected cutoffs and kernels remain fixed in this replacement.}
\end{lem}
\begin{proof}
{We first verify how the spatial truncation enters the outer-filtered KMS-to-GNS difference. Set}
\begin{align}
 {\vB:=\int_{\mathbb R^2}
 \big(\vA_{\vH'}(t+t_2)-\vA_{\vH+\eta\vV_\ell}(t+t_2)\big)_{\vH}(-t)
 g(t)h_+(t_2)\,\dd t\dd t_2.}
\end{align}
{Subtract \eqref{eq:linearization-error-term} with $\vH'$ replaced by $\vH+\eta\vV_\ell$ from the original expression. The inner-kernel representation, with the same $g$ and $h_+$, gives}
\begingroup
\begin{align*}
 &\frac{1}{4\pi\sigma}\int_{-\Omega_1}^{\Omega_1}\int_{\mathbb R^2}
 \Tr\Biggl[\vO^\dagger\Biggl(
 \sqrt{\vsigma}\Bigl(\sqrt{\vrho}^{-1}
 \bigl(\vA_{\vH'}(t+t_2)-\vA_{\vH+\eta\vV_\ell}(t+t_2)\bigr)
 \sqrt{\vrho}\Bigr)_{\vH}(-t,\omega_1)\sqrt{\vsigma}
 \\
 &\qquad\qquad
 -\bigl(\vA_{\vH'}(t+t_2)-\vA_{\vH+\eta\vV_\ell}(t+t_2)\bigr)_{\vH}(-t,\omega_1)\vsigma
 \Biggr)\Biggr]g(t)h_+(t_2)\,\dd t\dd t_2\dd\omega_1
 \\
 &=\frac{1}{4\pi\sigma}\int_{-\Omega_1}^{\Omega_1}
 \Tr\Biggl[\vO^\dagger\Biggl(
 \sqrt{\vsigma}\Bigl(\sqrt{\vrho}^{-1}\vB\sqrt{\vrho}\Bigr)_{\vH}(\omega_1)\sqrt{\vsigma}
 -\vB_{\vH}(\omega_1)\vsigma\Biggr)\Biggr]\,\dd\omega_1.
\end{align*}
\endgroup
{For the last equality, $\vH$-evolution commutes with Gibbs conjugation and the outer frequency filter. Their linearity allows both time integrals to be moved inside, where they give exactly $\vB$. These rearrangements are justified by the integrability of $|g(t)h_+(t_2)|$ and the boundedness of the real-time evolved operators.}

Write $S(r,\vA)$ for the number of allowed interaction terms at graph distance $r$ from $\supp(\vA)$ in the common interaction graph.
{The kernel bound gives $|g(t)h_+(t_2)|\leq e^{\poly(\beta^{\pm1})}e^{-\pi(|t|+|t_2|)/\beta}$. Apply \Cref{lem:basic-lr-perturbation} shell by shell to the terms at distance $r\geq\ell-1$. Split each time integral at ${|t|+|t_2|=r/(16ed)}$. On the inner region, the factorial term decays exponentially in $r$; on the outer region, use its $|t+t_2|$ bound and integrate the exponential kernel tail. This yields}
\begin{align}
 \|\vB\|
 &\leq\int_{\mathbb R^2}
 \|\vA_{\vH'}(t+t_2)-\vA_{\vH+\eta\vV_\ell}(t+t_2)\|
 |g(t)h_+(t_2)|\,\dd t\dd t_2\nonumber\\
 &\leq{\eta\|\vA\|e^{\poly(\beta^{\pm1})}
 \sum_{r\geq\ell-1}S(r,\vA)
 \left(e^{-c_1r}+e^{-c_2r/(d\beta)}\right)}
 \leq\eta\|\vA\|\exp\!\left[
 \poly(\beta^{\pm1})-\frac{c\ell}{{d}(1+\beta)}\right].
 \label{eq:spatial-truncation-inside-conjugation}
\end{align}
Here $c,c_1,c_2>0$ depend only on the fixed locality parameters.

However, this exponential spatial estimate cannot be passed directly through unrestricted Gibbs conjugation; the outer frequency cutoff controls that step. We now combine the spatial truncation with the outer
frequency cutoff. For any operator $\vB$, Gaussian
filtering and $\sigma=\beta^{-1}$ give, for
$|\omega_1|\leq\Omega_1$,
\begin{align}
    \|\vB_{\vH}(\omega_1)\|
    \leq
    \frac{\|\vB\|}{\sqrt{\sigma}(2\pi)^{1/4}},
    \qquad \left\|
        \Big(\sqrt{\vrho}^{-1}\vB\sqrt{\vrho}\Big)_{\vH}
        (\omega_1)
    \right\|
    \leq
    \frac{e^{\beta\Omega_1/2+1/4}}
    {\sqrt{\sigma}(2\pi)^{1/4}}\|\vB\|.
    \label{eq:outer-cutoff-operator-bounds}
\end{align}
Neither bound requires $\vB$ to be local.
By Schatten H\"older's inequality and
$\Tr(\vsigma)=1$, the KMS and GNS contributions therefore
satisfy
\begin{align}
    &\frac{1}{{4\pi\sigma}}
    \int_{-\Omega_1}^{\Omega_1}
    \left|
    \Tr\left[
        \vO^\dagger
        \left(
            \sqrt{\vsigma}
            \Big(\sqrt{\vrho}^{-1}\vB\sqrt{\vrho}\Big)_{\vH}
            (\omega_1)
            \sqrt{\vsigma}
            -
            \vB_{\vH}(\omega_1)\vsigma
        \right)
    \right]
    \right|
    \,\dd\omega_1
    \nonumber\\
    &\leq
    \frac{\|\vO\|}{{4\pi\sigma}}
    \int_{-\Omega_1}^{\Omega_1}
    \left(
        \left\|
            \Big(\sqrt{\vrho}^{-1}\vB\sqrt{\vrho}\Big)_{\vH}
            (\omega_1)
        \right\|
        +
        \|\vB_{\vH}(\omega_1)\|
    \right)
    \,\dd\omega_1
    \leq
    {\frac1{4\pi\sigma}}\,
    \frac{4\Omega_1 e^{\beta\Omega_1/2+1/4}}{{\sqrt{\sigma}(2\pi)^{1/4}}}
    \|\vO\|\|\vB\|.
    \label{eq:outer-cutoff-trace-bound}
\end{align}
The outer frequency cutoff thus allows us to control
the conjugated contribution by the unconjugated operator
norm, uniformly in the system size.

Apply \eqref{eq:outer-cutoff-trace-bound} to $\vB$ above. By \eqref{eq:spatial-truncation-inside-conjugation}, the change is at most
\begin{align}
 \eta\|\vO\|\|\vA\|
 \exp\!\left[\poly(\beta^{\pm1})-\frac{c\ell}{{d}(1+\beta)}\right].
 \label{eq:spatial-truncation-full-expression}
\end{align}
{Choose the polynomial radius $\ell$ after $\Omega_1,\Omega_2$ and the desired accuracy have been fixed. The outer-cutoff prefactor is then absorbed into $e^{\poly(\beta^{\pm1})}$. The radius is independent of $n$ and $\eta$; the background Hamiltonian $\vH$ is unchanged. The same estimate applies to $h_-$ and the opposite Gibbs conjugation.}
\end{proof}

We next control the spatially truncated contribution by
expanding the perturbed dynamics around $\vH$:
\begin{align}
    \vA_{\vH+\eta\vV_\ell}(t)
    &=
    \vA_{\vH}(t)
    +
    \sum_{m=1}^{\infty}(i\eta)^m
    \int_0^t\int_0^{s_1}\cdots\int_0^{s_{m-1}}
    \nonumber\\
    &\qquad\times
    \big[
        \vV_\ell(s_m),
        \big[
            \vV_\ell(s_{m-1}),\ldots,
            \big[
                \vV_\ell(s_1),\vA_{\vH}(t)
            \big]
            \ldots
        \big]
    \big]
    \,\dd s_m\cdots\dd s_1,
    \label{eq:spatially-truncated-dyson}
\end{align}
where $\vV_\ell(s):=e^{i\vH s}\vV_\ell e^{-i\vH s}$.
The series converges in operator norm for each fixed
real $t$, with oriented integrals when $t<0$.

The intuition behind the next two lemmas is that the ADB error of a low-frequency product of $q$ evolved single-site Pauli factors is bounded by a local ADB contribution multiplied by a factor exponential in $q$. At Dyson order $m$, the number of factors grows linearly in $m$, while spatial truncation ensures that the number of summands grows at most exponentially in $m$.

When the current coefficient error $\eta$ is smaller than a constant depending only on $\beta$ and the fixed locality parameters, the factor $\eta^m$ suppresses this exponential growth, making the integrated error bounds summable. \Cref{lem:zeroth-order-dyson,lem:dyson_term_adb_bound_m} make this intuition precise.

When $\eta$ is not sufficiently small, we instead truncate space, time, and the Dyson expansion, as described in \Cref{thm:coarse-adb-linearization}.

We first estimate the zeroth-order contribution to
\Cref{thm:kms-to-gns-reduction}, using the same
low-frequency ADB bound that will control the
higher-order terms.

\begin{lem}[Zeroth-order Dyson contribution]
\label{lem:zeroth-order-dyson}
For a single-site Pauli operator $\vA$ and any operator
$\vO$, define the zeroth-order contribution by
\begin{align}
    E_0
    &:=
    \frac{1}{{4\pi\sigma}}
    \int_{-\Omega_1}^{\Omega_1}
    \int_{-\infty}^{\infty}
    \int_{-\infty}^{\infty}
    \Tr\Big[
        \vO^\dagger
        \Big(
            \sqrt{\vsigma}
            \big(
                \vrho^{-1/2}
                \vA_{\vH}(t+t_2)
                \vrho^{1/2}
            \big)_{\vH}(-t,\omega_1)
            \sqrt{\vsigma}
    \nonumber\\
    &\hspace{5cm}
            -
            \big(\vA_{\vH}(t+t_2)\big)_{\vH}
            (-t,\omega_1)\vsigma
        \Big)
    \Big]
    g(t)h_+(t_2)\,\dd t_2\,\dd t\,\dd\omega_1.
    \label{eq:zeroth-order-contribution}
\end{align}
For the already fixed cutoffs $\Omega_1,\Omega_2$,
\begin{align}
    |E_0|
    \leq
    \|\vO\|e^{\poly(\beta^{\pm 1})}
    \mathsf{ADB}_{\vA}[\vsigma]^{\frac{1}{32\beta d+2}}.
    \label{eq:zeroth-order-adb-bound}
\end{align}
\end{lem}

\begin{proof}
Since $\vrho$ commutes with $\vH$, and frequency
filtering commutes with the $\vH$-evolution, we have
\begin{align}
    E_0
    =
    \frac{1}{{4\pi\sigma}}
    \int_{-\infty}^{\infty}
    \int_{-\infty}^{\infty}
    \Tr\Big[
        \vO^\dagger
        \Big(
            \sqrt{\vsigma}\vrho^{-1/2}
            \vA_{\vH}(t_2)_{\leq\Omega_1}\vrho^{1/2}
            -
            \vA_{\vH}(t_2)_{\leq\Omega_1}\sqrt{\vsigma}
        \Big)
        \sqrt{\vsigma}
    \Big]
    g(t)h_+(t_2)\,\dd t_2\,\dd t.
    \label{eq:zeroth-order-cancellation}
\end{align}

The symmetric frequency cutoff preserves Hermiticity.
Thus, taking the adjoint in the time-evolved assertion
of \Cref{lem:low_freq_bound}, with $\Delta=\Omega_1$,
gives
\begin{align}
    \left\|
        \sqrt{\vsigma}\vrho^{-1/2}
        \vA_{\vH}(t_2)_{\leq\Omega_1}\vrho^{1/2}
        -
        \vA_{\vH}(t_2)_{\leq\Omega_1}\sqrt{\vsigma}
    \right\|_2
    \leq
    e^{\poly(\beta^{\pm 1})}
    e^{\frac{\pi|t_2|}{\beta(16\beta d+1)}}
    \mathsf{ADB}_{\vA}[\vsigma]^{\frac{1}{32\beta d+2}},
    \label{eq:zeroth-order-filtered-adb}
\end{align}
where the cutoff-dependent factor has been absorbed
into the polynomial, since $\Omega_1$ is polynomially
bounded in $k,\beta^{\pm1}$.

Schatten H\"older's inequality and
$\|\sqrt{\vsigma}\|_2=1$ now imply
\begin{align}
    |E_0|
    \leq
    \frac{\|\vO\|}{{4\pi\sigma}}
    e^{\poly(\beta^{\pm 1})}
    \mathsf{ADB}_{\vA}[\vsigma]^{\frac{1}{32\beta d+2}}
    \times
    \int_{-\infty}^{\infty}
    \int_{-\infty}^{\infty}
    |g(t)h_+(t_2)|
    e^{\frac{\pi|t_2|}{\beta(16\beta d+1)}}
    \,\dd t_2\,\dd t
    \leq
    \|\vO\|e^{\poly(\beta^{\pm 1})}
    \mathsf{ADB}_{\vA}[\vsigma]^{\frac{1}{32\beta d+2}}.
\end{align}
For the last inequality, use the previously established
kernel bounds
\[
    \int_{-\infty}^{\infty}|g(t)|\,\dd t\leq {C_{\beta}},
    \qquad
    |h_+(t_2)|
    \leq
    {C_{\beta}}\frac{\sqrt{\sigma}}{\beta}
    e^{\beta\Omega_2}e^{-\sigma^2t_2^2}.
\]
All remaining prefactors are absorbed into
$e^{\poly(\beta^{\pm 1})}$.
\end{proof}

The zeroth-order Dyson contribution is independent of the current coefficient error $\eta$. It therefore leaves an ADB-dependent error floor in the learning guarantee, whereas this contribution vanishes for an exact Gibbs state.

We now turn to the Dyson terms of order $m\geq1$, each
of which carries an explicit factor $\eta^m$.
{The spatial truncation leaves at most $N_\ell\leq C_{\beta}(1+\ell)^D$ local Pauli terms in $\vV_\ell$. The order-$m$ expansion therefore contains at most $N_\ell^m$ coefficient-weighted nested commutators.}
Each nested commutator further expands into at most
$2^m$ ordered products. We therefore estimate the
contribution of each ordered tuple, keeping the
dependence on $m$ explicit for the subsequent
summation with the factor $\eta^m$.

For local Pauli terms $\vP_1,\ldots,\vP_m$ drawn from
$\vV_\ell$, write $\mathbf{P}=(\vP_1,\ldots,\vP_m)$ and
define
\begingroup
\begin{align}
    \operatorname{ad}^{(m)}_{\mathbf{P}}[\vA](0,t)
    :=
    \int_0^t\int_0^{t_1}\cdots\int_0^{t_{m-1}}
    \big[
        \vP_m(t_m),
        \ldots,
        \big[
            \vP_1(t_1),\vA_{\vH}(t)
        \big]
        \ldots
    \big]
    \,\dd t_m\cdots\dd t_1 .
\end{align}
\endgroup
All time evolutions in this definition are generated
by $\vH$. For $t<0$, the integrals retain their oriented
limits: reversing all limits contributes an overall
factor $(-1)^m$, which disappears in norm estimates.
For either sign of $t$, the ordered integration region
has volume $|t|^m/m!$.

With this notation, the $m$th-order Dyson term is
\[
    i^m\eta^m
    \sum_{\mathbf{P}}
    \operatorname{ad}^{(m)}_{\mathbf{P}}[\vA](0,t).
\]
Here and below, $\sum_{\mathbf{P}}$ includes the product
of the Pauli coefficients corresponding to
$\vP_1,\ldots,\vP_m$ in $\vV_\ell$.

The next lemma bounds the full coefficient-weighted contribution at each positive order.

\begin{lem}
    {Fix the cutoffs and radius selected above. For every integer $m\geq1$ and every auxiliary cutoff $\Omega\geq0$, denote the absolute order-$m$ contribution by $|E_m|$:}
    \begin{align}
    {|E_m|=}{}&\frac{\eta^m}{{4\pi\sigma}}
    \Bigg|
    \sum_{\mathbf P}
    \int_{-\Omega_1}^{\Omega_1}
    \int_{-\infty}^{\infty}
    \int_{-\infty}^{\infty}
    \Tr\Bigg[
        \vO^\dagger
        \Bigg(
            \sqrt{\vsigma}\vrho^{-1/2}
            \Big(
                \operatorname{ad}^{(m)}_{\mathbf P}
                [\vA](0,t+t_2)
            \Big)_{\vH}(-t,\omega_1)
            \vrho^{1/2}\sqrt{\vsigma}
    \nonumber\\
    &\hspace{5cm}
            -
            \Big(
                \operatorname{ad}^{(m)}_{\mathbf P}
                [\vA](0,t+t_2)
            \Big)_{\vH}(-t,\omega_1)\vsigma
        \Bigg)
    \Bigg]
    g(t)h_+(t_2)\,
    \dd t_2\,\dd t\,\dd\omega_1
    \Bigg| .
\end{align}
    {There is a constant $C_\beta\leq e^{\poly(\beta^{\pm1})}$, independent of $m,\eta,\Omega$ and $n$, such that}
\begin{align}
 {|E_m|}\leq{}&\frac{1}{{4\pi\sigma}}
 \|\vO\|\|\vA\|(\eta{C_{\beta}})^m
 \left(e^{-\Omega/(8d)}+
 {C_{\beta}\epsilon_{\mathrm{ADB}}^{\frac{1}{32\beta d+2}}}
 e^{\frac{16\beta^2d+1}{16\beta d+1}\Omega}\right).
\end{align}
\label{lem:dyson_term_adb_bound_m}
\end{lem}
\emph{Proof.}
We first estimate the low-frequency contribution. For real
endpoints $u,v$, define
\begin{align}
    \operatorname{ad}^{(m)}_{\mathbf P}[\vA](u,v;\Omega)
    &:=
    \frac{1}{(2\sigma\sqrt{2\pi})^{(m+1)/2}}
    \int_u^v\int_u^{s_1}\cdots\int_u^{s_{m-1}}
    \int_{\sum_{j=0}^{m}|\omega_j|\leq\Omega}
    \nonumber\\
    &\quad\times
    \big[
        \vP_m(s_m,\omega_m),\ldots,
        \big[
            \vP_1(s_1,\omega_1),\vA_{\vH}(v,\omega_0)
        \big]
        \ldots
    \big]
    \,\dd\vomega\,\dd s_m\cdots\dd s_1 .
    \label{eq:dyson-low-frequency-definition}
\end{align}
Here $\vomega=(\omega_0,\ldots,\omega_m)$ includes the
frequency of $\vA$, all evolutions are generated by $\vH$,
and the time integrals retain their oriented limits.
The prefactor is the Fourier reconstruction normalization for the $m+1$ factors.
{For arbitrary endpoints, $\operatorname{ad}^{(m)}_{\mathbf P}[\vA](u,v)$ denotes the same oriented time integral without frequency filtering, equivalently the integral over all $\mathbb R^{m+1}$ with this same reconstruction factor. Define the complementary term by}
\begin{align}
 {\operatorname{ad}^{(m)}_{\mathbf P}[\vA](u,v;>\Omega)
 :=\operatorname{ad}^{(m)}_{\mathbf P}[\vA](u,v)
 -\operatorname{ad}^{(m)}_{\mathbf P}[\vA](u,v;\Omega).}
\end{align}
{Two arguments $(s,\omega)$ on an operator denote its real-time evolution followed by its Gaussian frequency filter; $\mathbf P$ is an ordered tuple.}

Throughout this part, ${C_{\beta}}\geq1$ is independent of $m$, $\Omega$, the
time variables, and the system size, and satisfies
${C_{\beta}}\leq e^{\poly(\beta^{\pm 1})}$; its value may increase
between estimates. It suffices to consider $\|\vA\|=1$.

{At fixed ordered times and frequencies, expand the nested commutator into at most $2^m$ ordered products and estimate each product separately.} Consider one such product
and write
\[
    \vX_j:=\vQ_{j,\vH}(s_j,\omega_j),
    \qquad j=0,\ldots,m,
\]
where $\vQ_j$ are $\vA,\vP_1,\ldots,\vP_m$ in the
order of that product, with their corresponding times
and frequencies relabeled accordingly.  Telescope the product defect:
\begin{align}
    &\sqrt{\vsigma}\vrho^{-1/2}
    (\prod_{i=0}^{m}\vX_i)\vrho^{1/2}
    -(\prod_{i=0}^{m}\vX_i)\sqrt{\vsigma}
    =
    \sum_{j=0}^{m}
    (\prod_{i<j}\vX_i)
    \Big(
        \sqrt{\vsigma}\vrho^{-1/2}\vX_j\vrho^{1/2}
        -\vX_j\sqrt{\vsigma}
    \Big)
    \vrho^{-1/2}
    (\prod_{i>j}\vX_i)\vrho^{1/2}
    \label{eq:dyson-low-frequency-telescoping}
\end{align}
All products retain their displayed order, and empty
products are the identity.

Fix $j$ and all frequencies except $\omega_j$. Define
    $r_j:=\sum_{i\neq j}|\omega_i|$ and $\Delta_j:=\Omega-r_j.$
For $r_j\leq\Omega$, the remaining frequency ranges over
$[-\Delta_j,\Delta_j]$. The prefix and suffix in
\eqref{eq:dyson-low-frequency-telescoping} are independent
of $\omega_j$, so we integrate this frequency before
taking the norm. The middle factor becomes
\begin{align}
    \int_{-\Delta_j}^{\Delta_j}
    \Big(
        \sqrt{\vsigma}\vrho^{-1/2}\vX_j\vrho^{1/2}
        -\vX_j\sqrt{\vsigma}
    \Big)\,\dd\omega_j
    =
    \sqrt{\vsigma}\vrho^{-1/2}
    \vQ_{j,\vH}(s_j)_{\leq\Delta_j}\vrho^{1/2}
    -\vQ_{j,\vH}(s_j)_{\leq\Delta_j}\sqrt{\vsigma},
\end{align}
where the subscript $\leq\Delta_j$ denotes the unnormalized
frequency integral. Since this symmetric cutoff preserves
Hermiticity, the adjoint of the time-evolved estimate in
\Cref{lem:low_freq_bound} gives
\begin{align}
    \left\|
        \sqrt{\vsigma}\vrho^{-1/2}
        \vQ_{j,\vH}(s_j)_{\leq\Delta_j}\vrho^{1/2}
        -\vQ_{j,\vH}(s_j)_{\leq\Delta_j}\sqrt{\vsigma}
    \right\|_2
    \leq
    {C_{\beta}}\epsilon_{\mathrm{ADB}}^{\frac{1}{32\beta d+2}}
    e^{\frac{8\beta^2d+1/4}{16\beta d+1}\Delta_j}
    e^{\frac{\pi|s_j|}{\beta(16\beta d+1)}}.
    \label{eq:dyson-low-frequency-active-factor}
\end{align}
We now integrate over the frequency simplex while keeping the other factors in their original order.
\begin{align}
    &\frac{1}{(2\sigma\sqrt{2\pi})^{(m+1)/2}}
    \Bigg\|
    \int_{\sum_{i=0}^{m}|\omega_i|\leq\Omega}
    \Bigg[
        \sqrt{\vsigma}\vrho^{-1/2}
        \Big(\prod_{i=0}^{m}\vX_i\Big)\vrho^{1/2}
        -
        \Big(\prod_{i=0}^{m}\vX_i\Big)\sqrt{\vsigma}
    \Bigg]\,\dd\vomega
    \Bigg\|_2\\
    &\leq
    \frac{1}{(2\sigma\sqrt{2\pi})^{(m+1)/2}}
    \sum_{j=0}^{m}
    \int_{r_j\leq\Omega}
    \|\prod_{i<j}\vX_i\|
    \|
        \vrho^{-1/2}
        \Big(\prod_{i>j}\vX_i\Big)\vrho^{1/2}
    \|\\
    &\qquad\qquad\times
    \left\|
        \sqrt{\vsigma}\vrho^{-1/2}
        \vQ_{j,\vH}(s_j)_{\leq\Delta_j}\vrho^{1/2}
        -
        \vQ_{j,\vH}(s_j)_{\leq\Delta_j}\sqrt{\vsigma}
    \right\|_2
    \prod_{i\neq j}\dd\omega_i\\
\end{align}
Applying \Cref{eq:dyson-low-frequency-telescoping} {and}
integrating $\omega_j$ before taking the norm, we obtain the following bound for the equation above:
\begin{align}
    &\quad
    {C_{\beta}}^{m+1}
    \epsilon_{\mathrm{ADB}}^{\frac{1}{32\beta d+2}}
    \sum_{j=0}^{m}
    \int_{r_j\leq\Omega}
    e^{\frac{16\beta^2d+1/2}{16\beta d+1}r_j}
    e^{\frac{16\beta^2d+1/2}{16\beta d+1}\Delta_j}
    e^{\frac{\pi|s_j|}{\beta(16\beta d+1)}}
    e^{-\frac{8\beta^2d+1/4}{16\beta d+1}r_j}\prod_{i\neq j}\dd\omega_i
    \nonumber\\
    &\quad\leq
    (m+1){C_{\beta}}^{m+1}
    \epsilon_{\mathrm{ADB}}^{\frac{1}{32\beta d+2}}
    e^{\frac{16\beta^2d+1/2}{16\beta d+1}\Omega}
    e^{\frac{\pi\max_j|s_j|}{\beta(16\beta d+1)}}.
    \label{eq:dyson-low-frequency-word-bound}
\end{align}
For the remaining factors, the local Fourier bound combined with Gibbs conjugation gives
\begin{align}
 {\|\prod_{i<j}\vX_i\|
 \| \vrho^{-1/2}\Big(\prod_{i>j}\vX_i\Big)\vrho^{1/2}\|
 \leq C_{\beta}^m e^{\max\{0,\beta/2-1/(4d)\}r_j}
 \leq C_{\beta}^m e^{\frac{8\beta^2d+1/4}{16\beta d+1} r_j}.}
\end{align}
{Together with \eqref{eq:dyson-low-frequency-active-factor} and $\Delta_j\geq0$, this gives the inequality above, after absorbing the reconstruction factor into $C_{\beta}^{m+1}$.}
{The inequality uses $r_j+\Delta_j=\Omega$ and the following estimate of the remaining frequency integral:}
\begin{align}
 \int_{r_j\leq\Omega}e^{-{\frac{8\beta^2d+1/4}{16\beta d+1}} r_j}\prod_{i\neq j}\dd\omega_i
 &\leq{\prod_{i\neq j}\int_{\mathbb R}
 e^{-\frac{8\beta^2d+1/4}{16\beta d+1}|\omega_i|}\,\dd\omega_i
 =\left(\frac{2(16\beta d+1)}{8\beta^2d+1/4}\right)^m\leq C_{\beta}^m.}
\end{align}

We now apply this estimate inside the outer frequency
filter. For any operator $\vB$,
\begin{align}
    \int_{-\Omega_1}^{\Omega_1}
    \vB_{\vH}(-t,\omega_1)\,\dd\omega_1
    =
    \frac{1}{\sqrt{2\pi}}
    \int_{-\infty}^{\infty}
    \vB_{\vH}(t_1-t)f_\sigma(t_1)
    \frac{2\sin(\Omega_1t_1)}{t_1}\,\dd t_1.
\end{align}
Since $[\vH,\vrho]=0$, this identity applies inside the
Gibbs conjugation, and the common evolution shifts the
Dyson interval according to
\[
    \Big(
        \operatorname{ad}^{(m)}_{\mathbf P}[\vA]
        (0,t+t_2;\Omega)
    \Big)_{\vH}(t_1-t)
    =
    \operatorname{ad}^{(m)}_{\mathbf P}[\vA]
    (-t+t_1,t_1+t_2;\Omega).
\]
The ordered time region has volume $|t+t_2|^m/m!$,
and every factor in the shifted product has time
argument bounded by
$|t|+|t_1|+|t_2|$. Summing the $2^m$ products in the
nested commutator and applying
\eqref{eq:dyson-low-frequency-word-bound} therefore gives
\begin{align}
    &\left\|
        \sqrt{\vsigma}\vrho^{-1/2}
        \operatorname{ad}^{(m)}_{\mathbf P}[\vA]
        (-t+t_1,t_1+t_2;\Omega)\vrho^{1/2}
        -\operatorname{ad}^{(m)}_{\mathbf P}[\vA]
        (-t+t_1,t_1+t_2;\Omega)\sqrt{\vsigma}
    \right\|_2
    \nonumber\\
    &\leq
    \|\vA\|\,2^m(m+1){C_{\beta}}^{m+1}
    \frac{|t+t_2|^m}{m!}
    \epsilon_{\mathrm{ADB}}^{\frac{1}{32\beta d+2}}
    e^{\frac{16\beta^2d+1/2}{16\beta d+1}\Omega}
    e^{\frac{\pi{(|t|+|t_1|+|t_2|)}}{\beta(16\beta d+1)}}.
    \label{eq:dyson-low-frequency-ordered-bound}
\end{align}
The same estimate holds for either orientation of the
Dyson interval.

We now substitute the low-frequency Dyson terms into
the regularized KMS-to-GNS expression. By the preceding
time-domain representation of the outer frequency
filter, the common $\vH$-evolution shifts the Dyson
interval from $(0,t+t_2)$ to $(-t+t_1,t_1+t_2)$.

Factoring out the rightmost $\sqrt{\vsigma}$, the
absolute value of the low-frequency contribution is
\begin{align}
    &\frac{\eta^m}
    {{4\pi\sigma}\sqrt{2\pi}}
    \Bigg|
    \sum_{\mathbf P}
    \int_{-\infty}^{\infty}
    \int_{-\infty}^{\infty}
    \int_{-\infty}^{\infty}
    \Tr\Bigg[
        \vO^\dagger
        \Bigg(
            \sqrt{\vsigma}\vrho^{-1/2}
            \operatorname{ad}^{(m)}_{\mathbf P}[\vA]
            (-t+t_1,t_1+t_2;\Omega)\vrho^{1/2}
    \nonumber\\
    &\hspace{3cm}
            -
            \operatorname{ad}^{(m)}_{\mathbf P}[\vA]
            (-t+t_1,t_1+t_2;\Omega)\sqrt{\vsigma}
        \Bigg)
        \sqrt{\vsigma}
    \Bigg]
    g(t)h_+(t_2)f_\sigma(t_1)
    \frac{2\sin(\Omega_1t_1)}{t_1}
    \,\dd t_2\,\dd t_1\,\dd t
    \Bigg|
    \nonumber\\
    &\leq
    \frac{\eta^m\|\vO\|}
    {{4\pi\sigma}\sqrt{2\pi}}
    \int_{-\infty}^{\infty}
    \int_{-\infty}^{\infty}
    \int_{-\infty}^{\infty}
    \Bigg\|
    \sum_{\mathbf P}
    \Bigg(
            \sqrt{\vsigma}\vrho^{-1/2}
            \operatorname{ad}^{(m)}_{\mathbf P}[\vA]
            (-t+t_1,t_1+t_2;\Omega)\vrho^{1/2}
    \nonumber\\
    &\hspace{3cm}
            -
            \operatorname{ad}^{(m)}_{\mathbf P}[\vA]
            (-t+t_1,t_1+t_2;\Omega)\sqrt{\vsigma}
    \Bigg)
    \Bigg\|_2
    \left|
        g(t)h_+(t_2)f_\sigma(t_1)
        \frac{2\sin(\Omega_1t_1)}{t_1}
    \right|
    \,\dd t_2\,\dd t_1\,\dd t .
    \label{eq:dyson-low-frequency-trace-to-defect}
\end{align}
Here we used $|i^m|=1$ and
$\|\vO\sqrt{\vsigma}\|_2\leq\|\vO\|$.
The additional factor $1/\sqrt{2\pi}$ comes from
the outer Fourier transform.

Each summand inside the norm is precisely the defect
bounded in \eqref{eq:dyson-low-frequency-ordered-bound}.
Applying that estimate termwise and using
$|t+t_2|\leq |t|+|t_1|+|t_2|$, we obtain
\begin{align}
    &\Bigg\|
    \sum_{\mathbf P}
    \Bigg(
        \sqrt{\vsigma}\vrho^{-1/2}
        \operatorname{ad}^{(m)}_{\mathbf P}[\vA]
        (-t+t_1,t_1+t_2;\Omega)\vrho^{1/2}
        -
        \operatorname{ad}^{(m)}_{\mathbf P}[\vA]
        (-t+t_1,t_1+t_2;\Omega)\sqrt{\vsigma}
    \Bigg)
    \Bigg\|_2
    \nonumber\\
    &\leq
    \|\vA\|(m+1){C_{\beta}}^{m+1}
    \frac{(|t|+|t_1|+|t_2|)^m}{m!}
    \epsilon_{\mathrm{ADB}}^{\frac{1}{32\beta d+2}}
    e^{\frac{16\beta^2d+1/2}{16\beta d+1}\Omega}
    e^{\frac{\pi{(|t|+|t_1|+|t_2|)}}{\beta(16\beta d+1)}}.
\end{align}
The Pauli coefficient products have magnitude at most
one, and the number of ordered tuples is at most
${C_{\beta}}^m$. Their contribution and the powers of two in
\eqref{eq:dyson-low-frequency-ordered-bound} are absorbed
by increasing ${C_{\beta}}$.

Substituting this bound into
\eqref{eq:dyson-low-frequency-trace-to-defect}, and
absorbing the fixed Fourier normalization into ${C_{\beta}}$,
bounds the low-frequency contribution by
\begin{align}
    &\frac{\|\vO\|\|\vA\|}{{4\pi\sigma}}
    \eta^m(m+1){C_{\beta}}^{m+1}
    \epsilon_{\mathrm{ADB}}^{\frac{1}{32\beta d+2}}
    e^{\frac{16\beta^2d+1/2}{16\beta d+1}\Omega}
    \nonumber\\
    &\quad\times
    \int_{-\infty}^{\infty}
    \int_{-\infty}^{\infty}
    \int_{-\infty}^{\infty}
    \frac{(|t|+|t_1|+|t_2|)^m}{m!}
    \left|
        g(t)h_+(t_2)f_\sigma(t_1)
        \frac{2\sin(\Omega_1t_1)}{t_1}
    \right|
    e^{\frac{\pi{(|t|+|t_1|+|t_2|)}}{\beta(16\beta d+1)}}
    \,\dd t_2\,\dd t_1\,\dd t .
    \label{eq:dyson-low-frequency-time-integral}
\end{align}
{Set $R:=|t|+|t_1|+|t_2|$.} The established kernel bounds, $\sigma=\beta^{-1}$,
and $|\sin(\Omega_1t_1)/t_1|\leq\Omega_1$ imply the kernel of the integral is bounded by ${C_{\beta}} e^{-\pi R/\beta}$,
because the growth rate is at most $\pi/\beta$, while
$|g(t)|\leq {C_{\beta}} e^{-2\pi|t|/\beta}$ and the other two
kernels have Gaussian decay.
The remaining time moment is
\begin{align}
    &\int_{-\infty}^{\infty}
    \int_{-\infty}^{\infty}
    \int_{-\infty}^{\infty}
    \frac{R^m}{m!}e^{-\pi R/\beta}
    \,\dd t_2\,\dd t_1\,\dd t
    =
    4(m+1)(m+2)
    \left(\frac{\beta}{\pi}\right)^{m+3}.
\end{align}
Substituting into
\eqref{eq:dyson-low-frequency-time-integral}, absorbing
polynomial factors in $m$ into the order-$m$ constant,
and using $m\geq1$, we obtain the bound
\begin{align}
    \frac{\|\vO\|\|\vA\|}{{4\pi\sigma}}
    \Big(\eta e^{\poly(\beta^{\pm 1})}\Big)^m
    e^{\poly(\beta^{\pm 1})}
    \epsilon_{\mathrm{ADB}}^{\frac{1}{32\beta d+2}}
    e^{\frac{{16}\beta^2d+{1}}{16\beta d+1}\Omega}.
    \label{eq:dyson-low-frequency-final}
\end{align}
In the last step we also enlarged the coefficient of
$\Omega$ to match the statement of the lemma.
This bounds the low-frequency contribution.

For the high-frequency contribution, we directly estimate
the operator norm. We use the complementary cutoff
$\sum_{i=0}^{m}|\omega_i|>\Omega$.
Throughout, $m\geq1$ and ${C_{\beta}}\geq1$
may change from line to line. It depends only on the local
parameters and the already fixed cutoffs $\Omega_1,\Omega_2$,
with ${C_{\beta}}\leq e^{\poly(\beta^{\pm 1})}$, and is independent
of $m$, $\eta$, $\Omega$, the time variables, and the system size.

For any real endpoints $u,v$, the established local Fourier
bounds and the ordered integration volume $|v-u|^m/m!$ give
\begin{align}
    \left\|
        \operatorname{ad}^{(m)}_{\mathbf P}[\vA]
        (u,v;{>}\Omega)
    \right\|
    &\leq
    \|\vA\|{C_{\beta}}^{m+1}\frac{|v-u|^m}{m!}
    \int_{\sum_{i=0}^{m}|\omega_i|>\Omega}
    e^{-\sum_{i=0}^{m}|\omega_i|/(4d)}\,\dd\vomega
    \nonumber\\
    &\leq
    \|\vA\|{C_{\beta}}^{m+1}\frac{|v-u|^m}{m!}e^{-\Omega/(8d)}
    \int_{\sum_{i=0}^{m}|\omega_i|>\Omega}
    e^{-\sum_{i=0}^{m}|\omega_i|/(8d)}\,\dd\vomega\\
    &\leq
    \|\vA\|{C_{\beta}}^{m+1}\frac{|v-u|^m}{m!}
    e^{-\Omega/(8d)}.
    \label{eq:dyson-high-frequency-operator-bound}
\end{align}
The first inequality absorbs the $2^m$ commutator terms
and the reconstruction normalization into ${C_{\beta}}^{m+1}$.
For either orientation of the time interval, the absolute
ordered integration volume is $|v-u|^m/m!$.
{The last inequality uses}
\begin{align}
    \int_{\sum_{i=0}^{m}|\omega_i|>\Omega}
    e^{-\sum_{i=0}^{m}|\omega_i|/(8d)}\,\dd\vomega\leq \int_{\sum_{i=0}^{m}|\omega_i|\geq 0}
    e^{-\sum_{i=0}^{m}|\omega_i|/(8d)}\,\dd\vomega{=}(16d)^{{m+1}}.
\end{align}
The remaining factors are absorbed into ${C_{\beta}}^{m+1}$.

We first estimate the KMS contribution. The common
$\vH$-evolution still shifts the integration interval in the Dyson expansion exactly.
Using $|i^m|=1$ and Schatten
H\"older's inequality with $\Tr(\vsigma)=1$, we obtain
\begin{align}
    &\frac{\eta^m}{{4\pi\sigma}}
    \Bigg|
    \sum_{\mathbf P}
    \int_{-\Omega_1}^{\Omega_1}\!
    \int_{-\infty}^{\infty}\!
    \int_{-\infty}^{\infty}
    \Tr\Bigg[
        \vO^\dagger\sqrt{\vsigma}
        \times
        \Big(
            \vrho^{-1/2}
            \operatorname{ad}^{(m)}_{\mathbf P}[\vA]
            (0,t+t_2;{>}\Omega)\vrho^{1/2}
        \Big)_{\vH}(-t,\omega_1)
        \sqrt{\vsigma}
    \Bigg]
    g(t)h_+(t_2)\,\dd t_2\,\dd t\,\dd\omega_1
    \Bigg|
    \nonumber\\
    &\leq
    \frac{\eta^m\|\vO\|}{{4\pi\sigma}}
    \int_{-\infty}^{\infty}|g(t)|
    \Bigg\|
    \int_{-\Omega_1}^{\Omega_1}
    \Bigg(
        \vrho^{-1/2}
        \times
        \int_{-\infty}^{\infty}
        \sum_{\mathbf P}
        \operatorname{ad}^{(m)}_{\mathbf P}[\vA]
        (-t,t_2;{>}\Omega)h_+(t_2)\,\dd t_2
        \,\vrho^{1/2}
    \Bigg)_{\vH}(\omega_1)\,\dd\omega_1
    \Bigg\|\,\dd t
    \nonumber\\
    &\leq
    \frac{\eta^m\|\vO\|}{{4\pi\sigma}}
    \frac{2\Omega_1e^{\beta\Omega_1/2+1/4}}
    {\sqrt{\sigma}(2\pi)^{1/4}}
    \int_{-\infty}^{\infty}|g(t)|
        \times
    \Bigg\|
        \int_{-\infty}^{\infty}
        \sum_{\mathbf P}
        \operatorname{ad}^{(m)}_{\mathbf P}[\vA]
        (-t,t_2;{>}\Omega)h_+(t_2)\,\dd t_2
    \Bigg\|\,\dd t
    \nonumber\\
    &\leq
    \frac{\|\vO\|\|\vA\|}{{4\pi\sigma}}
    \eta^m {C_{\beta}}^{m+1}e^{-\Omega/(8d)}
    \int_{-\infty}^{\infty}
    \int_{-\infty}^{\infty}
    \frac{|t+t_2|^m}{m!}|g(t)h_+(t_2)|
    \,\dd t_2\,\dd t.
    \label{eq:dyson-high-frequency-kms-integral}
\end{align}
The second inequality is the previously established
outer-cutoff norm estimate with $\sigma=\beta^{-1}$.
For the last inequality, apply
\eqref{eq:dyson-high-frequency-operator-bound} with
$(u,v)=(-t,t_2)$. The number of Pauli tuples is at most
${C_{\beta}}^m$, and their coefficient products have magnitude at
most one. This count and the fixed outer-cutoff prefactor
are absorbed by increasing ${C_{\beta}}$.

The time integral is controlled by the same kernel bound
used above:
\begin{align}
    \int_{-\infty}^{\infty}
    \int_{-\infty}^{\infty}
    \frac{|t+t_2|^m}{m!}|g(t)h_+(t_2)|
    \,\dd t_2\,\dd t
    \leq
    {C_{\beta}}\int_{-\infty}^{\infty}
    \int_{-\infty}^{\infty}
    \frac{(|t|+|t_2|)^m}{m!}
    e^{-\pi(|t|+|t_2|)/\beta}
    \,\dd t_2\,\dd t
    =
    4{C_{\beta}}(m+1)\left(\frac{\beta}{\pi}\right)^{m+2}.
    \label{eq:dyson-high-frequency-time-moment}
\end{align}
Substituting into
\eqref{eq:dyson-high-frequency-kms-integral} and using
$m\geq1$, the KMS contribution is therefore at most
\begin{align}
    \frac{\|\vO\|\|\vA\|}{{4\pi\sigma}}
    \Big(\eta e^{\poly(\beta^{\pm 1})}\Big)^m
    e^{-\Omega/(8d)}.
    \label{eq:dyson-high-frequency-kms-bound}
\end{align}
All factors independent of $m$, as well as the remaining
polynomial factor in $m$, are absorbed into the
order-$m$ prefactor.

For the GNS contribution, use the same argument without
imaginary-time conjugation. The outer frequency integral
has norm at most $2\Omega_1/[\sqrt{\sigma}(2\pi)^{1/4}]$
times the input operator norm. Thus
\begin{align}
    &\frac{\eta^m}{{4\pi\sigma}}
    \Bigg|
    \sum_{\mathbf P}
    \int_{-\Omega_1}^{\Omega_1}
    \int_{-\infty}^{\infty}
    \int_{-\infty}^{\infty}
    \Tr\Bigg[
        \vO^\dagger
        \times
        \Big(
            \operatorname{ad}^{(m)}_{\mathbf P}[\vA]
            (0,t+t_2;{>}\Omega)
        \Big)_{\vH}(-t,\omega_1)\vsigma
    \Bigg]
    g(t)h_+(t_2)\,\dd t_2\,\dd t\,\dd\omega_1
    \Bigg|
    \nonumber\\
    &\leq
    \frac{\eta^m\|\vO\|}{{4\pi\sigma}}
    \frac{2\Omega_1}{\sqrt{\sigma}(2\pi)^{1/4}}
    \int_{-\infty}^{\infty}
    \int_{-\infty}^{\infty}
    \Bigg\|
        \sum_{\mathbf P}
        \operatorname{ad}^{(m)}_{\mathbf P}[\vA]
        (-t,t_2;{>}\Omega)
    \Bigg\|
    |g(t)h_+(t_2)|\,\dd t_2\,\dd t
    \nonumber\\
    &\leq
    \frac{\|\vO\|\|\vA\|}{{4\pi\sigma}}
    \Big(\eta e^{\poly(\beta^{\pm 1})}\Big)^m
    e^{-\Omega/(8d)}.
    \label{eq:dyson-high-frequency-gns-bound}
\end{align}
Adding the KMS and GNS bounds and increasing the
polynomial once more gives the high-frequency term in
\Cref{lem:dyson_term_adb_bound_m}. Combining it with the
preceding low-frequency estimate completes the proof.\qed

Combining \Cref{lem:dyson_term_adb_bound_m} with \Cref{lem:zeroth-order-dyson} gives the following proof of \Cref{thm:kms-to-gns-reduction}.
\begin{proof}[Proof of \Cref{thm:kms-to-gns-reduction}]
{Choose the spatial radius to give half the desired $\eta$-dependent budget. After the radius and both frequency cutoffs are fixed, \Cref{lem:dyson_term_adb_bound_m} gives a constant $C_\beta$ uniform in every positive order. Write $b_\beta=(16\beta^2d+1)/(16\beta d+1)$. If $\eta C_{\beta}\leq1/2$, then}
\begin{align}
 {\sum_{m\geq1}|E_m|}
 {\leq\frac{\|\vO\|\|\vA\|}{4\pi\sigma}
 \frac{\eta C_{\beta}}{1-\eta C_{\beta}}
 \left(e^{-\Omega/(8d)}+C_{\beta}\epsilon_{\mathrm{ADB}}^{\frac{1}{32\beta d+2}}e^{b_\beta\Omega}\right)}
    {\leq\eta\|\vO\|\|\vA\|e^{p(\beta^{\pm1})}
 \left(e^{-\Omega/(8d)}+\epsilon_{\mathrm{ADB}}^{\frac{1}{32\beta d+2}}e^{b_\beta\Omega}\right).}
\end{align}
{For target polynomial $P_1$, select a single cutoff and a state-independent ADB threshold satisfying}
\begin{align}
 {\frac{\Omega}{8d}\geq P_1+p+\log4,\qquad
 \log\frac1{\epsilon_{\mathrm{ADB}}}
 \geq(32\beta d+2)(P_1+p+b_\beta\Omega+\log4).}
\end{align}
{These choices make the sum at most $\eta\|\vO\|\|\vA\|e^{-P_1}/2$. They are polynomial in $\beta^{\pm1}$ and do not introduce powers of $\Omega$ depending on the Dyson order. The case $\epsilon_{\mathrm{ADB}}=0$ follows directly by omitting its term.}

{To justify termwise integration, the unfiltered Dyson series has the norm majorant $\|\vA\|e^{2\eta\|\vV_\ell\||t+t_2|}$. The outer-cutoff bound also controls its Gibbs-conjugated contribution. Impose $\eta\|\vV_\ell\|\leq\pi/(2\beta)$; the decay of $g$ and the Gaussian decay of $h_+$ then provide an integrable majorant. Add the spatial budget and the zeroth-order estimate from \Cref{lem:zeroth-order-dyson}. This proves the stated bound, including its ADB floor.}
\end{proof}
Finally, we prove \Cref{thm:coarse-adb-linearization} by truncating the variables.
\begin{proof}[Proof of \Cref{thm:coarse-adb-linearization}]
\begingroup
Fix the desired error $e^{-P(\beta^{\pm1})}$. Fix the observable cutoffs $\Omega_1,\Omega_2$ and then choose the spatial truncation radius polynomially large in $\beta^{\pm1}$. By \Cref{lem:kms-gns-perturbation-truncation}, the spatial replacement contributes at most a prescribed fraction of this error. Write the truncated candidate as $\vH+\eta\vV_\ell$, where $\eta\leq2$ and $\|\vV_\ell\|\leq\poly(\beta^{\pm1})$.

First restrict the time variables to $|t|,|t_2|\leq T$. The kernel bound $|g(t)h_+(t_2)|\leq e^{\poly(\beta^{\pm1})}e^{-\pi(|t|+|t_2|)/\beta}$ gives an exponentially decaying time integral. Applying {\eqref{eq:outer-cutoff-operator-bounds}} to the exact dynamics, the omitted contribution is bounded by
\begin{equation*}
 e^{\poly(\beta^{\pm1})}e^{-cT/\beta}
\end{equation*}
for a fixed $c>0$. Thus $T=\poly(\beta^{\pm1})$ can make the time tail smaller than another prescribed fraction of the error.

On this bounded time window, the ordered integrals in \eqref{eq:spatially-truncated-dyson} have volume $|t+t_2|^m/m!$. The order-$m$ contribution therefore obeys
\begin{equation*}
 \frac{(2\eta\|\vV_\ell\|\,|t+t_2|)^m}{m!}\|\vA\|
 \leq\left(\frac{8eT\|\vV_\ell\|}{m}\right)^m\|\vA\|.
\end{equation*}
For $m\geq8e^2T\|\vV_\ell\|$, this is at most $e^{-m}\|\vA\|$. Retaining a polynomial number of orders, sufficiently large also compared with $P$ and the logarithms of the fixed prefactors, therefore makes the entire discarded Dyson tail smaller than a prescribed fraction of $e^{-P(\beta^{\pm1})}$.

For the retained positive orders, apply the absolute-value estimates in the proof of \Cref{lem:dyson_term_adb_bound_m}. Their finite sum is bounded by
\begin{equation*}
 e^{\poly(\beta^{\pm1})}
 \left(e^{-\Omega/(8d)}+
 \epsilon_{\mathrm{ADB}}^{\frac{1}{32\beta d+2}}
 e^{\frac{16\beta^2d+1}{16\beta d+1}\Omega}\right).
\end{equation*}
Choose $\Omega=\poly(\beta^{\pm1})$ to control the first term. With this cutoff fixed, taking $P'$ sufficiently large controls the ADB term and the zeroth-order contribution from \Cref{lem:zeroth-order-dyson}; its absolute-kernel estimate also applies on the restricted time domain. Allocating the error among these contributions and adding them proves the claim.
\endgroup
\end{proof}
{We apply the refined or coarse estimate to both conjugation branches.}

{The second conjugation branch obeys the same bound. Indeed, the symmetric inner and outer frequency integrals preserve adjoints, $\vA$ is Hermitian, $g$ is real, and $h_- =\overline{h_+}$. Taking the adjoint and reversing the frequency variables converts the negative branch to the positive branch with test $\vO^\dagger$, which has the same support and norm. Allocate half the linearization budget to each branch.}

{Define $\mathcal I_\sigma$ by the following exactly normalized time integral; below we suppress its state subscript:}
\begin{align}
 \mathcal I(\vO,\vA;\vH,\vH')
 ={}&{\frac{1}{4\pi\sigma\sqrt{2\pi}}}\int_{\mathbb R^3}
 g(t){w_{\Omega_1}(t_1)}
 \Tr\Big[\vO^\dagger\Big(
 \big(\vA_{\vH'}(t+t_2)\big)_{\vH}(-t+t_1)\vsigma h_+(t_2)
 \nonumber\\
 &\hspace{4cm}
 -\vsigma\big(\vA_{\vH'}(t+t_2)\big)_{\vH}(-t+t_1)h_-(t_2)
 \Big)\Big]\dd t\dd t_1\dd t_2.
 \label{eq:linearized-I}
\end{align}
{Applying \eqref{eq:metastable-identifiability} to the inner-filtered operator and then applying the outer filter gives the two KMS branches with coefficient $1/(4\pi\sigma)$. Replacing each branch by its GNS term and converting the outer filter to time gives exactly \eqref{eq:linearized-I}. {In the refined regime of \Cref{thm:kms-to-gns-reduction}, this gives}}
\begin{align}
 {|\mathcal F_{\mathrm{cut}}-\mathcal I_\sigma|
 \leq\|\vO\|\|\vA\|\left(\eta e^{-P_1(\beta^{\pm1})}
 +e^{\poly(\beta^{\pm1})}\epsilon_{\mathrm{ADB}}^{\frac{1}{32\beta d+2}}\right).}
\end{align}

\subsection{Transfer of Hamiltonian dynamics}
\label{subsec:tested-dynamics-transfer}
We have now rewritten the nonlinear KMS expression in a form linear in the input state. The remaining task is to replace the unknown Hamiltonian dynamics by dynamics generated by a chosen test Hamiltonian, as explained in the next two subsections.

We first transfer the $\vH$-evolution from the state factor to the test observable and define
\begin{align}
 Q(\vO,\vA;\vH,\vH')
 ={}&{\frac{1}{4\pi\sigma\sqrt{2\pi}}}\int_{\mathbb R^3}
 g(t){w_{\Omega_1}(t_1)}
 \Tr\Big[\vO^\dagger_{\vH}(t-t_1)\Big(
 \vA_{\vH'}(t+t_2)\vsigma h_+(t_2)\nonumber\\
 &\hspace{4cm}-\vsigma\vA_{\vH'}(t+{t_2})h_-(t_2)
 \Big)\Big]\,\dd t\dd t_1\dd t_2.
 \label{eq:transferred-Q}
\end{align}
The main results of this subsection are two bounds on $Q-\mathcal I$: a global estimate with an explicit factor of $n$, and a stronger estimate for local tests that is independent of $n$.

\begin{thm}[Approximate transfer of the Hamiltonian dynamics]\label{thm:move-unknown-dynamics-global}
Let $\vA$ be a single-site Pauli operator and let $\vO$ be supported on at most $k$ sites. Then there exist polynomials $\poly$ in $\beta^{\pm1}$ such that
\begin{align}
 |Q(\vO,\vA;\vH,\vH')-\mathcal I(\vO,\vA;\vH,\vH')|
 \leq\|\vO\|\|\vA\|n\epsilon_{\mathrm{ADB}}^{\frac1{32\beta d+2}}
 e^{\poly(\beta^{\pm1})}.
\end{align}
\end{thm}

\begin{thm}[Local transfer of real-time dynamics]\label{thm:move-unknown-dynamics}
 Under the same setting of \Cref{thm:move-unknown-dynamics-global}, suppose $\vH'=\vH+\eta\vV$ is an allowed local Hamiltonian, with $|v_\gamma|\leq1$. For every prescribed polynomial $P(\beta^{\pm1})$, there are a polynomial $p$ and a {threshold} $\eta_0=e^{-\poly(\beta^{\pm1})}$ such that, for $0\leq\eta\leq\eta_0$ and $0\leq\epsilon_{\mathrm{ADB}}\leq1$,
\begin{align}
 |Q(\vO,\vA;\vH,\vH')-\mathcal I(\vO,\vA;\vH,\vH')|
 \leq\|\vO\|\left(e^{{p}(\beta^{\pm1})}
 \epsilon_{\mathrm{ADB}}^{{\frac{1}{\poly(\beta^{\pm1})}}}+
 \eta e^{-{P}(\beta^{\pm1})}\right).
\end{align}
The same form of bound holds on the coarse range $\eta_0\leq\eta\leq2$, after enlarging $p$. All constants are independent of $n$.
\end{thm}

To prove these two theorems, let $s=t-t_1$, $u=t+t_2$, and $\vsigma_{\vH}(s)=e^{is\vH}\vsigma e^{-is\vH}$. Cyclicity of the trace gives the exact difference
\begin{align}
 Q-\mathcal I
 ={}&{\frac{1}{4\pi\sigma\sqrt{2\pi}}}\int_{\mathbb R^3}
 g(t){w_{\Omega_1}(t_1)}
 \Tr\Big[\vO^\dagger_{\vH}(s)\Big(
 \vA_{\vH'}(u)(\vsigma-\vsigma_{\vH}(s))h_+(t_2)\nonumber\\
 &\hspace{3cm}-(\vsigma-\vsigma_{\vH}(s))\vA_{\vH'}(u)h_-({t_2})
 \Big)\Big]\,\dd t\dd t_1\dd t_2.
 \label{eq:transfer-state-difference}
\end{align}
{The underlying state-motion identity is}
\begin{align}
 e^{is\vH}\vsigma e^{-is\vH}-\vsigma
 ={i}\int_0^s[\vH,\vsigma]_{\vH}(v)\,\dd v.
 \label{eq:state-motion-identity}
\end{align}
The proof of \Cref{thm:move-unknown-dynamics-global} uses \eqref{eq:state-motion-identity} to control the motion of the entire state in trace norm. A small value of $\|[\vH,\vsigma]\|_1$ naturally bounds \Cref{eq:transfer-state-difference}. However, the global commutator estimate from \Cref{thm:approximate-commute-deduce} collects contributions throughout the system and therefore introduces a factor of $n$.

For \Cref{thm:move-unknown-dynamics}, we control only how the state motion affects the particular test expectations in \eqref{eq:transfer-state-difference}. Using $\Tr(\vB[\vsigma,\vH])=\Tr(\vsigma[\vH,\vB])$, we reduce the problem to the tested commutator estimates of \Cref{thm:weak}. Locality and the expansion into products of evolved Pauli operators allow these expectations to be bounded directly by the local ADB error. The bounds depend on the test support and evolution times, but not on the total system size.

We first prove the bound using the global commutator.

\begin{lem}[Approximate invariance under Hamiltonian evolution]\label{lem:Hamiltonian-stationarity}
For any state $\vsigma$ and Hamiltonian $\vH$, if $\|[\vH,\vsigma]\|_1\leq\epsilon$, then
\begin{align}
 \|e^{it\vH}\vsigma e^{-it\vH}-\vsigma\|_1\leq|t|\epsilon.
\end{align}
\end{lem}
\begin{proof}
{For either sign of $t$, \eqref{eq:state-motion-identity} and unitary invariance give}
\begin{align}
 \|e^{it\vH}\vsigma e^{-it\vH}-\vsigma\|_1
 &\leq\int_{{\min(0,t)}}^{{\max(0,t)}}
 \|[\vH,\vsigma]_{\vH}(v)\|_1\,\dd v
 \leq|t|\epsilon.
\end{align}
\end{proof}
\begin{proof}[Proof of \Cref{thm:move-unknown-dynamics-global}]
{Write $\epsilon_H=\|[\vH,\vsigma]\|_1$. Apply \Cref{lem:Hamiltonian-stationarity} to \eqref{eq:transfer-state-difference}, taking absolute values of all scalar kernels:}
\begin{align}
 |Q-\mathcal I|
 &\leq{\frac{1}{4\pi\sigma\sqrt{2\pi}}}\|\vO\|\|\vA\|\epsilon_H
 \int_{\mathbb R^3}{|t-t_1|\,|g(t)w_{\Omega_1}(t_1)|
 (|h_+(t_2)|+|h_-(t_2)|)}\,\dd t\dd t_1\dd t_2
 \nonumber\\
 &\leq\|\vO\|\|\vA\|\epsilon_H e^{\poly(\beta^{\pm1})}
 \leq\|\vO\|\|\vA\|n\epsilon_{\mathrm{ADB}}^{\frac1{32\beta d+2}}
 e^{\poly(\beta^{\pm1})}.
\end{align}
{The last inequality is \Cref{thm:approximate-commute-deduce}.}
\end{proof}
The examples in the final appendix show that a linear system-size factor is necessary for global stationarity bounds. This does not establish optimality of the ADB exponent or of the error for local tests. We now prove the bound independent of $n$ in \Cref{thm:move-unknown-dynamics}.

To prove this theorem, we use the Dyson expansion \eqref{eq:spatially-truncated-dyson} and the same coefficient-weighted ordered tuples as in the preceding subsection. At each order, we apply the tested commutator bound of \Cref{thm:weak}. This bound grows exponentially with the number of Pauli factors, so sufficiently small $\eta$ controls the growth after summation and integration, as in \Cref{thm:kms-to-gns-reduction}. We begin by spatially truncating the perturbation.
\begin{lem}[{Locality of the transfer error}]\label{lem:locality-of-observable}
{For every prescribed polynomial $P(\beta^{\pm1})$, a sufficiently large polynomial radius $\ell$ gives}
\begin{align}
 &|Q(\vO,\vA;\vH,\vH')-\mathcal I(\vO,\vA;\vH,\vH')|\nonumber\\
 &\quad-|Q(\vO,\vA;\vH,\vH+\eta\vV_\ell)
 -\mathcal I(\vO,\vA;\vH,\vH+\eta\vV_\ell)|
 \leq\|\vO\|\|\vA\|\eta e^{-{P}(\beta^{\pm1})}.
\end{align}
\end{lem}
\begin{proof}
By trace/operator-norm duality, the kernel bounds, and the perturbative Lieb--Robinson estimate,
\begin{align}
 &|Q(\vO,\vA;\vH,\vH')-Q(\vO,\vA;\vH,\vH+\eta\vV_\ell)|
 \nonumber\\
 &\quad\leq{\frac{1}{4\pi\sigma\sqrt{2\pi}}}\|\vO\|
 \int_{\mathbb R^3}\|\vA_{\vH'}(t+t_2)-\vA_{\vH+\eta\vV_\ell}(t+t_2)\|
 {|g(t)w_{\Omega_1}(t_1)|}
 \nonumber\\
 &\hspace{4cm}\times(|h_+(t_2)|+|h_-(t_2)|)\,\dd t\dd t_1\dd t_2
 \nonumber\\
 &\quad\leq\eta\|\vO\|\|\vA\|e^{\poly(\beta^{\pm1})}
 \sum_{r\geq{\ell-1}}S(r,\vA)
 \left(e^{-{c_1}r}+e^{-{c_2}r/(d\beta)}\right)
 \nonumber\\
 &\quad\leq\eta\|\vO\|\|\vA\|
 \exp\!\left[\poly(\beta^{\pm1})-\frac{c\ell}{{d}(1+\beta)}\right].
 \label{eq:Q-perturbation-locality}
\end{align}
{The shell estimate is the one used in \Cref{lem:kms-gns-perturbation-truncation}. The same bound holds for $\mathcal I$, since its additional $\vH$-evolution preserves the operator norm. Choose $\ell$ so that each difference is at most $\eta\|\vO\|\|\vA\|e^{-P}/2$, and add them by the triangle inequality.}
\end{proof}
After truncating the perturbation, we apply the Dyson expansion to the perturbed dynamics inside the integral. The following lemma gives the resulting bound.

\begin{lem}[Tested state motion for a truncated perturbation]\label{lem:move6-dyson}
For every $s,u\in\mathbb R$,
\begin{align}
 &\max\Big\{
 |\Tr[(\vO_{\vH}(s))^\dagger\vA_{\vH+\eta\vV_\ell}(u)
 (e^{is\vH}\vsigma e^{-is\vH}-\vsigma)]|,\nonumber\\
 &\hspace{1.5cm}|\Tr[\vA_{\vH+\eta\vV_\ell}(u)(\vO_{\vH}(s))^\dagger
 (e^{is\vH}\vsigma e^{-is\vH}-\vsigma)]|\Big\}\nonumber\\
 &\quad\leq\|\vO\|{C_{\beta}}
 \epsilon_{\mathrm{ADB}}^{{\frac{1}{\poly(\beta^{\pm1})}}}|s|
 \exp\!\left[\frac{\pi}{\beta(16\beta{d}+1)}(|s|+|u|)
 +\eta{C_{\beta}}|u|\right].
 \label{eq:move6-tested-bound}
\end{align}
{For the fixed polynomial radius, $C_\beta\leq e^{\poly(\beta^{\pm1})}$ depends only on the fixed local parameters, including $k$.}
\end{lem}
\begin{proof}
\begingroup
{For every fixed operator $\vB$, differentiation gives}
\begin{align}
    \tr\!\left[\vB(e^{is\vH}\vsigma e^{-is\vH}-\vsigma)\right]
    &=\tr\!\left[\vsigma(\vB_{\vH}(-s)-\vB)\right]\nonumber\\
    &=-i\int_0^s
      \tr\!\left(\vsigma[\vH,\vB_{\vH}(-v)]\right)\dd v.
    \label{eq:move6-tested-duhamel}
\end{align}
For $s<0$ the integral is oriented. The operator $\vB$ is held fixed throughout this identity.

At order $m$, expand the $m$ nested commutators and the $m$ copies of
$\vV_\ell=\sum_{\gamma\in\Gamma_\ell}v_\gamma\vP_\gamma$.
There are at most $2^m$ product terms for each choice of interactions, and their total absolute coefficient-integral weight is at most
\begin{equation}
    2^m\eta^m
    \left(\sum_{\gamma}|v_\gamma|\right)^m
    \frac{|u|^m}{m!}
    \leq \frac{(\eta C_{\beta}|u|)^m}{m!}.
    \label{eq:move6-order-weight}
\end{equation}
{Expand $\vO$ in Pauli strings; their absolute coefficient sum is at most $2^{k}\|\vO\|$, which is absorbed into $C_\beta$. After multiplication by $(\vO_{\vH}(s))^\dagger$, each product contains at most $km+k+1$ single-site $\vH$-evolved Pauli factors.}

Apply \eqref{eq:move6-tested-duhamel} to each of these products. After the common evolution by $-v$, its constituent times belong to
\begin{equation}
    \{s-v,u-v,r_1-v,\ldots,r_m-v\},
    \qquad
    \max\{|s-v|,|u-v|,|r_j-v|\}\leq |s|+|u|,
    \label{eq:move6-shifted-times}
\end{equation}
for $v$ between $0$ and $s$. Consequently, \Cref{thm:weak} in the form
\eqref{eq:weak-final} {bounds} the order-$m$ contribution by
\begin{equation}
    \|\vO\|C_{\beta}^{{km+k+1}}
    \epsilon_{\mathrm{ADB}}^{1/\poly(\beta^{\pm 1})} |s|
    e^{\frac{\pi}{\beta(16\beta{d}+1)}(|s|+|u|)}
    \frac{(\eta C_{\beta}|u|)^m}{m!}.
    \label{eq:move6-order-bound}
\end{equation}
{This includes $m=0$. Enlarging $C_\beta$ to absorb the fixed support factors and summing all orders yields}
\begin{align}
 &\bigl|\tr\!\left[(\vO_{\vH}(s))^\dagger\vA_{\vH+\eta\vV_\ell}(u)
       (e^{is\vH}\vsigma e^{-is\vH}-\vsigma)\right]\bigr|
       \nonumber\\
 &\quad\leq
       \|\vO\|C_{\beta}\epsilon_{\mathrm{ADB}}^{1/\poly(\beta^{\pm 1})} |s|
       e^{\frac{\pi}{\beta(16\beta{d}+1)}(|s|+|u|)}
       \sum_{m=0}^{\infty}
       \frac{(\eta C_{\beta}|u|)^m}{m!}
       \nonumber\\
 &\quad=
       \|\vO\|{C_\beta}\epsilon_{\mathrm{ADB}}^{1/\poly(\beta^{\pm 1})} |s|
       e^{\frac{\pi}{\beta(16\beta{d}+1)}(|s|+|u|)
                +\eta C_{\beta}|u|}.
    \label{eq:move6-Dyson-sum}
\end{align}
For fixed $s,u$, the series and the $v$ integral are absolutely convergent. {The same argument with the two operator factors reversed proves the other ordering in the statement.}
\endgroup
\end{proof}

\begin{proof}[Proof of \Cref{thm:move-unknown-dynamics}]
{First suppose $0\leq\eta\leq\eta_0$. Choose the radius in \Cref{lem:locality-of-observable} so that spatial truncation uses at most half the allowed error. Apply \Cref{lem:move6-dyson} to \eqref{eq:transfer-state-difference}, with $s=t-t_1$ and $u=t+t_2$. {Choose the threshold $\eta_0$ small enough that}}
\begin{align}
 {\eta C_{\beta}\leq\frac12\left(\frac{2\pi}{\beta}-2\frac{\pi}{\beta(16\beta d+1)}\right).}
 \label{eq:move6-smallness}
\end{align}
{The exponential decay of $g$ then dominates the growth in $t$, while the Gaussian decay of $w_{\Omega_1}$ and $h_\pm$ controls the other two time variables. Integrating the bound in \Cref{lem:move6-dyson} therefore costs only a factor $e^{\poly(\beta^{\pm1})}$. Adding the spatial-truncation error proves the refined estimate.}

{For $\eta_0\leq\eta\leq2$, use the finite-time, finite-order argument of \Cref{thm:coarse-adb-linearization}. Choose polynomial spatial and time cutoffs so that the discarded contributions use a fixed fraction of the error budget $\|\vO\|\eta_0e^{-P(\beta^{\pm1})}$. Estimate the time tail using the exact dynamics, which preserves operator norms, before making a Dyson expansion. On the remaining time interval, truncate the Dyson series at a polynomial order. The factorial decay of its coefficients makes the discarded orders fit within the remaining tail budget, uniformly for $\eta\leq2$.}

{Apply \Cref{thm:weak} to the retained ordered Pauli products. The finite sum of their coefficients and time-growth factors costs at most $e^{\poly(\beta^{\pm1})}$, giving the required ADB term with exponent $\frac{1}{\poly(\beta^{\pm1})}$. The spatial, time, and order tails together are bounded by $\|\vO\|\eta e^{-P(\beta^{\pm1})}$ because $\eta\geq\eta_0$. This proves the coarse estimate.}
\end{proof}

\subsection{Construction of measurable auxiliary tests}
\label{subsec:auxiliary-tests}
To construct a measurable family from \eqref{eq:transferred-Q}, replace the Hamiltonian $\vH$ acting on the test operator by a selected local Hamiltonian $\vG$ as in \Cref{defn:local-Hamiltonian}.
For these scalar kernels and frequency cutoffs, define the auxiliary test (again suppressing the state subscript)
\begin{align}
    Q(\vO,\vA;\vG,\vH')&=
    {\frac{1}{4\pi\sigma\sqrt{2\pi}}}\int_{-\infty}^{\infty}\int_{-\infty}^{\infty}\int_{-\infty}^{\infty}\Tr\Bigg[
    \vO^\dagger_{\vG}(t-t_1)
    \Big(
    \vA_{\vH'}(t+t_2)\vsigma \cdot h_+(t_2)\\
    &-
    \vsigma \vA_{\vH'}(t+t_2)\cdot h_-(t_2)\Big)\Bigg]g(t){w_{\Omega_1}(t_1)}
    \dd t \dd t_1 \dd t_2.
    \label{eq:measurable-Q}
\end{align}

This defines the test family and does not assert that an arbitrary $\vG$ approximates $\vH$. By cyclicity, its measurement operator is
\begin{align}
\vM_{\vO,\vA;\vG,\vH'}=
{\frac{1}{4\pi\sigma\sqrt{2\pi}}}\int_{\mathbb R^3}&g(t)w_{\Omega_1}(t_1)
\Bigl[h_+(t_2)(\vO^\dagger)_{\vG}(t-t_1)\vA_{\vH'}(t+t_2)\nonumber\\
&\qquad-h_-(t_2)\vA_{\vH'}(t+t_2)(\vO^\dagger)_{\vG}(t-t_1)\Bigr]
\dd t\dd t_1\dd t_2.
\label{eq:measurement-operator}
\end{align}
For $\vO=[\vA,\vP_\gamma]$ or $\vO=[\vA,\vG-\vH']/\eta$, we have $\vO^\dagger=-\vO$; since $h_-=\overline{h_+}$ and $g,w_{\Omega_1}$ are real, the two terms are adjoints and $\vM$ is Hermitian.
For the parameter search to succeed, the net must contain a candidate close to the true Hamiltonian whose test values are small for every auxiliary choice. \Cref{lem:q-exact-candidate} establishes this bound for the true Hamiltonian itself. Together with the robustness estimate in \Cref{lem:robustness}, it ensures that a nearby net point also has small test values, as shown in \Cref{lem:good-net-candidate}.

\begin{thm}[ADB control of the exact-candidate observable]\label{lem:q-exact-candidate}
For every auxiliary local Hamiltonian $\vG$ as in \Cref{defn:local-Hamiltonian},
\begin{align}
    \bigl|Q(\vO,\vA;\vG,\vH)\bigr|
    \leq{\|\vO\|}e^{\poly(\beta^{\pm 1})}\mathsf{ADB}_{\vA}[\vsigma]^{\frac{1}{32\beta d+2}}.
    \label{eq:exact-candidate-q-bound}
\end{align}
\end{thm}

\begin{proof}
{Use the unnormalized cutoff $\vA(t)_{\leq\Omega_2}$ from \eqref{eq:low-frequency-part-commutation}, with evolution generated by $\vH$, as in \Cref{lem:low_freq_bound}. This operator is Hermitian because $\vA$ is Hermitian and the frequency interval is symmetric.}
{Undoing only the $h_\pm$ representations in \eqref{eq:measurable-Q} gives}
\begin{align}
 Q(\vO,\vA;\vG,\vH)
 ={}&{\frac{1}{4\pi\sigma\sqrt{2\pi}}}\int_{\mathbb R^2}
 g(t){w_{\Omega_1}(t_1)}
 \Tr\Big[\vO_{\vG}^\dagger(t-t_1)\Big(
 \sqrt{\vrho}{\vA(t)_{\leq\Omega_2}}\sqrt{\vrho}^{-1}\vsigma
 \nonumber\\
 &\hspace{4cm}-\vsigma\sqrt{\vrho}^{-1}{\vA(t)_{\leq\Omega_2}}\sqrt{\vrho}
 \Big)\Big]\,\dd t\dd t_1.
\end{align}
{The inverse inner-kernel representation keeps the time argument $t$. Split the state-dependent block as}
\begin{align}
 &{\sqrt{\vrho}\vA(t)_{\leq\Omega_2}\sqrt{\vrho}^{-1}\vsigma
 -\vsigma\sqrt{\vrho}^{-1}\vA(t)_{\leq\Omega_2}\sqrt{\vrho}}\nonumber\\
 &\quad{=\big(\sqrt{\vrho}\vA(t)_{\leq\Omega_2}\sqrt{\vrho}^{-1}\sqrt{\vsigma}
 -\sqrt{\vsigma}\vA(t)_{\leq\Omega_2}\big)\sqrt{\vsigma}}\nonumber\\
 &\qquad{+\sqrt{\vsigma}\big(\vA(t)_{\leq\Omega_2}\sqrt{\vsigma}
 -\sqrt{\vsigma}\sqrt{\vrho}^{-1}\vA(t)_{\leq\Omega_2}\sqrt{\vrho}\big).}
\end{align}
{The two defects are bounded by \Cref{lem:low_freq_bound} and its adjoint, with $\Delta=\Omega_2$ and time $t$. Since $\Omega_2$ is polynomial, Schatten H\"older's inequality and $\|\sqrt{\vsigma}\|_2=1$ give}
\begin{align}
 &{\left\|\sqrt{\vrho}\vA(t)_{\leq\Omega_2}\sqrt{\vrho}^{-1}\vsigma
 -\vsigma\sqrt{\vrho}^{-1}\vA(t)_{\leq\Omega_2}\sqrt{\vrho}\right\|_1}\nonumber\\
 &\quad{\leq e^{\poly(\beta^{\pm1})}\mathsf{ADB}_{\vA}[\vsigma]^{\frac{1}{32\beta d+2}}
 e^{\frac{\pi}{\beta(16\beta d+1)}|t|}.}
\end{align}
{Finally, $\|\vO_{\vG}(t-t_1)\|=\|\vO\|$ for every allowed $\vG$, so}
\begin{align}
 |Q(\vO,\vA;\vG,\vH)|
 &\leq{\frac{1}{4\pi\sigma\sqrt{2\pi}}\|\vO\|e^{\poly(\beta^{\pm1})}
 \mathsf{ADB}_{\vA}[\vsigma]^{\frac{1}{32\beta d+2}}
 \int_{\mathbb R^2}|g(t)w_{\Omega_1}(t_1)|e^{\frac{\pi}{\beta(16\beta d+1)}|t|}\,\dd t\dd t_1}
 \nonumber\\
 &\leq{\|\vO\|}e^{\poly(\beta^{\pm1})}
 \mathsf{ADB}_{\vA}[\vsigma]^{\frac1{32\beta d+2}}.
\end{align}
{The last integral converges because $\frac{\pi}{\beta(16\beta d+1)}<2\pi/\beta$ and $w_{\Omega_1}$ has Gaussian decay. The estimate is uniform in the auxiliary Hamiltonian.}
\end{proof}

We now combine the preceding estimates to obtain the measurement family in \Cref{lem:tested-identification}. Each test is a bounded observable expectation, the true Hamiltonian has a small test value for every auxiliary choice, and taking the true Hamiltonian as the auxiliary generator gives an approximation to the unfiltered commutator pairing, with error bounds independent of the system size.

\begin{proof}[Proof of \Cref{lem:tested-identification}]
Use the family defined by \eqref{eq:measurable-Q}, with measurement operator \eqref{eq:measurement-operator}. Real-time evolution preserves operator norms, so for normalized test operators the triangle inequality bounds $\|\vM\|$ by the integrals of the absolute scalar kernels. The bounds on $g$, $h_\pm$, and $w_{\Omega_1}$ make these integrals at most $K_\beta=e^{\poly(\beta^{\pm1})}$, uniformly in the selected Hamiltonians and the system size. Measuring the Hermitian real and imaginary parts separately therefore proves the first assertion.

When the candidate is the true Hamiltonian, \Cref{lem:q-exact-candidate} gives a small test value for every auxiliary $\vG$. Since $\mathsf{ADB}_{\vA}[\vsigma]\leq\epsilon_{\mathrm{ADB}}$ for each single-site Pauli, this proves the second assertion after choosing the common constant and exponent below.

For the third assertion, set the auxiliary Hamiltonian to $\vH$ and recall $\eta=\|h-h'\|_\infty$. Fix the observable cutoffs using \Cref{thm:double-frequency-truncation} so that truncation contributes at most $a_\beta\eta/4$. Choose $\eta_0=e^{-\poly(\beta^{\pm1})}$ within the refined regimes of \Cref{thm:kms-to-gns-reduction,thm:move-unknown-dynamics}. For $\eta\leq\eta_0$, apply the refined linearization estimate to both conjugation branches. For $\eta_0\leq\eta\leq2$, apply \Cref{thm:coarse-adb-linearization} with budget $a_\beta\eta_0/8$ per branch, giving total error at most $a_\beta\eta/4$. Both arguments keep the observable cutoffs fixed, so they apply to the same measurement family. Finally, use \Cref{thm:move-unknown-dynamics} for the transfer step. Choosing the prescribed accuracy polynomials in these results gives
\begin{align}
 {|\mathcal F_\sigma-\mathcal F_{\mathrm{cut}}|}
 &{\leq\tfrac14a_\beta\eta,}\nonumber\\
 {|\mathcal F_{\mathrm{cut}}-\mathcal I_\sigma|}
 &{\leq\tfrac14a_\beta\eta+C_{\beta}\epsilon_{\mathrm{ADB}}^{\frac{1}{32\beta d+2}},}\nonumber\\
 {|\mathcal I_\sigma-Q_\sigma|}
 &{\leq\tfrac12a_\beta\eta+C_{\beta}\epsilon_{\mathrm{ADB}}^{\frac{1}{\poly(\beta^{\pm1})}}.}
 \label{eq:assembled-identification-budgets}
\end{align}
Choose one ADB threshold small enough for all of these estimates, including both regimes of $\eta$. The three lines of \eqref{eq:assembled-identification-budgets} control truncation, linearization, and transfer, respectively. Their sum compares the measured test with the original, unfiltered pairing: for every $0\leq\eta\leq2$,
\begin{align}
 |Q_\sigma-\mathcal F_\sigma|
 &\leq|\mathcal F_\sigma-\mathcal F_{\mathrm{cut}}|
 +|\mathcal F_{\mathrm{cut}}-\mathcal I_\sigma|
 +|\mathcal I_\sigma-Q_\sigma|
 {\leq a_\beta\eta+C_\beta\epsilon_{\mathrm{ADB}}^{\frac{1}{\poly(\beta^{\pm1})}}.}
\end{align}
Since $0\leq\epsilon_{\mathrm{ADB}}\leq1$, choosing a common inverse-polynomial exponent no larger than any of the input exponents bounds all ADB contributions in the stated form. Enlarging $C_\beta$ absorbs their prefactors and the true-candidate bound. All constants and cutoffs are independent of $n$, proving the three assertions. Coefficient recovery follows in \Cref{subsec:coefficient-recovery}.
\end{proof}

\clearpage
\makeatletter
\MT@newlabel{eq:coarse-measured-interface}
\MT@newlabel{{eq:coarse-measured-interface}}
\MT@newlabel{eq:coefficient-interface}
\MT@newlabel{{eq:coefficient-interface}}
\MT@newlabel{eq:coefficient-recovery-chain}
\MT@newlabel{{eq:coefficient-recovery-chain}}
\MT@newlabel{eq:completed-learning-cost}
\MT@newlabel{{eq:completed-learning-cost}}
\MT@newlabel{eq:contraction-allocation}
\MT@newlabel{{eq:contraction-allocation}}
\MT@newlabel{eq:full-candidate-hybrid}
\MT@newlabel{{eq:full-candidate-hybrid}}
\MT@newlabel{eq:full-refinement-budget}
\MT@newlabel{{eq:full-refinement-budget}}
\MT@newlabel{eq:global-regularized-adb}
\MT@newlabel{{eq:global-regularized-adb}}
\MT@newlabel{eq:good-net-candidate-bound}
\MT@newlabel{{eq:good-net-candidate-bound}}
\MT@newlabel{eq:learning-net-count}
\MT@newlabel{{eq:learning-net-count}}
\MT@newlabel{eq:learning-radii}
\MT@newlabel{{eq:learning-radii}}
\MT@newlabel{eq:learning-error-floor}
\MT@newlabel{{eq:learning-error-floor}}
\MT@newlabel{eq:learning-round-cost}
\MT@newlabel{{eq:learning-round-cost}}
\MT@newlabel{eq:learning-shots}
\MT@newlabel{{eq:learning-shots}}
\MT@newlabel{eq:prime-discretization}
\MT@newlabel{{eq:prime-discretization}}
\MT@newlabel{eq:robust-local-perturbation}
\MT@newlabel{{eq:robust-local-perturbation}}
\MT@newlabel{eq:robust-remote-perturbation}
\MT@newlabel{{eq:robust-remote-perturbation}}
\MT@newlabel{eq:robust-truncation}
\MT@newlabel{{eq:robust-truncation}}
\MT@newlabel{eq:rotated}
\MT@newlabel{{eq:rotated}}
\MT@newlabel{eq:selected-true-auxiliary}
\MT@newlabel{{eq:selected-true-auxiliary}}
\MT@newlabel{eq:virtual-regularization-bias}
\MT@newlabel{{eq:virtual-regularization-bias}}
\MT@newlabel{eq:faithfulness-haar-average}
\MT@newlabel{{eq:faithfulness-haar-average}}
\MT@newlabel{eq:faithfulness-partial-trace}
\MT@newlabel{{eq:faithfulness-partial-trace}}
\MT@newlabel{eq:faithfulness-frequency-split}
\MT@newlabel{{eq:faithfulness-frequency-split}}
\MT@newlabel{eq:faithfulness-high-frequency}
\MT@newlabel{{eq:faithfulness-high-frequency}}
\MT@newlabel{eq:faithfulness-low-frequency}
\MT@newlabel{{eq:faithfulness-low-frequency}}
\MT@newlabel{eq:faithfulness-rotation-tail}
\MT@newlabel{{eq:faithfulness-rotation-tail}}
\MT@newlabel{eq:faithfulness-low-kms}
\MT@newlabel{{eq:faithfulness-low-kms}}
\MT@newlabel{eq:faithfulness-three-terms}
\MT@newlabel{{eq:faithfulness-three-terms}}
\MT@newlabel{eq:faithfulness-first-holder}
\MT@newlabel{{eq:faithfulness-first-holder}}
\MT@newlabel{eq:faithfulness-filtered-adb}
\MT@newlabel{{eq:faithfulness-filtered-adb}}
\MT@newlabel{eq:faithfulness-first-term}
\MT@newlabel{{eq:faithfulness-first-term}}
\MT@newlabel{eq:faithfulness-conjugated-filter}
\MT@newlabel{{eq:faithfulness-conjugated-filter}}
\MT@newlabel{eq:faithfulness-second-term}
\MT@newlabel{{eq:faithfulness-second-term}}
\MT@newlabel{eq:faithfulness-third-term}
\MT@newlabel{{eq:faithfulness-third-term}}
\MT@newlabel{eq:faithfulness-trace-norm}
\MT@newlabel{{eq:faithfulness-trace-norm}}
\MT@newlabel{eq:faithfulness-rotated-bound}
\MT@newlabel{{eq:faithfulness-rotated-bound}}
\MT@newlabel{eq:faithfulness-combined-bound}
\MT@newlabel{{eq:faithfulness-combined-bound}}
\MT@newlabel{eq:faithfulness-tail-choice}
\MT@newlabel{{eq:faithfulness-tail-choice}}
\MT@newlabel{eq:faithfulness-defect-choice}
\MT@newlabel{{eq:faithfulness-defect-choice}}
\MT@newlabel{eq:faithfulness-final-gap}
\MT@newlabel{{eq:faithfulness-final-gap}}
\makeatother
\section{Learning protocol and error analysis}\label{sec:learning-protocol}
In \Cref{sec:identifiability-bounds}, we constructed measurable observables whose expectations approximate local KMS pairings with controlled error, providing a way to bound the squared KMS norms that encode the discrepancy between a candidate Hamiltonian and the true one.

Building on this construction, our goal here is to give a systematic learning algorithm for the model of \Cref{defn:local-Hamiltonian} and analyze its accuracy, sample complexity, and computational cost. The algorithm compares local Hamiltonian guesses through these observables and repeatedly updates their coefficients to improve the estimate. To justify these updates, we first need to show that small KMS errors imply small coefficient errors.

For clarity, \Crefrange{subsec:faithfulness}{subsec:learning-complexity} first present the algorithm and its analysis using independent copies of a single metastable state. We then prove \Cref{thm:intro-global-metastable-learning} in \Cref{sec:streaming-learning}, allowing independent preparations of different metastable states.

In the following subsection, we prove that the KMS norm faithfully controls the usual operator norm for operators of bounded support, with constants depending on the temperature and fixed locality parameters but not on the system size. Applied to local commutators of the Hamiltonian difference, this faithfulness result supplies the link between the measured tests and coefficient recovery that underlies the learning protocol.

\begingroup\thickmuskip=3mu plus 2mu minus 2mu\medmuskip=3mu plus 2mu minus 3mu
\subsection{Faithfulness}
\label{subsec:faithfulness}

We prove the metastable-state local KMS faithfulness estimate needed for coefficient recovery in \Cref{subsec:refinement}. This supplies the link from small local test pairings to small Hamiltonian coefficient errors.
\begin{thm}[Local faithfulness for metastable states]\label{thm:meta-local-faithfulness}
    Let $\vH$ be a local Hamiltonian as in \Cref{defn:local-Hamiltonian}. Let $\epsilon_{\mathrm{ADB}}$ denote the maximum ADB defect over the single-site Pauli jumps. {For operators $\vB$ of bounded support size, there exists a polynomial $\poly(\beta^{\pm 1})$ such that, if $\epsilon_{\mathrm{ADB}}<e^{-\poly(\beta^{\pm 1})}$, then $\|\vB\|_{\vsigma}\geq\|\vB\|_{\vtau}e^{-\poly(\beta^{\pm 1})}$. The polynomial may depend on the support-size bound and fixed locality parameters, but not on $n$.}
\end{thm}
The proof is organized into the following two lemmas.

\begin{lem}\label{lem:faithfulness-local-rotations} We relate $\|\vB\|_{\vtau}$ to the $\vsigma$-weighted KMS norm of a rotated version of $\vB$. In particular, there exist unitaries $\vU, \vV$ supported on $\supp(\vB)$ such that
\begin{align}
    \|\vB\|_{\vtau} \leq 2^{2|\supp(\vB)|} \|\vU^\dagger \vB \vV\|_{\vsigma}. \label{eq:rotated}
\end{align}
\end{lem}

\begin{proof}
We follow exactly the Haar-averaging strategy of Ref.~\cite[Lemma~III.6, Claim~1, Eq.~(3.5)]{chen2025learning}; this part of the argument applies to any state $\vsigma$.
Let $\mathsf{B}:=\supp(\vB)$, and let $|\vB|,|\vU|,|\vV|$ denote support cardinalities. Choose $\vU$ and $\vV$ independently from the Haar measure on $\mathsf{B}$. Averaging the squared KMS norm of the rotated operator separates the two state factors and yields
\begin{equation}
    \mathbb{E}_{\vU, \vV} \tr( \vB^\dagger \vU\sqrt{\vsigma} \vU^\dagger \vB \vV \sqrt{\vsigma} \vV^\dagger) =  \tr((\tr_{\mathsf{B}}(\sqrt{\vsigma}))^2) \frac
    {\|\vB\|_{\vtau}^2}{2^{|\mathsf{B}|}}.\label{eq:faithfulness-haar-average}
\end{equation}
Here the Haar average satisfies $\mathbb{E}_{\vU} \vU \sqrt{\vsigma} \vU^\dagger = \tr_{\mathsf{B}}(\sqrt{\vsigma}) \otimes \frac{\vI_{\mathsf{B}}}{2^{|\mathsf{B}|}}$. To obtain a lower bound on this mean, it suffices to control $\tr((\tr_{\mathsf{B}}(\sqrt{\vsigma}))^2)$. Positivity and normalization give
\begin{align}
    \sqrt{\vsigma} \leq 2^{|\mathsf{B}|} \vI_{\mathsf{B}} \otimes \tr_{\mathsf{B}}(\sqrt{\vsigma}) \Longrightarrow \ 1=\tr(\vsigma) \leq 2^{3|\mathsf{B}|}  \tr((\tr_{\mathsf{B}}(\sqrt{\vsigma}))^2).\label{eq:faithfulness-partial-trace}
\end{align}
The last step follows from $\tr(\vC^2)\leq\tr(\vD^2)$ whenever $0\leq\vC\leq\vD$, as seen from $\tr(\vD^2-\vC^2)=\tr[(\vD-\vC)(\vD+\vC)]$. Substituting \eqref{eq:faithfulness-partial-trace} into \eqref{eq:faithfulness-haar-average} bounds the mean squared norm from below. At least one pair of unitaries attains this mean or a larger value; taking square roots gives \eqref{eq:rotated}.
\end{proof}

\Cref{lem:faithfulness-operator-norm} provides the main distinction between our theorem and the Gibbs-state faithfulness result of Ref.~\cite[Lemma~III.6]{chen2025learning}: it establishes the norm bound for states with a small ADB defect. We split the local unitaries into low- and high-frequency components, control the high-frequency tails, and repeatedly apply \Cref{lem:low_freq_bound} to bound the errors in the low-frequency terms. Together with \Cref{lem:faithfulness-local-rotations}, these estimates yield the desired norm comparison for a suitable cutoff and sufficiently small ADB error.

\begin{lem}\label{lem:faithfulness-operator-norm} Under the hypotheses of \Cref{thm:meta-local-faithfulness}, for $\vB\neq0$, the $\vsigma$-weighted KMS norm and the usual operator norm of $\vB$ satisfy
\begin{align}
    \frac{\|\vB\|_{\vsigma}}{\|\vB\|}{\geq}e^{-\poly(\beta^{\pm 1})}.
\end{align}
\end{lem}

\begin{proof}
The available ADB estimate controls \(\|\sqrt{\vsigma}\vP-\sqrt{\vrho}\vP\sqrt{\vrho^{-1}}\sqrt{\vsigma}\|_2\) but does not directly control \(\|\vsigma^{\frac{1}{4}}\vP-\vrho^{\frac{1}{4}}\vP\vrho^{-\frac{1}{4}}\vsigma^{\frac{1}{4}}\|_2\). We therefore use a different estimate from the Gibbs-state learning proof.

Choose $\vU$ and $\vV$ as in \Cref{lem:faithfulness-local-rotations}, and decompose them into low- and high-frequency components to obtain
\begin{align}
    \vU= \frac{1}{\sqrt{2\sigma\sqrt{2\pi}}} \left(\int_{|\omega|\leq \Delta}\hat{\vU}(\omega)d\omega + \int_{|\omega|\geq \Delta}\hat{\vU}(\omega)d\omega \right) := \vU_{\leq \Delta} + \vU_{\geq \Delta},\label{eq:faithfulness-frequency-split}
\end{align}
for a cutoff $\Delta>0$, and analogously for $\vV$. {The Pauli expansion of $\vU$ has coefficient sum at most $2^{|\vB|}$. Applying the imaginary-time and Gaussian-shift estimates from the proof of \Cref{lem:weak-high-frequency} gives}
\begin{align}
    \|\vU_{\geq \Delta}\|
    \leq \frac{1}{\sigma\sqrt{2\pi}} \int_{|\omega|\geq \Delta} 4^{|\vB|}e^{{-|\omega|/(4d) + \sigma^2/(16d^2)}} d\omega
    \leq \frac{4d}{\sigma} 4^{|\vB|} e^{{\sigma^2/(16d^2)-\Delta/(4d)}}.
    \label{eq:faithfulness-high-frequency}
\end{align}
Using the normalization in \eqref{eq:faithfulness-frequency-split}, the low-frequency component satisfies
\begingroup
\begin{align}
    \|\vU_{\leq\Delta}\|
    &=\frac1{2\pi}\left\|\int_{-\infty}^{\infty}\vU(t)e^{-\sigma^2t^2}\frac{2\sin(\Delta t)}{t}\dd t\right\|\leq\frac{\beta\Delta}{\sqrt\pi}.
    \label{eq:faithfulness-low-frequency}
\end{align}
\endgroup
{The same bounds hold for $\vV$. Using \eqref{eq:faithfulness-high-frequency}, $\|\vU_{\leq\Delta}\|\leq1+\|\vU_{\geq\Delta}\|$, and $\|\cdot\|_{\vsigma}\leq\|\cdot\|$, we obtain}

\begin{align}
    \|\vU^\dagger \vB \vV\|_{\vsigma}
    \leq \|\vU_{\leq \Delta}^\dagger \vB \vV_{\leq \Delta}\|_{\vsigma} +(2+ {4d\beta 4^{|\vB|}} e^{{1/(16d^2\beta^2)-\Delta/(4d)}}){4d\beta 4^{|\vB|}} e^{{1/(16d^2\beta^2)-\Delta/(4d)}} \|\vB\|.
    \label{eq:faithfulness-rotation-tail}
\end{align}

{By \eqref{eq:faithfulness-low-frequency}, the remaining low-frequency contribution is bounded by}
\begin{align}
    \|\vU_{\leq \Delta}^\dagger \vB \vV_{\leq \Delta}\|^2_{\vsigma}
    &=\tr\L(\vV_{\leq \Delta}^\dagger \vB^\dagger \vU_{\leq \Delta} \sqrt{\vsigma} \vU_{\leq \Delta}^\dagger \vB \vV_{\leq \Delta} \sqrt{\vsigma}\R) \nonumber\\
    &\leq \|\vV_{\leq \Delta}^\dagger \vB^\dagger \vU_{\leq \Delta}\|\cdot\|\sqrt{\vsigma} \vU_{\leq \Delta}^\dagger \vB \vV_{\leq \Delta} \sqrt{\vsigma}\|_1 \nonumber\\
    &\leq \|\vB\|{\frac{\beta^2\Delta^2}{\pi}}\cdot \|\sqrt{\vsigma} \vU_{\leq \Delta}^\dagger \vB \vV_{\leq \Delta} \sqrt{\vsigma}\|_1.
    \label{eq:faithfulness-low-kms}
\end{align}

We bound the final trace norm using the decomposition
\begin{align}
    &\sqrt{\vsigma} \vU_{\leq \Delta}^\dagger \vB \vV_{\leq \Delta} \sqrt{\vsigma} =\L(\sqrt{\vsigma} \vU_{\leq \Delta}^\dagger-\sqrt{\vrho} \vU_{\leq \Delta}^\dagger\sqrt{\vrho^{-1}}\sqrt{\vsigma}\R)\vB \vV_{\leq \Delta} \sqrt{\vsigma} \nonumber\\
    &\quad +\sqrt{\vrho} \vU_{\leq \Delta}^\dagger\sqrt{\vrho^{-1}}\sqrt{\vsigma}\vB\L(\vV_{\leq \Delta} \sqrt{\vsigma}-\sqrt{\vsigma}\sqrt{\vrho^{-1}}\vV_{\leq \Delta} \sqrt{\vrho}\R) \nonumber\\
    &\quad +\sqrt{\vrho}\vU_{\leq \Delta}^\dagger\sqrt{\vrho^{-1}}\sqrt{\vsigma}\vB\sqrt{\vsigma}\sqrt{\vrho^{-1}}\vV_{\leq \Delta} \sqrt{\vrho}.
\label{eq:faithfulness-three-terms}
\end{align}
By the triangle inequality $\|\vA+\vB+\vC\|_1\leq\|\vA\|_1+\|\vB\|_1+\|\vC\|_1$, it suffices to estimate the three terms separately.

The first term satisfies
\begin{align}
    &\|\L(\sqrt{\vsigma} \vU_{\leq \Delta}^\dagger-\sqrt{\vrho} \vU_{\leq \Delta}^\dagger\sqrt{\vrho^{-1}}\sqrt{\vsigma}\R)\vB \vV_{\leq \Delta} \sqrt{\vsigma}\|_1\nonumber\\
    &\leq \|\L(\sqrt{\vsigma} \vU_{\leq \Delta}^\dagger-\sqrt{\vrho} \vU_{\leq \Delta}^\dagger\sqrt{\vrho^{-1}}\sqrt{\vsigma}\R)\|_2\cdot\|\vB\|\cdot\|\vV_{\leq \Delta}\|\cdot\|\sqrt{\vsigma}\|_2.
\label{eq:faithfulness-first-holder}
\end{align}
The nonlocal filtered factor is controlled by the low-frequency ADB estimate in \Cref{lem:low_freq_bound}.

The Pauli coefficients of $\vU$ have absolute sum at most $2^{|\vU|}\leq4^{|\vU|}$. Apply \Cref{lem:low_freq_bound} termwise, using $(\vU_{\leq\Delta})^\dagger=(\vU^\dagger)_{\leq\Delta}$ and absorbing the normalization factor $(2\sigma\sqrt{2\pi})^{-1/2}$ into $e^{\poly(\beta^{\pm1})}$. This gives
\begin{align}
    \|\sqrt{\vsigma} \vU_{\leq \Delta}^\dagger-\sqrt{\vrho} \vU_{\leq \Delta}^\dagger\sqrt{\vrho^{-1}}\sqrt{\vsigma}\|_2\leq 4^{|\vU|}\epsilon_{\mathrm{ADB}}^{\frac{1}{32\beta d+2}}\cdot e^{\frac{8\beta^2 d+1/4}{16\beta d+1}\Delta}e^{\poly(\beta^{\pm 1})}.
    \label{eq:faithfulness-filtered-adb}
\end{align}
{Substituting \eqref{eq:faithfulness-low-frequency} and \eqref{eq:faithfulness-filtered-adb} into \eqref{eq:faithfulness-first-holder} yields}
\begin{align}
    &\|\L(\sqrt{\vsigma} \vU_{\leq \Delta}^\dagger-\sqrt{\vrho} \vU_{\leq \Delta}^\dagger\sqrt{\vrho^{-1}}\sqrt{\vsigma}\R)\vB \vV_{\leq \Delta} \sqrt{\vsigma}\|_1\nonumber\\
    &\leq \|\L(\sqrt{\vsigma} \vU_{\leq \Delta}^\dagger-\sqrt{\vrho} \vU_{\leq \Delta}^\dagger\sqrt{\vrho^{-1}}\sqrt{\vsigma}\R)\|_2\cdot\|\vB\|\cdot\|\vV_{\leq \Delta}\|\cdot\|\sqrt{\vsigma}\|_2\nonumber\\
    &\leq 4^{|\vU|+|\vV|}\epsilon_{\mathrm{ADB}}^{\frac{1}{32\beta d+2}}\cdot e^{\frac{8\beta^2 d+1/4}{16\beta d+1}\Delta}e^{\poly(\beta^{\pm 1})}{\frac{\beta\Delta}{\sqrt\pi}}\|\vB\|.
\label{eq:faithfulness-first-term}
\end{align}
And the second term satisfies
\begin{align}
    &\|\sqrt{\vrho} \vU_{\leq \Delta}^\dagger\sqrt{\vrho^{-1}}\sqrt{\vsigma}\vB\L(\vV_{\leq \Delta} \sqrt{\vsigma}-\sqrt{\vsigma}\sqrt{\vrho^{-1}}\vV_{\leq \Delta} \sqrt{\vrho}\R)\|_1\nonumber\\
    &\leq  \|\sqrt{\vrho} \vU_{\leq \Delta}^\dagger\sqrt{\vrho^{-1}}\|\|\sqrt{\vsigma}\|_2\cdot\|\vB\|\cdot\|\L(\vV_{\leq \Delta} \sqrt{\vsigma}-\sqrt{\vsigma}\sqrt{\vrho^{-1}}\vV_{\leq \Delta} \sqrt{\vrho}\R)\|_2\nonumber\\
    &\leq  4^{|\vU|+|\vV|}e^{\beta \Delta/2}{\frac{e^{1/4}}{\sqrt\pi}}\epsilon_{\mathrm{ADB}}^{\frac{1}{32\beta d+2}}\cdot e^{\frac{8\beta^2 d+1/4}{16\beta d+1}\Delta}e^{\poly(\beta^{\pm 1})}\|\vB\|.
\label{eq:faithfulness-second-term}
\end{align}
For the third term, we have
\begin{align}
    &\|\sqrt{\vrho}\vU_{\leq \Delta}^\dagger\sqrt{\vrho^{-1}}\sqrt{\vsigma}\vB\sqrt{\vsigma}\sqrt{\vrho^{-1}}\vV_{\leq \Delta} \sqrt{\vrho}\|_1\nonumber\\
    &\leq \|\sqrt{\vrho}\vU_{\leq \Delta}^\dagger\sqrt{\vrho^{-1}}\|\cdot \|\vsigma^{1/4}\|_4\cdot \|\vsigma^{1/4}\vB\vsigma^{1/4}\|_2\cdot \|\vsigma^{1/4}\|_4\cdot \|\sqrt{\vrho^{-1}}\vV_{\leq \Delta} \sqrt{\vrho}\|\nonumber\\
    &\leq  4^{|\vU|+|\vV|}\|\vB\|_{\vsigma}e^{\beta \Delta}{\frac{e^{1/2}}{\pi}}.
\label{eq:faithfulness-third-term}
\end{align}
{Combining \eqref{eq:faithfulness-three-terms} with \eqref{eq:faithfulness-first-term}, \eqref{eq:faithfulness-second-term}, and \eqref{eq:faithfulness-third-term}, and enlarging the polynomials, gives}

\begin{align}
    \|\sqrt{\vsigma} \vU_{\leq \Delta}^\dagger \vB \vV_{\leq \Delta} \sqrt{\vsigma}\|_1
    &\leq 4^{|\vU|+|\vV|}e^{\poly(\beta^{\pm 1})\Delta}\epsilon_{\mathrm{ADB}}^{\frac{1}{32\beta d+2}}\cdot e^{\poly(\beta^{\pm 1})}\|\vB\|\nonumber\\
    &\quad+4^{|\vU|+|\vV|}\|\vB\|_{\vsigma}e^{\beta \Delta}{\frac{e^{1/2}}{\pi}}.
    \label{eq:faithfulness-trace-norm}
\end{align}
{Applying $\sqrt{a+b+c}\leq\sqrt a+\sqrt b+\sqrt c$ to \eqref{eq:faithfulness-low-kms} and \eqref{eq:faithfulness-trace-norm}, and then using \eqref{eq:faithfulness-rotation-tail}, gives}
\begin{align}
    \|\vU^\dagger \vB \vV\|_{\vsigma}
    &\leq (2+ {4d\beta 4^{|\vB|}} e^{{1/(16d^2\beta^2)-\Delta/(4d)}}){4d\beta 4^{|\vB|}} e^{{1/(16d^2\beta^2)-\Delta/(4d)}}\|\vB\| \nonumber\\
    &\quad +2^{|\vU|+|\vV|}e^{\poly(\beta^{\pm 1})\Delta}\epsilon_{\mathrm{ADB}}^{\frac{1}{64\beta d+4}}\cdot e^{\poly(\beta^{\pm 1})}\|\vB\| \nonumber\\
    &\quad +2^{|\vU|+|\vV|}\|\vB\|^{1/2}_{\vsigma}\|\vB\|^{1/2}e^{\beta \Delta/2}{\frac{e^{1/4}\beta\Delta}{\pi}}.
\label{eq:faithfulness-rotated-bound}
\end{align}

{Combining \eqref{eq:rotated} and \eqref{eq:faithfulness-rotated-bound}, dividing by $\|\vB\|$, and using $\|\vB\|_{\vtau}\geq2^{-|\vB|/2}\|\vB\|$, gives}
\begin{align}
    \frac{1}{2^{3|\vB|}}\leq \frac{\|\vB\|_{\vtau}}{2^{2|\vB|}\|\vB\|} &\leq (2+ {4d\beta 4^{|\vB|}} e^{{1/(16d^2\beta^2)-\Delta/(4d)}}){4d\beta 4^{|\vB|}} e^{{1/(16d^2\beta^2)-\Delta/(4d)}} \nonumber\\
    &\quad +2^{|\vU|+|\vV|}e^{\poly(\beta^{\pm 1})\Delta}\epsilon_{\mathrm{ADB}}^{\frac{1}{64\beta d+4}}\cdot e^{\poly(\beta^{\pm 1})} \nonumber\\
    &\quad +2^{|\vU|+|\vV|}{\frac{e^{1/4}\beta\Delta}{\pi}}\frac{\|\vB\|^{1/2}_{\vsigma}}{\|\vB\|^{1/2}}e^{\beta \Delta/2}.
\label{eq:faithfulness-combined-bound}
\end{align}

\begingroup
Choose $\Delta=\poly(\beta^{\pm1})$ sufficiently large and assume a sufficiently small upper bound on $\epsilon_{\mathrm{ADB}}$ so that
\begin{align}
    2^{-3|\vB|-1}
    &\leq2^{|\vU|+|\vV|}\frac{e^{1/4}\beta\Delta}{\pi}
    \frac{\|\vB\|_{\vsigma}^{1/2}}{\|\vB\|^{1/2}}e^{\beta\Delta/2}.
    \label{eq:faithfulness-final-gap}
\end{align}
Since $|\vU|,|\vV|\leq|\vB|$ are bounded and $\Delta=\poly(\beta^{\pm1})$, squaring and rearranging \eqref{eq:faithfulness-final-gap} gives $\|\vB\|_{\vsigma}/\|\vB\|\geq e^{-\poly(\beta^{\pm1})}$. This proves the lemma. Enlarging a single polynomial suffices for both the ADB threshold and the norm bound.

\endgroup
\end{proof}

\begin{proof}[Proof of \Cref{thm:meta-local-faithfulness}]
Apply \Cref{lem:faithfulness-operator-norm} and use $\|\vB\|_{\vtau}\leq\|\vB\|$; the conclusion also holds for $\vB=0$.
\end{proof}

\endgroup

\subsection{Putting everything together}
\label{subsec:refinement}
\label{subsec:coefficient-recovery}
\label{subsec:learning-summary}

Together, \Cref{lem:tested-identification,thm:meta-local-faithfulness} provide measurable tests that control the commutator error $[\vA,\vH-\vH']$ for single-site Paulis $\vA$, and hence the errors in the local Hamiltonian coefficients. Motivated by the iterative learning algorithm of Ref.~\cite[Algorithm~IV.3]{chen2025learning}, we use these tests to search over local parameter corrections, selecting the candidate with the smallest maximum test value over the auxiliary net. \Cref{lem:robustness} controls the errors from spatial truncation and parameter discretization, allowing us to use finite local nets. \Cref{lem:good-net-candidate} shows that each net contains a candidate with small test values for every auxiliary guess; minimizing the largest measured test value therefore also gives a small value, up to measurement error. \Cref{lem:local-candidate-accuracy} then guarantees that the selected local coefficients have error at most $\eta/{2}$, under the stated ADB and measurement error conditions. Updating each site thus reduces the full coefficient error by {a factor of two}. At fixed temperature and locality parameters, the local net sizes are independent of $n$, and the contraction requires only logarithmically many rounds to reach a target accuracy above the ADB error floor. These facts justify the iterative algorithm and explain the efficiency of its local search. We now specify the parameters and present one refinement step.

\makeatletter
\MT@newlabel{eq:local-net-small-tests}
\MT@newlabel{{eq:local-net-small-tests}}
\makeatother
\begingroup
At each refinement step, let $\vH_0$ be the current estimate with $\|h_0-h\|_\infty\leq\eta$ and $0<\eta\leq1$. We start with $\vH_0=0$ and $\eta=1$, since $|h_\gamma|\leq1$. Choose the frequency cutoffs as in \Cref{lem:tested-identification}, with $a_\beta$ sufficiently small, and set
\begin{align}
 \Omega_1,\Omega_2&=\poly(\beta^{\pm1}), & \ell_0&=\log(e/\eta)\poly(\beta^{\pm1}),\nonumber\\
 \ell_1&=\poly(\beta^{\pm1}), & \kappa_0&=e^{-\poly(\beta^{\pm1})}.
 \label{eq:learning-radii}
\end{align}
Here $\ell_1$ is the radius for searching over coefficient corrections, and $\ell_0$ is the radius retained when evaluating the test dynamics. Round the radii up to integers and choose their polynomials so that $\ell_0\geq\ell_1$, including the fixed enlargement needed to contain all test supports.

For each site, form a coefficient $\ell_\infty$ net of normalized corrections in $[-1,1]$ with mesh at most $\kappa_0$. Include the endpoints of each admissible interval so that the corrected Hamiltonian still has coefficients in $[-1,1]$. Multiplying a correction by $\eta$ gives physical mesh at most $\eta\kappa_0$; outside the search region, retain $\vH_0$. Thus every full candidate is within coefficient distance $2\eta\leq2$ of $\vH$, as required by \Cref{lem:tested-identification}. Assign each Pauli term to one fixed site in its support for recording its update.

\endgroup

\begin{algorithm}[H]
\caption{One local refinement step}
\label{alg:iterative-local-learning}
\begin{algorithmic}[1]
\Require Model of \Cref{defn:local-Hamiltonian}, copies of $\vsigma$, current estimate $\vH_0$ with $\|h_0-h\|_\infty\leq\eta\leq{1}$, measurement accuracy $\delta_Q$, and failure budget $\delta$.
\Ensure {Updated coefficients with $\|h_{\mathrm{new}}-h\|_\infty\leq\eta/{2}$, under the error conditions of \Cref{lem:local-candidate-accuracy}.}
\For{each site $i\in\Lambda$ belonging to an allowed interaction term}
 \State Let $\mathcal A_i=\{\vX_i,\vY_i,\vZ_i\}$ and $\Gamma_i=\{\gamma:i\in\supp(\vP_\gamma)\}$.
 \State Form a $\kappa_0$-net $\mathcal N_i$ of normalized corrections within radius $\ell_1$ of $i$, {in coefficient $\ell_\infty$ distance and covering all admissible corrections as above}.
 \For{each $\vW,\vU\in\mathcal N_i$}
  \State Set $\vG_{\vW}=\vH_0+\eta\vW$ and $\vK_{\vU}=\vH_0+\eta\vU$ within the allowed model.
  \State Truncate their dynamics to radius $\ell_0$ around the corresponding test supports.
  \For{each $\vA\in\mathcal A_i$ and $\gamma\in\Gamma_i$}
   \State Estimate $Q([\vA,\vP_\gamma],\vA;(\vG_{\vW})_{\ell_0},(\vK_{\vU})_{\ell_0})$ to accuracy $\delta_Q$.
   \State Denote the estimate by $\widehat Q_i(\vA,\gamma;\vW,\vU)$. Allocate failure probabilities to total at most $\delta$.
  \EndFor
 \EndFor
 \State Choose $\vU_i^*\in\arg\min_{\vU\in\mathcal N_i}\max_{\vW\in\mathcal N_i,\,\vA\in\mathcal A_i,\,\gamma\in\Gamma_i}|\widehat Q_i(\vA,\gamma;\vW,\vU)|$.
 \For{each term $\gamma$ assigned to site $i$}
  \State Set $h_{\mathrm{new},\gamma}=h_{0,\gamma}+\eta\,u^*_{i,\gamma}$.
 \EndFor
\EndFor
\State \Return $h_{\mathrm{new}}$.
\end{algorithmic}
\end{algorithm}
Compatible local measurements may be scheduled in parallel. Their statistical errors are included in the budgets below.

We first establish the robustness of the observable under spatial truncation and local perturbations. This lemma shows that the errors from truncating the Hamiltonian and searching on a finite parameter net are controlled by the chosen parameters.

The proof repeatedly invokes the Lieb--Robinson bounds collected in \Cref{lem:basic-lr}. Compare also the corresponding Gibbs-state argument in Ref.~\cite{chen2025learning}.

\begin{lem}[Robustness of the identifiability observable]\label{lem:robustness}
{For local Hamiltonians as in \Cref{defn:local-Hamiltonian}, a single-site Pauli $\vA$, and a fixed-support test $\vO$ with $\|\vO\|\leq1$, the observable $Q(\vO,\vA;\vG,\vH')$ obeys the following stability bounds. }
\begin{align}
    {\bigl|}Q(\vO,\vA;\vG,\vH')-Q(\vO,\vA;\vG_l,\vH'_l){\bigr|}\leq e^{\poly(\beta^{\pm 1})}\cdot {e^{-l/\poly(\beta^{\pm1})}}.
\label{eq:robust-truncation}
\end{align}
If $\vG-\vG'$ is supported within distance $l$ of $\supp(\vO)$, $\vK-\vK'$ is supported within distance $l$ of $\supp(\vA)$, and {each Pauli coefficient changes by at most $\eta$}, then
\begin{align}
    {\bigl|}Q(\vO,\vA;\vG,\vK)-Q(\vO,\vA;\vG',\vK'){\bigr|}\leq \eta {(1+l)^D} e^{\poly(\beta^{\pm 1})}.
\label{eq:robust-local-perturbation}
\end{align}
Finally, suppose $\vG$ and $\vG'$ agree inside radius $l$ and {each Pauli coefficient changes by at most $\eta$} outside that radius, and likewise for $\vK$ and $\vK'$. Then
\begin{align}
    {\bigl|}Q(\vO,\vA;\vG,\vK)-Q(\vO,\vA;\vG',\vK'){\bigr|}\leq \eta e^{\poly(\beta^{\pm 1})}{e^{-l/\poly(\beta^{\pm1})}}.
\label{eq:robust-remote-perturbation}
\end{align}
\end{lem}
\begin{proof}
In \eqref{eq:measurable-Q}, change one Hamiltonian at a time and use $\|\vsigma\|_1=1$ and preservation of operator norms under real-time evolution. Since $\|\vO\|\leq1$ and $\|\vA\|=1$,
\begin{align*}
 \bigl|Q(\vO,\vA;\vG,\vK)-Q(\vO,\vA;\vG',\vK')\bigr|
 \leq&\frac{1}{4\pi\sigma\sqrt{2\pi}}\int_{\mathbb R^3}
 \Bigl(\|\vO_{\vG}(t-t_1)-\vO_{\vG'}(t-t_1)\|
 +\|\vA_{\vK}(t+t_2)-\vA_{\vK'}(t+t_2)\|\Bigr)\\
 &|g(t)w_{\Omega_1}(t_1)|
 \bigl(|h_+(t_2)|+|h_-(t_2)|\bigr)\,\dd t\dd t_1\dd t_2.
\end{align*}
The scalar kernels have exponential or Gaussian decay, and their zeroth and first time moments are bounded by $e^{\poly(\beta^{\pm1})}$.

For \eqref{eq:robust-truncation}, take $\vG'=\vG_l$, $\vK=\vH'$, and $\vK'=\vH'_l$. On $|t|+|t_1|+|t_2|\leq l/(4ed)$, \Cref{lem:basic-lr} and the factorial bound give an error $\lesssim2^{-l}$ for each evolved operator, with the fixed support size absorbed into the constant. On the complementary region, use the trivial norm bound of two and the kernel tails. Integrating gives $e^{\poly(\beta^{\pm1})}e^{-l/\poly(\beta^{\pm1})}$.

For \eqref{eq:robust-local-perturbation}, Duhamel's formula gives
\begin{align*}
 \|\vA_{\vK}(t+t_2)-\vA_{\vK'}(t+t_2)\|
 \leq2|t+t_2|\|\vK-\vK'\|
 \leq C_\beta\eta(1+l)^D|t+t_2|.
\end{align*}
The last inequality counts the allowed terms in a radius-$l$ patch. The same estimate holds for the $\vO$ difference with $|t-t_1|$. The first time moments of the kernels therefore give the claimed bound.

For \eqref{eq:robust-remote-perturbation}, apply \Cref{lem:basic-lr-perturbation} to the changed terms outside radius $l$:
\begin{align*}
 \|\vA_{\vK}(t)-\vA_{\vK'}(t)\|
 \lesssim\frac{\eta}{d}\sum_{\ell=l}^{\infty}S(\ell,\vA)
 \min\left\{2|t|,\frac{(2d|t|)^{\ell+1}}{(\ell+1)!}\right\}.
\end{align*}
Here $S(\ell,\vA)$ is the number of allowed interaction terms at graph distance $\ell$ from $\supp(\vA)$. For each shell, split the kernel integral at $|t|+|t_1|+|t_2|=(\ell+1)/(4ed)$. The short-time factorial bound and the complementary kernel tails give exponential decay in $\ell/\poly(\beta^{\pm1})$. Summing the polynomial shell counts retains the factor $\eta$ and gives the stated decay in $l$. The $\vO$ difference is treated identically, proving the final bound.
\end{proof}

The following lemma shows that the local net always contains a candidate whose test values are uniformly small over all auxiliary guesses. It bounds these values by the ADB contribution and the errors from discretization and spatial truncation.

    \begin{lem}[Existence of a good local-net candidate]\label{lem:good-net-candidate}
{For each site $i$ and the local net above, there exists one candidate $\vK'_{\ell_1}=\vK_{\vU}$, with $\vU\in\mathcal N_i$, such that, simultaneously for every auxiliary $\vG'_{\ell_1}=\vG_{\vW}$ with $\vW\in\mathcal N_i$, every $\vA\in\mathcal A_i$, and every $\gamma\in\Gamma_i$,}
\begin{align}
\bigl|Q([\vA,\vP_{\gamma}],\vA;{(\vG'_{\ell_1})_{\ell_0},(\vK'_{\ell_1})_{\ell_0}})\bigr|
\leq e^{\poly(\beta^{\pm 1})}\mathsf{ADB}_{\vA}[\vsigma]^{\frac{1}{32\beta d+2}}+\eta e^{-\poly(\beta^{\pm 1})}.
\label{eq:good-net-candidate-bound}
\end{align}
    \end{lem}
\emph{Proof.} Choose $\vU'_{\ell_1}$ so that $\vK'_{\ell_1}=\vH_0+\eta\vU'_{\ell_1}$ is the closest net point to $\vH$ on the search region. Let $\vH_{\ell_1}$ and $\vH_{\ell_0}$ denote the truncations of $\vH$ to radii $\ell_1$ and $\ell_0$, respectively. {On the search region, the coefficient-wise discrepancy between $\vH$ and $\vK'_{\ell_1}$ is at most $\eta\kappa_0$; the subscript $\ell_1$ indicates the support of the correction, not a truncation of the full background $\vH_0$.} Define the hybrid Hamiltonian $\vH_K$ as follows.
    \begin{itemize}
        \item For each Pauli term $\vP_{\gamma}$ {retained in $\vH_{\ell_1}$}, its coefficient in $\vH_K$ equals the corresponding coefficient in $\vH_{\ell_1}$.
        \item For {terms retained in $\vH_{\ell_0}$ but outside the search region, the coefficient} in $\vH_K$ equals that in $\vK'_{\ell_1}$, and hence that in $\vH_0$.
        \item For Pauli terms {outside the measurement region}, the coefficient in $\vH_K$ is zero.
    \end{itemize}
The second bound in \Cref{lem:robustness} gives
    \begin{align}
        {\bigl|}Q([\vA,\vP_{\gamma}],\vA;{(\vG'_{\ell_1})_{\ell_0},(\vK'_{\ell_1})_{\ell_0}})-Q([\vA,\vP_{\gamma}],\vA;{(\vG'_{\ell_1})_{\ell_0}},\vH_K){\bigr|}\leq {\eta} e^{-\poly(\beta^{\pm 1})},
    \end{align}
because {$(\vK'_{\ell_1})_{\ell_0}$ and $\vH_K$} coincide outside radius $\ell_1$ and {their coefficients differ by at most $\eta\kappa_0$ inside it. Choose $\kappa_0$ small enough to absorb the factor $(1+\ell_1)^D e^{\poly(\beta^{\pm1})}$ in the robustness bound.}

Next, comparing $\vH_K$ with the radius-$\ell_0$ truncation of $\vH$ gives
    \begin{align}
        {\bigl|}Q([\vA,\vP_{\gamma}],\vA;{(\vG'_{\ell_1})_{\ell_0}},\vH_K)-Q([\vA,\vP_{\gamma}],\vA;{(\vG'_{\ell_1})_{\ell_0}},\vH_{\ell_0}){\bigr|}\leq {\eta} e^{-\poly(\beta^{\pm 1})}.
    \end{align}
The two Hamiltonians agree within radius $\ell_1$ and {differ coefficient-wise by at most $\eta$} outside that radius, so the third bound in \Cref{lem:robustness} applies.

Finally, because $\ell_0={\log(e/\eta)}\poly(\beta^{\pm 1})$, the first bound in \Cref{lem:robustness} {compares $Q([\vA,\vP_\gamma],\vA;(\vG'_{\ell_1})_{\ell_0},\vH_{\ell_0})$ with $Q([\vA,\vP_\gamma],\vA;\vG'_{\ell_1},\vH)$} at cost ${\eta} e^{-\poly(\beta^{\pm 1})}$.

Combining these three comparisons with \Cref{lem:q-exact-candidate} proves the lemma. \hfill{$\square$}

The next lemma shows that a local candidate with small test values for every auxiliary guess has coefficients close to those of the true Hamiltonian. Under the stated error conditions, the local coefficient error is at most $\eta/{2}$, which gives the required improvement at each refinement step.

\begingroup
\begin{lem}[Small local tests imply accurate coefficients]\label{lem:local-candidate-accuracy}
Fix a site $i$ and a full candidate $\vH'=\vK_{\vU}$ with $\vU\in\mathcal N_i$. Suppose that
\begin{equation}
 \max_{\substack{\vW\in\mathcal N_i,\,\vA\in\mathcal A_i\\\gamma\in\Gamma_i}}
 \bigl|Q([\vA,\vP_\gamma],\vA;(\vG_{\vW})_{\ell_0},(\vH')_{\ell_0})\bigr|\leq\xi.
 \label{eq:local-net-small-tests}
\end{equation}
Under the hypotheses of \Cref{lem:tested-identification,thm:meta-local-faithfulness},
\begin{equation}
 \max_{\gamma\in\Gamma_i}|h_\gamma-h'_\gamma|
 \leq C_\beta\bigl(\xi+2a_\beta\eta+\eta e^{-\poly(\beta^{\pm1})}\bigr)
       +C'_\beta\epsilon_{\mathrm{ADB}}^{\frac{1}{\poly(\beta^{\pm1})}}.
 \label{eq:coefficient-interface}
\end{equation}
In particular, under a suitable choice of parameters, the local coefficient error is at most $\eta/{2}$.
\end{lem}
\begin{proof}
Choose the auxiliary net point closest to $\vH$ on the search region. Applying the three comparisons in the proof of \Cref{lem:good-net-candidate} to the auxiliary argument, and removing the measurement truncations by \Cref{lem:robustness}, gives
\begin{equation}
 \bigl|Q([\vA,\vP_\gamma],\vA;\vH,\vH')\bigr|
 \leq\xi+\eta e^{-\poly(\beta^{\pm1})}
 \qquad(\vA\in\mathcal A_i,\ \gamma\in\Gamma_i).
 \label{eq:selected-true-auxiliary}
\end{equation}
For $\vB=[\vA,\vH-\vH']$, local faithfulness and Pauli orthogonality give
\begin{align}
 \|\vB\|_{\vsigma}^2
 &=\sum_{\gamma\in\Gamma_i}(h_\gamma-h'_\gamma)
 \langle[\vA,\vP_\gamma],\vB\rangle_{\vsigma},\nonumber\\
 4\sum_{\gamma:\{\vA,\vP_\gamma\}=0}|h_\gamma-h'_\gamma|^2
 &=\|\vB\|_{\vtau}^2\leq C_\beta^2\|\vB\|_{\vsigma}^2.
 \label{eq:coefficient-recovery-chain}
\end{align}
At most $d$ terms contribute, so Cauchy--Schwarz bounds their coefficient errors linearly by the largest absolute pairing. Every $\gamma\in\Gamma_i$ anticommutes with some $\vA\in\mathcal A_i$. Apply \Cref{lem:tested-identification} to \eqref{eq:selected-true-auxiliary}, using $\|h-h'\|_\infty\leq2\eta$, and then \eqref{eq:coefficient-recovery-chain} to obtain \eqref{eq:coefficient-interface}.
\end{proof}

On the event that every measured test is accurate to $\delta_Q$, \Cref{lem:good-net-candidate} and empirical min--max selection ensure that the selected candidate satisfies \eqref{eq:local-net-small-tests} with
\begin{equation}
 \xi\leq e^{\poly(\beta^{\pm1})}\epsilon_{\mathrm{ADB}}^{\frac{1}{32\beta d+2}}
       +\eta e^{-\poly(\beta^{\pm1})}+2\delta_Q.
 \label{eq:full-refinement-budget}
\end{equation}
The two measurement errors account for comparing the selected empirical value with the good candidate and then returning to exact test values. {Choose $a_\beta$, the net and truncation parameters, and $\delta_Q$ so that their contribution to the coefficient error is at most $\eta/4$. For $\eta\geq C_\beta\epsilon_{\mathrm{ADB}}^{1/\poly(\beta^{\pm1})}$, the ADB contribution is also at most $\eta/4$, after enlarging $C_\beta$. Then \Cref{lem:local-candidate-accuracy} gives accuracy $\eta/2$ at each site. Recording each assigned coefficient once yields $\|h_{\mathrm{new}}-h\|_\infty\leq\eta/2$, with probability at least $1-\delta$.}
\endgroup

\subsection{Learning precision analysis}
\label{subsec:learning-precision}
\begingroup
We analyze the learning precision under global and local metastability, beginning with the accuracy attainable from a bound on the uniform ADB error.

\paragraph*{Learning precision under ADB.}
Let $\epsilon_{\mathrm{ADB}}=\max_a\mathsf{ADB}_a[\vsigma]$ as in \eqref{eq:intro-uniform-ADB-error}, {and assume the ADB smallness conditions required by local faithfulness and \Cref{lem:tested-identification}. Starting from $\vH_0=0$ and $\eta=1$, we use the assembled tests of \Cref{lem:tested-identification}, which already include the local dynamics-transfer error in the ADB term.}

The min--max argument in \Cref{subsec:refinement}, followed by local faithfulness, bounds each recorded coefficient error by the ADB contribution and the controllable errors. {Choose $a_\beta$, the net spacing, and the spatial truncation and measurement accuracies so that the sum of their contributions after coefficient recovery, including the factor $2\delta_Q$ from empirical selection, is at most $\eta/4$.} Enlarging $C_\beta=e^{\poly(\beta^{\pm1})}$ once to include the fixed comparison constants, each refinement step then satisfies
\begin{equation*}
    \|h_{\mathrm{new}}-h\|_\infty
    \leq \frac{\eta}{4}+C_\beta\epsilon_{\mathrm{ADB}}^{\frac{1}{\poly(\beta^{\pm1})}}.
\end{equation*}
Thus the error is at most $\eta/2$ whenever $\eta\geq4C_\beta\epsilon_{\mathrm{ADB}}^{\frac{1}{\poly(\beta^{\pm1})}}$. Repeating the refinement until the certified error is at most $\epsilon$ gives an estimate $h'$ satisfying
\begin{equation*}
    \max_\gamma|h'_\gamma-h_\gamma|\leq\epsilon,
    \qquad \epsilon\geq4C_\beta\epsilon_{\mathrm{ADB}}^{\frac{1}{\poly(\beta^{\pm1})}}.
\end{equation*}
 Allocating total failure probability $\delta$ across all tests and rounds makes this guarantee hold with probability at least $1-\delta$. Also, this sufficient precision threshold is independent of $n$.

\paragraph*{Learning precision under local metastability.}
Suppose $\epsilon_{\mathrm{ms}}^{\mathrm{loc}}=\max_{a\in\mathcal A}\|\CL_a[\vsigma]\|_1$ is sufficiently small. By \Cref{thm:local-adb}, the uniform ADB error satisfies
\begin{equation*}
    \epsilon_{\mathrm{ADB}}
    \leq C_\beta\epsilon_{\mathrm{ms}}^{\mathrm{loc}}
    \left[\log\frac{e}{\epsilon_{\mathrm{ms}}^{\mathrm{loc}}}\right]^{D+1}.
\end{equation*}
Substituting this estimate into the preceding ADB precision bound and absorbing the fixed prefactors into $C_\beta$, the algorithm learns the coefficients to accuracy $\max_\gamma|h'_\gamma-h_\gamma|\leq\epsilon$, with probability at least $1-\delta$, whenever
\begin{equation*}
    \epsilon\geq C_\beta\bigl(\epsilon_{\mathrm{ms}}^{\mathrm{loc}}\bigr)^{\frac{1}{\poly(\beta^{\pm1})}}
    \left[\log\frac{e}{\epsilon_{\mathrm{ms}}^{\mathrm{loc}}}\right]^{\frac{D+1}{\poly(\beta^{\pm1})}}.
\end{equation*}
Both the local metastability-to-ADB bound and the coefficient recovery have constants independent of $n$, so this precision threshold is also independent of the system size.

\paragraph*{Learning precision under global metastability.}
Suppose $0<\epsilon_{\mathrm{ms}}^{\mathrm{glo}}\leq1/2$ and the ADB bound in \Cref{thm:global-adb} satisfies the smallness assumptions above. Substituting that bound into the ADB precision guarantee and absorbing the fixed prefactors into $C_\beta$ gives coefficient accuracy $\max_\gamma|h'_\gamma-h_\gamma|\leq\epsilon$, with probability at least $1-\delta$, whenever
\begin{equation*}
    \epsilon\geq C_\beta\bigl(\epsilon_{\mathrm{ms}}^{\mathrm{glo}}\bigr)^{\frac{1}{\poly(\beta^{\pm1})}}
    \left[n+\log\frac1{\epsilon_{\mathrm{ms}}^{\mathrm{glo}}}\right]^{\frac{1}{\poly(\beta^{\pm1})}}.
\end{equation*}
\endgroup
\subsection{Learning complexity analysis}
\label{subsec:learning-complexity}
For any Hermitian observable $\vE$ with $\|\vE\|\leq1$, its expectation $\Tr(\vsigma\vE)$ can be estimated to additive error $\epsilon$, with probability at least $1-\delta$, using $O(\epsilon^{-2}\log(1/\delta))$ independent copies of $\vsigma$. We apply this standard measurement argument as in Ref.~\cite[Appendix A]{chen2025learning}. For the tests $\vO=[\vA,\vP_\gamma]$, the operator in \eqref{eq:measurement-operator} is Hermitian and has norm at most $C_\beta=e^{\poly(\beta^{\pm1})}$, by the scalar kernel bounds. Thus estimating one test to accuracy $\delta_Q$ requires $O(C_\beta^2\delta_Q^{-2}\log(1/\delta))$ copies. We now describe an efficient implementation of this measurement.

At refinement scale $\eta$, take $\delta_Q=\eta/C_\beta$ {with $C_\beta$ sufficiently large,} and use the time cutoff $T$ and grid spacing $\Delta$ given by
\begin{equation*}
 T=\poly(\beta^{\pm1})\log\frac{C_\beta}{\delta_Q},
 \qquad
 \Delta=\frac{\delta_Q}{C_\beta(1+\ell_0)^D(1+T)^3}.
\end{equation*}
First truncate the three time integrals in \eqref{eq:measurement-operator} to $[-T,T]^3$. The exponential decay of $g$ and Gaussian decay of $w_{\Omega_1}$ and $h_\pm$ make the discarded operator norm at most $\delta_Q/4$. Next discretize each integral on this interval, replacing the times by grid points and $\dd t\dd t_1\dd t_2$ by $\Delta^3$. Bounds on the derivatives of the scalar kernels and local evolutions give a discretization error at most $\delta_Q/4$ with the chosen mesh.

The resulting finite sum is implemented by the Hamiltonian-simulation and block-encoding tools used in Ref.~\cite[Appendix A]{chen2025learning}. The truncated Hamiltonians contain $O((1+\ell_0)^D)$ local terms, so their time evolutions and the corresponding evolved Pauli products admit efficient block encodings. The scalar coefficients factor into three explicit one-dimensional kernels and have total absolute weight at most $C_\beta$. {Expand the test operator into Paulis and write the grid sum as $\sum_{j=1}^{L}c_j\vB_j$. The kernel bounds give $0\leq c_j=O(C_\beta\Delta^3)$ and $L=O((1+T/\Delta)^3)$, hence $Lc_j\leq C_\beta(1+T)^3$ after enlarging $C_\beta$. Pad the sum with zero terms so that $L$ is a power of two, prepare a uniform index register, use $Lc_j/[C_\beta(1+T)^3]$ as the success amplitude of a controlled rotation, apply the controlled $\vB_j$, and unprepare the index register. This gives normalization $C_\beta(1+T)^3$ without enumerating the grid. Arithmetic, rotations, and simulations need only $O(\log(C_\beta(1+\ell_0)^D(1+T)^3/\delta_Q))$ bits of accuracy, so the preprocessing and per-shot costs remain polynomial in the quantities below.} To operator accuracy $\delta_Q/4$, the gate cost per measurement shot is
\begin{equation*}
 C_\beta\poly\!\left((1+\ell_0)^D,T,\log\frac{T}{\Delta},\log\frac1{\delta_Q}\right)
 =e^{\poly(\beta^{\pm1})}\polylog\frac1\eta.
\end{equation*}
Each shot uses one copy of $\vsigma$. {The normalization gives $O(C_\beta^2(1+T)^6\delta_Q^{-2}\log(1/\delta))$ copies per test; the extra factor $(1+T)^6$ is polylogarithmic in $1/\eta$.} The remaining error budget is assigned to statistical estimation. Classical preparation of the local Hamiltonians and evaluation of the scalar kernels obey the same polynomial cost bound.

Finally, \eqref{eq:learning-radii} implies $|\mathcal N_i|\leq(1+2/\kappa_0)^{O((1+\ell_1)^D)}\leq C_\beta$, so each round requires at most $C_\beta n$ tests. Partition them into $C_\beta(1+\ell_0)^D$ groups with disjoint supports and measure each group in parallel on the same copies. Allocate failure probability $\delta/[C_\beta n\log(e/\epsilon)]$ to each test. Using $\ell_0=\poly(\beta^{\pm1})\log(e/\eta)$ and summing over the {halving scales $\eta=1,1/2,1/4,\ldots$ down to $\epsilon$, with the last scales dominating the sum of $\eta^{-2}$,} gives
\begin{align}
 N_{\mathrm{samp}}
 &=O\!\left(\frac{e^{\poly(\beta^{\pm1})}}{\epsilon^2}
 \log\frac{n}{\delta}\,\polylog\frac1\epsilon\right),
 T_{\mathrm{tot}}
 =O\!\left(\frac{n e^{\poly(\beta^{\pm1})}}{\epsilon^2}
 \log\frac{n}{\delta}\,\polylog\frac1\epsilon\right).
 \label{eq:completed-learning-cost}
\end{align}
{Here $T_{\mathrm{tot}}$ counts all quantum gates and classical operations for the local-net candidates, with the input states supplied; parallel execution affects the depth, not this total. The measurements use ancillas and controlled operations on local patches, simulating only the known candidates.}
\subsection{Learning algorithm on different metastable states}
\label{sec:streaming-learning}

We now assume we receive $N$ independent preparations
$\vsigma_1\otimes\cdots\otimes\vsigma_N$, each satisfying the same
local or global metastability bound from \Cref{subsec:learning-precision} for the same
Hamiltonian and temperature, but the states themselves may differ. Different from the previous section, the states are not identical. Within each round, the product input
is fixed independently of the private random measurement assignment and measurement results.

We will prove that after redesigning the sampling algorithm, the same learning precision and complexity bounds hold for this scenario.

{The analysis in \Crefrange{subsec:faithfulness}{subsec:learning-complexity} assumes independent copies of the same metastable state. With different preparations, the test values need not refer to a common state, so the preceding argument does not apply directly. Although the ADB error is nonlinear in the state, the generators defining metastability are linear. Hence, for the classical average $\overline{\vsigma}=\frac{1}{N}\sum_{i=1}^{N}\vsigma_i$, the triangle inequality gives, under the respective global and local assumptions,}
\begin{align}
    \|\mathcal{L}(\overline{\vsigma})\|_1&\leq\frac{1}{N}\sum_{i=1}^{N}\|\mathcal{L}(\vsigma_i)\|_1\leq \epsilon_{\mathrm{ms}}^{\mathrm{glo}}\\
    \max_{a}\|\mathcal{L}_a(\overline{\vsigma})\|_1&\leq\max_{a}\frac{1}{N}\sum_{i=1}^{N}\|\mathcal{L}_a(\vsigma_i)\|_1\leq \epsilon_{\mathrm{ms}}^{\mathrm{loc}}
\end{align}
{The averaged state therefore satisfies the same metastability bound, and \Cref{thm:local-adb,thm:global-adb} give the corresponding ADB control. For each fixed test $Q$, the defining operator has bounded norm, so its measurement outcomes are bounded for every input state. Hoeffding's inequality therefore controls independent measurements on different preparations: their average concentrates around the expectation on $\overline{\vsigma}$, with a number of preparations quadratic in the inverse accuracy and logarithmic in the inverse failure probability. To use this averaging argument for coefficient recovery, we next modify the test observable so that a suitable auxiliary test controls the coefficient error under its own average state.}

{At refinement scale $\eta$, retain the candidates and parameters of \Cref{alg:iterative-local-learning}.}
For every candidate $\vH'$ and $\vG$ on the local net, consider the following choice for operator $\vO$:
\begin{align}
    \vO=\frac{[\vA,\vG-\vH' ]}{\eta}
\end{align}
where $\vA$ is every possible single site Pauli operator in the local net.

Notice that the test operator $\vO$ here and below may be supported on up to $dk$ sites rather
than $k$. Accordingly, when applying the preceding bounds, we replace
the fixed locality parameter $k$ by a larger constant
$k_{\mathrm{eff}}=O_{d,k}(1)$ that bounds the relevant support sizes
and interaction degrees, including those obtained by adjoining the
bounded-support Pauli expansion of $\vO$ to the Hamiltonian terms.
Whenever unit norm is required, we normalize $\vO$ and restore its
norm factor, which is also $O(1)$ for these tests.
These adjustments affect only constants depending on the fixed
locality and geometric parameters, leaving the stated dependence
on $n$, $\eta$, and $\delta$ unchanged.

Define the corresponding test observable as
\begin{align}
    Q'(\vA;\vG,\vH')=\frac{Q([\vA,\vG-\vH'],\vA;\vG,\vH')}{\eta}.
\end{align}

{This choice turns one suitably chosen auxiliary test into an approximate squared KMS norm under that test's own average state. It avoids summing pairings evaluated on different average states.} In the following we will propose an algorithm that estimates $Q'(\vA;\vG,\vH')$ on the classically averaged state $\overline{\vsigma}$.

For each fixed choice of $\vA,\vG,\vH'$, write
    $Q'_{\overline{\vsigma}}(\vA;\vG,\vH')=\tr(\vM\overline{\vsigma}).$
Here $\overline{\vsigma}$ denotes the average of the preparations
used for this test; different tests may use different averages.
Since both candidates differ coefficientwise from the current
estimate by at most $\eta$, locality gives $ \|[\vA,\vG-\vH']/\eta\|\leq 4d.$
We therefore apply the measurement construction of Appendix~D\,4
directly to this normalized test operator. For measurement accuracy
$\delta_Q>0$, truncate and discretize the time integrals and include
the implementation errors so that
\begin{align}
    \widetilde{\vM}
    =\sum_{j=1}^{L}c_j\vB_j,\qquad {c_j\geq0,}\quad \|\vB_j\|\leq 1,\qquad
    W=\sum_{j=1}^{L}c_j\leq C_\beta,\qquad
    \|\widetilde{\vM}-\vM\|\leq\frac{\delta_Q}{2}.
    \label{eq:prime-discretization}
\end{align}
The coefficients include the measurement
normalizations, with their phases absorbed into $\vB_j$. {The grid construction in \Cref{subsec:learning-complexity} gives $Lc_j\leq C_\beta(1+T)^3$, and each coefficient and controlled Pauli product can be evaluated with the per-shot cost stated there.}

After this decomposition, for every local Pauli operator $\vA$ and $\vG$, $\vH'$ in the local net, consider the following randomized sampling procedure (\Cref{alg:randomized-Qprime-measurement}):

\begin{algorithm}[htbp]
\caption{Randomized measurement of
$Q'(\vA;\vG,\vH')$ on nonidentical preparations}
\label{alg:randomized-Qprime-measurement}
\begin{algorithmic}[1]
\Require A fixed test $(\vA,\vG,\vH')$, its decomposition
    \eqref{eq:prime-discretization}, measurement accuracy
    $\delta_Q$, failure probability $\delta$, and {access to fresh independent preparations}.
\Ensure An estimate of
    $Q'_{\overline{\vsigma}}(\vA;\vG,\vH')$
    with error less than $\delta_Q$, except with probability $\delta$, where $\overline{\vsigma}=\frac{1}{N}\sum_{i=1}^{N}\vsigma_i$.
\State Set
    $N=\left\lceil
       64C_\beta^2{(1+T)^6}\delta_Q^{-2}\log(4/\delta)
       \right\rceil$.
\For{$i=1,\ldots,N$}
    \State Receive the preparation $\vsigma_i$.
    \State Independently sample $j$ {uniformly from $\{1,\ldots,L\}$}
        and toss a fair coin.
    \If{the coin selects the real part}
        \State Perform the local measurement with outcome
            $X_i\in[-1,1]$ and mean
            $\tr((\vB_j+\vB_j^\dagger)\vsigma_i)/2$.
        \State Set $Z_i=2{Lc_j}X_i$.
    \Else
        \State Perform the local measurement with outcome
            $X_i\in[-1,1]$ and mean
            $\tr((\vB_j-\vB_j^\dagger)\vsigma_i)/(2i)$.
        \State Set $Z_i=2i{Lc_j}X_i$.
    \EndIf
\EndFor
\State \Return
    $\widehat Q'(\vA;\vG,\vH')
       :=N^{-1}\sum_{i=1}^{N}Z_i$.
\end{algorithmic}
\end{algorithm}

Indeed, averaging over the sampled term, the coin, and the
measurement outcome gives
\[
    \mathbb E Z_i
    =\sum_{j=1}^{L}c_j\tr(\vB_j\vsigma_i),
    \qquad
    \mathbb E\widehat Q'
    =\tr\!\left(
       \widetilde{\vM}\frac1N\sum_{i=1}^{N}\vsigma_i
      \right).
\]
Thus each shot estimates the entire discretized test on its own
input state, rather than estimating its summands separately. Furthermore,
\Cref{lem:randomized-kernel-measurement} gives
\[
    \Pr\!\left[
       \left|\widehat Q'-
       Q'_{\overline{\vsigma}}(\vA;\vG,\vH')\right|
       \geq\delta_Q
    \right]
    \leq
    4\exp\!\left(
       -\frac{N\delta_Q^2}{64C_\beta^2{(1+T)^6}}
      \right)
    \leq\delta.
\]

For simultaneous estimation of $M$ tests in one refinement step, where $M=C_{\beta}n$ in our setting,
we need to replace the failure budget in each call by $\delta/M$.
As in Appendix~D\,4, partition the tests into at most
$C_\beta(1+\ell_0)^D$ groups with disjoint measurement supports.
All tests in a group can be performed in parallel on the same
preparation, independently sampling their kernel terms and
real or imaginary parts. Use a fresh batch for each group.
Condition on the preceding rounds, so the current estimate and tests are fixed. Independence across fresh preparations gives concentration for each test; correlations between tests on the same preparation do not affect the union bound.
A union bound guarantees accuracy $\delta_Q$ for every test,
each evaluated on its group's average state, using
\[
    O\!\left(
       \frac{e^{\poly(\beta^{\pm1})}{(1+T)^6}(1+\ell_0)^D}{\delta_Q^2}
       \log\frac{4M}{\delta}
      \right)
\]
preparations {in this round. Summing over halving scales with the failure allocation of \Cref{subsec:learning-complexity} gives \eqref{eq:completed-learning-cost}.}
\begin{lem}[Randomized measurement of an integrated observable]
\label{lem:randomized-kernel-measurement}
Fix a test $Q_{\vsigma}=\tr(\vM\vsigma)$.
After absorbing coefficient phases into the operators, suppose its
discretization satisfies \eqref{eq:prime-discretization}.
Given an arbitrary product input
$\vsigma_1\otimes\cdots\otimes\vsigma_N$, {with $N\geq64C_\beta^2(1+T)^6\delta_Q^{-2}\log(4/\delta)$,} \Cref{alg:randomized-Qprime-measurement} {estimates}
$Q_{\overline{\vsigma}}$, where
$\overline{\vsigma}=N^{-1}\sum_{i=1}^{N}\vsigma_i$,
to error less than ${\delta_Q}$ with probability at least $1-\delta$.
\end{lem}

\begin{proof}
{The variables $Z_i$ are independent, and their real and imaginary parts lie in $[-2C_\beta(1+T)^3,2C_\beta(1+T)^3]$.} Applying Hoeffding's inequality to each
coordinate and taking a union bound gives
\[
    \Pr\!\left[
       \left|\widehat Q-
       \tr(\widetilde{\vM}\overline{\vsigma})\right|\geq u
    \right]
    \leq4\exp\!\left(-\frac{Nu^2}{16{C_\beta^2(1+T)^6}}\right),
\]
since a complex error of magnitude at least $u$ has a real or
imaginary part of magnitude at least $u/\sqrt2$.
{Taking $u=\delta_Q/2$ and adding the bias $\|\widetilde{\vM}-\vM\|\leq\delta_Q/2$ proves the claim with the stated $N$.}
\end{proof}

{
After estimating the truncated tests, select
\[
 U_i^*\in\arg\min_{U\in\mathcal N_i}
 \max_{W\in\mathcal N_i,\,\vA\in\mathcal A_i}
 \bigl|\widehat Q'_i(\vA;\vG_W,\vK_U)\bigr|.
\]
For each coefficient assigned to site $i$, record $h_{\mathrm{new},\gamma}=h_{0,\gamma}+\eta u^*_{i,\gamma}$, as in \Cref{alg:iterative-local-learning}. The maximization over $\gamma$ is absent. For different choices of $U,W,\vA$, the measurements may use different batches, so their averaged states need not coincide. Thus \Cref{lem:good-net-candidate,lem:local-candidate-accuracy} cannot be applied directly as stated for a common state; in particular, the coefficient-recovery proof combines KMS pairings evaluated on that state. \Cref{lem:good-Qprime-candidate} supplies a single candidate with small test values for all auxiliary choices and all admissible average states. \Cref{lem:Qprime-coefficient-recovery} obtains a squared KMS norm from one suitably chosen auxiliary test under its own average state and uses faithfulness to bound the coefficient error. Together, these lemmas justify the min--max selection even when the tests use different average states.
}

\begin{lem}[Existence of a good local-net candidate]
\label{lem:good-Qprime-candidate}
For each site $i$, there exists $U_*\in\mathcal N_i$ such that,
for every $W\in\mathcal N_i$, $\vA\in\mathcal A_i$, and every
state $\vsigma$ satisfying the above assumptions,
\[
    \left|
    Q'_{\vsigma}\!\left(
       \vA;(\vG_W)_{\ell_0},(\vK_{U_*})_{\ell_0}
    \right)
    \right|
    \leq
    C_\beta\epsilon_{\mathrm{ADB}}^{1/(32\beta d+2)}
    +\eta e^{-\poly(\beta^{\pm1})}.
\]
In particular, the same candidate works for all admissible input states.
\end{lem}

\begin{proof}
Choose $\vK_{U_*}$ nearest to $\vH$ on the search region, with
coefficient error at most $\eta\kappa_0$ there.
For fixed $\vA,W$, set
$\vO=[\vA,\vG_W-\vK_{U_*}]/\eta$, so $\|\vO\|\leq4d$.
Keeping $\vO$ fixed, the net and spatial comparisons in
\Cref{lem:robustness} replace the measured test by
$Q_{\vsigma}(\vO,\vA;\vG_W,\vH)$
at cost $\eta e^{-\poly(\beta^{\pm1})}$.
\Cref{lem:q-exact-candidate} bounds the latter by
$C_\beta\epsilon_{\mathrm{ADB}}^{1/(32\beta d+2)}$.
\end{proof}
After that, we only need to prove that a good local candidate implies accurate coefficient recovery:
\begin{lem}[Small tests imply accurate coefficients]
\label{lem:Qprime-coefficient-recovery}
Use the local net and measurement truncations above, choosing their
errors sufficiently small compared with $a_\beta\eta$.
Fix a single-site Pauli $\vA$ and a candidate $\vH'$ with
$\|h-h'\|_\infty\leq2\eta$.
Suppose that
\[
    |Q'(\vA;\vG,\vH')|\leq\xi
    \qquad\text{for every auxiliary $\vG$ in the local net}.
\]
Each test may use a different state, provided all states satisfy
the hypotheses of Theorems~C.1 and~D.1 with the same bound
$\epsilon_{\mathrm{ADB}}$.
Then every $\gamma$ with $\{\vA,\vP_\gamma\}=0$ satisfies
\[
    |h_\gamma-h'_\gamma|^2
    \leq C_\beta\eta
    \left(
       \xi+a_\beta\eta
       +\epsilon_{\mathrm{ADB}}^{1/\poly(\beta^{\pm1})}
    \right).
\]
\end{lem}

\begin{proof}
Choose $\vG$ nearest to $\vH$ on the search region, and let
$\vsigma$ be the state used for this test. Since the region
contains all terms meeting $\operatorname{supp}(\vA)$,
\[
    \left\|\frac{[\vA,\vG-\vH]}{\eta}\right\|
    \leq2d\kappa_0
\]
Then we can apply \Cref{lem:robustness} to the following estimate:
\begin{align}
    Q'(\vA;\vG,\vH') &= Q([\vA,\vG-\vH']/\eta,\vA;\vG,\vH')\\
    &\geq\frac{\beta}{2\eta}\langle[\vA,\vG-\vH'],[\vA,\vH-\vH']\rangle_{\vsigma}-a_\beta\eta
    -C_{\beta}\epsilon_{\mathrm{ADB}}^{1/\poly(\beta^{\pm1})}\\
    &\geq \frac{\beta}{2\eta}\|[\vA,\vH-\vH']\|^2_{\vsigma}-4\beta d^2\kappa_0\eta-a_\beta\eta
    -C_{\beta}\epsilon_{\mathrm{ADB}}^{1/\poly(\beta^{\pm1})}\\
    &\geq \frac{1}{C^2_\beta}\frac{\beta}{2\eta}\|[\vA,\vH-\vH']\|^2_{\vtau}-4\beta d^2\kappa_0\eta-a_\beta\eta
    -C_{\beta}\epsilon_{\mathrm{ADB}}^{1/\poly(\beta^{\pm1})}
\end{align}
where the first inequality is due to \Cref{lem:tested-identification}, the second inequality is due to the Cauchy--Schwarz inequality, and the last inequality is due to \Cref{thm:meta-local-faithfulness}. After enlarging $C_\beta$ and taking a sufficiently small $a_{\beta}$, we have
\begin{align*}
    4\sum_{\gamma:\{\vA,\vP_\gamma\}=0}
       |h_\gamma-h'_\gamma|^2
    &=\|[\vA,\vH-\vH']\|_\tau^2\\
    &\leq C_\beta\eta
       \left(
          \xi+a_\beta\eta
          +\epsilon_{\mathrm{ADB}}^{1/\poly(\beta^{\pm1})}
       \right).
\end{align*}
So the lemma is proved.
\end{proof}
On the event of measurement accuracy $\delta_Q$, empirical
min--max selection and Lemma~\ref{lem:good-Qprime-candidate} give
\[
    \xi\leq
    C_\beta\epsilon_{\mathrm{ADB}}^{1/(32\beta d+2)}
    +\eta e^{-\poly(\beta^{\pm1})}+2\delta_Q.
\]
{
Substituting this bound into \Cref{lem:Qprime-coefficient-recovery}, taking a common smaller inverse-polynomial exponent, and recording each assigned coefficient once gives
\[
 \|h_{\mathrm{new}}-h\|_\infty^2
 \leq C_\beta\eta\left[
 (a_\beta+e^{-\poly(\beta^{\pm1})})\eta+2\delta_Q
 +\epsilon_{\mathrm{ADB}}^{1/\poly(\beta^{\pm1})}\right].
\]
Choose the controllable constants sufficiently small and enlarge $C_\beta$ in the following conditions. Then
\[
 \|h_{\mathrm{new}}-h\|_\infty\leq\frac\eta2
 \quad\text{when}\quad
 \eta\geq C_\beta\epsilon_{\mathrm{ADB}}^{1/\poly(\beta^{\pm1})},
 \qquad \delta_Q\leq\frac\eta{C_\beta}.
\]
The factor $\eta$ in the squared-error bound preserves the precision floor of \Cref{subsec:learning-precision}. Repeating this halving step with the failure allocation and cost count in \Cref{subsec:learning-complexity} gives the same final precision and complexity bounds.
}

\clearpage
\section{Lower bounds}\label{sec:lower-bounds}
We use the error measures defined in
\Cref{subsec:metastability-notions}, with Gaussian width $1/\beta$.
The examples below limit coefficient precision and local-to-global
stationarity bounds. Write
\[
 \kappa:=\int_{\mathbb R}\gamma(\omega)
 |\widehat f_{1/\beta}(\omega)|^2\dd\omega>0.
\]

\subsection{An ADB-limited learning lower bound}
\label{sec:adb-learning-lower-bound}

\begin{prop}[Information-theoretic ADB floor]
\label{prop:adb-information-floor}
For an ADB budget $\epsilon>0$, let
\begin{equation}
 \eta_*:=\min\left\{1,\frac{\log(1+\sqrt{2\epsilon})}{\beta}\right\}.
 \label{eq:adb-eta-star}
\end{equation}
There are two allowed target Hamiltonians separated by $\eta_*$ in
coefficient $\ell_\infty$ norm, with the same supplied state satisfying
$\epsilon_{\mathrm{ADB}}\leq\epsilon$ and $\epsilon_H=0$.
Every learner receiving only copies of this state has worst-case
expected coefficient error at least $\eta_*/2$.
No such learner can achieve error strictly below $\eta_*/2$ with
probability greater than $1/2$ on both instances.
\end{prop}
\begin{proof}
Take
\[
 \vH^{(0)}=0,\qquad \vH^{(1)}=\eta_*\vZ_1,
 \qquad \vsigma_*=\vI/2^n.
\]
The supplied state commutes with both targets and is the Gibbs state of
$\vH^{(0)}$. For $\vH^{(1)}$, only the $\vX_1,\vY_1$ defects can be
nonzero. The other sites factor out of their squared Hilbert--Schmidt
norms. On site $1$, Gibbs conjugation multiplies the two off-diagonal
entries by $e^{\pm\beta\eta_*}$. Thus, using $\int g=1/2$,
$\gamma\leq1$, and the normalized Gaussian transform from
\Cref{subsec:operator-fourier},
\begin{equation}
 \mathsf{ADB}_{\vX_1}[\vsigma_*]
 =\mathsf{ADB}_{\vY_1}[\vsigma_*]
 \leq\frac{(e^{\beta\eta_*}-1)^2+(1-e^{-\beta\eta_*})^2}{4}
 \leq\frac{(e^{\beta\eta_*}-1)^2}{2}\leq\epsilon.
 \label{eq:adb-second-instance-budget}
\end{equation}
Every measurement strategy has the same output distribution on these
identical input states, regardless of the number of copies. For its
estimate $\widehat h_1$, the inequality
$|\widehat h_1|+|\widehat h_1-\eta_*|\geq\eta_*$ proves the expected-error
claim. The two success intervals of radius less than $\eta_*/2$ are
disjoint, proving the probability claim.
\end{proof}
For fixed $\beta$ and $\epsilon\to0$, this gives a precision floor
$\eta_*/2\sim\sqrt{\epsilon}/(\sqrt2\beta)$ even under exact Hamiltonian
stationarity. It does not establish tightness of the learning upper bound.

The same example limits learning from metastable states, with no dependence on $n$.
For $\vH^{(0)}$, the state $\vsigma_*$ is the Gibbs state, so all its metastability errors vanish.
For $\vH^{(1)}$, the Hamiltonian part of $\CL$ vanishes on $\vsigma_*$ because $[\vH^{(1)},\vsigma_*]=0$.
Every ADB defect other than those of $\vX_1$ and $\vY_1$ vanishes, so \eqref{eq:adb-to-local-metastability} forces the corresponding components $\CL_a[\vsigma_*]$ to vanish as well.
The two remaining components are bounded by $C_\beta\sqrt{\epsilon}$, and therefore
\[
 {\epsilon_{\mathrm{ms}}^{\mathrm{loc}}}(\vsigma_*)\leq C_\beta\sqrt{\epsilon},
 \qquad
 {\epsilon_{\mathrm{ms}}^{\mathrm{glo}}}(\vsigma_*)\leq 2C_\beta\sqrt{\epsilon},
\]
uniformly in $n$.
Given a global or local metastability budget $\epsilon_{\mathrm{ms}}$, choosing $\sqrt{\epsilon}=\epsilon_{\mathrm{ms}}/(2C_\beta)$ gives a precision floor of order $\epsilon_{\mathrm{ms}}/(C_\beta\beta)$ that no number of copies can overcome.
This floor is independent of $n$, so the example does not decide whether the system-size factor in \Cref{thm:intro-global-metastable-learning} is necessary.

\subsection{Necessary system-size factors in global stationarity}
\label{subsec:commutation-error-lower-bound}

The same family witnesses both the ADB-to-global-commutation and the
local-to-global-metastability obstructions.
Fix $\beta>0$, $0<h\leq1$, and $|\theta|\leq\pi/2$, and set
\begin{equation}
 \vH=h\sum_{j=1}^n\vZ_j,\qquad
 \vG=\vX^{\otimes n},\qquad
 \vsigma_\theta=e^{-i\theta\vG/2}\vrho e^{i\theta\vG/2},
 \qquad \vrho=\frac{e^{-\beta\vH}}{(2\cosh(\beta h))^n}.
 \label{eq:comm-lb-state}
\end{equation}
\begin{prop}[Necessary linear system-size dependence]
\label{prop:commutation-error-lower-bound}
Uniformly in $n\geq1$ and $|\theta|\leq\pi/2$,
\begin{align}
 \epsilon_H(\vsigma_\theta)&=2nh\tanh(\beta h)|\sin\theta|,
 \label{eq:comm-lb-exact-commutator}\\
 \epsilon_{\mathrm{ADB}}&=\Theta_{\beta,h}(\sin^2\theta),
 &{\epsilon_{\mathrm{ms}}^{\mathrm{loc}}}(\vsigma_\theta)&=\Theta_{\beta,h}(|\sin\theta|),
 \label{eq:comm-lb-two-errors}\\
 {\epsilon_{\mathrm{ms}}^{\mathrm{glo}}}(\vsigma_\theta)&=\Theta_{\beta,h}(n|\sin\theta|).
 \label{eq:comm-lb-global-ms}
\end{align}
\end{prop}
\begin{proof}
Put $m=\tanh(\beta h)$ and $q_n=\operatorname{sech}^n(\beta h)$.
Since $\vG\vH\vG=-\vH$,
\[
 [\vH,\vsigma_\theta]
 =i\sin\theta\,\vH(\vrho-\vG\vrho\vG)\vG.
\]
The diagonal matrix $\vH(\vrho-\vG\vrho\vG)$ is negative semidefinite,
with trace $2\Tr(\vH\vrho)=-2nhm$, proving
\eqref{eq:comm-lb-exact-commutator}.

The same rotation applied to $\sqrt{\vrho}$ gives
\[
 \|\sqrt{\vsigma_\theta}-\sqrt{\vrho}\|_2^2
 =(1-\cos\theta)(1-q_n),\qquad
 \|[\vZ_i,\sqrt{\vsigma_\theta}]\|_2^2
 =2(1-q_n)\sin^2\theta.
\]
The $\vZ_i$ jump commutes with $\vH$, so
$\mathsf{ADB}_{\vZ_i}[\vsigma_\theta]=\kappa(1-q_n)\sin^2\theta$.
Lemma~\ref{lem:state-cont}, applied with the exact Gibbs state, bounds
every local ADB defect by $4(1-\cos\theta)(1-q_n)$.
Since $1-q_n\geq1-\operatorname{sech}(\beta h)>0$, these bounds give
the asserted ADB scaling. Equation~\eqref{eq:adb-to-local-metastability}
then gives the local-MS upper bound. For the lower bound,
\[
 \|\CL_{\vZ_i}[\vsigma_\theta]\|_1
 =\kappa|\sin\theta|\,\|\vrho-\vG\vrho\vG\|_1
 \geq2\kappa m|\sin\theta|,
\]
where the last inequality follows by tracing out all but one site.

It remains to rule out cancellation in the full generator. For this
product Hamiltonian, each filtered jump acts on one site. Its squared
modulus is diagonal, so the coherent correction
\eqref{eq:thermal-coherent-correction} vanishes by oddness of $c(t)$.
The filtered $\vX_i,\vY_i$ jumps are respectively real and purely
imaginary off-diagonal matrices, while the $\vZ_i$ jump is diagonal.
Their adjoint generators therefore map $\vX_i$ to a scalar multiple of
$\vX_i$, and hence $\CL_a^\dagger[\vG]$ is a scalar multiple of $\vG$.
As $\Tr(\vG\vsigma_\theta)=0$, trace duality yields
\begin{equation}
 {\epsilon_{\mathrm{ms}}^{\mathrm{glo}}}(\vsigma_\theta)
 \geq|\Tr(\vG\CL[\vsigma_\theta])|
 =|\Tr(\vG(-i[\vH,\vsigma_\theta]))|
 =2nhm|\sin\theta|.
 \label{eq:comm-lb-global-witness}
\end{equation}
The matching upper bound follows from \eqref{eq:relations-triangle}
and the local-MS estimate above.
\end{proof}
Thus $\epsilon_H=\Theta_{\beta,h}(n\sqrt{\epsilon_{\mathrm{ADB}}})$ and
${\epsilon_{\mathrm{ms}}^{\mathrm{glo}}}=\Theta_{\beta,h}(n{\epsilon_{\mathrm{ms}}^{\mathrm{loc}}})$.
At any fixed sufficiently small nonzero $\theta$, both local errors are
small independently of $n$, while both global errors grow linearly
in $n$. Thus neither conversion can replace its system-size factor by
$o(n)$ at a fixed local error budget. Letting $\theta\to0$ also excludes
an ADB exponent greater than $1/2$ for global commutation. These examples
do not determine the optimal error exponents of the general upper bounds.

\bibliography{metastable_references}

\begin{thebibliography}{77}%
\makeatletter
\providecommand \@ifxundefined [1]{%
 \@ifx{#1\undefined}
}%
\providecommand \@ifnum [1]{%
 \ifnum #1\expandafter \@firstoftwo
 \else \expandafter \@secondoftwo
 \fi
}%
\providecommand \@ifx [1]{%
 \ifx #1\expandafter \@firstoftwo
 \else \expandafter \@secondoftwo
 \fi
}%
\providecommand \natexlab [1]{#1}%
\providecommand \enquote  [1]{``#1''}%
\providecommand \bibnamefont  [1]{#1}%
\providecommand \bibfnamefont [1]{#1}%
\providecommand \citenamefont [1]{#1}%
\providecommand \href@noop [0]{\@secondoftwo}%
\providecommand \href [0]{\begingroup \@sanitize@url \@href}%
\providecommand \@href[1]{\@@startlink{#1}\@@href}%
\providecommand \@@href[1]{\endgroup#1\@@endlink}%
\providecommand \@sanitize@url [0]{\catcode `\\12\catcode `\$12\catcode
  `\&12\catcode `\#12\catcode `\^12\catcode `\_12\catcode `\%12\relax}%
\providecommand \@@startlink[1]{}%
\providecommand \@@endlink[0]{}%
\providecommand \url  [0]{\begingroup\@sanitize@url \@url }%
\providecommand \@url [1]{\endgroup\@href {#1}{\urlprefix }}%
\providecommand \urlprefix  [0]{URL }%
\providecommand \Eprint [0]{\href }%
\providecommand \doibase [0]{https://doi.org/}%
\providecommand \selectlanguage [0]{\@gobble}%
\providecommand \bibinfo  [0]{\@secondoftwo}%
\providecommand \bibfield  [0]{\@secondoftwo}%
\providecommand \translation [1]{[#1]}%
\providecommand \BibitemOpen [0]{}%
\providecommand \bibitemStop [0]{}%
\providecommand \bibitemNoStop [0]{.\EOS\space}%
\providecommand \EOS [0]{\spacefactor3000\relax}%
\providecommand \BibitemShut  [1]{\csname bibitem#1\endcsname}%
\let\auto@bib@innerbib\@empty
\bibitem [{\citenamefont {Gebhart}\ \emph {et~al.}(2023)\citenamefont
  {Gebhart}, \citenamefont {Santagati}, \citenamefont {Gentile}, \citenamefont
  {Gauger}, \citenamefont {Craig}, \citenamefont {Ares}, \citenamefont
  {Banchi}, \citenamefont {Marquardt}, \citenamefont {Pezz{\`e}},\ and\
  \citenamefont {Bonato}}]{gebhartLearningQuantumSystems2023}%
  \BibitemOpen
  \bibfield  {author} {\bibinfo {author} {\bibfnamefont {V.}~\bibnamefont
  {Gebhart}}, \bibinfo {author} {\bibfnamefont {R.}~\bibnamefont {Santagati}},
  \bibinfo {author} {\bibfnamefont {A.~A.}\ \bibnamefont {Gentile}}, \bibinfo
  {author} {\bibfnamefont {E.~M.}\ \bibnamefont {Gauger}}, \bibinfo {author}
  {\bibfnamefont {D.}~\bibnamefont {Craig}}, \bibinfo {author} {\bibfnamefont
  {N.}~\bibnamefont {Ares}}, \bibinfo {author} {\bibfnamefont {L.}~\bibnamefont
  {Banchi}}, \bibinfo {author} {\bibfnamefont {F.}~\bibnamefont {Marquardt}},
  \bibinfo {author} {\bibfnamefont {L.}~\bibnamefont {Pezz{\`e}}},\ and\
  \bibinfo {author} {\bibfnamefont {C.}~\bibnamefont {Bonato}},\ }\bibfield
  {title} {\bibinfo {title} {Learning quantum systems},\ }\href
  {https://doi.org/10.1038/s42254-022-00552-1} {\bibfield  {journal} {\bibinfo
  {journal} {Nat. Rev. Phys.}\ }\textbf {\bibinfo {volume} {5}},\ \bibinfo
  {pages} {141} (\bibinfo {year} {2023})}\BibitemShut {NoStop}%
\bibitem [{\citenamefont {Anshu}\ and\ \citenamefont
  {Arunachalam}(2024)}]{anshuSurveyComplexityLearning2023}%
  \BibitemOpen
  \bibfield  {author} {\bibinfo {author} {\bibfnamefont {A.}~\bibnamefont
  {Anshu}}\ and\ \bibinfo {author} {\bibfnamefont {S.}~\bibnamefont
  {Arunachalam}},\ }\bibfield  {title} {\bibinfo {title} {A survey on the
  complexity of learning quantum states},\ }\href
  {https://doi.org/10.1038/s42254-023-00662-4} {\bibfield  {journal} {\bibinfo
  {journal} {Nat. Rev. Phys.}\ }\textbf {\bibinfo {volume} {6}},\ \bibinfo
  {pages} {59} (\bibinfo {year} {2024})}\BibitemShut {NoStop}%
\bibitem [{\citenamefont {Huang}\ \emph {et~al.}(2022)\citenamefont {Huang},
  \citenamefont {Broughton}, \citenamefont {Cotler}, \citenamefont {Chen},
  \citenamefont {Li}, \citenamefont {Mohseni}, \citenamefont {Neven},
  \citenamefont {Babbush}, \citenamefont {Kueng}, \citenamefont {Preskill},\
  and\ \citenamefont {McClean}}]{huangQuantumAdvantageLearning2022}%
  \BibitemOpen
  \bibfield  {author} {\bibinfo {author} {\bibfnamefont {H.-Y.}\ \bibnamefont
  {Huang}}, \bibinfo {author} {\bibfnamefont {M.}~\bibnamefont {Broughton}},
  \bibinfo {author} {\bibfnamefont {J.}~\bibnamefont {Cotler}}, \bibinfo
  {author} {\bibfnamefont {S.}~\bibnamefont {Chen}}, \bibinfo {author}
  {\bibfnamefont {J.}~\bibnamefont {Li}}, \bibinfo {author} {\bibfnamefont
  {M.}~\bibnamefont {Mohseni}}, \bibinfo {author} {\bibfnamefont
  {H.}~\bibnamefont {Neven}}, \bibinfo {author} {\bibfnamefont
  {R.}~\bibnamefont {Babbush}}, \bibinfo {author} {\bibfnamefont
  {R.}~\bibnamefont {Kueng}}, \bibinfo {author} {\bibfnamefont
  {J.}~\bibnamefont {Preskill}},\ and\ \bibinfo {author} {\bibfnamefont
  {J.~R.}\ \bibnamefont {McClean}},\ }\bibfield  {title} {\bibinfo {title}
  {Quantum advantage in learning from experiments},\ }\href
  {https://doi.org/10.1126/science.abn7293} {\bibfield  {journal} {\bibinfo
  {journal} {Science}\ }\textbf {\bibinfo {volume} {376}},\ \bibinfo {pages}
  {1182} (\bibinfo {year} {2022})}\BibitemShut {NoStop}%
\bibitem [{\citenamefont {Haah}\ \emph {et~al.}(2016)\citenamefont {Haah},
  \citenamefont {Harrow}, \citenamefont {Ji}, \citenamefont {Wu},\ and\
  \citenamefont {Yu}}]{haahSampleoptimalTomographyQuantum2016}%
  \BibitemOpen
  \bibfield  {author} {\bibinfo {author} {\bibfnamefont {J.}~\bibnamefont
  {Haah}}, \bibinfo {author} {\bibfnamefont {A.~W.}\ \bibnamefont {Harrow}},
  \bibinfo {author} {\bibfnamefont {Z.}~\bibnamefont {Ji}}, \bibinfo {author}
  {\bibfnamefont {X.}~\bibnamefont {Wu}},\ and\ \bibinfo {author}
  {\bibfnamefont {N.}~\bibnamefont {Yu}},\ }\bibfield  {title} {\bibinfo
  {title} {Sample-optimal tomography of quantum states},\ }in\ \href
  {https://doi.org/10.1145/2897518.2897585} {\emph {\bibinfo {booktitle}
  {Proceedings of the Forty-Eighth Annual ACM Symposium on Theory of
  Computing}}}\ (\bibinfo {year} {2016})\ pp.\ \bibinfo {pages}
  {913--925}\BibitemShut {NoStop}%
\bibitem [{\citenamefont {O'Donnell}\ and\ \citenamefont
  {Wright}(2016)}]{odonnellEfficientQuantumTomography2016}%
  \BibitemOpen
  \bibfield  {author} {\bibinfo {author} {\bibfnamefont {R.}~\bibnamefont
  {O'Donnell}}\ and\ \bibinfo {author} {\bibfnamefont {J.}~\bibnamefont
  {Wright}},\ }\bibfield  {title} {\bibinfo {title} {Efficient quantum
  tomography},\ }in\ \href {https://doi.org/10.1145/2897518.2897544} {\emph
  {\bibinfo {booktitle} {Proceedings of the Forty-Eighth Annual ACM Symposium
  on Theory of Computing}}}\ (\bibinfo {year} {2016})\ pp.\ \bibinfo {pages}
  {899--912}\BibitemShut {NoStop}%
\bibitem [{\citenamefont {Cramer}\ \emph {et~al.}(2010)\citenamefont {Cramer},
  \citenamefont {Plenio}, \citenamefont {Flammia}, \citenamefont {Somma},
  \citenamefont {Gross}, \citenamefont {Bartlett}, \citenamefont
  {{Landon-Cardinal}}, \citenamefont {Poulin},\ and\ \citenamefont
  {Liu}}]{cramerEfficientQuantumState2010}%
  \BibitemOpen
  \bibfield  {author} {\bibinfo {author} {\bibfnamefont {M.}~\bibnamefont
  {Cramer}}, \bibinfo {author} {\bibfnamefont {M.~B.}\ \bibnamefont {Plenio}},
  \bibinfo {author} {\bibfnamefont {S.~T.}\ \bibnamefont {Flammia}}, \bibinfo
  {author} {\bibfnamefont {R.}~\bibnamefont {Somma}}, \bibinfo {author}
  {\bibfnamefont {D.}~\bibnamefont {Gross}}, \bibinfo {author} {\bibfnamefont
  {S.~D.}\ \bibnamefont {Bartlett}}, \bibinfo {author} {\bibfnamefont
  {O.}~\bibnamefont {{Landon-Cardinal}}}, \bibinfo {author} {\bibfnamefont
  {D.}~\bibnamefont {Poulin}},\ and\ \bibinfo {author} {\bibfnamefont {Y.-K.}\
  \bibnamefont {Liu}},\ }\bibfield  {title} {\bibinfo {title} {Efficient
  quantum state tomography},\ }\href {https://doi.org/10.1038/ncomms1147}
  {\bibfield  {journal} {\bibinfo  {journal} {Nat. Commun.}\ }\textbf {\bibinfo
  {volume} {1}},\ \bibinfo {pages} {149} (\bibinfo {year} {2010})}\BibitemShut
  {NoStop}%
\bibitem [{\citenamefont {Flammia}\ \emph {et~al.}(2012)\citenamefont
  {Flammia}, \citenamefont {Gross}, \citenamefont {Liu},\ and\ \citenamefont
  {Eisert}}]{flammiaQuantumTomographyCompressed2012}%
  \BibitemOpen
  \bibfield  {author} {\bibinfo {author} {\bibfnamefont {S.~T.}\ \bibnamefont
  {Flammia}}, \bibinfo {author} {\bibfnamefont {D.}~\bibnamefont {Gross}},
  \bibinfo {author} {\bibfnamefont {Y.-K.}\ \bibnamefont {Liu}},\ and\ \bibinfo
  {author} {\bibfnamefont {J.}~\bibnamefont {Eisert}},\ }\bibfield  {title}
  {\bibinfo {title} {Quantum tomography via compressed sensing: Error bounds,
  sample complexity and efficient estimators},\ }\href
  {https://doi.org/10.1088/1367-2630/14/9/095022} {\bibfield  {journal}
  {\bibinfo  {journal} {New J. Phys.}\ }\textbf {\bibinfo {volume} {14}},\
  \bibinfo {pages} {095022} (\bibinfo {year} {2012})}\BibitemShut {NoStop}%
\bibitem [{\citenamefont {Huang}\ \emph {et~al.}(2020)\citenamefont {Huang},
  \citenamefont {Kueng},\ and\ \citenamefont
  {Preskill}}]{huangPredictingManyProperties2020}%
  \BibitemOpen
  \bibfield  {author} {\bibinfo {author} {\bibfnamefont {H.-Y.}\ \bibnamefont
  {Huang}}, \bibinfo {author} {\bibfnamefont {R.}~\bibnamefont {Kueng}},\ and\
  \bibinfo {author} {\bibfnamefont {J.}~\bibnamefont {Preskill}},\ }\bibfield
  {title} {\bibinfo {title} {Predicting many properties of a quantum system
  from very few measurements},\ }\href
  {https://doi.org/10.1038/s41567-020-0932-7} {\bibfield  {journal} {\bibinfo
  {journal} {Nat. Phys.}\ }\textbf {\bibinfo {volume} {16}},\ \bibinfo {pages}
  {1050} (\bibinfo {year} {2020})}\BibitemShut {NoStop}%
\bibitem [{\citenamefont
  {Jaynes}(1957)}]{jaynesInformationTheoryStatistical1957}%
  \BibitemOpen
  \bibfield  {author} {\bibinfo {author} {\bibfnamefont {E.~T.}\ \bibnamefont
  {Jaynes}},\ }\bibfield  {title} {\bibinfo {title} {Information theory and
  statistical mechanics},\ }\href {https://doi.org/10.1103/PhysRev.106.620}
  {\bibfield  {journal} {\bibinfo  {journal} {Phys. Rev.}\ }\textbf {\bibinfo
  {volume} {106}},\ \bibinfo {pages} {620} (\bibinfo {year}
  {1957})}\BibitemShut {NoStop}%
\bibitem [{\citenamefont
  {Alhambra}(2023)}]{alhambraQuantumManyBodySystems2023}%
  \BibitemOpen
  \bibfield  {author} {\bibinfo {author} {\bibfnamefont {{\'A}.~M.}\
  \bibnamefont {Alhambra}},\ }\bibfield  {title} {\bibinfo {title} {Quantum
  many-body systems in thermal equilibrium},\ }\href
  {https://doi.org/10.1103/PRXQuantum.4.040201} {\bibfield  {journal} {\bibinfo
   {journal} {PRX Quantum}\ }\textbf {\bibinfo {volume} {4}},\ \bibinfo {pages}
  {040201} (\bibinfo {year} {2023})}\BibitemShut {NoStop}%
\bibitem [{\citenamefont {Ackley}\ \emph {et~al.}(1985)\citenamefont {Ackley},
  \citenamefont {Hinton},\ and\ \citenamefont
  {Sejnowski}}]{ackleyLearningAlgorithmBoltzmann1985}%
  \BibitemOpen
  \bibfield  {author} {\bibinfo {author} {\bibfnamefont {D.}~\bibnamefont
  {Ackley}}, \bibinfo {author} {\bibfnamefont {G.}~\bibnamefont {Hinton}},\
  and\ \bibinfo {author} {\bibfnamefont {T.}~\bibnamefont {Sejnowski}},\
  }\bibfield  {title} {\bibinfo {title} {A learning algorithm for {Boltzmann}
  machines},\ }\href {https://doi.org/10.1016/S0364-0213(85)80012-4} {\bibfield
   {journal} {\bibinfo  {journal} {Cogn. Sci.}\ }\textbf {\bibinfo {volume}
  {9}},\ \bibinfo {pages} {147} (\bibinfo {year} {1985})}\BibitemShut {NoStop}%
\bibitem [{\citenamefont
  {Bresler}(2015)}]{breslerEfficientlyLearningIsing2014}%
  \BibitemOpen
  \bibfield  {author} {\bibinfo {author} {\bibfnamefont {G.}~\bibnamefont
  {Bresler}},\ }\bibfield  {title} {\bibinfo {title} {Efficiently learning
  {Ising} models on arbitrary graphs},\ }in\ \href
  {https://doi.org/10.1145/2746539.2746631} {\emph {\bibinfo {booktitle}
  {Proceedings of the Forty-Seventh Annual ACM Symposium on Theory of
  Computing}}}\ (\bibinfo {year} {2015})\ pp.\ \bibinfo {pages}
  {771--782}\BibitemShut {NoStop}%
\bibitem [{\citenamefont {Vuffray}\ \emph {et~al.}(2016)\citenamefont
  {Vuffray}, \citenamefont {Misra}, \citenamefont {Lokhov},\ and\ \citenamefont
  {Chertkov}}]{vuffrayInteractionScreeningEfficient2016}%
  \BibitemOpen
  \bibfield  {author} {\bibinfo {author} {\bibfnamefont {M.}~\bibnamefont
  {Vuffray}}, \bibinfo {author} {\bibfnamefont {S.}~\bibnamefont {Misra}},
  \bibinfo {author} {\bibfnamefont {A.~Y.}\ \bibnamefont {Lokhov}},\ and\
  \bibinfo {author} {\bibfnamefont {M.}~\bibnamefont {Chertkov}},\ }\bibfield
  {title} {\bibinfo {title} {Interaction screening: Efficient and
  sample-optimal learning of {Ising} models},\ }in\ \href@noop {} {\emph
  {\bibinfo {booktitle} {Proceedings of the 30th International Conference on
  Neural Information Processing Systems}}},\ \bibinfo {series and number}
  {NIPS'16}\ (\bibinfo {year} {2016})\ pp.\ \bibinfo {pages}
  {2603--2611}\BibitemShut {NoStop}%
\bibitem [{\citenamefont {Klivans}\ and\ \citenamefont
  {Meka}(2017)}]{klivansLearningGraphicalModels2017}%
  \BibitemOpen
  \bibfield  {author} {\bibinfo {author} {\bibfnamefont {A.}~\bibnamefont
  {Klivans}}\ and\ \bibinfo {author} {\bibfnamefont {R.}~\bibnamefont {Meka}},\
  }\bibfield  {title} {\bibinfo {title} {Learning graphical models using
  multiplicative weights},\ }in\ \href {https://doi.org/10.1109/FOCS.2017.39}
  {\emph {\bibinfo {booktitle} {2017 IEEE 58th Annual Symposium on Foundations
  of Computer Science (FOCS)}}}\ (\bibinfo {year} {2017})\ pp.\ \bibinfo
  {pages} {343--354}\BibitemShut {NoStop}%
\bibitem [{\citenamefont {Amin}\ \emph {et~al.}(2018)\citenamefont {Amin},
  \citenamefont {Andriyash}, \citenamefont {Rolfe}, \citenamefont
  {Kulchytskyy},\ and\ \citenamefont
  {Melko}}]{aminQuantumBoltzmannMachine2018}%
  \BibitemOpen
  \bibfield  {author} {\bibinfo {author} {\bibfnamefont {M.~H.}\ \bibnamefont
  {Amin}}, \bibinfo {author} {\bibfnamefont {E.}~\bibnamefont {Andriyash}},
  \bibinfo {author} {\bibfnamefont {J.}~\bibnamefont {Rolfe}}, \bibinfo
  {author} {\bibfnamefont {B.}~\bibnamefont {Kulchytskyy}},\ and\ \bibinfo
  {author} {\bibfnamefont {R.}~\bibnamefont {Melko}},\ }\bibfield  {title}
  {\bibinfo {title} {Quantum {Boltzmann} machine},\ }\href
  {https://doi.org/10.1103/PhysRevX.8.021050} {\bibfield  {journal} {\bibinfo
  {journal} {Phys. Rev. X}\ }\textbf {\bibinfo {volume} {8}},\ \bibinfo {pages}
  {021050} (\bibinfo {year} {2018})}\BibitemShut {NoStop}%
\bibitem [{\citenamefont {Samarakoon}\ \emph {et~al.}(2020)\citenamefont
  {Samarakoon}, \citenamefont {Barros}, \citenamefont {Li}, \citenamefont
  {Eisenbach}, \citenamefont {Zhang}, \citenamefont {Ye}, \citenamefont
  {Sharma}, \citenamefont {Dun}, \citenamefont {Zhou}, \citenamefont {Grigera},
  \citenamefont {Batista},\ and\ \citenamefont
  {Tennant}}]{samarakoonMachinelearningassistedInsightSpin2020}%
  \BibitemOpen
  \bibfield  {author} {\bibinfo {author} {\bibfnamefont {A.~M.}\ \bibnamefont
  {Samarakoon}}, \bibinfo {author} {\bibfnamefont {K.}~\bibnamefont {Barros}},
  \bibinfo {author} {\bibfnamefont {Y.~W.}\ \bibnamefont {Li}}, \bibinfo
  {author} {\bibfnamefont {M.}~\bibnamefont {Eisenbach}}, \bibinfo {author}
  {\bibfnamefont {Q.}~\bibnamefont {Zhang}}, \bibinfo {author} {\bibfnamefont
  {F.}~\bibnamefont {Ye}}, \bibinfo {author} {\bibfnamefont {V.}~\bibnamefont
  {Sharma}}, \bibinfo {author} {\bibfnamefont {Z.~L.}\ \bibnamefont {Dun}},
  \bibinfo {author} {\bibfnamefont {H.}~\bibnamefont {Zhou}}, \bibinfo {author}
  {\bibfnamefont {S.~A.}\ \bibnamefont {Grigera}}, \bibinfo {author}
  {\bibfnamefont {C.~D.}\ \bibnamefont {Batista}},\ and\ \bibinfo {author}
  {\bibfnamefont {D.~A.}\ \bibnamefont {Tennant}},\ }\bibfield  {title}
  {\bibinfo {title} {Machine-learning-assisted insight into spin ice
  {Dy$_2$Ti$_2$O$_7$}},\ }\href {https://doi.org/10.1038/s41467-020-14660-y}
  {\bibfield  {journal} {\bibinfo  {journal} {Nat. Commun.}\ }\textbf {\bibinfo
  {volume} {11}},\ \bibinfo {pages} {892} (\bibinfo {year} {2020})}\BibitemShut
  {NoStop}%
\bibitem [{\citenamefont {Wiebe}\ \emph {et~al.}(2014)\citenamefont {Wiebe},
  \citenamefont {Granade}, \citenamefont {Ferrie},\ and\ \citenamefont
  {Cory}}]{wiebeHamiltonianLearningCertification2014}%
  \BibitemOpen
  \bibfield  {author} {\bibinfo {author} {\bibfnamefont {N.}~\bibnamefont
  {Wiebe}}, \bibinfo {author} {\bibfnamefont {C.}~\bibnamefont {Granade}},
  \bibinfo {author} {\bibfnamefont {C.}~\bibnamefont {Ferrie}},\ and\ \bibinfo
  {author} {\bibfnamefont {D.~G.}\ \bibnamefont {Cory}},\ }\bibfield  {title}
  {\bibinfo {title} {{Hamiltonian} learning and certification using quantum
  resources},\ }\href {https://doi.org/10.1103/PhysRevLett.112.190501}
  {\bibfield  {journal} {\bibinfo  {journal} {Phys. Rev. Lett.}\ }\textbf
  {\bibinfo {volume} {112}},\ \bibinfo {pages} {190501} (\bibinfo {year}
  {2014})}\BibitemShut {NoStop}%
\bibitem [{\citenamefont {Eisert}\ \emph {et~al.}(2020)\citenamefont {Eisert},
  \citenamefont {Hangleiter}, \citenamefont {Walk}, \citenamefont {Roth},
  \citenamefont {Markham}, \citenamefont {Parekh}, \citenamefont {Chabaud},\
  and\ \citenamefont {Kashefi}}]{eisertQuantumCertificationBenchmarking2020}%
  \BibitemOpen
  \bibfield  {author} {\bibinfo {author} {\bibfnamefont {J.}~\bibnamefont
  {Eisert}}, \bibinfo {author} {\bibfnamefont {D.}~\bibnamefont {Hangleiter}},
  \bibinfo {author} {\bibfnamefont {N.}~\bibnamefont {Walk}}, \bibinfo {author}
  {\bibfnamefont {I.}~\bibnamefont {Roth}}, \bibinfo {author} {\bibfnamefont
  {D.}~\bibnamefont {Markham}}, \bibinfo {author} {\bibfnamefont
  {R.}~\bibnamefont {Parekh}}, \bibinfo {author} {\bibfnamefont
  {U.}~\bibnamefont {Chabaud}},\ and\ \bibinfo {author} {\bibfnamefont
  {E.}~\bibnamefont {Kashefi}},\ }\bibfield  {title} {\bibinfo {title} {Quantum
  certification and benchmarking},\ }\href
  {https://doi.org/10.1038/s42254-020-0186-4} {\bibfield  {journal} {\bibinfo
  {journal} {Nat. Rev. Phys.}\ }\textbf {\bibinfo {volume} {2}},\ \bibinfo
  {pages} {382} (\bibinfo {year} {2020})}\BibitemShut {NoStop}%
\bibitem [{\citenamefont {Carrasco}\ \emph {et~al.}(2021)\citenamefont
  {Carrasco}, \citenamefont {Elben}, \citenamefont {Kokail}, \citenamefont
  {Kraus},\ and\ \citenamefont
  {Zoller}}]{carrascoTheoreticalExperimentalPerspectives2021}%
  \BibitemOpen
  \bibfield  {author} {\bibinfo {author} {\bibfnamefont {J.}~\bibnamefont
  {Carrasco}}, \bibinfo {author} {\bibfnamefont {A.}~\bibnamefont {Elben}},
  \bibinfo {author} {\bibfnamefont {C.}~\bibnamefont {Kokail}}, \bibinfo
  {author} {\bibfnamefont {B.}~\bibnamefont {Kraus}},\ and\ \bibinfo {author}
  {\bibfnamefont {P.}~\bibnamefont {Zoller}},\ }\bibfield  {title} {\bibinfo
  {title} {Theoretical and experimental perspectives of quantum verification},\
  }\href {https://doi.org/10.1103/PRXQuantum.2.010102} {\bibfield  {journal}
  {\bibinfo  {journal} {PRX Quantum}\ }\textbf {\bibinfo {volume} {2}},\
  \bibinfo {pages} {010102} (\bibinfo {year} {2021})}\BibitemShut {NoStop}%
\bibitem [{\citenamefont {Kokail}\ \emph {et~al.}(2021)\citenamefont {Kokail},
  \citenamefont {van Bijnen}, \citenamefont {Elben}, \citenamefont
  {Vermersch},\ and\ \citenamefont
  {Zoller}}]{kokailEntanglementHamiltonianTomography2021}%
  \BibitemOpen
  \bibfield  {author} {\bibinfo {author} {\bibfnamefont {C.}~\bibnamefont
  {Kokail}}, \bibinfo {author} {\bibfnamefont {R.}~\bibnamefont {van Bijnen}},
  \bibinfo {author} {\bibfnamefont {A.}~\bibnamefont {Elben}}, \bibinfo
  {author} {\bibfnamefont {B.}~\bibnamefont {Vermersch}},\ and\ \bibinfo
  {author} {\bibfnamefont {P.}~\bibnamefont {Zoller}},\ }\bibfield  {title}
  {\bibinfo {title} {Entanglement {Hamiltonian} tomography in quantum
  simulation},\ }\href {https://doi.org/10.1038/s41567-021-01260-w} {\bibfield
  {journal} {\bibinfo  {journal} {Nat. Phys.}\ }\textbf {\bibinfo {volume}
  {17}},\ \bibinfo {pages} {936} (\bibinfo {year} {2021})}\BibitemShut
  {NoStop}%
\bibitem [{\citenamefont {Anshu}\ \emph {et~al.}(2021)\citenamefont {Anshu},
  \citenamefont {Arunachalam}, \citenamefont {Kuwahara},\ and\ \citenamefont
  {Soleimanifar}}]{anshuSampleefficientLearningQuantum2021}%
  \BibitemOpen
  \bibfield  {author} {\bibinfo {author} {\bibfnamefont {A.}~\bibnamefont
  {Anshu}}, \bibinfo {author} {\bibfnamefont {S.}~\bibnamefont {Arunachalam}},
  \bibinfo {author} {\bibfnamefont {T.}~\bibnamefont {Kuwahara}},\ and\
  \bibinfo {author} {\bibfnamefont {M.}~\bibnamefont {Soleimanifar}},\
  }\bibfield  {title} {\bibinfo {title} {Sample-efficient learning of
  interacting quantum systems},\ }\href
  {https://doi.org/10.1038/s41567-021-01232-0} {\bibfield  {journal} {\bibinfo
  {journal} {Nat. Phys.}\ }\textbf {\bibinfo {volume} {17}},\ \bibinfo {pages}
  {931} (\bibinfo {year} {2021})}\BibitemShut {NoStop}%
\bibitem [{\citenamefont {Haah}\ \emph {et~al.}(2022)\citenamefont {Haah},
  \citenamefont {Kothari},\ and\ \citenamefont {Tang}}]{haah2022optimal}%
  \BibitemOpen
  \bibfield  {author} {\bibinfo {author} {\bibfnamefont {J.}~\bibnamefont
  {Haah}}, \bibinfo {author} {\bibfnamefont {R.}~\bibnamefont {Kothari}},\ and\
  \bibinfo {author} {\bibfnamefont {E.}~\bibnamefont {Tang}},\ }\bibfield
  {title} {\bibinfo {title} {Optimal learning of quantum {Hamiltonians} from
  high-temperature {Gibbs} states},\ }in\ \href
  {https://doi.org/10.1109/FOCS54457.2022.00020} {\emph {\bibinfo {booktitle}
  {2022 IEEE 63rd Annual Symposium on Foundations of Computer Science
  (FOCS)}}}\ (\bibinfo {year} {2022})\ pp.\ \bibinfo {pages}
  {135--146}\BibitemShut {NoStop}%
\bibitem [{\citenamefont {Bakshi}\ \emph
  {et~al.}(2024{\natexlab{a}})\citenamefont {Bakshi}, \citenamefont {Liu},
  \citenamefont {Moitra},\ and\ \citenamefont
  {Tang}}]{bakshiLearningQuantumHamiltonians2023}%
  \BibitemOpen
  \bibfield  {author} {\bibinfo {author} {\bibfnamefont {A.}~\bibnamefont
  {Bakshi}}, \bibinfo {author} {\bibfnamefont {A.}~\bibnamefont {Liu}},
  \bibinfo {author} {\bibfnamefont {A.}~\bibnamefont {Moitra}},\ and\ \bibinfo
  {author} {\bibfnamefont {E.}~\bibnamefont {Tang}},\ }\bibfield  {title}
  {\bibinfo {title} {Learning quantum {Hamiltonians} at any temperature in
  polynomial time},\ }in\ \href {https://doi.org/10.1145/3618260.3649619}
  {\emph {\bibinfo {booktitle} {Proceedings of the 56th Annual ACM Symposium on
  Theory of Computing}}}\ (\bibinfo {year} {2024})\ pp.\ \bibinfo {pages}
  {1470--1477}\BibitemShut {NoStop}%
\bibitem [{\citenamefont {Chen}\ \emph
  {et~al.}(2025{\natexlab{a}})\citenamefont {Chen}, \citenamefont {Anshu},\
  and\ \citenamefont {Nguyen}}]{chen2025learning}%
  \BibitemOpen
  \bibfield  {author} {\bibinfo {author} {\bibfnamefont {C.-F.}\ \bibnamefont
  {Chen}}, \bibinfo {author} {\bibfnamefont {A.}~\bibnamefont {Anshu}},\ and\
  \bibinfo {author} {\bibfnamefont {Q.~T.}\ \bibnamefont {Nguyen}},\ }\bibfield
   {title} {\bibinfo {title} {Learning quantum {Gibbs} states locally and
  efficiently},\ }in\ \href {https://doi.org/10.1109/FOCS63196.2025.00076}
  {\emph {\bibinfo {booktitle} {2025 IEEE 66th Annual Symposium on Foundations
  of Computer Science (FOCS)}}}\ (\bibinfo {year} {2025})\ pp.\ \bibinfo
  {pages} {1458--1464},\ \Eprint {https://arxiv.org/abs/2504.02706}
  {arXiv:2504.02706} \BibitemShut {NoStop}%
\bibitem [{\citenamefont
  {Barahona}(1982)}]{barahonaComputationalComplexityIsing1982}%
  \BibitemOpen
  \bibfield  {author} {\bibinfo {author} {\bibfnamefont {F.}~\bibnamefont
  {Barahona}},\ }\bibfield  {title} {\bibinfo {title} {On the computational
  complexity of {Ising} spin glass models},\ }\href
  {https://doi.org/10.1088/0305-4470/15/10/028} {\bibfield  {journal} {\bibinfo
   {journal} {J. Phys. A: Math. Gen.}\ }\textbf {\bibinfo {volume} {15}},\
  \bibinfo {pages} {3241} (\bibinfo {year} {1982})}\BibitemShut {NoStop}%
\bibitem [{\citenamefont
  {Sly}(2010)}]{slyComputationalTransitionUniqueness2010}%
  \BibitemOpen
  \bibfield  {author} {\bibinfo {author} {\bibfnamefont {A.}~\bibnamefont
  {Sly}},\ }\bibfield  {title} {\bibinfo {title} {Computational transition at
  the uniqueness threshold},\ }in\ \href {https://doi.org/10.1109/FOCS.2010.34}
  {\emph {\bibinfo {booktitle} {2010 IEEE 51st Annual Symposium on Foundations
  of Computer Science}}}\ (\bibinfo {year} {2010})\ pp.\ \bibinfo {pages}
  {287--296}\BibitemShut {NoStop}%
\bibitem [{\citenamefont {Sly}\ and\ \citenamefont
  {Sun}(2014)}]{slyCountingTwospinModels2014}%
  \BibitemOpen
  \bibfield  {author} {\bibinfo {author} {\bibfnamefont {A.}~\bibnamefont
  {Sly}}\ and\ \bibinfo {author} {\bibfnamefont {N.}~\bibnamefont {Sun}},\
  }\bibfield  {title} {\bibinfo {title} {Counting in two-spin models on
  {$d$}-regular graphs},\ }\href {https://doi.org/10.1214/13-AOP888} {\bibfield
   {journal} {\bibinfo  {journal} {Ann. Probab.}\ }\textbf {\bibinfo {volume}
  {42}},\ \bibinfo {pages} {2383} (\bibinfo {year} {2014})}\BibitemShut
  {NoStop}%
\bibitem [{\citenamefont {Kitaev}\ \emph {et~al.}(2002)\citenamefont {Kitaev},
  \citenamefont {Shen},\ and\ \citenamefont
  {Vyalyi}}]{kitaevClassicalQuantumComputation2002}%
  \BibitemOpen
  \bibfield  {author} {\bibinfo {author} {\bibfnamefont {A.}~\bibnamefont
  {Kitaev}}, \bibinfo {author} {\bibfnamefont {A.}~\bibnamefont {Shen}},\ and\
  \bibinfo {author} {\bibfnamefont {M.}~\bibnamefont {Vyalyi}},\ }\href
  {https://doi.org/10.1090/gsm/047} {\emph {\bibinfo {title} {Classical and
  Quantum Computation}}},\ \bibinfo {series} {Graduate Studies in Mathematics},
  Vol.~\bibinfo {volume} {47}\ (\bibinfo  {publisher} {American Mathematical
  Society},\ \bibinfo {address} {Providence, RI},\ \bibinfo {year}
  {2002})\BibitemShut {NoStop}%
\bibitem [{\citenamefont {Kempe}\ \emph {et~al.}(2006)\citenamefont {Kempe},
  \citenamefont {Kitaev},\ and\ \citenamefont
  {Regev}}]{kempeComplexityLocalHamiltonian2006}%
  \BibitemOpen
  \bibfield  {author} {\bibinfo {author} {\bibfnamefont {J.}~\bibnamefont
  {Kempe}}, \bibinfo {author} {\bibfnamefont {A.}~\bibnamefont {Kitaev}},\ and\
  \bibinfo {author} {\bibfnamefont {O.}~\bibnamefont {Regev}},\ }\bibfield
  {title} {\bibinfo {title} {The complexity of the local {Hamiltonian}
  problem},\ }\href {https://doi.org/10.1137/S0097539704445226} {\bibfield
  {journal} {\bibinfo  {journal} {SIAM J. Comput.}\ }\textbf {\bibinfo {volume}
  {35}},\ \bibinfo {pages} {1070} (\bibinfo {year} {2006})}\BibitemShut
  {NoStop}%
\bibitem [{\citenamefont {Chen}\ \emph
  {et~al.}(2025{\natexlab{b}})\citenamefont {Chen}, \citenamefont {Huang},
  \citenamefont {Preskill},\ and\ \citenamefont
  {Zhou}}]{chenLocalMinimaQuantum2023}%
  \BibitemOpen
  \bibfield  {author} {\bibinfo {author} {\bibfnamefont {C.-F.}\ \bibnamefont
  {Chen}}, \bibinfo {author} {\bibfnamefont {H.-Y.}\ \bibnamefont {Huang}},
  \bibinfo {author} {\bibfnamefont {J.}~\bibnamefont {Preskill}},\ and\
  \bibinfo {author} {\bibfnamefont {L.}~\bibnamefont {Zhou}},\ }\bibfield
  {title} {\bibinfo {title} {Local minima in quantum systems},\ }\href
  {https://doi.org/10.1038/s41567-025-02781-4} {\bibfield  {journal} {\bibinfo
  {journal} {Nat. Phys.}\ }\textbf {\bibinfo {volume} {21}},\ \bibinfo {pages}
  {654} (\bibinfo {year} {2025}{\natexlab{b}})}\BibitemShut {NoStop}%
\bibitem [{\citenamefont {Rouz{\'e}}\ \emph {et~al.}(2026)\citenamefont
  {Rouz{\'e}}, \citenamefont {Stilck~Fran{\c c}a},\ and\ \citenamefont
  {Alhambra}}]{rouzeEfficientThermalizationUniversal2026}%
  \BibitemOpen
  \bibfield  {author} {\bibinfo {author} {\bibfnamefont {C.}~\bibnamefont
  {Rouz{\'e}}}, \bibinfo {author} {\bibfnamefont {D.}~\bibnamefont
  {Stilck~Fran{\c c}a}},\ and\ \bibinfo {author} {\bibfnamefont {{\'A}.~M.}\
  \bibnamefont {Alhambra}},\ }\bibfield  {title} {\bibinfo {title} {Efficient
  thermalization and universal quantum computing with quantum {Gibbs}
  samplers},\ }\href {https://doi.org/10.1038/s41567-026-03246-y} {\bibfield
  {journal} {\bibinfo  {journal} {Nat. Phys.}\ }\textbf {\bibinfo {volume}
  {22}},\ \bibinfo {pages} {1123} (\bibinfo {year} {2026})}\BibitemShut
  {NoStop}%
\bibitem [{\citenamefont {Gamarnik}\ \emph {et~al.}(2024)\citenamefont
  {Gamarnik}, \citenamefont {Kiani},\ and\ \citenamefont
  {Zlokapa}}]{gamarnikSlowMixingQuantum2024}%
  \BibitemOpen
  \bibfield  {author} {\bibinfo {author} {\bibfnamefont {D.}~\bibnamefont
  {Gamarnik}}, \bibinfo {author} {\bibfnamefont {B.~T.}\ \bibnamefont
  {Kiani}},\ and\ \bibinfo {author} {\bibfnamefont {A.}~\bibnamefont
  {Zlokapa}},\ }\bibfield  {title} {\bibinfo {title} {Slow mixing of quantum
  {Gibbs} samplers},\ }\Eprint {https://arxiv.org/abs/2411.04300}
  {arXiv:2411.04300}  (\bibinfo {year} {2024})\BibitemShut {NoStop}%
\bibitem [{\citenamefont {Rakovszky}\ \emph {et~al.}(2026)\citenamefont
  {Rakovszky}, \citenamefont {Placke}, \citenamefont {Breuckmann},\ and\
  \citenamefont {Khemani}}]{rakovszkyBottlenecksQuantumChannels2024}%
  \BibitemOpen
  \bibfield  {author} {\bibinfo {author} {\bibfnamefont {T.}~\bibnamefont
  {Rakovszky}}, \bibinfo {author} {\bibfnamefont {B.}~\bibnamefont {Placke}},
  \bibinfo {author} {\bibfnamefont {N.~P.}\ \bibnamefont {Breuckmann}},\ and\
  \bibinfo {author} {\bibfnamefont {V.}~\bibnamefont {Khemani}},\ }\bibfield
  {title} {\bibinfo {title} {Bottlenecks in quantum channels and finite
  temperature phases of matter},\ }\href {https://doi.org/10.1103/y9z2-52rl}
  {\bibfield  {journal} {\bibinfo  {journal} {Phys. Rev. Lett.}\ }\textbf
  {\bibinfo {volume} {137}},\ \bibinfo {pages} {050402} (\bibinfo {year}
  {2026})}\BibitemShut {NoStop}%
\bibitem [{\citenamefont {Berthier}\ and\ \citenamefont
  {Biroli}(2011)}]{berthierTheoreticalPerspectiveGlass2011}%
  \BibitemOpen
  \bibfield  {author} {\bibinfo {author} {\bibfnamefont {L.}~\bibnamefont
  {Berthier}}\ and\ \bibinfo {author} {\bibfnamefont {G.}~\bibnamefont
  {Biroli}},\ }\bibfield  {title} {\bibinfo {title} {Theoretical perspective on
  the glass transition and amorphous materials},\ }\href
  {https://doi.org/10.1103/RevModPhys.83.587} {\bibfield  {journal} {\bibinfo
  {journal} {Rev. Mod. Phys.}\ }\textbf {\bibinfo {volume} {83}},\ \bibinfo
  {pages} {587} (\bibinfo {year} {2011})}\BibitemShut {NoStop}%
\bibitem [{\citenamefont {Brown}\ \emph {et~al.}(2016)\citenamefont {Brown},
  \citenamefont {Loss}, \citenamefont {Pachos}, \citenamefont {Self},\ and\
  \citenamefont {Wootton}}]{brownQuantumMemoriesFinite2016}%
  \BibitemOpen
  \bibfield  {author} {\bibinfo {author} {\bibfnamefont {B.~J.}\ \bibnamefont
  {Brown}}, \bibinfo {author} {\bibfnamefont {D.}~\bibnamefont {Loss}},
  \bibinfo {author} {\bibfnamefont {J.~K.}\ \bibnamefont {Pachos}}, \bibinfo
  {author} {\bibfnamefont {C.~N.}\ \bibnamefont {Self}},\ and\ \bibinfo
  {author} {\bibfnamefont {J.~R.}\ \bibnamefont {Wootton}},\ }\bibfield
  {title} {\bibinfo {title} {Quantum memories at finite temperature},\ }\href
  {https://doi.org/10.1103/RevModPhys.88.045005} {\bibfield  {journal}
  {\bibinfo  {journal} {Rev. Mod. Phys.}\ }\textbf {\bibinfo {volume} {88}},\
  \bibinfo {pages} {045005} (\bibinfo {year} {2016})}\BibitemShut {NoStop}%
\bibitem [{\citenamefont {Bergamaschi}\ \emph {et~al.}(2025)\citenamefont
  {Bergamaschi}, \citenamefont {Chen},\ and\ \citenamefont
  {Vazirani}}]{bergamaschi2025metastability}%
  \BibitemOpen
  \bibfield  {author} {\bibinfo {author} {\bibfnamefont {T.}~\bibnamefont
  {Bergamaschi}}, \bibinfo {author} {\bibfnamefont {C.-F.}\ \bibnamefont
  {Chen}},\ and\ \bibinfo {author} {\bibfnamefont {U.}~\bibnamefont
  {Vazirani}},\ }\bibfield  {title} {\bibinfo {title} {A structural theory of
  quantum metastability: {Markov} properties and area laws},\ }\Eprint
  {https://arxiv.org/abs/2510.08538} {arXiv:2510.08538}  (\bibinfo {year}
  {2025})\BibitemShut {NoStop}%
\bibitem [{\citenamefont {Davies}(1974)}]{daviesMarkovianMasterEquations1974}%
  \BibitemOpen
  \bibfield  {author} {\bibinfo {author} {\bibfnamefont {E.~B.}\ \bibnamefont
  {Davies}},\ }\bibfield  {title} {\bibinfo {title} {{Markovian} master
  equations},\ }\href {https://doi.org/10.1007/BF01608389} {\bibfield
  {journal} {\bibinfo  {journal} {Commun. Math. Phys.}\ }\textbf {\bibinfo
  {volume} {39}},\ \bibinfo {pages} {91} (\bibinfo {year} {1974})}\BibitemShut
  {NoStop}%
\bibitem [{\citenamefont {Chen}\ \emph
  {et~al.}(2025{\natexlab{c}})\citenamefont {Chen}, \citenamefont {Kastoryano},
  \citenamefont {Brand{\~a}o},\ and\ \citenamefont
  {Gily{\'e}n}}]{chen2025thermal}%
  \BibitemOpen
  \bibfield  {author} {\bibinfo {author} {\bibfnamefont {C.-F.}\ \bibnamefont
  {Chen}}, \bibinfo {author} {\bibfnamefont {M.}~\bibnamefont {Kastoryano}},
  \bibinfo {author} {\bibfnamefont {F.~G. S.~L.}\ \bibnamefont {Brand{\~a}o}},\
  and\ \bibinfo {author} {\bibfnamefont {A.}~\bibnamefont {Gily{\'e}n}},\
  }\bibfield  {title} {\bibinfo {title} {Efficient quantum thermal
  simulation},\ }\href {https://doi.org/10.1038/s41586-025-09583-x} {\bibfield
  {journal} {\bibinfo  {journal} {Nature}\ }\textbf {\bibinfo {volume} {646}},\
  \bibinfo {pages} {561} (\bibinfo {year} {2025}{\natexlab{c}})}\BibitemShut
  {NoStop}%
\bibitem [{\citenamefont {Chen}\ \emph
  {et~al.}(2023{\natexlab{a}})\citenamefont {Chen}, \citenamefont
  {Kastoryano},\ and\ \citenamefont {Gily{\'e}n}}]{chen2023exact}%
  \BibitemOpen
  \bibfield  {author} {\bibinfo {author} {\bibfnamefont {C.-F.}\ \bibnamefont
  {Chen}}, \bibinfo {author} {\bibfnamefont {M.~J.}\ \bibnamefont
  {Kastoryano}},\ and\ \bibinfo {author} {\bibfnamefont {A.}~\bibnamefont
  {Gily{\'e}n}},\ }\bibfield  {title} {\bibinfo {title} {An efficient and exact
  noncommutative quantum {Gibbs} sampler},\ }\Eprint
  {https://arxiv.org/abs/2311.09207} {arXiv:2311.09207}  (\bibinfo {year}
  {2023}{\natexlab{a}})\BibitemShut {NoStop}%
\bibitem [{\citenamefont {Scandi}\ and\ \citenamefont
  {Alhambra}(2026)}]{scandiThermalizationOpenManyBody2026}%
  \BibitemOpen
  \bibfield  {author} {\bibinfo {author} {\bibfnamefont {M.}~\bibnamefont
  {Scandi}}\ and\ \bibinfo {author} {\bibfnamefont {{\'A}.~M.}\ \bibnamefont
  {Alhambra}},\ }\bibfield  {title} {\bibinfo {title} {Thermalization in open
  many-body systems and {KMS} detailed balance},\ }\href
  {https://doi.org/10.1103/sfp3-3sqf} {\bibfield  {journal} {\bibinfo
  {journal} {Phys. Rev. X}\ }\textbf {\bibinfo {volume} {16}},\ \bibinfo
  {pages} {011040} (\bibinfo {year} {2026})}\BibitemShut {NoStop}%
\bibitem [{\citenamefont {Temme}\ \emph {et~al.}(2011)\citenamefont {Temme},
  \citenamefont {Osborne}, \citenamefont {Vollbrecht}, \citenamefont {Poulin},\
  and\ \citenamefont {Verstraete}}]{temmeQuantumMetropolisSampling2011}%
  \BibitemOpen
  \bibfield  {author} {\bibinfo {author} {\bibfnamefont {K.}~\bibnamefont
  {Temme}}, \bibinfo {author} {\bibfnamefont {T.~J.}\ \bibnamefont {Osborne}},
  \bibinfo {author} {\bibfnamefont {K.~G.}\ \bibnamefont {Vollbrecht}},
  \bibinfo {author} {\bibfnamefont {D.}~\bibnamefont {Poulin}},\ and\ \bibinfo
  {author} {\bibfnamefont {F.}~\bibnamefont {Verstraete}},\ }\bibfield  {title}
  {\bibinfo {title} {Quantum {Metropolis} sampling},\ }\href
  {https://doi.org/10.1038/nature09770} {\bibfield  {journal} {\bibinfo
  {journal} {Nature}\ }\textbf {\bibinfo {volume} {471}},\ \bibinfo {pages}
  {87} (\bibinfo {year} {2011})}\BibitemShut {NoStop}%
\bibitem [{\citenamefont {Cleve}\ and\ \citenamefont
  {Wang}(2017)}]{cleveEfficientQuantumAlgorithms2017}%
  \BibitemOpen
  \bibfield  {author} {\bibinfo {author} {\bibfnamefont {R.}~\bibnamefont
  {Cleve}}\ and\ \bibinfo {author} {\bibfnamefont {C.}~\bibnamefont {Wang}},\
  }\bibfield  {title} {\bibinfo {title} {Efficient quantum algorithms for
  simulating {Lindblad} evolution},\ }in\ \href
  {https://doi.org/10.4230/LIPICS.ICALP.2017.17} {\emph {\bibinfo {booktitle}
  {44th International Colloquium on Automata, Languages, and Programming (ICALP
  2017)}}},\ \bibinfo {series} {Leibniz International Proceedings in
  Informatics (LIPIcs)}, Vol.~\bibinfo {volume} {80}\ (\bibinfo {year} {2017})\
  pp.\ \bibinfo {pages} {17:1--17:14}\BibitemShut {NoStop}%
\bibitem [{\citenamefont {Chen}\ \emph
  {et~al.}(2023{\natexlab{b}})\citenamefont {Chen}, \citenamefont {Kastoryano},
  \citenamefont {Brand{\~a}o},\ and\ \citenamefont
  {Gily{\'e}n}}]{chen2023thermal}%
  \BibitemOpen
  \bibfield  {author} {\bibinfo {author} {\bibfnamefont {C.-F.}\ \bibnamefont
  {Chen}}, \bibinfo {author} {\bibfnamefont {M.~J.}\ \bibnamefont
  {Kastoryano}}, \bibinfo {author} {\bibfnamefont {F.~G. S.~L.}\ \bibnamefont
  {Brand{\~a}o}},\ and\ \bibinfo {author} {\bibfnamefont {A.}~\bibnamefont
  {Gily{\'e}n}},\ }\bibfield  {title} {\bibinfo {title} {Quantum thermal state
  preparation},\ }\Eprint {https://arxiv.org/abs/2303.18224} {arXiv:2303.18224}
   (\bibinfo {year} {2023}{\natexlab{b}})\BibitemShut {NoStop}%
\bibitem [{\citenamefont {Ding}\ \emph {et~al.}(2025)\citenamefont {Ding},
  \citenamefont {Li},\ and\ \citenamefont
  {Lin}}]{dingEfficientQuantumGibbs2024}%
  \BibitemOpen
  \bibfield  {author} {\bibinfo {author} {\bibfnamefont {Z.}~\bibnamefont
  {Ding}}, \bibinfo {author} {\bibfnamefont {B.}~\bibnamefont {Li}},\ and\
  \bibinfo {author} {\bibfnamefont {L.}~\bibnamefont {Lin}},\ }\bibfield
  {title} {\bibinfo {title} {Efficient quantum {Gibbs} samplers with
  {Kubo}--{Martin}--{Schwinger} detailed balance condition},\ }\href
  {https://doi.org/10.1007/s00220-025-05235-3} {\bibfield  {journal} {\bibinfo
  {journal} {Commun. Math. Phys.}\ }\textbf {\bibinfo {volume} {406}},\
  \bibinfo {pages} {67} (\bibinfo {year} {2025})}\BibitemShut {NoStop}%
\bibitem [{\citenamefont {Gily{\'e}n}\ \emph {et~al.}(2024)\citenamefont
  {Gily{\'e}n}, \citenamefont {Chen}, \citenamefont {Doriguello},\ and\
  \citenamefont {Kastoryano}}]{gilyen2024glauber}%
  \BibitemOpen
  \bibfield  {author} {\bibinfo {author} {\bibfnamefont {A.}~\bibnamefont
  {Gily{\'e}n}}, \bibinfo {author} {\bibfnamefont {C.-F.}\ \bibnamefont
  {Chen}}, \bibinfo {author} {\bibfnamefont {J.~F.}\ \bibnamefont
  {Doriguello}},\ and\ \bibinfo {author} {\bibfnamefont {M.~J.}\ \bibnamefont
  {Kastoryano}},\ }\bibfield  {title} {\bibinfo {title} {Quantum
  generalizations of {Glauber} and {Metropolis} dynamics},\ }\Eprint
  {https://arxiv.org/abs/2405.20322} {arXiv:2405.20322}  (\bibinfo {year}
  {2024})\BibitemShut {NoStop}%
\bibitem [{\citenamefont {Breuer}\ and\ \citenamefont
  {Petruccione}(2007)}]{breuerTheoryOpenQuantum2007}%
  \BibitemOpen
  \bibfield  {author} {\bibinfo {author} {\bibfnamefont {H.-P.}\ \bibnamefont
  {Breuer}}\ and\ \bibinfo {author} {\bibfnamefont {F.}~\bibnamefont
  {Petruccione}},\ }\href
  {https://doi.org/10.1093/acprof:oso/9780199213900.001.0001} {\emph {\bibinfo
  {title} {The Theory of Open Quantum Systems}}},\ \bibinfo {edition} {1st}\
  ed.\ (\bibinfo  {publisher} {Oxford University Press},\ \bibinfo {address}
  {Oxford},\ \bibinfo {year} {2007})\BibitemShut {NoStop}%
\bibitem [{\citenamefont {Mozgunov}\ and\ \citenamefont
  {Lidar}(2020)}]{mozgunovCompletelyPositiveMaster2020}%
  \BibitemOpen
  \bibfield  {author} {\bibinfo {author} {\bibfnamefont {E.}~\bibnamefont
  {Mozgunov}}\ and\ \bibinfo {author} {\bibfnamefont {D.}~\bibnamefont
  {Lidar}},\ }\bibfield  {title} {\bibinfo {title} {Completely positive master
  equation for arbitrary driving and small level spacing},\ }\href
  {https://doi.org/10.22331/q-2020-02-06-227} {\bibfield  {journal} {\bibinfo
  {journal} {Quantum}\ }\textbf {\bibinfo {volume} {4}},\ \bibinfo {pages}
  {227} (\bibinfo {year} {2020})}\BibitemShut {NoStop}%
\bibitem [{\citenamefont
  {Narayanan}(2025)}]{narayananImprovedAlgorithmsLearning2024}%
  \BibitemOpen
  \bibfield  {author} {\bibinfo {author} {\bibfnamefont {S.}~\bibnamefont
  {Narayanan}},\ }\bibfield  {title} {\bibinfo {title} {Improved algorithms for
  learning quantum {Hamiltonians}, via flat polynomials},\ }in\ \href@noop {}
  {\emph {\bibinfo {booktitle} {Proceedings of Thirty Eighth Conference on
  Learning Theory}}}\ (\bibinfo {year} {2025})\ pp.\ \bibinfo {pages}
  {4360--4385}\BibitemShut {NoStop}%
\bibitem [{\citenamefont {Lokhov}\ \emph {et~al.}(2018)\citenamefont {Lokhov},
  \citenamefont {Vuffray}, \citenamefont {Misra},\ and\ \citenamefont
  {Chertkov}}]{lokhovOptimalStructureParameter2018}%
  \BibitemOpen
  \bibfield  {author} {\bibinfo {author} {\bibfnamefont {A.~Y.}\ \bibnamefont
  {Lokhov}}, \bibinfo {author} {\bibfnamefont {M.}~\bibnamefont {Vuffray}},
  \bibinfo {author} {\bibfnamefont {S.}~\bibnamefont {Misra}},\ and\ \bibinfo
  {author} {\bibfnamefont {M.}~\bibnamefont {Chertkov}},\ }\bibfield  {title}
  {\bibinfo {title} {Optimal structure and parameter learning of {Ising}
  models},\ }\href {https://doi.org/10.1126/sciadv.1700791} {\bibfield
  {journal} {\bibinfo  {journal} {Sci. Adv.}\ }\textbf {\bibinfo {volume}
  {4}},\ \bibinfo {pages} {e1700791} (\bibinfo {year} {2018})}\BibitemShut
  {NoStop}%
\bibitem [{\citenamefont {Bresler}\ \emph {et~al.}(2018)\citenamefont
  {Bresler}, \citenamefont {Gamarnik},\ and\ \citenamefont
  {Shah}}]{breslerLearningGraphicalModels2018}%
  \BibitemOpen
  \bibfield  {author} {\bibinfo {author} {\bibfnamefont {G.}~\bibnamefont
  {Bresler}}, \bibinfo {author} {\bibfnamefont {D.}~\bibnamefont {Gamarnik}},\
  and\ \bibinfo {author} {\bibfnamefont {D.}~\bibnamefont {Shah}},\ }\bibfield
  {title} {\bibinfo {title} {Learning graphical models from the {Glauber}
  dynamics},\ }\href {https://doi.org/10.1109/TIT.2017.2713828} {\bibfield
  {journal} {\bibinfo  {journal} {IEEE Trans. Inform. Theory}\ }\textbf
  {\bibinfo {volume} {64}},\ \bibinfo {pages} {4072} (\bibinfo {year}
  {2018})}\BibitemShut {NoStop}%
\bibitem [{\citenamefont {Gaitonde}\ \emph {et~al.}(2025)\citenamefont
  {Gaitonde}, \citenamefont {Moitra},\ and\ \citenamefont
  {Mossel}}]{gaitondeBypassingNoisyParity2024}%
  \BibitemOpen
  \bibfield  {author} {\bibinfo {author} {\bibfnamefont {J.}~\bibnamefont
  {Gaitonde}}, \bibinfo {author} {\bibfnamefont {A.}~\bibnamefont {Moitra}},\
  and\ \bibinfo {author} {\bibfnamefont {E.}~\bibnamefont {Mossel}},\
  }\bibfield  {title} {\bibinfo {title} {Bypassing the noisy parity barrier:
  Learning higher-order {Markov} random fields from dynamics},\ }in\ \href
  {https://doi.org/10.1145/3717823.3718231} {\emph {\bibinfo {booktitle}
  {Proceedings of the 57th Annual ACM Symposium on Theory of Computing}}}\
  (\bibinfo {year} {2025})\ pp.\ \bibinfo {pages} {348--359}\BibitemShut
  {NoStop}%
\bibitem [{\citenamefont {Liu}\ \emph {et~al.}(2024)\citenamefont {Liu},
  \citenamefont {Mohanty}, \citenamefont {Raghavendra}, \citenamefont
  {Rajaraman},\ and\ \citenamefont
  {Wu}}]{liuLocallyStationaryDistributions2024}%
  \BibitemOpen
  \bibfield  {author} {\bibinfo {author} {\bibfnamefont {K.}~\bibnamefont
  {Liu}}, \bibinfo {author} {\bibfnamefont {S.}~\bibnamefont {Mohanty}},
  \bibinfo {author} {\bibfnamefont {P.}~\bibnamefont {Raghavendra}}, \bibinfo
  {author} {\bibfnamefont {A.}~\bibnamefont {Rajaraman}},\ and\ \bibinfo
  {author} {\bibfnamefont {D.~X.}\ \bibnamefont {Wu}},\ }\bibfield  {title}
  {\bibinfo {title} {Locally stationary distributions: A framework for
  analyzing slow-mixing {Markov} chains},\ }in\ \href
  {https://doi.org/10.1109/FOCS61266.2024.00022} {\emph {\bibinfo {booktitle}
  {2024 IEEE 65th Annual Symposium on Foundations of Computer Science
  (FOCS)}}}\ (\bibinfo {year} {2024})\ pp.\ \bibinfo {pages}
  {203--215}\BibitemShut {NoStop}%
\bibitem [{\citenamefont {Jayakumar}\ \emph {et~al.}(2026)\citenamefont
  {Jayakumar}, \citenamefont {Lokhov}, \citenamefont {Misra},\ and\
  \citenamefont {Vuffray}}]{jayakumarDiscreteDistributionsAre2026}%
  \BibitemOpen
  \bibfield  {author} {\bibinfo {author} {\bibfnamefont {A.}~\bibnamefont
  {Jayakumar}}, \bibinfo {author} {\bibfnamefont {A.~Y.}\ \bibnamefont
  {Lokhov}}, \bibinfo {author} {\bibfnamefont {S.}~\bibnamefont {Misra}},\ and\
  \bibinfo {author} {\bibfnamefont {M.}~\bibnamefont {Vuffray}},\ }\bibfield
  {title} {\bibinfo {title} {Discrete distributions are learnable from
  metastable samples},\ }\href {https://doi.org/10.1038/s41467-026-75439-1}
  {\bibfield  {journal} {\bibinfo  {journal} {Nat. Commun.}\ }\textbf {\bibinfo
  {volume} {17}},\ \bibinfo {pages} {8957} (\bibinfo {year}
  {2026})}\BibitemShut {NoStop}%
\bibitem [{\citenamefont {Stilck~Fran{\c c}a}\ \emph
  {et~al.}(2024)\citenamefont {Stilck~Fran{\c c}a}, \citenamefont {Markovich},
  \citenamefont {Dobrovitski}, \citenamefont {Werner},\ and\ \citenamefont
  {Borregaard}}]{stilckfrancaEfficientRobustEstimation2024}%
  \BibitemOpen
  \bibfield  {author} {\bibinfo {author} {\bibfnamefont {D.}~\bibnamefont
  {Stilck~Fran{\c c}a}}, \bibinfo {author} {\bibfnamefont {L.~A.}\ \bibnamefont
  {Markovich}}, \bibinfo {author} {\bibfnamefont {V.~V.}\ \bibnamefont
  {Dobrovitski}}, \bibinfo {author} {\bibfnamefont {A.~H.}\ \bibnamefont
  {Werner}},\ and\ \bibinfo {author} {\bibfnamefont {J.}~\bibnamefont
  {Borregaard}},\ }\bibfield  {title} {\bibinfo {title} {Efficient and robust
  estimation of many-qubit {Hamiltonians}},\ }\href
  {https://doi.org/10.1038/s41467-023-44012-5} {\bibfield  {journal} {\bibinfo
  {journal} {Nat. Commun.}\ }\textbf {\bibinfo {volume} {15}},\ \bibinfo
  {pages} {311} (\bibinfo {year} {2024})}\BibitemShut {NoStop}%
\bibitem [{\citenamefont {Yu}\ \emph {et~al.}(2023)\citenamefont {Yu},
  \citenamefont {Sun}, \citenamefont {Han},\ and\ \citenamefont
  {Yuan}}]{yuRobustEfficientHamiltonian2023a}%
  \BibitemOpen
  \bibfield  {author} {\bibinfo {author} {\bibfnamefont {W.}~\bibnamefont
  {Yu}}, \bibinfo {author} {\bibfnamefont {J.}~\bibnamefont {Sun}}, \bibinfo
  {author} {\bibfnamefont {Z.}~\bibnamefont {Han}},\ and\ \bibinfo {author}
  {\bibfnamefont {X.}~\bibnamefont {Yuan}},\ }\bibfield  {title} {\bibinfo
  {title} {Robust and efficient {Hamiltonian} learning},\ }\href
  {https://doi.org/10.22331/q-2023-06-29-1045} {\bibfield  {journal} {\bibinfo
  {journal} {Quantum}\ }\textbf {\bibinfo {volume} {7}},\ \bibinfo {pages}
  {1045} (\bibinfo {year} {2023})}\BibitemShut {NoStop}%
\bibitem [{\citenamefont {Huang}\ \emph {et~al.}(2023)\citenamefont {Huang},
  \citenamefont {Tong}, \citenamefont {Fang},\ and\ \citenamefont
  {Su}}]{huangLearningManyBodyHamiltonians2023}%
  \BibitemOpen
  \bibfield  {author} {\bibinfo {author} {\bibfnamefont {H.-Y.}\ \bibnamefont
  {Huang}}, \bibinfo {author} {\bibfnamefont {Y.}~\bibnamefont {Tong}},
  \bibinfo {author} {\bibfnamefont {D.}~\bibnamefont {Fang}},\ and\ \bibinfo
  {author} {\bibfnamefont {Y.}~\bibnamefont {Su}},\ }\bibfield  {title}
  {\bibinfo {title} {Learning many-body {Hamiltonians} with
  {Heisenberg}-limited scaling},\ }\href
  {https://doi.org/10.1103/PhysRevLett.130.200403} {\bibfield  {journal}
  {\bibinfo  {journal} {Phys. Rev. Lett.}\ }\textbf {\bibinfo {volume} {130}},\
  \bibinfo {pages} {200403} (\bibinfo {year} {2023})}\BibitemShut {NoStop}%
\bibitem [{\citenamefont {Dutkiewicz}\ \emph {et~al.}(2024)\citenamefont
  {Dutkiewicz}, \citenamefont {O'Brien},\ and\ \citenamefont
  {Schuster}}]{dutkiewiczAdvantageQuantumControl2024}%
  \BibitemOpen
  \bibfield  {author} {\bibinfo {author} {\bibfnamefont {A.}~\bibnamefont
  {Dutkiewicz}}, \bibinfo {author} {\bibfnamefont {T.~E.}\ \bibnamefont
  {O'Brien}},\ and\ \bibinfo {author} {\bibfnamefont {T.}~\bibnamefont
  {Schuster}},\ }\bibfield  {title} {\bibinfo {title} {The advantage of quantum
  control in many-body {Hamiltonian} learning},\ }\href
  {https://doi.org/10.22331/q-2024-11-26-1537} {\bibfield  {journal} {\bibinfo
  {journal} {Quantum}\ }\textbf {\bibinfo {volume} {8}},\ \bibinfo {pages}
  {1537} (\bibinfo {year} {2024})}\BibitemShut {NoStop}%
\bibitem [{\citenamefont {Li}\ \emph {et~al.}(2024)\citenamefont {Li},
  \citenamefont {Tong}, \citenamefont {Gefen}, \citenamefont {Ni},\ and\
  \citenamefont {Ying}}]{liHeisenberglimitedHamiltonianLearning2023}%
  \BibitemOpen
  \bibfield  {author} {\bibinfo {author} {\bibfnamefont {H.}~\bibnamefont
  {Li}}, \bibinfo {author} {\bibfnamefont {Y.}~\bibnamefont {Tong}}, \bibinfo
  {author} {\bibfnamefont {T.}~\bibnamefont {Gefen}}, \bibinfo {author}
  {\bibfnamefont {H.}~\bibnamefont {Ni}},\ and\ \bibinfo {author}
  {\bibfnamefont {L.}~\bibnamefont {Ying}},\ }\bibfield  {title} {\bibinfo
  {title} {{Heisenberg}-limited {Hamiltonian} learning for interacting
  bosons},\ }\href {https://doi.org/10.1038/s41534-024-00881-2} {\bibfield
  {journal} {\bibinfo  {journal} {npj Quantum Inf.}\ }\textbf {\bibinfo
  {volume} {10}},\ \bibinfo {pages} {83} (\bibinfo {year} {2024})}\BibitemShut
  {NoStop}%
\bibitem [{\citenamefont {M{\"o}bus}\ \emph {et~al.}(2025)\citenamefont
  {M{\"o}bus}, \citenamefont {Bluhm}, \citenamefont {Gefen}, \citenamefont
  {Tong}, \citenamefont {Werner},\ and\ \citenamefont
  {Rouz{\'e}}}]{mobusHeisenberglimitedHamiltonianLearning2025}%
  \BibitemOpen
  \bibfield  {author} {\bibinfo {author} {\bibfnamefont {T.}~\bibnamefont
  {M{\"o}bus}}, \bibinfo {author} {\bibfnamefont {A.}~\bibnamefont {Bluhm}},
  \bibinfo {author} {\bibfnamefont {T.}~\bibnamefont {Gefen}}, \bibinfo
  {author} {\bibfnamefont {Y.}~\bibnamefont {Tong}}, \bibinfo {author}
  {\bibfnamefont {A.~H.}\ \bibnamefont {Werner}},\ and\ \bibinfo {author}
  {\bibfnamefont {C.}~\bibnamefont {Rouz{\'e}}},\ }\bibfield  {title} {\bibinfo
  {title} {{Heisenberg}-limited {Hamiltonian} learning continuous variable
  systems via engineered dissipation},\ }\Eprint
  {https://arxiv.org/abs/2506.00606} {arXiv:2506.00606}  (\bibinfo {year}
  {2025})\BibitemShut {NoStop}%
\bibitem [{\citenamefont {Ni}\ \emph {et~al.}(2024)\citenamefont {Ni},
  \citenamefont {Li},\ and\ \citenamefont
  {Ying}}]{niQuantumHamiltonianLearning2024}%
  \BibitemOpen
  \bibfield  {author} {\bibinfo {author} {\bibfnamefont {H.}~\bibnamefont
  {Ni}}, \bibinfo {author} {\bibfnamefont {H.}~\bibnamefont {Li}},\ and\
  \bibinfo {author} {\bibfnamefont {L.}~\bibnamefont {Ying}},\ }\bibfield
  {title} {\bibinfo {title} {Quantum {Hamiltonian} learning for the
  {Fermi}-{Hubbard} model},\ }\href
  {https://doi.org/10.1007/s10440-024-00651-4} {\bibfield  {journal} {\bibinfo
  {journal} {Acta Appl. Math.}\ }\textbf {\bibinfo {volume} {191}},\ \bibinfo
  {pages} {2} (\bibinfo {year} {2024})}\BibitemShut {NoStop}%
\bibitem [{\citenamefont {Mirani}\ and\ \citenamefont
  {Hayden}(2024)}]{miraniLearningInteractingFermionic2024}%
  \BibitemOpen
  \bibfield  {author} {\bibinfo {author} {\bibfnamefont {A.}~\bibnamefont
  {Mirani}}\ and\ \bibinfo {author} {\bibfnamefont {P.}~\bibnamefont
  {Hayden}},\ }\bibfield  {title} {\bibinfo {title} {Learning interacting
  fermionic {Hamiltonians} at the {Heisenberg} limit},\ }\href
  {https://doi.org/10.1103/PhysRevA.110.062421} {\bibfield  {journal} {\bibinfo
   {journal} {Phys. Rev. A}\ }\textbf {\bibinfo {volume} {110}},\ \bibinfo
  {pages} {062421} (\bibinfo {year} {2024})}\BibitemShut {NoStop}%
\bibitem [{\citenamefont {Bakshi}\ \emph
  {et~al.}(2024{\natexlab{b}})\citenamefont {Bakshi}, \citenamefont {Liu},
  \citenamefont {Moitra},\ and\ \citenamefont
  {Tang}}]{bakshiStructureLearningHamiltonians2024}%
  \BibitemOpen
  \bibfield  {author} {\bibinfo {author} {\bibfnamefont {A.}~\bibnamefont
  {Bakshi}}, \bibinfo {author} {\bibfnamefont {A.}~\bibnamefont {Liu}},
  \bibinfo {author} {\bibfnamefont {A.}~\bibnamefont {Moitra}},\ and\ \bibinfo
  {author} {\bibfnamefont {E.}~\bibnamefont {Tang}},\ }\bibfield  {title}
  {\bibinfo {title} {Structure learning of {Hamiltonians} from real-time
  evolution},\ }in\ \href {https://doi.org/10.1109/FOCS61266.2024.00069} {\emph
  {\bibinfo {booktitle} {2024 IEEE 65th Annual Symposium on Foundations of
  Computer Science (FOCS)}}}\ (\bibinfo {year} {2024})\ pp.\ \bibinfo {pages}
  {1037--1050}\BibitemShut {NoStop}%
\bibitem [{\citenamefont {Zhao}(2025)}]{zhaoLearningStructureAny2025}%
  \BibitemOpen
  \bibfield  {author} {\bibinfo {author} {\bibfnamefont {A.}~\bibnamefont
  {Zhao}},\ }\bibfield  {title} {\bibinfo {title} {Learning the structure of
  any {Hamiltonian} from minimal assumptions},\ }in\ \href
  {https://doi.org/10.1145/3717823.3718115} {\emph {\bibinfo {booktitle}
  {Proceedings of the 57th Annual ACM Symposium on Theory of Computing}}}\
  (\bibinfo {year} {2025})\ pp.\ \bibinfo {pages} {1201--1211}\BibitemShut
  {NoStop}%
\bibitem [{\citenamefont {Ma}\ \emph {et~al.}(2026)\citenamefont {Ma},
  \citenamefont {Flammia}, \citenamefont {Preskill},\ and\ \citenamefont
  {Tong}}]{maLearningKbodyHamiltonians2026}%
  \BibitemOpen
  \bibfield  {author} {\bibinfo {author} {\bibfnamefont {M.}~\bibnamefont
  {Ma}}, \bibinfo {author} {\bibfnamefont {S.~T.}\ \bibnamefont {Flammia}},
  \bibinfo {author} {\bibfnamefont {J.}~\bibnamefont {Preskill}},\ and\
  \bibinfo {author} {\bibfnamefont {Y.}~\bibnamefont {Tong}},\ }\bibfield
  {title} {\bibinfo {title} {Learning {$k$}-body {Hamiltonians} via compressed
  sensing},\ }\href {https://doi.org/10.1109/TIT.2026.3720996} {\bibfield
  {journal} {\bibinfo  {journal} {IEEE Trans. Inform. Theory}\ }\textbf
  {\bibinfo {volume} {72}},\ \bibinfo {pages} {7553} (\bibinfo {year}
  {2026})}\BibitemShut {NoStop}%
\bibitem [{\citenamefont {Hu}\ \emph {et~al.}(2025)\citenamefont {Hu},
  \citenamefont {Ma}, \citenamefont {Gong}, \citenamefont {Ye}, \citenamefont
  {Tong}, \citenamefont {Flammia},\ and\ \citenamefont
  {Yelin}}]{huAnsatzfreeHamiltonianLearning2025}%
  \BibitemOpen
  \bibfield  {author} {\bibinfo {author} {\bibfnamefont {H.-Y.}\ \bibnamefont
  {Hu}}, \bibinfo {author} {\bibfnamefont {M.}~\bibnamefont {Ma}}, \bibinfo
  {author} {\bibfnamefont {W.}~\bibnamefont {Gong}}, \bibinfo {author}
  {\bibfnamefont {Q.}~\bibnamefont {Ye}}, \bibinfo {author} {\bibfnamefont
  {Y.}~\bibnamefont {Tong}}, \bibinfo {author} {\bibfnamefont {S.~T.}\
  \bibnamefont {Flammia}},\ and\ \bibinfo {author} {\bibfnamefont {S.~F.}\
  \bibnamefont {Yelin}},\ }\bibfield  {title} {\bibinfo {title} {Ansatz-free
  {Hamiltonian} learning with {Heisenberg}-limited scaling},\ }\href
  {https://doi.org/10.1103/j7b8-pb77} {\bibfield  {journal} {\bibinfo
  {journal} {PRX Quantum}\ }\textbf {\bibinfo {volume} {6}},\ \bibinfo {pages}
  {040315} (\bibinfo {year} {2025})}\BibitemShut {NoStop}%
\bibitem [{\citenamefont {Zhou}\ and\ \citenamefont
  {Gong}(2026)}]{zhouOptimalAnsatzfreeHamiltonian2026}%
  \BibitemOpen
  \bibfield  {author} {\bibinfo {author} {\bibfnamefont {T.}~\bibnamefont
  {Zhou}}\ and\ \bibinfo {author} {\bibfnamefont {W.}~\bibnamefont {Gong}},\
  }\bibfield  {title} {\bibinfo {title} {Optimal ansatz-free {Hamiltonian}
  learning in situ},\ }\Eprint {https://arxiv.org/abs/2606.19486}
  {arXiv:2606.19486}  (\bibinfo {year} {2026})\BibitemShut {NoStop}%
\bibitem [{\citenamefont {Chen}\ \emph
  {et~al.}(2025{\natexlab{d}})\citenamefont {Chen}, \citenamefont {Li},\ and\
  \citenamefont {Slote}}]{chenLowerBoundsLearning2026}%
  \BibitemOpen
  \bibfield  {author} {\bibinfo {author} {\bibfnamefont {Z.}~\bibnamefont
  {Chen}}, \bibinfo {author} {\bibfnamefont {J.}~\bibnamefont {Li}},\ and\
  \bibinfo {author} {\bibfnamefont {J.}~\bibnamefont {Slote}},\ }\bibfield
  {title} {\bibinfo {title} {Lower bounds for learning {Hamiltonians} from time
  evolution},\ }\Eprint {https://arxiv.org/abs/2509.20665} {arXiv:2509.20665}
  (\bibinfo {year} {2025}{\natexlab{d}})\BibitemShut {NoStop}%
\bibitem [{\citenamefont {{Cedillo Vayson de Pradenne}}\ \emph
  {et~al.}(2026)\citenamefont {{Cedillo Vayson de Pradenne}}, \citenamefont
  {Cotler},\ and\ \citenamefont
  {Huang}}]{pradenneLearningHamiltoniansLong2026}%
  \BibitemOpen
  \bibfield  {author} {\bibinfo {author} {\bibfnamefont {C.}~\bibnamefont
  {{Cedillo Vayson de Pradenne}}}, \bibinfo {author} {\bibfnamefont
  {J.}~\bibnamefont {Cotler}},\ and\ \bibinfo {author} {\bibfnamefont {H.-Y.}\
  \bibnamefont {Huang}},\ }\bibfield  {title} {\bibinfo {title} {Learning
  {Hamiltonians} at long times},\ }\Eprint {https://arxiv.org/abs/2606.05690}
  {arXiv:2606.05690}  (\bibinfo {year} {2026})\BibitemShut {NoStop}%
\bibitem [{\citenamefont {Lewis}\ \emph {et~al.}(2026)\citenamefont {Lewis},
  \citenamefont {Tang},\ and\ \citenamefont
  {Wright}}]{lewisLearningStructureOpen2026}%
  \BibitemOpen
  \bibfield  {author} {\bibinfo {author} {\bibfnamefont {L.}~\bibnamefont
  {Lewis}}, \bibinfo {author} {\bibfnamefont {E.}~\bibnamefont {Tang}},\ and\
  \bibinfo {author} {\bibfnamefont {J.}~\bibnamefont {Wright}},\ }\bibfield
  {title} {\bibinfo {title} {Learning the structure of open quantum systems},\
  }\Eprint {https://arxiv.org/abs/2606.30358} {arXiv:2606.30358}  (\bibinfo
  {year} {2026})\BibitemShut {NoStop}%
\bibitem [{\citenamefont {M{\"o}bus}\ \emph {et~al.}(2026)\citenamefont
  {M{\"o}bus}, \citenamefont {Bergamaschi}, \citenamefont {Stilck~Fran{\c
  c}a},\ and\ \citenamefont {Rouz{\'e}}}]{mobusRobustStructureLearning2026}%
  \BibitemOpen
  \bibfield  {author} {\bibinfo {author} {\bibfnamefont {T.}~\bibnamefont
  {M{\"o}bus}}, \bibinfo {author} {\bibfnamefont {T.}~\bibnamefont
  {Bergamaschi}}, \bibinfo {author} {\bibfnamefont {D.}~\bibnamefont
  {Stilck~Fran{\c c}a}},\ and\ \bibinfo {author} {\bibfnamefont
  {C.}~\bibnamefont {Rouz{\'e}}},\ }\bibfield  {title} {\bibinfo {title}
  {Robust structure learning of {$k$}-local {Lindbladians}},\ }\Eprint
  {https://arxiv.org/abs/2606.23652} {arXiv:2606.23652}  (\bibinfo {year}
  {2026})\BibitemShut {NoStop}%
\bibitem [{\citenamefont {Arad}\ \emph {et~al.}(2026)\citenamefont {Arad},
  \citenamefont {Chen}, \citenamefont {Guo}, \citenamefont {Rebentrost},\ and\
  \citenamefont {Yu}}]{aradNearOptimalLearningLocal2026}%
  \BibitemOpen
  \bibfield  {author} {\bibinfo {author} {\bibfnamefont {I.}~\bibnamefont
  {Arad}}, \bibinfo {author} {\bibfnamefont {Z.}~\bibnamefont {Chen}}, \bibinfo
  {author} {\bibfnamefont {N.}~\bibnamefont {Guo}}, \bibinfo {author}
  {\bibfnamefont {P.}~\bibnamefont {Rebentrost}},\ and\ \bibinfo {author}
  {\bibfnamefont {Z.}~\bibnamefont {Yu}},\ }\bibfield  {title} {\bibinfo
  {title} {Near-optimal learning of local {Lindbladians}},\ }\Eprint
  {https://arxiv.org/abs/2606.20535} {arXiv:2606.20535}  (\bibinfo {year}
  {2026})\BibitemShut {NoStop}%
\bibitem [{\citenamefont {Qi}\ and\ \citenamefont
  {Ranard}(2019)}]{qiDeterminingLocalHamiltonian2019a}%
  \BibitemOpen
  \bibfield  {author} {\bibinfo {author} {\bibfnamefont {X.-L.}\ \bibnamefont
  {Qi}}\ and\ \bibinfo {author} {\bibfnamefont {D.}~\bibnamefont {Ranard}},\
  }\bibfield  {title} {\bibinfo {title} {Determining a local {Hamiltonian} from
  a single eigenstate},\ }\href {https://doi.org/10.22331/q-2019-07-08-159}
  {\bibfield  {journal} {\bibinfo  {journal} {Quantum}\ }\textbf {\bibinfo
  {volume} {3}},\ \bibinfo {pages} {159} (\bibinfo {year} {2019})}\BibitemShut
  {NoStop}%
\bibitem [{\citenamefont {Bairey}\ \emph {et~al.}(2019)\citenamefont {Bairey},
  \citenamefont {Arad},\ and\ \citenamefont
  {Lindner}}]{baireyLearningLocalHamiltonian2019}%
  \BibitemOpen
  \bibfield  {author} {\bibinfo {author} {\bibfnamefont {E.}~\bibnamefont
  {Bairey}}, \bibinfo {author} {\bibfnamefont {I.}~\bibnamefont {Arad}},\ and\
  \bibinfo {author} {\bibfnamefont {N.~H.}\ \bibnamefont {Lindner}},\
  }\bibfield  {title} {\bibinfo {title} {Learning a local {Hamiltonian} from
  local measurements},\ }\href {https://doi.org/10.1103/PhysRevLett.122.020504}
  {\bibfield  {journal} {\bibinfo  {journal} {Phys. Rev. Lett.}\ }\textbf
  {\bibinfo {volume} {122}},\ \bibinfo {pages} {020504} (\bibinfo {year}
  {2019})}\BibitemShut {NoStop}%
\bibitem [{\citenamefont {Bairey}\ \emph {et~al.}(2020)\citenamefont {Bairey},
  \citenamefont {Guo}, \citenamefont {Poletti}, \citenamefont {Lindner},\ and\
  \citenamefont {Arad}}]{baireyLearningDynamicsOpen2020}%
  \BibitemOpen
  \bibfield  {author} {\bibinfo {author} {\bibfnamefont {E.}~\bibnamefont
  {Bairey}}, \bibinfo {author} {\bibfnamefont {C.}~\bibnamefont {Guo}},
  \bibinfo {author} {\bibfnamefont {D.}~\bibnamefont {Poletti}}, \bibinfo
  {author} {\bibfnamefont {N.~H.}\ \bibnamefont {Lindner}},\ and\ \bibinfo
  {author} {\bibfnamefont {I.}~\bibnamefont {Arad}},\ }\bibfield  {title}
  {\bibinfo {title} {Learning the dynamics of open quantum systems from their
  steady states},\ }\href {https://doi.org/10.1088/1367-2630/ab73cd} {\bibfield
   {journal} {\bibinfo  {journal} {New J. Phys.}\ }\textbf {\bibinfo {volume}
  {22}},\ \bibinfo {pages} {032001} (\bibinfo {year} {2020})}\BibitemShut
  {NoStop}%
\bibitem [{\citenamefont {Chen}\ and\ \citenamefont
  {Rouz{\'e}}(2025)}]{chen2025GibbsMarkov}%
  \BibitemOpen
  \bibfield  {author} {\bibinfo {author} {\bibfnamefont {C.-F.}\ \bibnamefont
  {Chen}}\ and\ \bibinfo {author} {\bibfnamefont {C.}~\bibnamefont
  {Rouz{\'e}}},\ }\bibfield  {title} {\bibinfo {title} {Quantum {Gibbs} states
  are locally {Markovian}},\ }\Eprint {https://arxiv.org/abs/2504.02208}
  {arXiv:2504.02208}  (\bibinfo {year} {2025})\BibitemShut {NoStop}%
\bibitem [{\citenamefont {Lieb}\ and\ \citenamefont
  {Robinson}(1972)}]{lieb1972finite}%
  \BibitemOpen
  \bibfield  {author} {\bibinfo {author} {\bibfnamefont {E.~H.}\ \bibnamefont
  {Lieb}}\ and\ \bibinfo {author} {\bibfnamefont {D.~W.}\ \bibnamefont
  {Robinson}},\ }\bibfield  {title} {\bibinfo {title} {The finite group
  velocity of quantum spin systems},\ }\href
  {https://doi.org/10.1007/BF01645779} {\bibfield  {journal} {\bibinfo
  {journal} {Commun. Math. Phys.}\ }\textbf {\bibinfo {volume} {28}},\ \bibinfo
  {pages} {251} (\bibinfo {year} {1972})}\BibitemShut {NoStop}%
\bibitem [{\citenamefont {Powers}\ and\ \citenamefont
  {St{\o}rmer}(1970)}]{powers1970free}%
  \BibitemOpen
  \bibfield  {author} {\bibinfo {author} {\bibfnamefont {R.~T.}\ \bibnamefont
  {Powers}}\ and\ \bibinfo {author} {\bibfnamefont {E.}~\bibnamefont
  {St{\o}rmer}},\ }\bibfield  {title} {\bibinfo {title} {Free states of the
  canonical anticommutation relations},\ }\href
  {https://doi.org/10.1007/BF01645492} {\bibfield  {journal} {\bibinfo
  {journal} {Commun. Math. Phys.}\ }\textbf {\bibinfo {volume} {16}},\ \bibinfo
  {pages} {1} (\bibinfo {year} {1970})}\BibitemShut {NoStop}%
\end{thebibliography}%
\end{document}